\documentclass[nobib]{tufte-book-tai}

\usepackage{amsmath}
\usepackage{amssymb}
\usepackage{cancel}
\usepackage{amsthm}
\usepackage{mathrsfs}
\usepackage[dvipsnames]{xcolor}
\usepackage{graphicx}
\usepackage{epigraph}
\usepackage{booktabs}
\usepackage{relsize}
\usepackage{hyperref}
\usepackage{multicol}
\usepackage{floatrow}
	\newfloatcommand{capbtabbox}{table}[][\FBwidth]
\usepackage{kbordermatrix}
\usepackage{transparent}
\usepackage[normalem]{ulem}
\hypersetup{colorlinks} 
\hypersetup{ %uncomment to override tufte colors
    colorlinks=true,    
    linkcolor =RubineRed,
    citecolor=RubineRed
}

\usepackage{array}

\usepackage[shortlabels]{enumitem}
\setlist[enumerate]{leftmargin=*, align=left}
\setlist[itemize]{topsep=1ex}

\usepackage{tikz}
\usetikzlibrary{cd,calc, arrows, shapes, matrix, positioning, intersections, decorations.markings, arrows.meta, decorations.pathmorphing}
  \tikzcdset{arrow style=tikz, diagrams={>=stealth}}
	\tikzstyle{vertex}=[fill=black,circle,scale=0.6]
	\tikzset{thick/.style={line width=.6mm}}
  \tikzstyle{tensor}=[circle,thick,draw=black,fill=blue!60!green!40!white,minimum size=4mm]
  \tikzstyle{witharrow} = [thick,decoration={
      markings,mark=at position 0.75 with {\arrow[scale=1.5,>=stealth]{>}}},postaction={decorate}]
  \tikzstyle{littletensor}=[circle,thick,draw=black,fill=red!30,minimum size=.5mm] 

\newcommand{\ar}[1]{\xrightarrow{\ensuremath{#1}}}

\newcommand{\V}{\mathcal{V}}
\renewcommand{\C}{\mathsf{C}}
\newcommand{\D}{\mathsf{D}}
\newcommand{\E}{\mathsf{E}}
\renewcommand{\o}{{\color{YellowOrange}o}}
\newcommand{\g}{{\color{Green}g}}
\newcommand{\p}{{\color{Purple}p}}
\newcommand{\<}{\langle}
\renewcommand{\>}{\rangle}
\newenvironment{bsmallmatrix}
  {\left[\begin{smallmatrix}}
  {\end{smallmatrix}\right]}
\let\emph\relax % there's no \RedeclareTextFontCommand
\DeclareTextFontCommand{\emph}{\bfseries}

\newcommand{\End}{\operatorname{End}}
\newcommand{\op}{\operatorname{op}}
\newcommand{\tr}{\operatorname{tr}}
\DeclareMathOperator{\id}{id}
\DeclareMathOperator{\eval}{eval}

\newcommand{\adj}[4]{
\begin{tikzcd}[ampersand replacement=\&, column sep=4ex, cramped]
   #1 \colon #2	\ar[yshift=+.6ex]{r}
\& #3 \colon #4	\ar[yshift=-.4ex]{l}
\end{tikzcd}
}

\theoremstyle{plain}
\newtheorem{lemma}{Lemma}[chapter]
\newtheorem{proposition}{Proposition}[chapter]
\newtheorem{corollary}{Corollary}[chapter]
\newtheorem{takeaway}{Takeaway}%[section]

\theoremstyle{definition}
\newtheorem{definition}{Definition}[chapter]
\newtheorem{example}{Example}[chapter]

{\bfseries}{\itshape}
\newtheorem*{UP_vect}{Universal Property for Vector Spaces}{\bfseries}{\itshape}
\newtheorem*{UP_meet}{Universal Property for Free Meet Semilattices}{\bfseries}{\itshape}
\newtheorem*{UP_join}{Universal Property for Free Join Semilattices}{\bfseries}{\itshape}
\newtheorem*{UP_cocomplete}{Universal Property for the Free Cocompletion of a Category}{\bfseries}{\itshape}
\newtheorem*{UP_complete}{Universal Property for the Free Completion of a Category}{\bfseries}{\itshape}
\newtheorem*{mainprocedure}{Main Procedure}{\bfseries}{\itshape}

\title{At the Interface \\ \noindent of  Algebra and Statistics}
\author{Tai-Danae Bradley}
\usepackage{epigraph}

\usepackage{makeidx}
\makeindex

\begin{document}
%\raggedbottom
\maketitlepage

% \let\cleardoublepage\clearpage %this gets rid of blank pages. comment out for actual printing

% ----------------------------COPYRIGHT PAGE--------------------------------- %
\newpage
\begin{fullwidth}
~\vfill
\thispagestyle{empty}
\setlength{\parindent}{0pt}
\setlength{\parskip}{\baselineskip}
Copyright \copyright\ \the\year %\ \thanklessauthor

\par \thanklessauthor

\par\smallcaps{All Rights Reserved}

\par A dissertation submitted to the Graduate Faculty in Mathematics in partial fulfillment of the requirements for the degree of Doctor of Philosophy, The City University of New York

\end{fullwidth}
\clearpage

% ----------------------------------TOC------------------------------------- %
\setcounter{page}{3}
\setcounter{tocdepth}{2}
\tableofcontents

\frontmatter
% ------------------------------FRONT MATTER--------------------------------- %
\setcounter{page}{5}

\chapter*{A Note to the Reader}
\addcontentsline{toc}{chapter}{A Note to the Reader}

%\newthought{Here are a few} things to know before reading this dissertation.
\newthought{Some things in life} are too good not to share, and mathematics is one of those things. For this reason, \textit{this dissertation was written with a wide audience in mind}. There is a great deal of exposition woven into the mathematics that provides intuition and motivation for the ideas. I've also sprinkled several ``behind the scenes'' snippets throughout, and alongside the propositions, lemmas, and corollaries there are  \textit{Takeaways} that summarize key ideas.\sidenote{And almost every page is laced with side notes with additional information.} Several of these key ideas are introduced through simple examples that are placed before---not after---the theory they're meant to illustrate. And in a happy turn of events, there is a low entrance fee for following the mathematics. The main tools are linear algebra and basic probability theory, and I suspect you already have these in your toolbox. This brings us to the next point.

This thesis concerns the equation: \textbf{Algebra $+$ Statistics $=$ ?} Said more carefully, this thesis stems from a desire to understand mathematical structure that has both algebraic and statistical properties. Here, ``algebra'' refers to the basic sense in which things come together to create something new. In an algebra, vectors can be multiplied together to create a new vector. Another word for this idea is \textit{compositionality}, where small things assemble together to build a larger construction, and where knowledge of this larger construction comes through understanding the individual parts together with the rules for combining them. \textit{But what if those rules are statistical?} What if the rule for multiplying vectors in an algebra is mediated by probability? Now we are at the interface of algebra and statistics, and one wonders what kind of mathematical structure is found there. To investigate it, we look for an example. Amazingly, we needn't look far. \textit{It's right in front of us, carefully knitting together each word on this page}. That is, natural language---English, Greek, Tagalog,\ldots---exhibits mathematical structure that is both algebraic and statistical. It's algebraic in that words come together to form larger expressions. The words \textit{orange} and \textit{fruit} can be concatenated to form \textit{orange fruit}.\marginnote[-1cm]{\begin{center}\includegraphics[width=\linewidth]{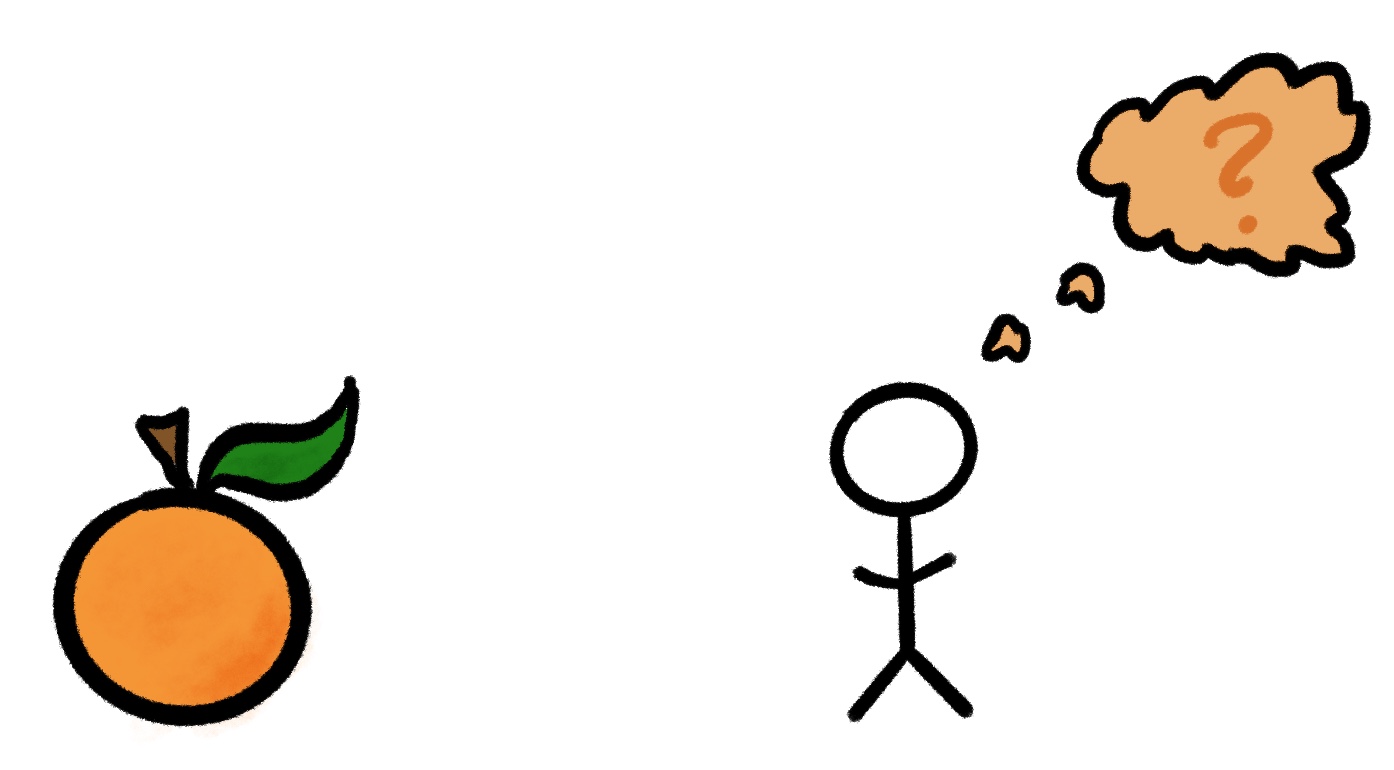}\end{center}} But language is also statistical, as some expressions occur more often than others: \textit{orange fruit} occurs more frequently than \textit{orange idea}, and this contributes to the meanings of these expressions. The probability of reading or saying ``orange fruit'' is higher than the probability of ``orange idea.'' But what \textit{is} this mathematical structure? We look for a preliminary set of mathematical tools to start exploring it. To identify this toolbox, we search for a second example of where compositionality and statistics meet. Again we needn't look far. The world of quantum many body physics involves precisely these ideas. Small systems compose to form larger composite systems, and various properties of these quantum systems are driven by statistics. So that is where we begin. This thesis uses basic tools of quantum physics to understand compositional and statistical mathematical structure.

% These two ideas---compositionality and statistics---and the mathematics behind them, are showcased right here before us, in the words on this page.  That is,  \textit{natural language}---as in, English, Spanish, Tagalog, Greek,\ldots---is one example where compositional and statistical structure emerges.  \tai{FINISH---Two examples of compositional/distributional structure are quantum physics and language. so let's use the tools of the former to explore the latter.} \tai{put this in ch 1?}
%[meaning is by Yoneda---environment plus statistics. Note we get this with entailment, that's why we map into suffix subsystem.] mention how words go to vectors, and phrases go to the tp of the words. .... say language is compositional and statistical.. just like my talks! (so maybe put this early on? or early and again here.) See Dissertation Fellowship.

\newthought{Here's a brief word} about the content itself. This work is ready-made to be digested in bits and pieces. Chapter \ref{ch:introduction} is the introduction and sets the layout of the land. You can read it, stop there, and have a good idea of what this work is about. But I hope you'll keep reading! Chapter \ref{ch:preliminaries} contains a series of tools and results from linear algebra called \textit{quantum probability theory}, but I've included a number of expositional treats that may not be found in standard introductions to the subject. Importantly, I do not assume the reader has a background in physics; the examples and exposition are all anchored in the world of mathematics. Chapter \ref{ch:probability} uses these tools to construct a mathematical machine, and with a clear understanding of the theory, we'll be able to ``look under the hood'' to see why it works. Chapter \ref{ch:application} describes an application with a concrete experiment that showcases the theory in action. The application falls under the genre of \textit{unsupervised machine learning,} though the reader is not assumed to have familiarity with these words. The emphasis is placed on the underlying mathematics, and specialized language is introduced only when it helps to put the math in context. 

Chapter \ref{ch:categorytheory} is largely independent from the rest of the text. The first four chapters illustrate the utility of a certain construction, which, as it turns out, arises in different areas of mathematics in different guises. The best way to see this is by lifting our feet off of the ground to get a clear, holistic view of the landscape. To do so, we'll use the language of category theory---a branch of mathematics whose \textit{modus operandi} is identifying these kinds of sweeping connections. Chapter \ref{ch:categorytheory} assumes familiarity with the basics of category theory; I do not provide an introduction. Fortunately, no reader will be left behind---I've placed the tantalizing connection in simple terms in Chapter \ref{ch:introduction}.
		\marginnote[-1.5cm]{
		\newthought{Colophon}. This document was typeset using \LaTeX with minor modifications to \texttt{tufte-latex}, based on Edward Tufte's \textit{Beautiful Evidence}. Images were hand drawn in Procreate using an Apple Pencil (2nd generation) and iPad Pro.
		}
\medskip

\noindent Let's begin.
 
\chapter*{Abstract}
\addcontentsline{toc}{chapter}{Abstract}

\newthought{This thesis takes} inspiration from quantum physics to investigate mathematical structure that lies at the interface of algebra and statistics.  The starting point is a passage from classical probability theory to quantum probability theory. The quantum version of a probability distribution is a density operator,  the quantum version of marginalizing is an operation called the partial trace, and the quantum version of a marginal probability distribution is a reduced density operator. Every joint probability distribution on a finite set can be modeled as a rank one density operator. By applying the partial trace, we obtain reduced density operators whose diagonals recover classical marginal probabilities. In general, these reduced densities will have rank higher than one, and their eigenvalues and eigenvectors will contain extra information that encodes subsystem interactions governed by statistics. We decode this information---and show it is akin to conditional probability---and then investigate the extent to which the eigenvectors capture ``concepts'' inherent in the original joint distribution. The theory is then illustrated with an experiment. In particular, we show how to reconstruct a joint probability distribution on a set of data by gluing together the spectral information of reduced densities operating on small subsystems. The algorithm naturally leads to a tensor network model, which we test on the even-parity dataset. Turning to a more theoretical application, we also discuss a preliminary framework for modeling entailment and concept hierarchy in natural language---namely, by representing expressions in the language as densities. Finally, initial inspiration for this thesis comes from formal concept analysis, which finds many striking parallels with the linear algebra. The parallels are not coincidental, and a common blueprint is found in category theory. We close with an exposition on free (co)completions and how the free-forgetful adjunctions in which they arise strongly suggest that in certain categorical contexts, the ``fixed points'' of a morphism with its adjoint encode interesting information.

\chapter*{Acknowledgments}
\addcontentsline{toc}{chapter}{Acknowledgments}

\newthought{My time in graduate} school can be characterized as a series of uniquely wonderful experiences, and this is largely due to the kindness, vision, and insight of my thesis advisor, John Terilla. I've deeply enjoyed working on and exploring the mathematics with John, and through his supervision, each year has been filled with more joy than the previous. He created \textit{the} most perfect, vibrant environment in which I could thrive and has been a constant source of wholehearted support. I am immensely grateful.

I have also benefited from many helpful conversations with a number of wonderful people: John Baez, Tyler Bryson, Brendan Fong, Joey Hirsh, Steve Lack, Martha Lewis, Fosco Loregian, Jade Master, Jacob Miller, Arthur Parzygnat, Raymond Puzio, Nissim Ranade, Emily Riehl, Jackie Shadlen, David Spivak, Jim Stasheff, James Stokes, Miles Stoudenmire, Jay Sulezburg, Dennis Sullivan, Brad Theilman, Thomas Tradler, Yiannis Vlassopoulos, Simon Willerton, Scott Wilson, Liang Ze Wong, Mahmoud Zeinalian, and many others. \marginnote{Many thanks to the readers of Math3ma and to everyone who has reached out to me through the blog.  I've loved sharing mathematics with you, and your enthusiasm to learn with me continues to fuel this endeavor.} I especially thank Scott and Thomas for serving on my defense committee and for their helpful feedback. My thanks go to Brooke Feigon for her supervision throughout my undergraduate research, and for her remarkable foresight in suggesting I reach out to John Terilla as my thesis advisor. I also thank Clark Barwick, who first shared with me the elegance of \texttt{tufte-latex}. Warmest thanks extend to Christine Chang, Yin Choi Cheng, Rynae Rasley, and Jacob Russell for their gift of friendship.
	\marginnote[2cm]{I gratefully acknowledge Tunnel for a vibrant working environment, the Tunnel Technologies-CUNY GC Dissertation Fellowship for support during my fifth year, and the Mina Rees Dissertation Fellowship in the Sciences for support during my sixth year.}

Finally, I am deeply grateful for Cheikh Mboup, without whom I would not have known the riches of mathematics. His patience, humility, and vision first ignited my enthusiasm for the subject ten years ago, and continual exposure to his masterful instruction over the years laid a solid foundation for the way I think and write about math today. And my utmost appreciation extends to my mom and dad for their love, hugs, and never-ending support, and for teaching me the  value of perseverance and endurance. I treasure their guidance as they lead by example with the highest standard of excellence.

%Appreciation also extends to Jack Hidary and X for their interest in my work, their support for Math3ma, and their invitation to join the team.

% ---------------------------------CHAPTERS---------------------------------- %
\mainmatter
\setcounter{secnumdepth}{2}

%chapter 1
\chapter{Introduction}\label{ch:introduction}
 
\epigraph{The beauty of mathematics only shows itself to more patient followers.}{Maryam Mirzakhani \cite{mirzakhani2008} }

\newthought{This thesis features} a passage from classical probability theory to quantum probability theory by modeling probability distributions with particular linear operators. We'll investigate the extent to which the eigenvectors embody concepts inherent in the probability distribution and use some experimental results as our guide. The linear algebra is presented in Chapters \ref{ch:preliminaries} and \ref{ch:probability}, and the experimental considerations are given in Chapter \ref{ch:application}. Initial inspiration is drawn from formal concept analysis, a mathematical framework that uses lattice theory to identify concepts and concept hierarchies in data \cite{davey2002introduction}. We give a brief introduction below. Think of this chapter as the prelude to the main piece, which begins in Chapter \ref{ch:preliminaries}.

\section{Formal Concepts}\label{sec:ch1_FCA}
A \emph{formal concept}\index{formal concept} is the name given to a complete bipartite subgraph of a particular bipartite graph. Intuitively, one views the two sets of vertices as \textit{objects} and \textit{attributes}, and an edge indicates whether an object possesses a certain attribute. A formal concept is thus a subset $A$ of objects and a subset $B$ of attributes, with the property that every object in $A$ possesses all the attributes in $B,$ and that each attribute in $B$ is possessed by all the objects in $A$. Typically, however, formal concepts are defined in terms of finite sets and relations rather than bipartite graphs. This set-theoretic definition mirrors a construction that will reappear in the work to come, so let's describe it more carefully.
	\begin{marginfigure}
	\vspace{-4.5cm}
	\includegraphics[width=\linewidth]{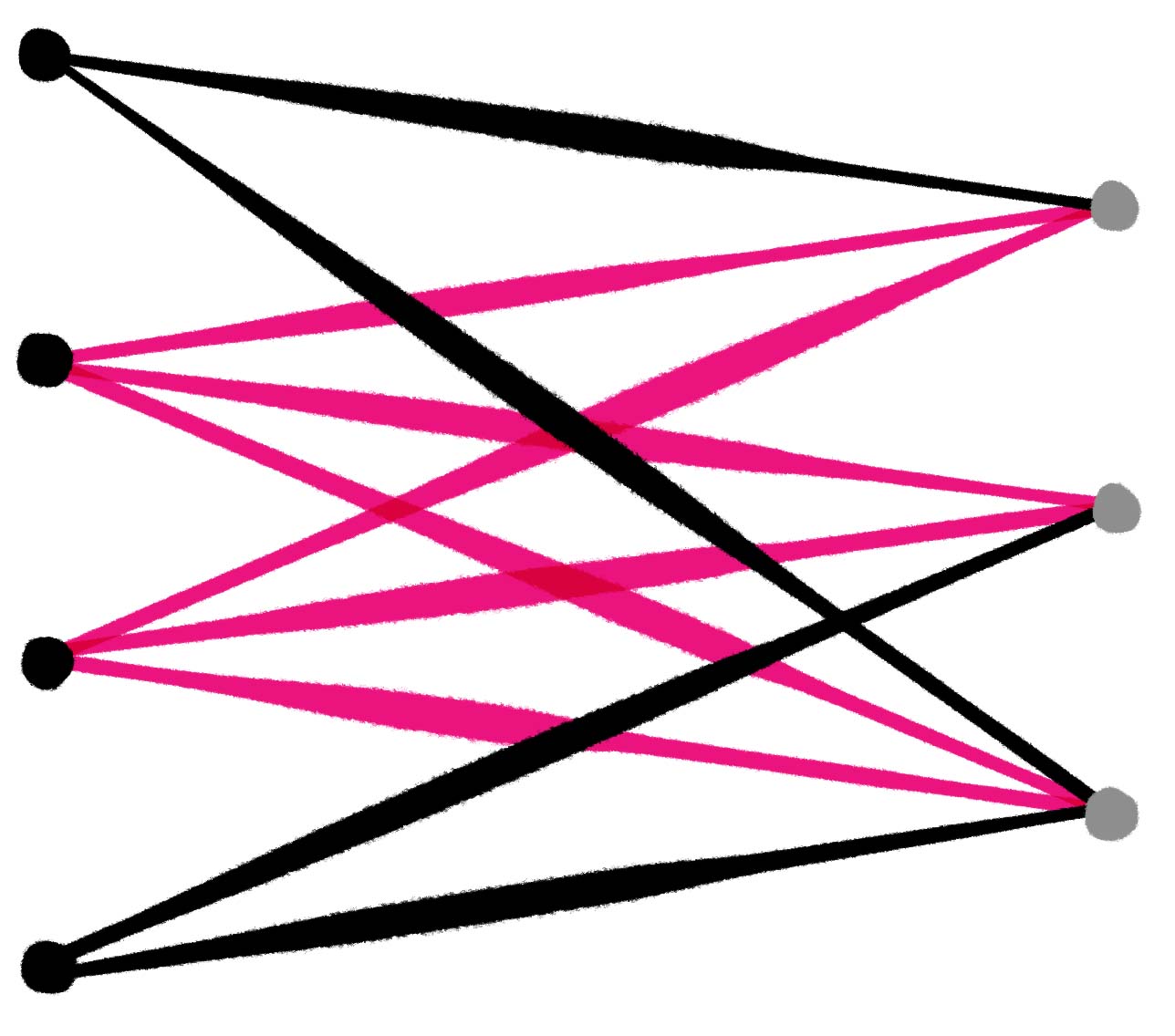}
	\caption{The subgraph in \textcolor{RubineRed}{red} is complete bipartite. This is one way to visualize a formal concept.}
	\label{fig:FCA}
	\end{marginfigure}

Let $X$ and $Y$ be finite sets and consider a function $R\colon X\times Y\to \{0,1\}$, which can be identified with a relation on, or equivalently a subset of, $X\times Y$. Intuitively, if $X$ is a set of objects and $Y$ is a set of attributes, then $R(x,y)=1$ if and only if object $x$ possesses attribute $y$. This relation determines a function $a$ from $X$ to the power set $2^Y$ of $Y$ that associates to each element $x$ the set of all elements in $Y$ that relate to $x$. Likewise, there is a function $b\colon Y\to 2^X$ that associates to each element $y$ the set of all elements in $X$ that relate to it.
\begin{equation}\label{eq:ab}
a(x):=\{y\in Y\mid R(x,y)=1\}\qquad
b(y):=\{x\in X\mid R(x,y)=1\}
\end{equation}
A similar assignment can be given to a subset $A\in 2^X$ by forming the set---call it $f(A)$---of all elements in $Y$ that relate to \textit{every} element in $A$. Think of this assignment as an \textit{extension} of $a$. In other words, there is a natural inclusion $X\to 2^X$ that maps each $x$ to the singleton set $\{x\}$, and the function $a$ extends uniquely along this inclusion to give a function $f\colon 2^X\to 2^Y$ defined on each $A\in 2^X$ by 
\begin{equation}\label{eq:f}
	f(A):=\bigcap_{x\in A}a(x).
\end{equation}
Now is a good time to point out that there is additional structure lying around: Power sets are more than mere sets. They are \emph{posets} (partially ordered sets) with subset containment providing the partial order. What's more, the function $f$ preserves this structure in a way that is \textit{order-reversing}: if $A'\subseteq A$ then $f(A)\subseteq f(A')$. To see this, it helps to think intuitively again: if $y\in f(A)$, then all objects in $A$---and in particular, those in $A'$---possess attribute $y,$ and so $y\in f(A')$. It's possible, however, for the objects in $A'$ to have a common attribute that is \textit{not} shared by another object $x\in A\smallsetminus A'$. That is, there's no reason to expect an equality $f(A)=f(A')$ in general. So in the same way, the function $b$ extends uniquely along the inclusion $Y\to 2^Y$ to give an order-reversing function $g\colon 2^Y\to 2^X$ defined on each $B\in 2^Y$ by
\begin{equation}\label{eq:g}
	g(B):= \bigcap_{y\in B}b(y)
\end{equation}
which is the set containing all $x\in X$ that relate to all $y\in B$. Now observe that for any pair of subsets $A\in 2^X$ and $B\in 2^Y$,
\begin{equation}\label{galois}
	A\subseteq g(B) \quad \text{if and only if} \quad B \subseteq f(A).
\end{equation}
The left-hand side states that the objects in $A$ have all the attributes in $B$ in common, though perhaps more. The right-hand side states that the attributes common to all objects in $A$ are at least those in $B$. These are, of course, two ways of saying the same thing. Pairs $(A,B)$ for which \textit{equality} holds in both sides of (\ref{galois}) are given a special name.

\begin{definition}
Given finite sets $X$ and $Y$ and a function $R\colon X\times Y\to \{0,1\}$, a pair $(A,B)\in 2^X\times 2^Y$ is a \textbf{formal concept of $R$}\index{formal concept} if $f(A)=B$ and $g(B)=A$, where $f$ and $g$ are the functions defined above.
\end{definition}

\noindent This definition admits a notion of concept hierarchy, as the set of all formal concepts of $R$ is itself partially ordered. Define $(A',B')\leq (A,B)$ whenever $A'\subseteq A$ or equivalently whenever $B'\supseteq B.$ Here is an example.

\begin{example}\label{ex:FCA} Consider the sets
	\marginnote{\smallcaps{Behind the Scenes.} In looking forward to language-inspired applications in Section \ref{sec:entailment} and Chapter \ref{ch:application}, I've chosen the elements of $X$ and $Y$ be English words rather than generic $x_i$ and $y_j$.}
\[X = \{\text{orange, green, purple}\} \qquad Y = \{\text{fruit, vegetable}\}\]
and let $R\colon X\times Y\to \{0,1\}$ be the function defined by
\begin{align*}
R(\text{orange, fruit}) &= 1\\
R(\text{green, fruit}) &= 1\\
R(\text{purple, vegetable}) &= 1\\
\end{align*}
with $R(x,y)=0$ for all other $(x,y)\in X\times Y$. This function can be represented by a $2\times 3$ table with a $1$ in the $xy$th entry if $R(x,y)=1$ and with $0$s elsewhere.
\begin{table}[h!]
	\centering
	\begin{tabular}{r|ccc}
		& \text{orange} & \text{green} & \text{purple}\\[5pt]
	\hline\\[-8pt]
	\text{fruit} & $1$ & $1$ & $0$\\[5pt]
	\text{vegetable} & $0$ & $0$& $1$
	\end{tabular}\\[10pt]
	\label{table1}
	\caption{The relation $R\colon X\times Y\to \{0,1\}$}
\end{table}

\noindent In this example, the values of the function $a\colon X\to 2^Y$ defined in Equation (\ref{eq:ab}) are as follows.
\begin{align*}
	a(\text{orange})&=\{\text{fruit}\}\\
	a(\text{green})&=\{\text{fruit}\}\\
	a(\text{purple})&=\{\text{vegetable}\}
\end{align*}
Notice that $a(\text{orange})$ can be read off from the nonzero entry in the first column of Table \ref{table1}, while $a(\text{green})$ and $a(\text{purple})$ may be read off from the second and third columns. The values of the function $b\colon Y\to 2^X$ are similarly determined by the nonzero entries of the rows of the table.
\begin{align*}
	b(\text{fruit})&=\{\text{orange, green}\}\\
	b(\text{vegetable})&= \{\text{purple}\}
\end{align*}
Given $a$ and $b$, the functions $\adj{f}{2^X}{2^Y}{g}$ are below.
\begin{fullwidth}
\begin{multicols}{2}
	\begin{align*}
	f\colon 2^X&\longrightarrow 2^Y\\
	\{\text{orange}\}&\mapsto\{\text{fruit}\}\\
	\{\text{green}\}&\mapsto\{\text{fruit}\}\\
	\{\text{purple}\}&\mapsto\{\text{vegetable}\}\\
	\{\text{orange, green}\}&\mapsto\{\text{fruit}\}\\
	\{\text{orange, purple}\}&\mapsto\varnothing\\
	\{\text{green, purple}\}&\mapsto\varnothing\\
	\{\text{orange, green, purple}\}&\mapsto\varnothing
	\end{align*}
\columnbreak \\
	\begin{align*}
	g\colon 2^Y&\longrightarrow 2^X\\
	\{\text{fruit}\}&\mapsto\{\text{orange, green}\}\\
	\{\text{vegetable}\}&\mapsto\{\text{purple}\}\\
	\{\text{fruit, vegetable}\}&\mapsto\varnothing
	\end{align*}
\end{multicols}
\end{fullwidth}
A formal concept is a pair of subsets $(A,B)\in 2^X\times 2^Y$ so that $fA=B$ and $gB=A$. Looking at the assignments above, we find that there are two formal concepts: the pair \{orange, green\} and \{fruit\}, and the pair \{purple\} and \{vegetable\}.
\begin{center}
\begin{tikzpicture}
\node at (3,0) {\includegraphics[scale=0.15]{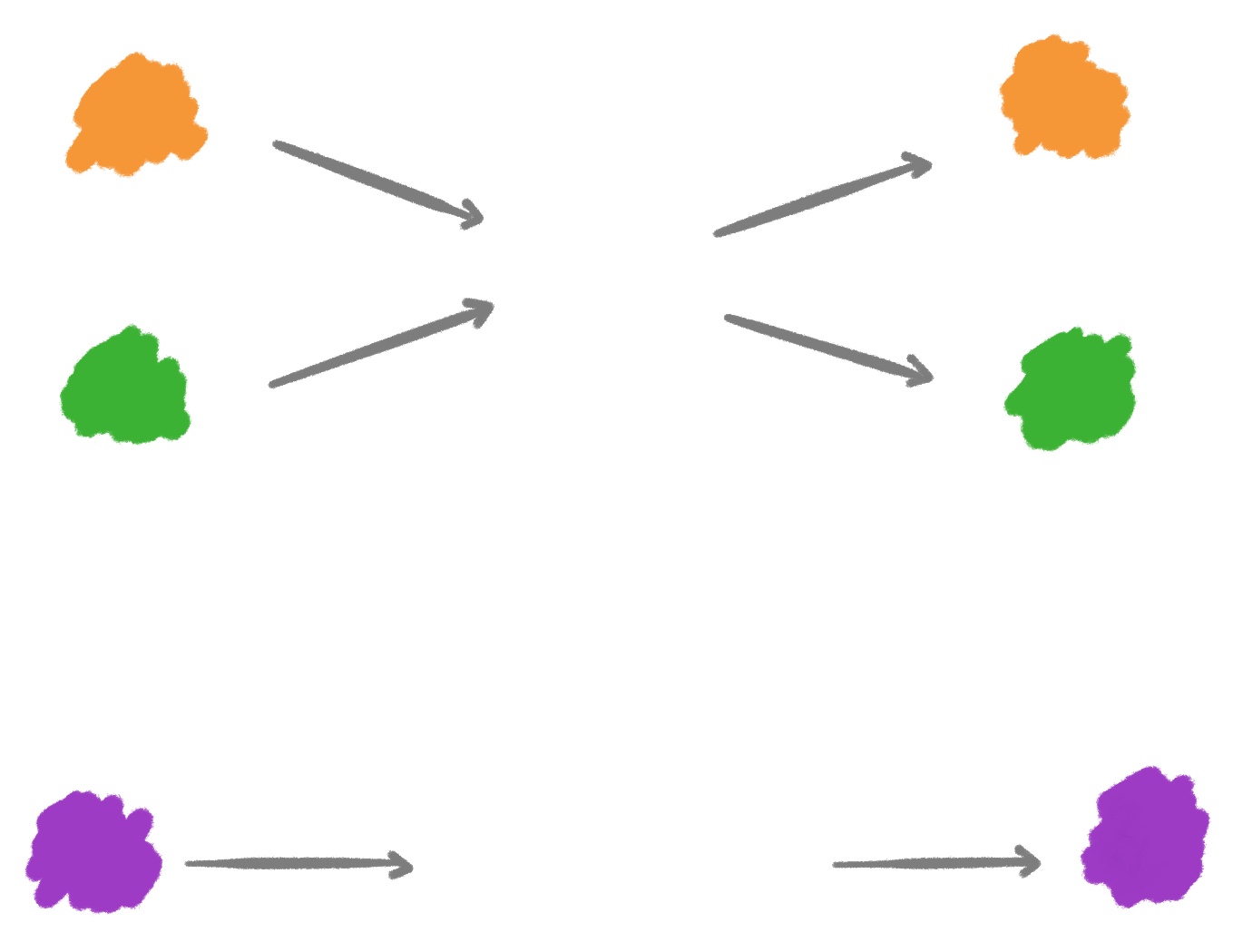}};
\node at (2,2) {$f$};
\node at (4,2) {$g$};
\node at (3,1.25) {fruit};
\node at (10,1.25) {(\{{\color{YellowOrange}orange}, {\color{Green}green}\},\{\text{fruit}\})};

\node at (1.2,-2) {$f$};
\node at (4.8,-2) {$g$};
\node at (10,-2.25) {(\{{\color{Purple}purple}\},\{\text{vegetable}\})};
\node at (3,-2.25) {vegetable};
\end{tikzpicture}
\end{center}
The first formal concept is, perhaps, the concept of citrus fruits. The second formal concept is that of a purple vegetable. \textit{There are plenty of these:} red cabbage, purple cauliflower, purple carrots, purple asparagus, endive, and eggplant---though that's technically a fruit!---to name a few.

\newthought{As noted in the} introduction, there is also a simple, graph theoretic way to understand formal concepts. In place of Table \ref{table1}, the relation $R$ may be represented as a bipartite graph. The sets $X$ and $Y$ provide the two sets of vertices, and there is an edge joining $x\in X$ and $y\in Y$ if and only if $R(x,y)=1$.
\begin{center}
\begin{tikzpicture}[y=1.2cm, x=1.2cm]
\node at (0,0) {\includegraphics[scale=0.1]{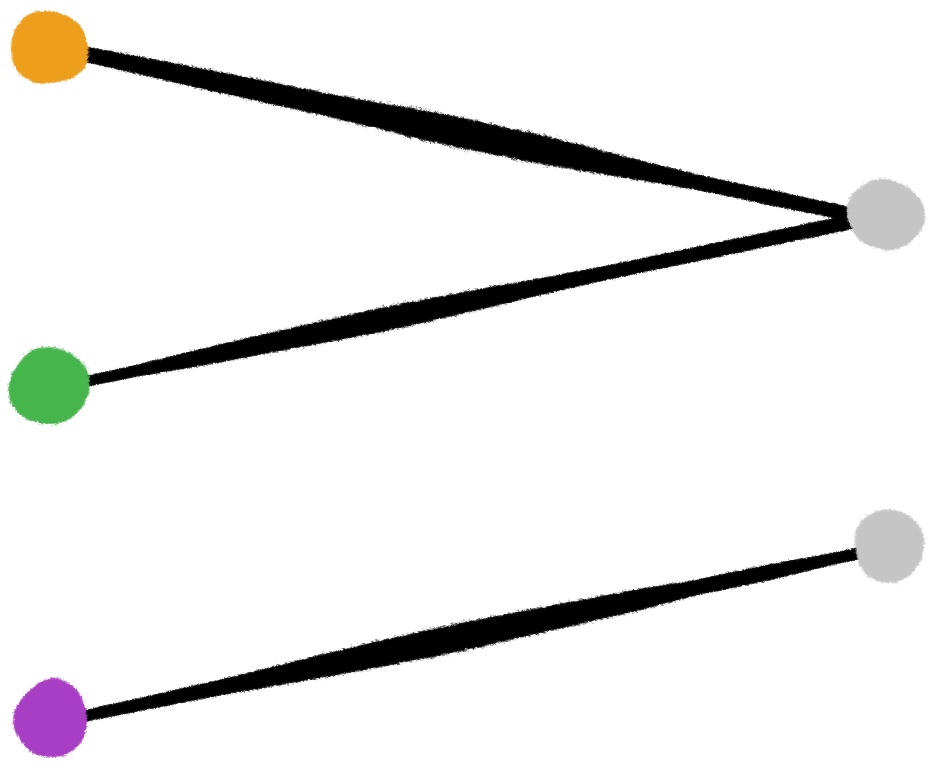}};
\node at (-2,1) {{\color{YellowOrange}orange}};
\node at (-2,0) {{\color{Green}green}};
\node at (-2,-1) {{\color{Purple}purple}};

\node at (2,.5) {fruit};
\node at (2,-.5) {vegetable};
\end{tikzpicture}
\end{center}
Formal concepts coincide with complete bipartite subgraphs. Indeed, this graph has two complete bipartite subgraphs---which happen to be disjoint in this example---and they correspond the two formal concepts listed above.
\end{example}

In general, formal concepts may also be identified with invariant subsets or ``fixed points'' of the composite functions $fg$ and $gf$. Let Fix$fg$ denote the set of all $B\in 2^Y$ satisfying $fg(B)=B$ and similarly for \text{Fix}$gf$. There are bijections
\[
	\text{formal concepts of $R$} 
	\quad \cong \quad \text{Fix}fg 
	\quad \cong \quad 
	\text{Fix}gf.
\]
If $(A,B)$ is a formal concept, then $B$ is a fixed point of $fg$, and $A$ is a fixed point of $gf$:
\[
	f(gB)=fA=B \qquad g(fA)=gB=A.
\]
Conversely, if $B\in \text{Fix}fg$, then $(gB,B)$ is a formal concept, and if $A\in\text{Fix}gf$, then $(A,fA)$ is also a formal concept. Let's summarize the discussion so far.

\begin{takeaway}\label{takeaway:FCA}
Given finite sets $X$ and $Y$, any function $R\colon X\times Y\to \{0,1\}$ induces two functions $a\colon X\to 2^Y$ and $b\colon Y\to 2^X$ that lift to order-reversing functions $f$ and $g$ so that the following diagrams commute. 
\[
	\begin{tikzcd}
	2^X \arrow[r, "f", dashed]         & 2^Y && 2^Y \arrow[r, "g", dashed]         & 2^X \\
	X \arrow[ru, "a"'] \arrow[u, hook] &     && Y \arrow[u, hook] \arrow[ru, "b"'] &    
	\end{tikzcd}
\]
Moreover, $f$ and $g$ satisfy
\[
	A\subseteq g(B) \quad \text{if and only if} \quad B \subseteq f(A),
\]
for all $A\in 2^X$ and $B\in 2^Y$, and the invariant subsets of the compositions $fg$ and $gf$ are formal concepts. Moreover, there is a one-to-one correspondence between them.
\end{takeaway}

The ideas here are an example of a construction that appears often throughout mathematics. The relationship between $f$ and $g$ as seen in (\ref{galois}) is closely related to the correspondence between groups and fields in Galois theory and the correspondence between fundamental groups and covering spaces in algebraic topology. They are all instances of a general construction called a \textit{Galois connection}, which is the name given to a pair of maps between posets that satisfy a property like that in (\ref{galois}). Galois connections themselves are examples of a more general construction known as an \emph{adjunction} in category theory.\sidenote{Given two categories $\C$ and $\D$, a pair of functors $F\colon \C\to\D$ and $G\colon \D\to\C$ form an \emph{adjunction} if \begin{equation}\label{eq:adjoint} \hom_\C(FA,B)\cong \hom_\D(A,GB)\end{equation} for all objects $A$ in $\C$ and $B$ in $\D$. In other words, $F$ and $G$ are \emph{adjoint functors} if the set of morphisms $FA\to B$ are in (natural) bijection with the set of morphisms $A\to GB$. Compare Equation (\ref{eq:adjoint}) with Equation (\ref{adjoint}) to see where ``adjoint'' functors get their name.} Indeed, the power set $2^X$ is, like any poset, a category whose objects are subsets of $X$ and whose morphisms are provided by the partial order. The order-reversing functions $f$ and $g$ are contravariant functors, and the statement in (\ref{galois}) witnesses an adjunction $\adj{f}{2^X}{2^Y}{g}$. As it turns out, pairs of adjoint functors on posets arise frequently in mathematics; the one we've discussed in this section is just a special case. Not coincidentally, the terms ``adjunction'' and ``adjoint'' arise from similarity with \textit{adjoint linear maps}. (See Equation (\ref{eq:adjoint}) in the margin.) In fact, when we exchange $\mathbb{C}$ for $\{0,1\}$ in our discussion above, the summary in Takeaway \ref{takeaway:FCA} has a familiar linear algebraic version that leads to several key ideas for this thesis.

\section{A Linear Algebraic Version}\label{sec:ch1_linalg}
Let $X$ and $Y$ be finite sets as before, and consider a complex-valued function $M\colon X\times Y\to \mathbb{C}$. Let $\mathbb{C}^X$ denote the free complex vector space on $X$. For the moment, it will be helpful to think of vectors $v\in\mathbb{C}^X$ as functions $v\colon X\to\mathbb{C}$, just as a subset $A\in 2^X$ coincides with a function $A\colon X\to 2$. The function $M$ gives rise to functions $\alpha\colon X\to\mathbb{C}^Y$ and $\beta\colon Y\to\mathbb{C}^X$. For each $x$ the function  $\alpha (x)\colon Y\to\mathbb{C}$ assigns $M(x,y)$ to each $y$. Likewise for each $y$ define the function $\beta (y)\colon X\to \mathbb{C}$ to assign $\overline{M(x,y)}$ to each $x$. Now observe that there is a natural inclusion $X\to \mathbb{C}^X$ that associates to each $x$ the function $v_x\colon X\to \mathbb{C}$ whose value $v_x(x')$ is $1$ if $x'=x$ and is $0$ otherwise. By the universal property of vector spaces, the function $\alpha$  uniquely extends along this inclusion to a linear map\sidenote{Said more informally, defining a linear map $\mathbb{C}^X\to\mathbb{C}^Y$ is as simple as specifying where the basis vectors go. That's what $\alpha$ is doing.} $M\colon \mathbb{C}^X\to\mathbb{C}^Y$. The same letter $M$ is used to denote this linear map since its matrix representation is precisely the function $M\colon X\times Y\to\mathbb{C}$ viewed as a $|Y|\times |X|$ matrix. In fact, $\alpha(x)$ is simply the $x$th column of $M$. In the same way, the function $\beta$ extends to a linear map $M^\dagger\colon\mathbb{C}^Y\to\mathbb{C}^X$, and $\beta(y)$ is the (complex conjugate of the) $y$th row of $M$. Dagger notation is used since the matrix representing $M^\dagger$ is the conjugate transpose of $M$. 
\[
	\kbordermatrix{
		& &   &  \\
	    &  & \mid &  \\[5pt]
	    &\phantom{\cdots} & \alpha(x) &\phantom{\cdots}   \\[5pt]
	    & & \mid & 
	    }
	   =M=
	\kbordermatrix{
		& &  &  \\
	    & &   &  \\[5pt]
	     & \text{---} & \beta(y) & \text{---}  \\[5pt]
	    & &  & 
	    }
\]
In other words, $M^\dagger\colon \mathbb{C}^Y\to\mathbb{C}^X$ is the linear adjoint of $M\colon \mathbb{C}^X\to\mathbb{C}^Y$ since for all $v\in\mathbb{C}^X$ and $w\in \mathbb{C}^Y$,
	\marginnote{Compare Equation (\ref{adjoint}) with Equation (\ref{eq:adjoint}) to see where ``adjoint functors'' get their name.}
\begin{equation}\label{adjoint}
\langle Mv,w\rangle = \langle v,M^\dagger w\rangle.
\end{equation}
``Fixed points'' of the composite maps $M^\dagger M$ and $MM^\dagger$ are  their eigenvectors, and we will show in Proposition \ref{prop:evects} that there is a one-to-one correspondence between them\sidenote{...as long as $M$ is ``properly normalized'' in a sense that the proposition makes precise.}. Here's the takeaway of this discussion. Notice the similarity with Takeaway \ref{takeaway:FCA}.

\begin{takeaway}\label{takeaway:LA}
Given finite sets $X$ and $Y$, any function $M\colon X\times Y\to \mathbb{C}$ induces two functions $\alpha\colon X\to \mathbb{C}^Y$ and $\beta\colon Y\to\mathbb{C}^X$ that lift to linear maps $M$ and $M^\dagger$ so that the following diagrams commute.
\[
	\begin{tikzcd}
	\mathbb{C}^X \arrow[r, "M", dashed]         & \mathbb{C}^Y && \mathbb{C}^Y \arrow[r, "M^\dagger", dashed]         & \mathbb{C}^X \\
	X \arrow[ru, "\alpha"'] \arrow[u, hook] &     && Y \arrow[u, hook] \arrow[ru, "\beta"'] &    
	\end{tikzcd}
\]
Moreover, $M$ and $M^\dagger$ satisfy 
\[
\langle Mv,w\rangle = \langle v,M^\dagger w\rangle
\]
for all $v\in\mathbb{C}^X$ and $w\in\mathbb{C}^Y$, and the one-dimensional invariant subspaces of the compositions $M^\dagger M$ and $MM^\dagger$ correspond to their eigenvectors. Moreover, there is a one-to-one correspondence between them.
\end{takeaway}

\newthought{Compare the linear} algebra in Takeaway \ref{takeaway:LA} and Equation (\ref{adjoint}) with the set theory in Takeaway \ref{takeaway:FCA} and Statement (\ref{galois}). The similarities are unmistakable: a relation $R$ is like a matrix $M$, the poset maps $f,g$ are like the linear maps $M,M^\dagger$, and formal concepts are like eigenvectors. So we ask the natural question, \textit{Do eigenvectors of $M^\dagger M$ and $MM^\dagger$ capture concepts inherent in the matrix $M$?} A main goal of this thesis is to show that the answer is ``Yes.'' We further wish to understand the general construction for which the diagrams in Takeaways \ref{takeaway:FCA} and \ref{takeaway:LA} are a special case. A first step in this direction is to notice that the diagrams are rather imprecise. Some of the arrows are functions, while others are functions that preserve structure.\sidenote[][-1cm]{For instance, $\alpha$ is a function while $M$ is a linear map. In fact, take a closer look at $\alpha\colon X\to\mathbb{C}^Y$. We've referred to $\alpha$ as a ``function,'' but the domain and codomain of a function are \textit{sets}. The domain of $\alpha$ is certainly a set, but its codomain is a set \textit{equipped with} a vector space structure. So $\alpha$ is not well-typed. It would be better for us to say that $\alpha$ is ``a function from $X$ to the \textit{underlying set} of the vector space $\mathbb{C}^Y$.'' Getting this right is more than pedantry. It plays an essential role in discovering why Takeaways \ref{takeaway:FCA} and \ref{takeaway:LA} are so similar. As we'll see in Chapter \ref{ch:categorytheory}, the answer is found in category theory.} It's not clear which \textit{category} the mathematics is taking place in. A more careful treatment will clarify the situation. So another goal of this thesis is to shed light on the category theory underlying the two constructions discussed here.

\section{Outline of Contents}
Here's a preview of the ideas to come. First, for the applications we have in mind, a main focus will be on matrices $M\colon X\times Y\to\mathbb{C}$ with the property that $\sum_{x,y}|M(x,y)|^2=1$; that is, the absolute square of the entries of $M$ define a probability distribution on $X\times Y$. This matrix represents a linear map $\mathbb{C}^X\to\mathbb{C}^Y$, and so it can be identified with a unit vector in the tensor product $\mathbb{C}^X\otimes\mathbb{C}^Y$. Orthogonal projection onto this unit vector defines a certain type of linear operator known as a \textit{density operator} or synonymously a \textit{quantum state}. This operator induces two \textit{reduced density operators}---one on $\mathbb{C}^X$ and one on $\mathbb{C}^Y$---whose matrix representations are precisely $M^\dagger M$ and $MM^\dagger$. If the ranks of these operators are greater than $1$, then their eigenvectors will be shown to harness valuable information within the original quantum state, and this information will turn out to be useful in an applied setting. We will introduce all of these ideas from the ground up. Here's the path we'll take to do so.

\newthought{Chapter Two}. We begin with preliminaries in Chapter \ref{ch:preliminaries}. Section \ref{sec:example} motivates the linear algebra in the previous paragraph by presenting an elementary-yet-illuminating example that does not use the language of quantum mechanics. Section \ref{sec:prelim} then introduces \textit{bra-ket} notation and \textit{tensor network diagram} notation, which will be used throughout. We will hit the ground running in Sections \ref{sec:density}, \ref{sec:partial_trace}, and \ref{sec:red_density}, which cover the basics of \textit{density operators}, the \textit{partial trace}, and \textit{reduced density operators}. The main theme in this chapter is that density operators are the quantum version of probability distributions, reduced density operators are the quantum version of marginal probability, and the partial trace is the quantum version of marginalizing.

\newthought{Chapter Three}. Armed with the prerequisites, we present the main contribution in Chapter \ref{ch:probability}. It is a procedure for modeling any joint probability distribution by a rank $1$ density operator together with a decryption of the information carried in the eigenvalues and eigenvectors of its reduced density operators. The procedure is outlined in Section \ref{sec:procedure}, and an extended illustration is given in Example \ref{ex:main}. The example identifies a particularly simple scenario, namely  that when the probability distribution being modeled is an \textit{empirical} one, the reduced densities and their spectral information have simple combinatorial interpretations that can be read off visually from a graph representing the distribution. Section \ref{sec:graphs} explains this visual, combinatorial idea. The mathematics then branches off into a couple of directions. Section \ref{sec:evectsFCA} revisits the connection with formal concept analysis with a more informed perspective, and Section \ref{sec:entailment} outlines a preliminary framework for modeling entailment and concept hierarchy using density operators. 

\newpage
\newthought{Chapter Four}. Chapter \ref{ch:application} is Chapter \ref{ch:probability} in live action. We will discover that the spectral information of reduced densities of a \textit{rank $1$ density} can be pieced together to reconstruct the original state. We describe an algorithm for doing this, and along the way address a common puzzle in modeling probability distributions:  Given a dataset of samples drawn from a possibly unknown probability distribution $\pi$, how might we use the samples to model a \textit{new} probability distribution that estimates $\pi$, with the purpose of generating new data from it? Such models are called \textit{generative models}. In Section \ref{sec:experiment} we'll show how the information stored in the eigenvectors and eigenvalues of reduced densities can be pieced together to build one of these models. The process is given by a deterministic algorithm that naturally leads to a \textit{tensor network} presentation of the model. Because the algorithm and tools come from simple linear algebra, it's possible to predict how well the model will work, just given the number of samples used to build it. We will describe a concrete experiment on a dataset of even-parity bitstrings in Section \ref{sec:theexperiment}. The application, algorithm, and experiment are work originally shared in \cite{bradley2019modeling}.

\newthought{Chapter Five}. The last chapter reviews the mathematics of Chapters \ref{ch:introduction}--\ref{ch:application} from a high-level perspective. After distilling the essential ideas, it is apparent that the \textit{fixed points of the composition of a map with its adjoint are interesting}. In the context of linear algebra, this refers to eigenvectors of reduced densities; in the context of sets and order theory, it refers to formal concepts. Both constructions are remarkably similar. In Chapter \ref{ch:categorytheory} we use the language of category theory to search for the reason for this similarity. As we'll see, both constructions arise from a free-forgetful adjunction. In Section \ref{sec:5_LA} we'll recall the free-forgetful adjunction between the category of sets and the category of vector spaces. Special emphasis will be on the unit of the adjunction, which fits into a universal property that recovers Takeaway \ref{takeaway:LA}. In Section \ref{sec:5_CT} we will remark on \textit{free (co)completions} and their universal properties. Along the way, we'll comment on a number of striking analogies between category theory and linear algebra and a clear dictionary between the two. Section \ref{sec:5_FCA} then specializes the category theory to an \textit{enrichment over truth values}, where a certain universal property will recover Takeaway \ref{takeaway:FCA}. Know that Chapter \ref{ch:categorytheory} simply collects a number of known results into a single exposition with the purpose of calling attention to the similarity with our contributions in Chapters \ref{ch:introduction}--\ref{ch:application}. Although we do not recover our linear algebraic constructions from the categorical blueprint, the pattern is crystal clear. We close in Section \ref{sec:5behind} with a remark on this.

\clearpage

%chapter 2
\chapter{Preliminaries}\label{ch:preliminaries}

\epigraph{Networks transport classical things like power, water, oil, and cars. Tensor networks transport linear algebraic things like rank and entanglement and should be thought of as the quantum analogy.}{Shawn X. Cui et. al. \cite{Cui2015}}

\newthought{In this chapter} we'll describe a passage from classical probability to quantum probability by modeling any probability distribution $\pi$ as a pure quantum state; that is, a rank $1$ density operator $\rho$. When the state is entangled, the reduced density operators of $\rho$ recover the classical marginal probability distributions of $\pi$ along their diagonals. Because the original state is entangled, the eigenvectors of the reduced densities  contain additional information from $\pi,$ which we see as conditional probability. The goal of this chapter is to present these ideas from the ground up.

To begin, we think of density operators as the quantum version of probability distributions, and of reduced density operators as the quantum versions of marginal probability distributions. The operation that plays the role of marginalizing is called the partial trace. So we give an overview of density operators in Section \ref{sec:density}, of the partial trace in Section \ref{sec:partial_trace}, and of reduced density operators in Section \ref{sec:red_density}. Know that these are basic constructions that may be found in many excellent texts on quantum information theory and quantum computing, such as \cite{kitaev2002classical, nielsen2000quantum, watrous2018theory, witten2018miniintroduction,preskill229notes}, though I've included a number of expositional treats that may not be found in standard introductions. Although the language of quantum mechanics is used, the mathematics may appeal to a wider mathematical audience with different applications in mind. So to help facilitate the transfer knowledge, Section \ref{sec:example} motivates this chapter with an elementary example that does not use the language of quantum mechanics. In the remaining sections, terminology is only introduced as needed. The opening example is followed by Section \ref{sec:prelim}, which reviews bra-ket notation along with tensor network diagrams.

\section{A Motivating Example}\label{sec:example} 
The best way to introduce the ideas in this chapter is through an example. The example in this section is remarkably elementary---understanding it requires no more than a few basic definitions from probability theory. Even so, it highlights interesting mathematics that is the keynote of this thesis.

\subsection{Some Elementary Probability}\label{ssec:intro_probability}
Before the example, here's some basic terminology. A \emph{probability distribution}\index{probability distribution} (or simply, \emph{distribution}) on a finite set $S$ is a function $\pi\colon S\to \mathbb{R}$ satisfying
\[
	\sum_{s\in S}\pi(s) = 1, \qquad \pi(s)\geq 0 \text{ for all $s\in S.$}
\]
If the elements of $S$ are ordered so that $S=\{s_1,\ldots,s_n\}$, then a probability distribution may be written as a tuple $(\pi(s_1),\ldots, \pi(s_n))$ of nonnegative real numbers whose sum is $1$. A \emph{joint probability distribution}\index{probability distribution!joint} will refer to a probability distribution on a Cartesian product of finite sets; that is, a function $\pi\colon X\times Y\to \mathbb{R}$ satisfying 
\[
	\sum_{(x,y)\in X\times Y}\pi(x,y)=1.
\]
Every joint distribution gives rise to \emph{marginal probability distributions}\index{probability distribution!marginal} by summing probabilities over one of the factors, sometimes referred to as ``integrating out.'' So the marginal distribution $\pi_X\colon X\to \mathbb{R}$ on $X$ and the marginal distribution $\pi_Y\colon Y\to \mathbb{R}$ on $Y$ are defined by
\[
	\pi_X(x)=\sum_{y\in Y}\pi(x,y)
	\qquad
	\pi_Y(y)=\sum_{x\in X}\pi(x,y).
\]
The process of marginalizing loses information. Informally, \textit{marginal probability doesn't have memory.} Let's illustrate with an example. Suppose $X$ and $Y$ contain words from the English language.
\[
X=\{\text{orange, green, purple}\}
	\qquad
Y=\{\text{fruit, vegetable}\}
\]
Consider the joint distribution on $X\times Y$ indicated
	\marginnote{The $2\times3$ table on this page may remind you of the nearly identical $2\times 3$ table used to illustrate formal concepts in Example \ref{ex:FCA}. This is not a coincidence. We will make the connection explicit in Section \ref{sec:evectsFCA}.}
in the $2\times3$ table below. For simplicity, suppose this distribution is an \textit{empirical} one, so that for a fixed pair $(x,y)\in X\times Y$ the probability $\pi(x,y)$ is just the number of $(x,y)$ divided by the total number of pairs.
	\marginnote[1cm]{
	\[
	\begin{array}{c}
	\text{\Large {\color{YellowOrange}fruit}} \\[5pt]
	\text{\Large {\color{Green}fruit}} \\[5pt]
	\text{\Large {\color{Purple}vegetable}} \\
	\end{array}
	\]
	}
\[
\begin{array}{r|ccc}
	& \text{orange} & \text{green} & \text{purple}\\[5pt]
\hline\\[-8pt]
\text{fruit} & \frac{1}{3} & \frac{1}{3} &0\\[5pt]
\text{vegetable} & 0 & 0 & \frac{1}{3} \\
\end{array}
\]
So there are two kinds of food in three colors. You might imagine a very, very small corpus of text---a (rather dull) children's book, for instance---consisting of the three phrases
\begin{quote}
\textit{orange fruit \qquad green fruit \qquad purple vegetable}.
\end{quote}
Marginalizing provides a way to isolate and summarize the statistics of one feature at a time. For example, we can sum over $Y$ to ignore food and only focus on colors alone. The resulting marginal probability distribution $\left(\frac{1}{3},\frac{1}{3},\frac{1}{3}\right)$ on $X$ is computed by summing along the three columns of the table. In the same way, we can sum over $X$ to ignore colors and obtain information about food alone. The resulting marginal probability distribution $\left(\frac{2}{3},\frac{1}{3}\right)$ on $Y$ is computed by summing across the two rows of the table.
\[
\begin{array}{lllll}
	\pi_X=\left(\frac{1}{3},\frac{1}{3},\frac{1}{3}\right) 
		& &
	\leftrightsquigarrow
		& &
	\begin{array}{l}
		\text{{\color{YellowOrange}orange}} \\
		\text{{\color{Green}green}} \\
		\text{{\color{Purple}purple}} \\
		\end{array}
		\\[20pt]
		\hline
		\\
	\pi_Y=\left(\frac{2}{3},\frac{1}{3}\right)
		& &
	\leftrightsquigarrow
		& &
	\begin{array}{l}
		\text{fruit} \\
		\text{fruit} \\
		\text{vegetable}
	\end{array}
\end{array}
\]
Now it's apparent that marginal probability doesn't have memory.
\begin{itemize}
	\item The marginal probability of \textit{fruit} is $\frac{2}{3}$, though that number alone does not indicate that half of the fruits are {\color{YellowOrange}orange} and half are {\color{Green}green}.
	\item The marginal probability of \textit{vegetable} is $\frac{1}{3}$, though that number alone does not indicate that all of the vegetables are {\color{Purple}purple}.
	\end{itemize}
So marginal probability is forgetful. It summarizes information and therefore leaves behind certain details from the joint distribution.

\newthought{But there is \textit{another}} way to compute marginal probability so that the information lost is in fact \textit{not} lost but is readily available. Think of it as marginal probability with memory. We introduce it now.

\subsection{Marginal Probability With Memory}\label{ssec:memory}
Here is a different way to compute marginal probabilities. Start with the same joint distribution on $X\times Y$. For reference, here it is again:
\[
	M=
	\begin{bmatrix}
	\sqrt{\frac{1}{3}} & \sqrt{\frac{1}{3}} & 0\\[5pt]
	0 & 0 &\sqrt{\frac{1}{3}}
	\end{bmatrix}
\]
Some cosmetic changes have been made. First, the $2\times 3$ table is now a $2\times 3$ matrix, $M$. We may still think of the columns as being labeled by the colors \textit{orange, green, purple}, and the rows as being labeled by \textit{fruit} and \textit{vegetable}. Also, square roots have been judiciously added. The reason will become clear soon. Now, without yet providing the reason \textit{why}, let's consider the product of $M$ with its (conjugate) transpose $M^\dagger$.
\[
	M^\dagger M = 
	\kbordermatrix{
		& & &\\
	    & \frac{1}{3} & \frac{1}{3} & 0  \\[5pt]
	    & \frac{1}{3} & \frac{1}{3} & 0  \\[5pt]
	    & 0 & 0 & \frac{1}{3} 
	    }
\]
This matrix has several interesting features. First, it has three rows and three columns, which correspond to the three elements of the ordered set $X=\{\text{orange, green, purple}\}$. The $i$th diagonal entry can therefore be associated with the $i$th element of $X$. In particular, the diagonal entries of $M^\dagger M$ are precisely the marginal probabilities on $X$. So we turned the joint distribution into a matrix $M$ and computed marginal probability as the diagonal of $M^\dagger M$.
\[
\begin{array}{llllll}
	\kbordermatrix{
    & \text{{\color{YellowOrange}orange}} & \text{{\color{Green}green}} & \text{{\color{Purple}purple}}\\
    \text{{\color{YellowOrange}orange}} & {\color{YellowOrange}\frac{1}{3}} & \frac{1}{3} & 0  \\[10pt]
    \text{{\color{Green}green}} & \frac{1}{3} & {\color{Green}\frac{1}{3}} & 0  \\[10pt]
    \text{{\color{Purple}purple}} & 0 & 0 & {\color{Purple}\frac{1}{3}} \\
  	}

	&& \leftrightsquigarrow &&

	\pi_X=({\color{YellowOrange}\frac{1}{3}},{\color{Green}\frac{1}{3}},{\color{Purple}\frac{1}{3}})
\end{array}
\]
But there's something else to notice. In addition to the diagonal entries, $M^\dagger M$ also contains non-zero off-diagonal entries. This guarantees that its \textit{eigenvectors} are interesting.\sidenote{If $M^\dagger M$ had \textit{no} nonzero off diagonals, then its eigenvectors would correspond to the elements in $X$, which is not too interesting since we wouldn't have recovered anything new.} Since $M^\dagger M$ is a rank $2$ matrix, it has two eigenvectors.  Not coincidentally, the set $Y=\{\text{fruit, vegetable}\}$ has two elements. Below are the two eigenvectors of $M^\dagger M$ along with a suggestive interpretation of their relation to the two elements of $Y$. In short, the squares of the entries of the eigenvectors define conditional probability distributions on $X$.
\[
\begin{array}{rccrcc}
	\kbordermatrix{
		& \\
	    \pi(\text{{\color{YellowOrange}orange}}|\text{fruit}) & \sqrt{\frac{1}{2}} \\[10pt]
	    \pi(\text{{\color{Green}green}}|\text{fruit}) & \sqrt{\frac{1}{2}} \\[10pt]
	    \pi(\text{{\color{Purple}purple}}|\text{fruit}) & 0  
	}
	&& \overset{\text{fruit}}{\leftrightsquigarrow} &&
	\begin{array}{l}
	\text{{\color{YellowOrange}orange} fruit} \\ 
	\text{{\color{Green}green} fruit}
	\end{array}
	\\[20pt]
	\kbordermatrix{
		& \\
	    \pi(\text{{\color{YellowOrange}orange}}|\text{vegetable}) & 0 \\[10pt]
	    \pi(\text{{\color{Green}green}}|\text{vegetable}) & 0 \\[10pt]
	    \pi(\text{{\color{Purple}purple}}|\text{vegetable}) & 1  
	}
	&& \overset{\text{vegetable}}{\leftrightsquigarrow} &&
	\begin{array}{ll}
	\text{{\color{Purple}purple} vegetable}
	\end{array}
\end{array}
\]
\medskip

\noindent Concretely, the squares of the entries of the first eigenvector define a probability distribution on $X$, conditioned on the first element of $Y$. For instance $\pi(\text{{\color{YellowOrange}orange}}|\text{fruit})=\left(\sqrt{1/2}\right)^2=1/2$, which coincides with the fact that \textit{half} the fruit are orange. Likewise, the squares of the entries of the second eigenvector define a probability distribution on $X$ conditioned on the second element of $Y.$ For example, $\pi(\text{{\color{Purple}purple}}|\text{vegetable})=1$, which coincides with the fact that \textit{all} of the vegetables are purple. So the eigenvectors of $M^\dagger M$ evidently capture \emph{conditional probability}, which is precisely the information lost when marginalizing in the usual way as in Section \ref{ssec:intro_probability}. A similar phenomena happens for the matrix $MM^\dagger = \begin{bsmallmatrix} 2/3& 0 \\ 0 & 1/3\end{bsmallmatrix}$. Its diagonal recovers the marginal distribution $\pi_Y=\left(\frac{2}{3},\frac{1}{3}\right)$ on $Y$, and the entries of its two eigenvectors have the interpretation of conditional probabilities.

\newthought{So we've presented} a method to compute marginal probabilities in such a way that we have ready access to the information lost when marginalizing in the usual way. The general procedure, along with the theory behind it and the experimental results supporting it, are at the heart of this thesis. Here's the summary.

\begin{mainprocedure} Start with finite sets $X=\{x_1,\ldots,x_n\}$ and $Y=\{y_1,\ldots,y_m\}$ and a probability distribution $\pi\colon X\times Y\to \mathbb{R}$.
	\marginnote{\smallcaps{Behind the Scenes.} It might seem backwards to associate the $ij$th entry to the pair $(x_j,y_i)$, but doing so means $M$ corresponds to a linear map $\mathbb{C}^X\to\mathbb{C}^Y$, which is what we want. That is, for the application to come in Chapter \ref{ch:application}, having $\mathbb{C}^Y$ as the target space will be important.}
Define an $m\times n$ matrix $M$ whose $ij$th entry is the square root of the probability of $(x_j,y_i)$.
\[M_{ij}:=\sqrt{\pi(x_j,y_i)}\]
The diagonal of $M^\dagger M$ is the marginal probability distribution on $X$. The diagonal of $MM^\dagger$ is the marginal probability distribution on $Y$. Their eigenvectors capture information akin to conditional probability.
\end{mainprocedure}
\noindent Think of the Main Procedure as being bottom-heavy. It takes no work to define $M$, and yet this simple passage from probabilities on a \textit{set} to a linear map of \textit{vector spaces} opens a floodgate of theory and applications.
	\marginnote[-1cm]{
	\begin{center}
	\begin{tikzpicture}
	\node at (2,0) {\includegraphics[scale=0.13]{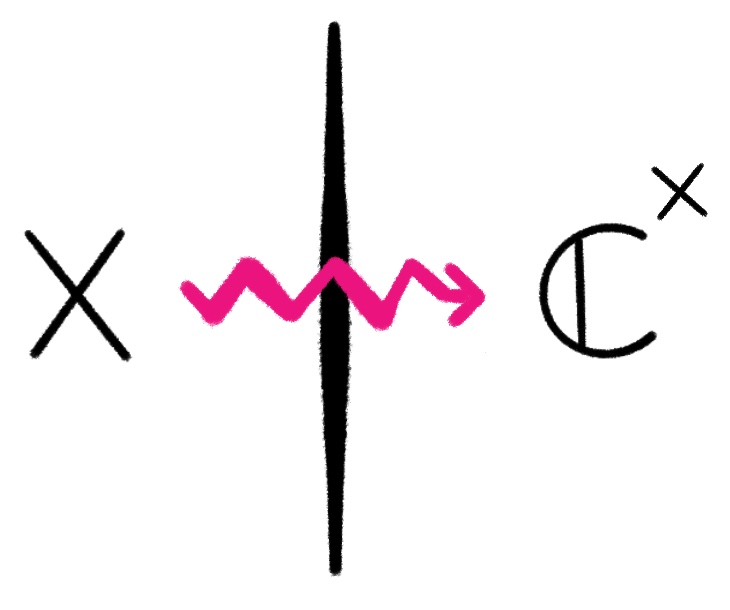}};
	\node at (.7,1.1) {\smallcaps{sets}};
	\node at (3.2,1.1) {\smallcaps{vector spaces}};
	\node at (3.4,-1) {\textit{more structure here!}};
	\end{tikzpicture}
	\end{center}
	}
Understanding and exploiting the last sentence of the Main Procedure is the bulk of this thesis. The phrase ``akin to'' is particularly important. In special cases, the entries of the eigenvectors are conditional probabilities \textit{on the nose}, as in our example. For more complicated $M$ this won't be the case, but there is a very precise sense in which the information encoded by the eigenvectors is always conditional in nature. The precise statement will be given in Takeaway \ref{takeaway:densities} and will be illustrated in Figure \ref{fig:reconstruct}. 

\newthought{Let's now pull back} the curtain and speak plainly. The mathematics here is \emph{quantum probability theory}. Borrowing terminology that will be defined in Section \ref{sec:density}, here's what's really going on.
	\begin{itemize}
		\item Every joint probability distribution on a finite set $\pi\colon X\times Y\to \mathbb{R}$ defines a rank $1$ linear operator---called a \emph{density operator} or synonymously a \emph{pure quantum state}---on the tensor product of vector spaces $\mathbb{C}^X\otimes \mathbb{C}^Y$. \textit{Density operators are the quantum version of probability distributions.}
			\marginnote[-2.5cm]{As we'll see in Section \ref{sec:density}, there is more than one way to construct a rank $1$ density operator from a joint probability distribution $\pi.$ This thesis centers on a \textit{very specific} rank $1$ density, which will explain the phenomena observed in this section's motivating example. See Section \ref{ssec:revisit} for the punchline.}

		\item The rank $1$ density operator gives rise to two \emph{reduced density operators}, one on $\mathbb{C}^X$ and another on $\mathbb{C}^Y$. The matrices $M^\dagger M$ and $MM^\dagger$ in our example are representations of reduced density operators. \textit{Reduced density operators are the quantum version of marginal probability distributions}. 

		\item  The eigenvalues and eigenvectors of these reduced density operators encode information about interactions between the $\mathbb{C}^X$ and $\mathbb{C}^Y$ subsystems, and we'll show that it is tantamount to conditional probability. More generally, the spectral information of reduced densities can be used to reconstruct a pure quantum state. When the state is that defined by a classical probability distribution, this ability to reconstruct suggests a new algorithm for reconstructing a classical joint probability distribution, which lends an application to data science. We will return to these points in Chapters \ref{ch:probability} and \ref{ch:application}.
	\end{itemize}
So the Main Procedure describes a particular passage from classical probability to quantum probability that defines special operators $M^\dagger M$ and $M M^\dagger$. The rest of this chapter is a careful walk-through of this passage. It is primarily a passage from  \textit{sets} to \textit{functions on sets} (that is, to \textit{vector spaces}), and the next section begins with some notational preliminaries. We'll introduce bra-ket notation as well as tensor network diagram notation, which is a clean way to visualize common objects and operations in linear algebra.

\section{Notation}\label{sec:prelim}
The passage from classical probability to quantum probability begins with the passage from sets to \textit{functions on sets,} namely vector spaces, which have much richer structure than the sets themselves. Unless said otherwise, we will only consider finite sets and finite-dimensional vector spaces.

\subsection{Bra-Ket Notation}\label{ssec:notation}
Given a finite set $S$, the \emph{free vector space $V=\mathbb{C}^S$ on $S$} consists of complex-valued functions on $S$, which is a \emph{Hilbert space} with inner product 
	\marginnote{A \emph{Hilbert space}\index{Hilbert space} is a vector space equipped with an inner product such that it is a \textit{complete metric space} with respect to the metric induced by the inner product. If a vector space is finite dimensional, then the induced metric always has this property. So any finite-dimensional vector space with an inner product is a Hilbert space.}
\begin{equation*}
\<v|w\> = \sum_{s \in S}\overline{v(s)} w(s).
\end{equation*}
We will sometimes denote $v(s)$ by $v_s$. The free vector space comes with a natural map
	\marginnote[1cm]{\smallcaps{Behind the Scenes.} The notation $S\to \mathbb{C}^S$ is a bit sloppy: $S$ is a set while $\mathbb{C}^S$ is a vector space. A mapping between different objects makes little mathematical sense. So by ``a natural map $S\to \mathbb{C}^S$'' I really mean a \textit{function} from $S$ to the \textit{underlying set} of $\mathbb{C}^S$. This is much more than pedantry. Category theory gives a nice way to indicate this in the notation, and it will play a major role in Chapter \ref{ch:categorytheory}.}
from $S\to \mathbb{C}^S$, defined shortly.  But first, to avoid confusion, it is helpful to use notation that distinguishes between an element $s\in S$ and its image in $\mathbb{C}^S$, which is a vector.  The vector image of $s$ is sometimes denoted with a boldface font $\mathbf{s}$ or an overset arrow $\overrightarrow{s}$.  We'll use \emph{bra and ket notation}\index{bra-ket notation}, which is especially nice when dealing with inner products.  For any $s\in S$, let $|s\>$ denote the function $S \to \mathbb{C}$ that sends $s\mapsto 1$ and $s'\mapsto 0$ for $s' \neq s$.   The set $\{|s\>\}$ is a linearly independent orthonormal spanning set for $V$, and is sometimes called the \emph{computational basis}\index{computational basis}.  If an ordering is chosen on the set $S$, say $S=\{s_1. \ldots, s_n\}$, then $|s_i\>$ is identified with the $i$th standard basis vector in $\mathbb{C}^n$.
\[
|s_i\> = 
\begin{bmatrix} 0\\ \vdots\\ 1\\ \vdots \\ 0\end{bmatrix}
\begin{matrix} \phantom{0}\\ \phantom{\vdots}\\ \leftarrow \text{$i$th entry}\\ \phantom{\vdots} \\ \phantom{0}\end{matrix}
\]
This defines an isomorphism $\mathbb{C}^S\cong\mathbb{C}^n$.  More generally we'll denote elements in $V$ by \emph{ket}\index{bra-ket notation!ket} notation $|v\> \in V$, which is just a linear combination of basis vectors $|s\>$. As an array, it is a column vector of complex numbers.
\[
	|v\> = \sum_{s\in S}v(s)|s\> =
	\begin{bmatrix}
	v(s_1) \\ \vdots \\ v(s_n)
	\end{bmatrix}
	\qquad v(s)\in\mathbb{C}
\]  
For any $|v\> \in V$, there is a linear functional in $V^*:=\hom(V,\mathbb{C})$\marginnote{The dual space $V^*=\hom(V,\mathbb{C})$ consists of all linear maps $V\to\mathbb{C}$.} whose value on $|v'\>\in V$ is the inner product $\<v|v'\>$.  We'll denote this linear functional by the \emph{bra}\index{bra-ket notation!bra} notation $\<v| \in V^*$. 
\[
\begin{tikzcd}
V \ar[r,"\<v|"] & \mathbb{C}\\[-20pt]
{|v'\>} \ar[r,mapsto] & {\<v|v'\>}
\end{tikzcd}
\]
As an array with respect to the dual basis $\{\<s|\}$, this is represented by the conjugate transpose of $|v\>$; that is $|v\>^\dagger=\<v|$ and likewise $\<v|^\dagger=|v\>$.
	\marginnote[-1cm]{I realize that bra-ket notation may not be familiar to all, but the concept behind it is not foreign to mathematics. In category theory, for instance, certain functors are commonly denoted by what they \textit{do}, which is the idea behind bra-kets. Here's an example. Fix any set $X$ and consider the functor $F$ from the category $\mathsf{Set}$ of sets to itself whose assignment on a set $Y$ is the Cartesian product $X\times Y$.
	\[F(X)=X\times Y\]
	In lieu of a generic ``$F$,'' one usually denotes this assignment by the descriptive notation $X\times-$. The empty space tells us precisely what the functor does. Just stick any set $Y$ in the blank spot to find out!
	\[
	\begin{tikzcd}[ampersand replacement = \&]
	\mathsf{Set} \ar[r,"X\times -"] \& \mathsf{Set}\\[-10pt]
	Y \ar[r,mapsto] \& X\times Y
	\end{tikzcd}
	\]
	This is the same idea behind bra-ket notation. A bra $\<v|$, which could also be written $\<v|-\>$, is analogous to the symbol $X\times -$.  The empty space tells us precisely what the linear map does. Just stick any vector $v'$ in the blank spot to find out!
	\[
	\begin{tikzcd}[ampersand replacement = \&]
	V\ar[r,"{\<v|-\>}"] \& \mathbb{C}\\[-10pt]
	v' \ar[r,mapsto] \& {\<v|v'\>}
	\end{tikzcd}
	\]
	Instead of carrying around $\<v|-\>$, which feels a little lopsided, it's better to \textit{declare} that vectors $v'$ are denoted $|v'\>$ so that ``sticking in a vector into the blank spot'' simply corresponds to  sandwiching the bra $\<v|$ and the ket $|v'\>$ together.} 
\[
	\<v| = \sum_{s\in S}\overline{v(s)}\<s| =
	\begin{bmatrix}
	\overline{v(s_1)} & \cdots & \overline{v(s_n)}
	\end{bmatrix}
	\qquad v(s)\in\mathbb{C}
\]  
Every linear functional in $V^*$ is of the form $\<v|$ for some $|v\>\in V$.  So we have vectors $|v\> \in V$ and \emph{covectors} $\<v| \in V^*$ and the map 
\[ |v\> \longleftrightarrow \<v|
\]
defines a natural isomorphism between $V$ and $V^*$. We have chosen to distinguish between vectors and covectors with bra and ket notation, but we will not give any special meaning to upper and lower indices. Given another vector space $W$, elements in the tensor product $V\otimes W$ will be denoted $|v\>\otimes |w\>$ or sometimes  $|vw\>$, when $|v\>\in V$ and $|w\>\in W.$ Further, the expression $|w\>\<v|$ will be used for the tensor product $|w\>\otimes \<v|$, which is an element of $W\otimes V^*$. It is naturally identified with the following map $V \to W$. 
\[
\begin{tikzcd}
V \ar[r,"|w\>\<v|"] & W\\[-20pt]
{|v'\>} \ar[r,mapsto] & {|w\>\<v|v'\>}
\end{tikzcd}
\]
In particular, the expression $|v\>\<v|$ is an element in $\End(V)$.  Here, $\End(V)$ denotes the space of all linear operators on $V$, and in the presence of a basis is identified with $\dim(V)\times \dim(V)$ matrices. If $|\psi\>$ is a unit vector, then the operator $|\psi\>\<\psi|$ is the orthogonal projection onto $|\psi\>$. It maps $|\psi\>$ to $|\psi\>$, and it maps every vector perpendicular to $|\psi\>$ to zero. To encourage fluency, Table \ref{tab:dictionary} contains a dictionary summarizing the translation between bra-ket notation, the usual star notation, and array representations (column versus row vectors, and so on). Here are a few more helpful facts to know.
\begin{table}[h!]
  \begin{center}
    \begin{tabular}{@{}lccc@{}}
    \toprule
	& \text{star}& \text{bra-ket} & \text{array}\\
	\midrule
	\text{vector} & $v$ 		 & ${|v\>}$		  & $\begin{bmatrix}\vdots \end{bmatrix}$ \\[20pt]
	\text{covector} & $v^*$  & ${\<v|}$		  & $\left[\:\cdots \right]$  \\[20pt]
	\text{element in $V\otimes W$} & $v\otimes w$ 		 & $|v\>\otimes |w\>$	  & $\begin{bmatrix}\vdots \end{bmatrix}\otimes \begin{bmatrix}\vdots \end{bmatrix}$ 
	\\[20pt]
	\text{inner product} & $\<v,w\>$ & $\<v|w\>$ & $\left[\:\cdots \right] \begin{bmatrix}\vdots \end{bmatrix}$ 
	\\[20pt]
	\text{linear map $V\to W$} & $w\otimes v^*$ & $|w\>\<v|$ & $\begin{bmatrix}\vdots \end{bmatrix}\left[\:\cdots \right]\begin{bmatrix}\phantom{\vdots} \end{bmatrix}$
	\\[20pt]
	\text{orthogonal projection} & $\psi\otimes \psi^*$ & $|\psi\>\<\psi|$ & $\begin{bmatrix}\vdots \end{bmatrix} \left[\:\cdots \right]$
	\\
    \bottomrule
   \end{tabular}
  \end{center}
  \caption{A dictionary for translation between bra-ket notation and star notation.}
  \label{tab:dictionary}
\end{table}

\begin{enumerate}
	\item \emph{The tensor product of vectors is the outer product of their column vector representations.} In other words,
	\begin{equation}\label{eq:outer_prod}
	|w\>\otimes|v\>=|w\>\<v|.
	\end{equation}
	To elaborate, recall that an element $|w\>\<v|$ of the tensor product $W\otimes V^*$ corresponds to a linear map $V\to W$ as mentioned above. Under the identification $V\cong V^*$ it also corresponds to the element $|w\>\otimes|v\>$ in $W\otimes V$. If bases are chosen for $W$ and $V$ then $|w\>\otimes|v\>$ is the $\dim(W)\times \dim(V)$ matrix obtained by multiplying the $\dim(W)\times 1$ column vector $|w\>$ with the $1\times \dim(V)$ row vector $\<v|$, which is Equation (\ref{eq:outer_prod}). This matrix product is called the \emph{outer product}\index{outer product} of $|w\>$ and $|v\>$. 

	\begin{example}\label{ex:outer_prod}
	As a simple example, suppose
		\marginnote{Writing the vector $|w\>\otimes|v\>$ may be puzzling since we often think of vectors as \textit{columns}. But there's no confusion here. The $3\times 2$ matrix on the left may be effortlessly reshaped into a $6\times 1$ vector, which we can also identify with the tensor product $|w\>\otimes|v\>$.
		\[
		|w\>\otimes|v\> = 
		\begin{bmatrix}-4i \\5\\-8i\\10\\12\\15i
		\end{bmatrix}
		\]
		So whether we think of tensor products of vectors as linear maps $V\to W$ or as elements of $W\otimes V$ will determine whether we use a matrix representation or a column vector.
		}
	\[|w\>=
	\begin{bmatrix} 1 \\ 2 \\ 3i
	\end{bmatrix}\qquad
	|v\> = 
	\begin{bmatrix} 4i \\ 5
	\end{bmatrix}
	\]
	Then the tensor product is a matrix
	\[
	|w\>\otimes|v\> = |w\>\<v|=
	\begin{bmatrix} 1 \\ 2 \\ 3i
	\end{bmatrix}
	\begin{bmatrix} -4i & 5
	\end{bmatrix} =
	\begin{bmatrix}
	-4i & 5\\
	-8i & 10\\
	12 & 15i
	\end{bmatrix}
	\]
	called the \textit{outer product} of $|w\>$ and $|v\>.$
	\end{example}
	\vspace{1cm}

	\item \emph{The trace of the linear operator $|v'\>\<v|$ is the inner product $\<v|v'\>$.} In other words, 
	\begin{equation*}
		\tr |v'\>\<v| = \<v|v'\>
	\end{equation*}
	for any vectors $|v\>,|v'\>\in V$. Indeed, if $\{|s\>\}$ is an orthonormal basis for $V$ then the matrix representation of $|v'\>\<v|$ is the outer product
	\[
	|v'\>\<v| =
	\begin{bmatrix}
	\vdots \\ v'(s) \\ \vdots
	\end{bmatrix}
	\begin{bmatrix}
	\cdots & \overline{v(s)} & \cdots
	\end{bmatrix}
	% =
	% \begin{bmatrix}
	% v'(s_1)\overline{v(s_1)} & v'(s_1)\overline{v(s_2)} & \cdots  \\
	% v'(s_2)\overline{v(s_1)} & v'(s_2)\overline{v(s_2)} & \cdots\\
	% \vdots & \vdots & \ddots \\
 % 	\end{bmatrix}
	\]
	which has trace equal to $\sum_s\overline{v(s)}v'(s)=\<v|v'\>.$\\
	\vspace{1cm}

	\item \emph{An outer product of tensor products is the tensor product of outer products.} In other words, for any vectors $|v\>,|v'\>\in V$ and $|w\>,|w'\>\in W$, one has $|v\>\otimes|w\>\<v'|\otimes\<w'|=|v\>\<v'|\otimes|w\>\<w'|$. But this looks a bit confusing, so let's suppress the tensor product symbol and write $|vw\>:=|v\>\otimes |w\>$. The claim is then
		\begin{equation}\label{eq:outer_tensor_product}
		|vw\>\<v'w'|=|v\>\<v'|\otimes|w\>\<w'|.
		\end{equation}
	Observe that both the left- and the right-hand side are linear operators $V\otimes W\to V\otimes W$. To verify their equality, let's check that their assignments on basis vectors are the same. First, recall that if $f\colon A\to B$ and $f'\colon A'\to B'$ are linear maps, then their tensor product is a linear map  $f\otimes f'\colon A\otimes A'\to B\otimes B'$ defined on any basis vector $|a\>\otimes |a'\>$ of $A\otimes A'$ by
	\begin{equation}\label{eq:f_tensor}
	(f\otimes f')|a\>\otimes |a'\>:= f|a\>\otimes f|a'\>.
	\end{equation}
	Now suppose $\{|s\>\}$ is a basis for $V$ and $\{|t\>\}$ is a basis for $W$.
			\begin{description}
			\item \textit{(left-hand side)} The expression $|vw\>\<v'w'|$ on the left-hand side of Equation (\ref{eq:outer_tensor_product}) is the linear map sending the vector $|st\>$ to $|vw\>\<v'w'|st\>$, where the inner product is equal to
			\begin{align*}
			\<v'w'|st\> &=(\<v'|\otimes\<w'|)  |s\>\otimes|t\>\\
			&=\<v'|s\>\<w'|t\>,
			\end{align*}
			which follows from Equation (\ref{eq:f_tensor}) with $f=\<v'|$ and $f'=\<w'|$. So the image of $|st\>$ under the operator on the left-hand side of Equation (\ref{eq:outer_tensor_product}) is 
			\[|vw\>\overbrace{\<v'|s\>\<w'|t\>}^{\text{a number}}.\] 
			\smallskip

			\item  \textit{(right-hand side)} The image of any basis vector $|st\>$ under the right-hand side of (\ref{eq:outer_tensor_product})  follows from a direct application of Equation (\ref{eq:f_tensor}).
			\begin{align*}
			(|v\>\<v'|\otimes|w\>\<w'|)|s\>\otimes|t\>
			&=|v\>\<v'|s\> \otimes |w\>\<w'|t\>\\
			&=|v\> \otimes |w\> \<v'|s\>\<w'|t\>\\
			&=|vw\>\overbrace{\<v'|s\>\<w'|t\>}^{\text{a number}}.
			\end{align*}
		\end{description}
	So we have equality as claimed.
\end{enumerate}

Let's move on to another way to denote familiar tools in linear algebra, namely as \textit{tensor network diagrams}, based on Roger Penrose's graphical notation \cite{penrose71}.

\subsection{Tensor Network Diagrams}\label{ssec:TNs}
We'll start with some terminology. A \emph{tensor of order $n$}\index{tensor} is an $n$-dimensional array of (real or complex) numbers.\marginnote{Sometimes a ``tensor of order $n$'' is also called a ``rank $n$ tensor,'' though we will not use this expression. Rank already has a meaning in linear algebra!} Such an array represents a vector in, or a mapping between, tensor products of finite-dimensional vector spaces. A factorization of a tensor into smaller ones is called a  \emph{tensor network}\index{tensor network}. Every matrix factorization (singular value decomposition, QR decomposition, Cholesky decomposition, \ldots) is an example.\sidenote{The expression ``tensor product'' might not come to mind when working with everyday matrices, but it should. Every linear map---and hence every matrix---corresponds to a vector in a tensor product of vector spaces. We'll say more about this in Section \ref{sec:red_density}.} More complicated factorizations of larger-dimensional tensors can be a notational chore to write down. Fortunately, there is a clean visual way to represent them---a \emph{tensor network diagram}\index{tensor network diagram}. In a tensor network diagram, an $n$-tensor is represented by a node, and each of the $n$ vector spaces, or \textit{dimensions}, is drawn as an edge incident to the node. 

A scalar, for example, can be thought of as a zero-dimensional array---a point. It is a $0$-tensor and corresponds to a node with no edges. A vector is a $1$-tensor, a one-dimensional array, and hence corresponds to a node with one edge. A matrix is a $2$-tensor and hence a node with two edges. A $3$-tensor is a node with three edges and so on. The edges may be further decorated with indices to distinguish each dimension. For instance, to specify an $m\times n$ matrix $M$, one must specify all $mn$ entries $M_{ij}$, where $i$ indexes the number of rows and $j$ indexes the number of columns of $M$.
\begin{center}
\begin{tikzpicture}[y=1.2cm]
\node at (0,0) {\includegraphics[scale=0.06]{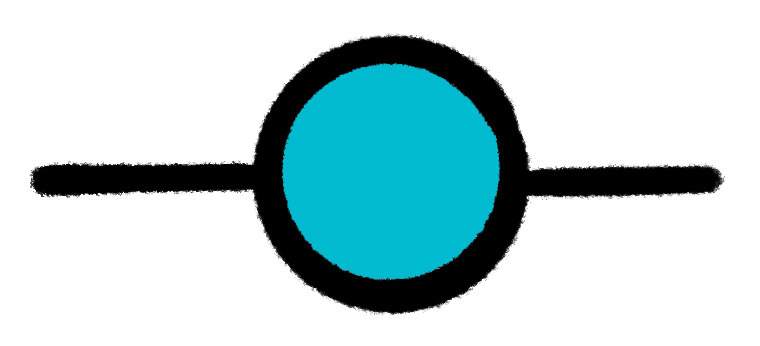}};
\node at (-1,0) {$i$};
\node at (1,0) {$j$};
\node at (0,1) {a matrix $M_{ij}$ is the diagram};
\end{tikzpicture}
\end{center}
One often thinks of ``$M_{ij}$'' not just as the $ij$th entry but also as a stand-in for all $i,j$ entries at once, resulting in the full matrix $M$. In this same spirit we may drop the indices from the diagram and simply write
\begin{tikzpicture}[x=0.5cm,baseline=-2pt]
\draw[thick] (-1,0)--(1,0);
\node[tensor,fill={rgb, 255:red,0;green,194;blue,214}] at (0,0) {};
\end{tikzpicture}
as representing the linear map itself. More generally, we have the following diagrams.
	\marginnote{Though these pictures might be new to some, I suspect the idea is familiar to all. When teaching students about functions, for instance, one often says, ``A function is like a machine. You feed it an input, the machine does its job, and then it spits out the output.'' The accompanying picture is something like this: 
	\begin{center}
	\begin{tikzpicture}
	\node at (0,0) {\includegraphics[scale=0.07]{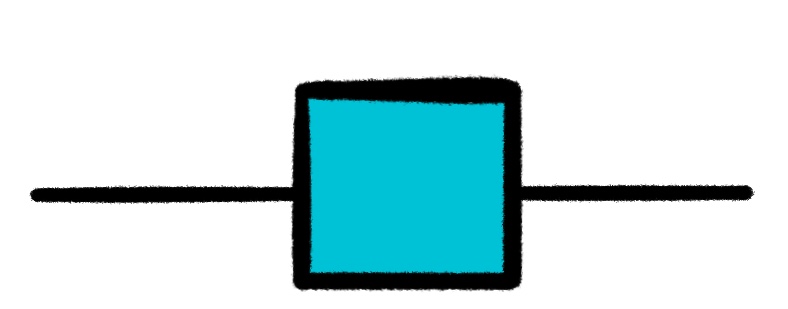}};
	\node at (-.6,.15) {\color{gray}in};
	\node at (.7,.15) {\color{gray}out};
	\node at (0,-1) {a function is a machine};
	\end{tikzpicture}
	\end{center}
	This is an example of a tensor diagram, though we're now interested in \textit{linear} functions. Think of the node as the linear mapping and think of the edges as the input (domain) and output (codomain) vector spaces. I prefer to draw the nodes as circles, rather than squares, though it doesn't matter.
	\begin{center}
	\begin{tikzpicture}
	\node at (0,0) {\includegraphics[scale=0.07]{figures/ch2/2tensor}};
	\node at (0,-1) {a matrix is a node};
	\end{tikzpicture}
	\end{center}
	 }
  \begin{center}
  \begin{tikzpicture}

    \node at (-2,3) {c};
    \node at (-2,2) {$v_i$};
    \node at (-2,1) {$M_{ij}$};
    \node at (-2,0) {$T_{ijk}$};

    \node at (-4,3) {scalar};
    \node at (-4,2) {vector};
    \node at (-4,1) {matrix};
    \node at (-4,0) {3-tensor};

    \node at (0,3) {\includegraphics[scale=0.06]{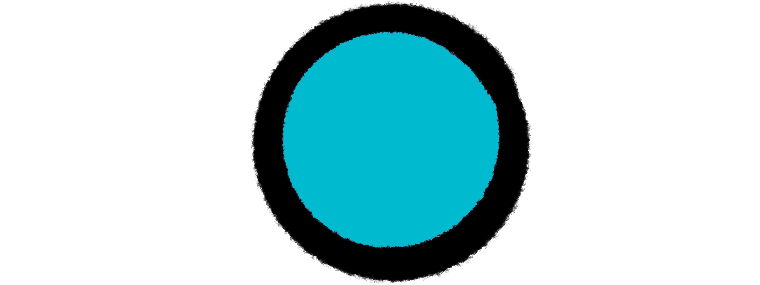}};

    \node at (0,2) {{\includegraphics[scale=0.06]{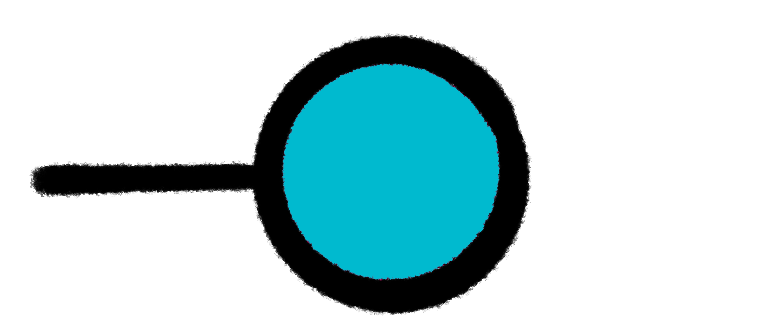}}};
    \node at (-1,2) {$i$};

    \node at (0,1) {{\includegraphics[scale=0.06]{figures/ch2/2tensor}}};
    \node at (-1,1) {$i$};
    \node at (1,1) {$j$};

    \node at (0,-.15) {{\includegraphics[scale=0.06]{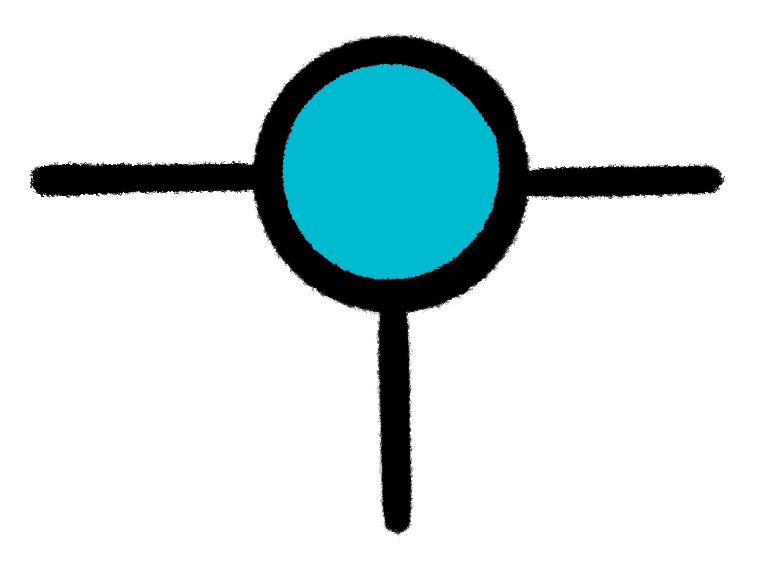}}};
    \node at (-1,0) {$i$};
    \node at (1,0) {$j$};
    \node at (.05,-.9) {$k$};

  \end{tikzpicture}
  \end{center}

\noindent In this graphical notation, familiar notions have elegant pictures. Here is a brief showcase. 

  \begin{enumerate}
      \item \emph{Composition is tensor contraction}.
       Tensors can be composed along dimensions of matching indices, and tensor contraction corresponds to summing along this common index. Graphically, this corresponds to joining the corresponding edges between diagrams. For example, the product of two matrices
      \begin{center}
      \begin{tikzpicture}
        \node (M) at (-2,0) {\includegraphics[scale=0.08]{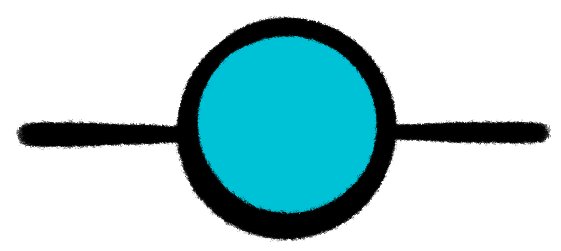}};
        \node at (-2.9,0) {$i$};
        \node at (-1.1,0) {$j$};

        \node (N) at (2,0) {\includegraphics[scale=0.08]{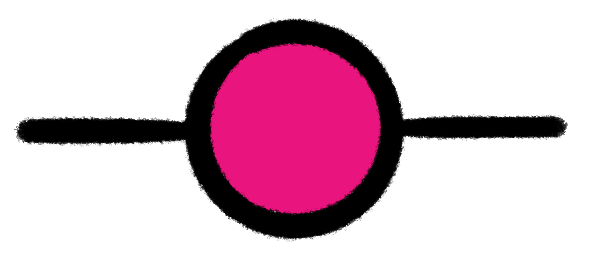}};
        \node at (1.1,0) {$j$};
        \node at (2.9,0) {$k$};

        \node[above = .4cm] at (M) {$M_{ij}$};
        \node[above = .4cm] at (N) {$N_{jk}$};

      \end{tikzpicture}
      \end{center}
      is illustrated by ``gluing'' the two edges labeled $j$ and then fusing the two nodes into a single node.
      \begin{center}
      \begin{tikzpicture}
          \node (N) at (-.7,0) {\includegraphics[scale=0.08]{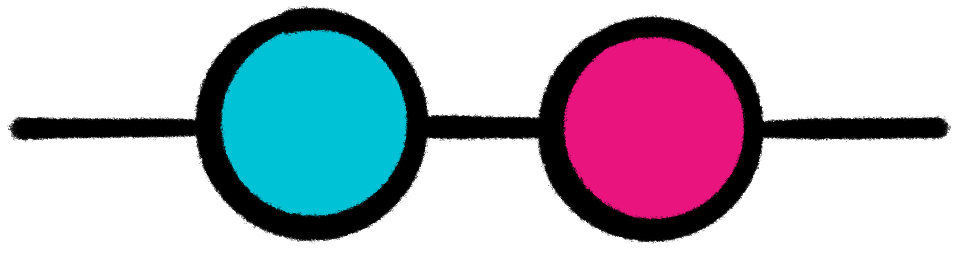}};
          \node at (-2.16,0) {$i$};
          \node at (.8,0) {$k$};
          
          \draw[thick] (2.3,0) node[left] {$i$} -- (3.7,0) node[right] {$k$} {};

          \node (MN) at (3,0) {\includegraphics[scale=0.08]{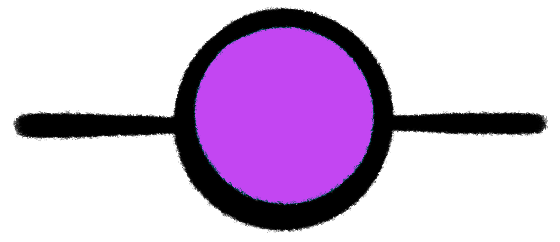}};
          
          \node[above = .4cm, xshift = .17cm] at (MN) {$(MN)_{ik}$};
          \node[above = .4cm] at (N) {$\sum_j M_{ij}N_{jk}$};
          \node  at (1.45,0) {=};
        \end{tikzpicture}
        \end{center}

      The resulting diagram has two free indices, $i$ and $k$, which indeed specify a new matrix. As another example, the product of a matrix $M$ with a vector $|v\>$ results in another vector $M|v\>$, which is a node with one free edge.
      \begin{center}
      \begin{tikzpicture}
      \node at (0,0) {\includegraphics[scale=0.16]{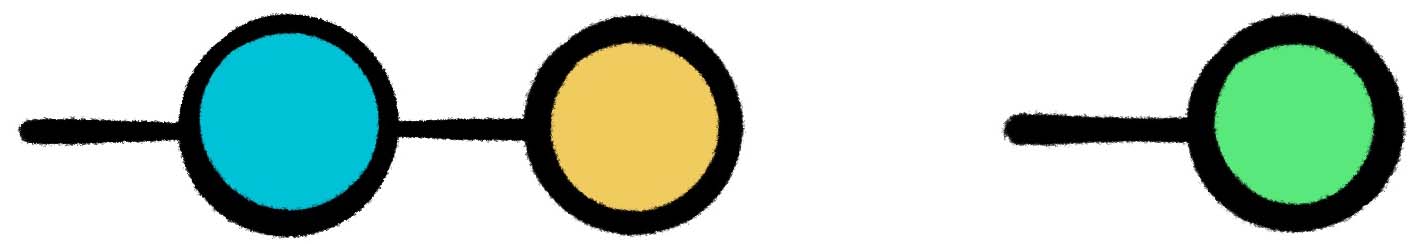}};
      \node at (-1.3,1) {$M$};
      \node at (-.2,1) {$|v\>$};
      \node at (.5,0) {=};
      \node at (1.8,1) {$M|v\>$};
      \end{tikzpicture}
      \end{center}

      To keep the picture clean, we've now dropped the indices. More generally, the composition of two or more tensors is represented by a cluster of nodes and edges where the contractions occur along edges with matching indices.
      \begin{center}
      \includegraphics[scale=0.08]{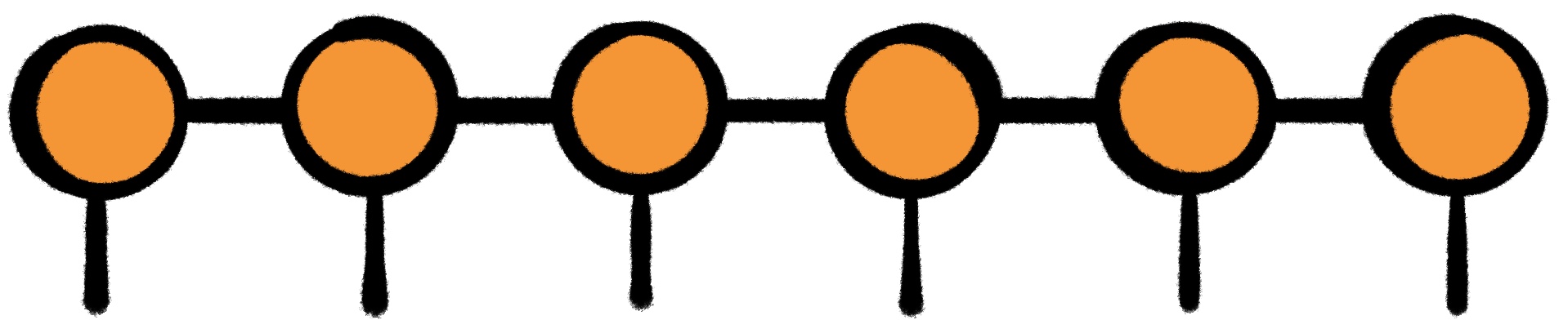}
      \end{center}

      The diagram here illustrates another important point. There is great flexibility in how one chooses to orient the diagrams spatially. A vector, for instance, is characterized by the fact that it is \textit{one node with one edge}. We will not imbue additional meaning to whether the edge is horizontal or vertical or otherwise. For example we take both \begin{tikzpicture}[baseline=-2pt] \draw[thick] (0,0)--(-.5,0) {}; \node[tensor,fill={rgb, 255:red,244;green,149;blue,54}] at (0,0) {}; \end{tikzpicture} and \begin{tikzpicture}[baseline=-2pt]  \draw[thick] (0,0)--(0,-.5) {}; \node[tensor,fill={rgb, 255:red,244;green,149;blue,54}] at (0,0) {};\end{tikzpicture} to represent the same vector.

      \item \emph{The shape of a node may convey additional meaning.} There is flexibility in the shapes used for nodes, as convention varies across the literature. This allows for creativity in how information can be conveyed through a diagram. When working with 2-tensors, for instance, we may wish to use a symmetric shape for symmetric matrices only. Then the dual mapping can be represented by reflecting its diagram,

      \begin{center}
      \begin{tikzpicture}
          \node (A) at (1,0) {\includegraphics[scale=0.08]{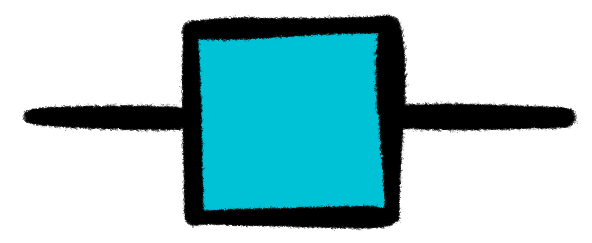}};

          \node (B) at (4,0) {\includegraphics[scale=0.09]{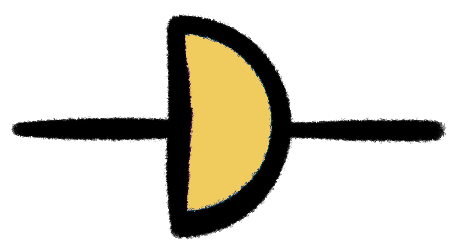}};  

          \node at (2.4,-2) (D) {\includegraphics[scale=0.09]{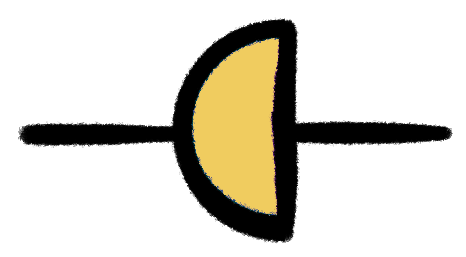}};

          \node[above = .4cm] at (A) {symmetric};
          \node[above = .4cm]  at (B) {not symmetric};
          \node[above = .4cm] at (D) {dual};

        \end{tikzpicture}
        \end{center}

      so that the symmetry is preserved in the notation.

      \begin{center}
      \begin{tikzpicture}
          \node (A) at (1,0) {\includegraphics[scale=0.08]{figures/ch2/TNs/tn5}};
          
          \node (B) at (4,0) {\includegraphics[scale=0.08]{figures/ch2/TNs/tn5}};  

          \node[above = .4cm] at (A) {$M$};
          \node[right = 1.2cm] at (A) {$=$};
          \node[above = .4cm]  at (B) {$M^\dagger$};

        \end{tikzpicture}
        \end{center}

      Another useful choice is to represent isometric embeddings as triangles: 
      \begin{center}
      \begin{tikzpicture}

          \node  (A) at (0.1,0) {\includegraphics[scale=0.08]{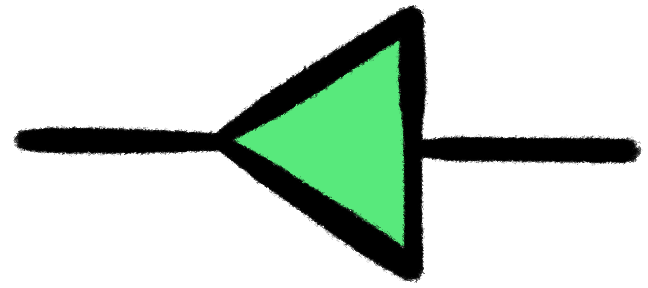}};

      \end{tikzpicture}
      \end{center}

      An \emph{isometric embedding} $U$ is a linear map from a space $V$ to a space $W$ of larger dimension that preserves the lengths of vectors. Such a map satisfies $U^\dagger U=\text{id}_V$ but $UU^\dagger \neq \text{id}_W$. In words, projection of the large space $W$ onto the embedded image $UV\subseteq W$ won't distort the vectors in $V$. This operation is the identity on $V$. On the other hand, compressing $W$ onto $V$ necessarily loses information, so to speak. The asymmetry of the triangle serves as a visual reminder of this: the base ($W$) is larger than its tip $(V)$.

      \begin{center}
      \begin{tikzpicture}
      	\node at (-.8,.6) {{\color{gray}\scriptsize{small}}};
      	\node at (.5,.6) {{\color{gray}large}};
      	\node at (1.8,.6) {{\color{gray}\scriptsize{small}}};
        \node (C) at (.5,0) {\includegraphics[scale=0.08]{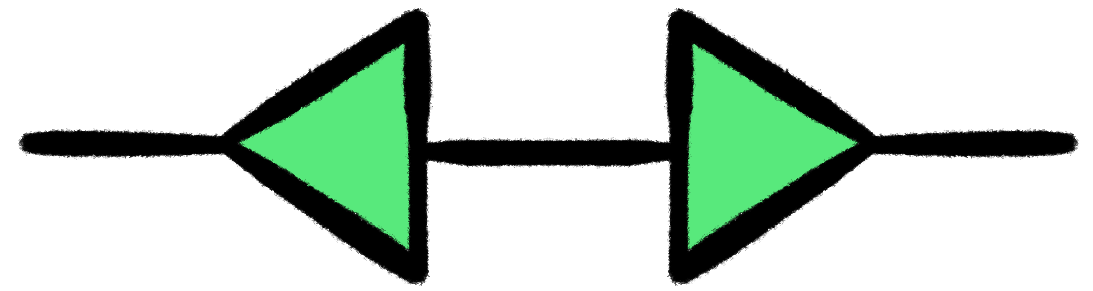}};

        \node at (-.8,-1.5) {{\color{gray}large}};
      	\node at (.5,-1.5) {{\color{gray}\scriptsize{small}}};
      	\node at (1.8,-1.5) {{\color{gray}large}};
        \node (B) at (.5,-2) {\includegraphics[scale=0.08]{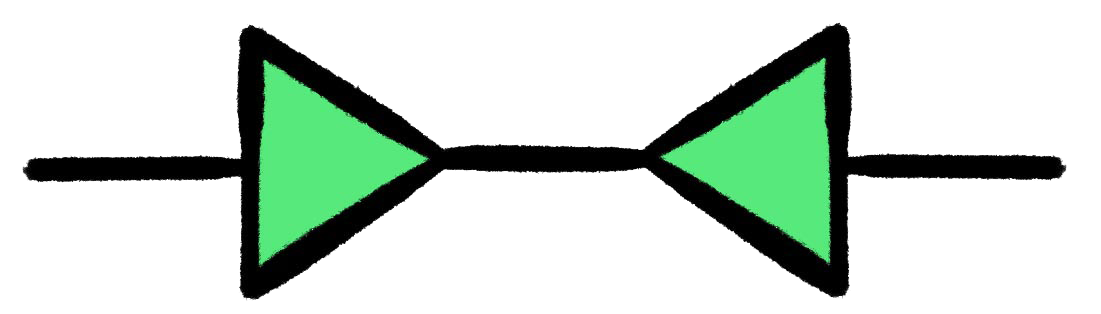}};

          \node at (2.6,0) {$=$};
          \node at (2.6,-2) {$\neq$};

          \node at (4.5,0) {\includegraphics[scale=0.07]{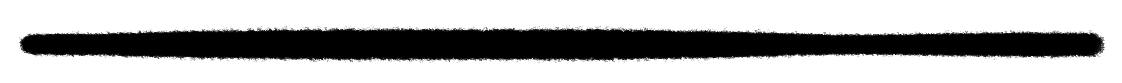}};
          \node at (4.5,-2) {\includegraphics[scale=0.07]{figures/ch2/TNs/tn10}};

      \end{tikzpicture}
      \end{center}

      When $W=V$ and when $U$ satisfies both equalities $UU^\dagger=U^\dagger U=\id_V$, then it is called a \emph{unitary operator}. This illustrates another useful convention: the identity mapping is often represented as an edge with no node. Indeed, contraction with an identity leaves a tensor and its corresponding diagram unchanged.
      \begin{center}
      \begin{tikzpicture}
      \node at (0,0) {\includegraphics[scale=0.13]{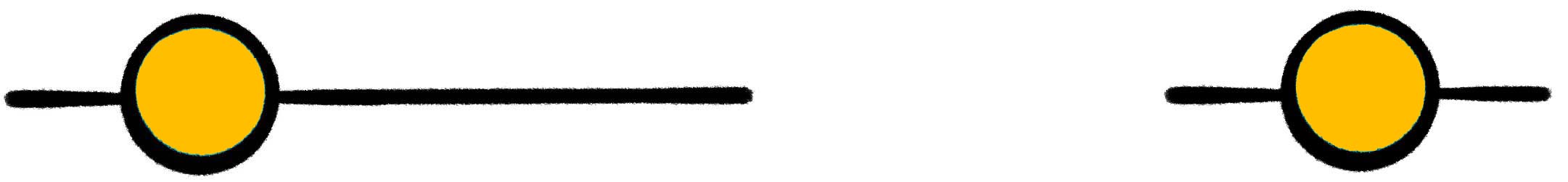}}; 
      \node at (.6,0) {=};
      \end{tikzpicture}
      \end{center}

      \item \emph{Tensor decomposition is node decomposition.} The flexibility in choosing different node shapes provides useful pictures for tensor decomposition. For example, the singular value decomposition of a matrix $M=VDU^\dagger$ (see Section \ref{sec:red_density}) can be illustrated as:

      \begin{center}
      \begin{tikzpicture}
        \node (M) at (-3,0) {\includegraphics[scale=0.07]{figures/ch2/TNs/tn1}};
        
 		\node (M) at (1,0) {\includegraphics[scale=0.13]{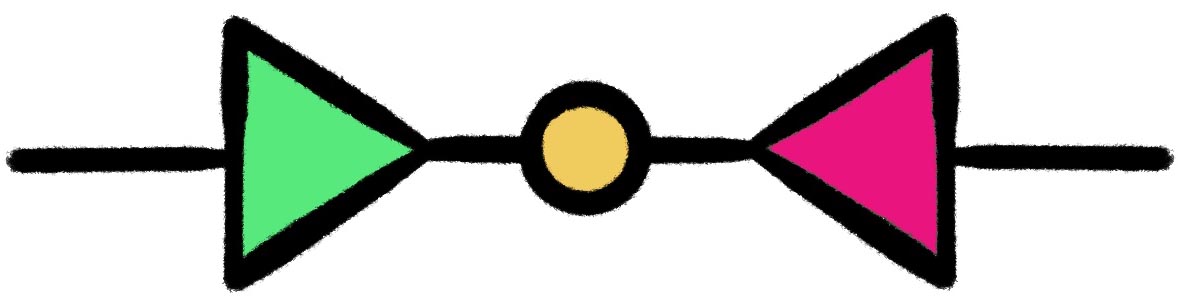}};

        \node at (-1.4,0) {$=$};
      \end{tikzpicture}
      \end{center}

      Here, $U$ and $V$ are unitary operators, hence isometries and hence triangles, while $D$ is a diagonal operator drawn as a circle. More generally, tensor decomposition is the decomposition of one node  into multiple nodes, while tensor composition is the fusion of multiple nodes into a single node.
      \begin{center}
      \begin{tikzpicture}[y=1.7cm, x=0.5cm]
      \node at (0,0) {\includegraphics[scale=0.13]{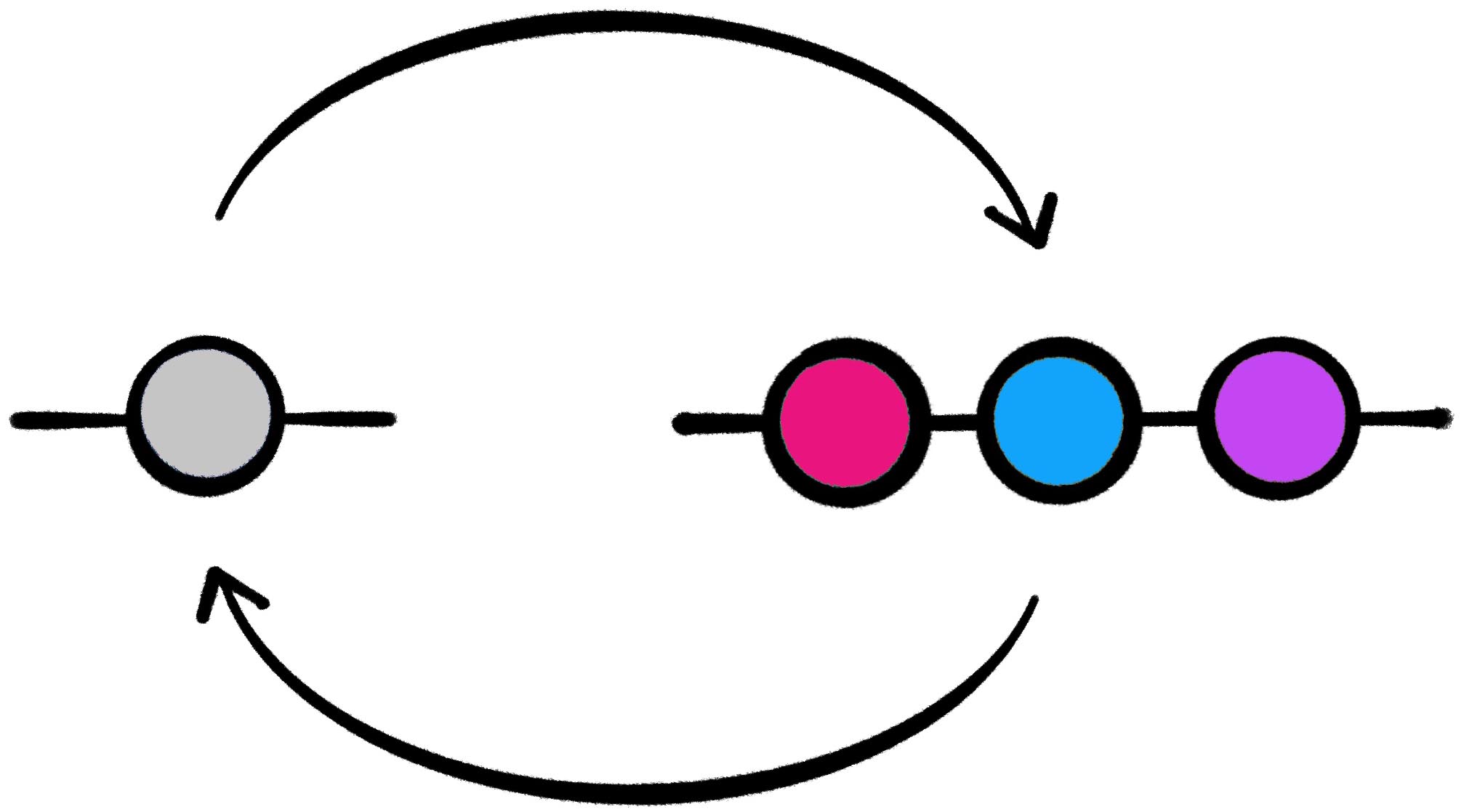}}; 
      \node at (-.8,1) {\textit{de}compose};
      \node at (-.8,-1.1) {compose};
      \end{tikzpicture}
      \end{center}

      \item \emph{Proofs have simple pictures}. Another feature of tensor diagram notation is the visual simplicity of proofs. Consider the trace of a matrix, for instance. It is the sum along a common index, and so the corresponding diagram is a loop:

      \begin{center}
      \begin{tikzpicture}
        \node (L) at (0,0) {\includegraphics[scale=0.13]{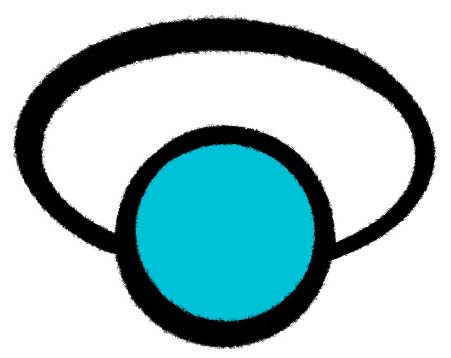}};

        \node[above = .6cm] at (L) {$\sum_i M_{ii}$};
      \end{tikzpicture}
      \end{center}

      Note that a loop has no free indices, which is consistent with the idea that a scalar is a $0$-tensor. Proving that the trace is invariant under cyclic permutations becomes as simple as sliding beads along a necklace.

      \begin{center}
      \begin{tikzpicture}
        \node at (.5,.2) (A) {\includegraphics[scale=0.07]{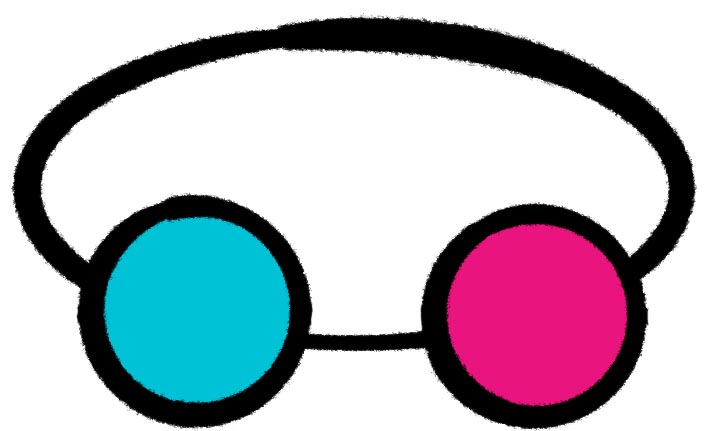}};

        \node at (2.5,0) {=};

         \node at (4.5,.2) (B) {\includegraphics[scale=0.07]{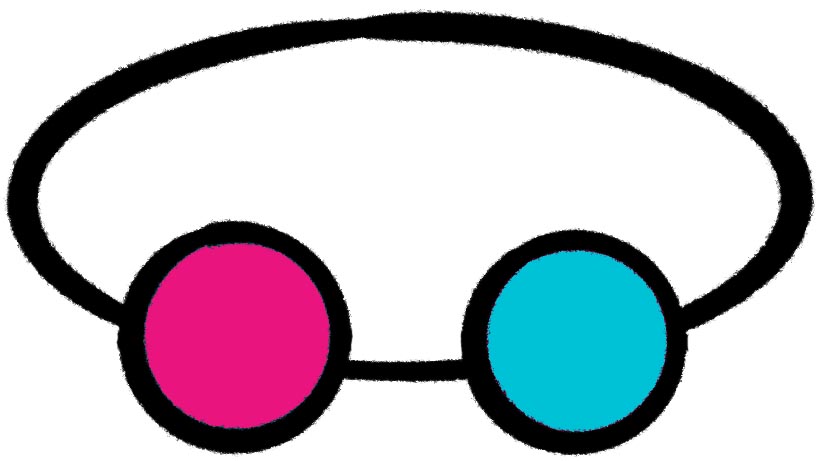}};

        \node[above = .5cm] at (A) {$tr(NM)$};
        \node[above = .5cm]  at (B) {$tr(MN)$};

      \end{tikzpicture}
      \end{center}

     \item \emph{Tensor products are illustrated by stacking in parallel}. The tensor product of vector spaces $V\otimes W$ is represented by two parallel edges, one associated to $V$ and the other associated to $W$. An arbitrary vector $|\phi\>$ in this space is then drawn as a node with two parallel edges. If $|\phi\>=|v\>\otimes |w\>$ decomposes as the tensor product of two vectors $|v\>\in V$ and $|w\>\in W$, then we illustrate $|\phi\>$ by placing the diagrams for the vectors $|v\>$ and $|w\>$ side-by-side. Similarly, an arbitrary linear map between tensor products $h\colon V\otimes V'\to W\otimes W'$ is represented by a node with four edges, one for each space. If $h$ decomposes as a tensor product $h = f\otimes f'$ for linear maps $f\colon V\to W$ and $f'\colon V'\to W'$, then the diagram for $h$ is drawn as the diagrams for $f$ and $f'$ stacked side-by-side.\sidenote{Recall that $f\otimes f' \colon V\otimes V'\to W\otimes W'$ is the linear map defined in Equation (\ref{eq:f_tensor}).}
    \begin{center}
    \begin{tikzpicture}[y=1.2cm]
    \node at (0,0) {\includegraphics[scale=0.07]{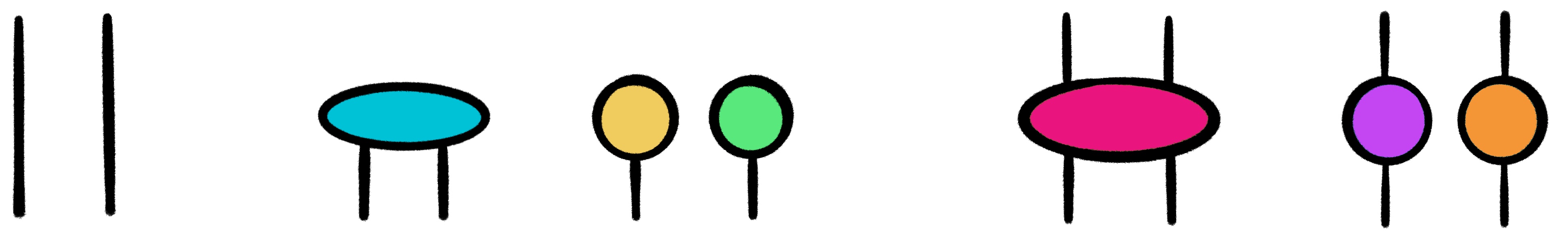}};
    \node at (-4.3,-1) {$V\otimes W$};
    \node at (-2.3,-1) {$|\phi\>$};
    \node at (-.5,-1) {$|v\>\otimes |w\>$};
    \node at (2.1,-1) {$h$};
    \node at (4,-1) {$f\otimes f'$};
    \end{tikzpicture}
    \end{center}
    The diagram shown previously, for example, is a vector in the six-fold tensor product of vector spaces $V_1,V_2,\ldots,V_6.$

    \begin{center}
    \begin{tikzpicture}
    \node (M) at (0,0) {\includegraphics[scale=0.08]{figures/ch2/TNs/mps}};
    \node[right = 2.7cm] at (M) {$\in V_1\otimes V_2\otimes V_3\otimes V_4\otimes V_5 \otimes V_6$};
    \end{tikzpicture}
    \end{center}
    \end{enumerate}

\noindent We've mentioned this particular orange diagram twice now. As it turns out, diagrams of this shape---a one-dimensional string of nodes with parallel edges sticking down---will resurface in a major way in Chapter \ref{ch:application}. Understanding it further illustrates a main advantage to working with tensor networks, so let's say a few words about that now. Consider the $N$-fold tensor product of a $d$-dimensional vector space $V$.  If $V=\mathbb{C}^X$ for some set $X=\{x_1,\ldots,x_d\}$, then the Cartesian product $X^N=\{(x_{i_1},\ldots,x_{i_N})\mid x_{i_k}\in X\}$ is a basis for the tensor product $V^{\otimes N}\cong\mathbb{C}^{X^N}$. So an arbitrary vector in $V^{\otimes N}$ is of the form 
\[|\psi\>=\sum_{(x_{i_1},\ldots,x_{i_N})\in X^N}\psi_{i_1i_2\cdots i_N}|x_{i_1}\>\otimes|x_{i_2}\>\otimes\cdots\otimes|x_{i_N}\>.\]
To specify any vector $|\psi\>\in V^{\otimes N}$ we must therefore specify $d^N$ numbers $\psi_{i_1i_2\cdots i_N}$. For large $N$, storing these numbers on a computer quickly becomes impractical. To make the situation manageable, one often looks for another vector in $V^{\otimes N}$ that is a good \textit{approximation} to $|\psi\>$ and that can be specified using far fewer parameters. This is where tensor networks come in. For example, certain $|\psi\>$ are well-approximated by a \emph{matrix product state}\index{matrix product state} (MPS), also called a \emph{tensor train},\index{tensor train} which is a vector in $V^{\otimes N}$ whose coefficients $\psi_{i_1i_2\cdots i_N}$ are obtained as a product of just a few (relatively speaking) numbers. As a diagram, MPS are depicted as a string of 3-tensors, capped off with two matrices on the left and right. This is the diagram we saw above, with $N=6.$ 
\begin{center}
  \begin{tikzpicture}
    \node (M) at (0,0) {\includegraphics[scale=0.08]{figures/ch2/TNs/mps}};
    \node[below = 1cm] at (M) {a matrix product state};
  \end{tikzpicture}
\end{center}
For now, let's not worry about the precise definition of an MPS---we'll build one from scratch in Section \ref{ssec:bird}. In the mean time, know that tensor networks come in many ``flavors,'' and an MPS is one example. More generally, tensor networks are useful in modeling the states of quantum many body systems. Different systems have different properties, and some tensor networks are better than others at capturing those properties. For a brief sketch of widely-used tensor networks and main algorithms used to work with them, see the recent survey \cite{orus2019}. For a more thorough introduction to tensor networks and their diagrammatic notation, see \cite{tensornetwork,Bridgeman2017,orus2014,Biamonte2017,Biamonte2019,penrose71,roberts2019,schollwock} and other treatments.
  
\newthought{With these notational} preliminaries in hand, think back to the motivating example in Section \ref{ssec:memory}. There we introduced an advantage of working with \textit{density operators} in lieu of probability distributions, namely access to ``marginal probability with memory.'' This advantage ultimately---and quite simply---arises from access to results and tools of linear algebra. Let's introduce this idea more formally now, starting with the definition of a density operator. 

\section{Density Operators}\label{sec:density}

\begin{definition} A \emph{density operator},\index{density operator} or simply \emph{density}, is a linear operator that is Hermitian, positive semidefinite, and has unit trace. A density operator is also called \emph{quantum state}.\index{quantum state}
\end{definition}

\noindent We think of density operators as the quantum version of probability distributions. The analogy is clear after unwinding the definition. Let's take a closer look at the three properties.
	\begin{description}
		\item \textit{Hermitian operators are like real numbers.}  An operator $f\colon V\to V$ is \emph{Hermitian}\index{Hermitian operator} if $\<fv|w\>=\<v|fw\>$ for all $v,w\in V$. The matrix representation of a Hermitian operator is equal to its conjugate transpose $f=f^\dagger,$ and the eigenvalues of Hermitian operators are always real. Analogously, probabilities are real numbers.

		\item  \textit{Positive semidefinite operators are like nonnegative numbers.} An operator $f\colon V\to V$ is \emph{positive semidefinite}\index{positive semidefinite operator} if $\<v|f|v\>\geq 0$ for all $v\in V$. Loosely speaking, if $f$ is positive semidefinite, then the image of the vector $|v\>$ under $f$'s transformation doesn't point in the ``opposite'' direction of $|v\>$. The eigenvalues of positive semidefinite operators are always nonnegative. Analogously, probabilities are nonnegative numbers.

		\item  \textit{The trace is a sum of numbers.} The trace of an operator is the sum of the diagonal entries of a matrix representing it. The trace is also equal to the sum of the eigenvalues of the operator. A density operator has \emph{trace equal to $1$}. Analogously, probabilities sum to $1$.
	\end{description}

\noindent So density operators generalize probability distributions.

\begin{center}
	\begin{tabular}{@{}l|l@{}}
	      %\toprule
	      \textbf{density operator} & \textbf{probability distribution} \\
	      \midrule
	      real eigenvalues & real numbers \\[20pt]
	      \parbox[c]{1cm}{nonnegative\\ eigenvalues} & \parbox[c]{1cm}{nonnegative\\ numbers} \\[20pt]
	       sum to $1$ & sum to $1$\\
	     %\bottomrule
	\end{tabular}
\end{center}

\noindent But more is true. Every density operator defines a classical probability distribution, and every classical probability defines a density operator.  
	\begin{center}
	\begin{tikzpicture}
	\node at (0,0) {\includegraphics[scale=0.1]{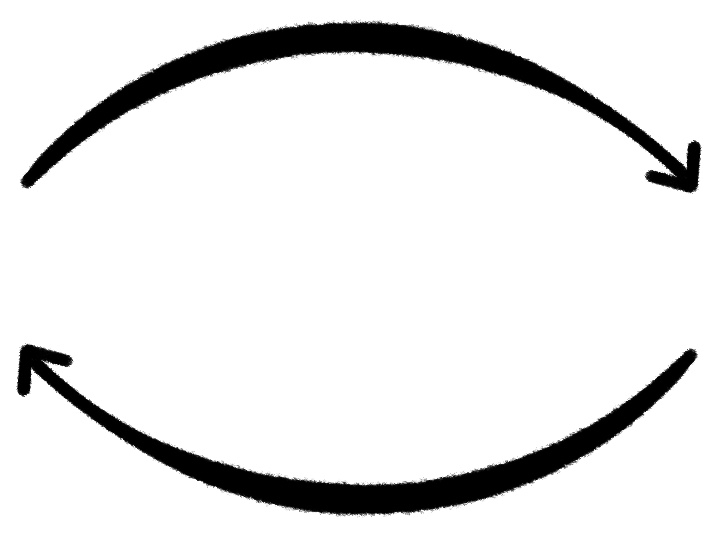}};
	\node at (-2,0) {density operator};
	\node at (2.2,0) {probability distribution};
	\end{tikzpicture}
	\end{center}
	% \marginnote{\[
	% \begin{tikzcd}[ampersand replacement = \&]
	% \text{ a density on } \mathbb{C}^S \arrow[rr, bend left,thick] \&  \& \text{probabilities on } S \arrow[ll, bend left,gray]
	% \end{tikzcd}
	% \]}

\noindent Let's understand the first claim first.

\bigskip
\noindent\textit{\Large{Every density on $\mathbb{C}^S$ defines a probability distribution on $S$.}}
% in latex book class can replace \noindent\textbf with \subsubsection
\bigskip

\noindent Let $S$ be a finite set. Any density operator $\rho\colon \mathbb{C}^S \to \mathbb{C}^S$ defines a probability distribution $\pi_\rho\colon S\to \mathbb{R}$  by the \emph{Born rule},\index{Born rule}
\begin{equation}
\pi_\rho(s)=\<s|\rho|s\> \qquad s\in S.
\end{equation}
The $\pi_\rho(s)$ are the diagonal entries of the matrix representation of $\rho$ with respect to the basis $\{|s\>\}$. These values are real and nonnegative since $\rho$ is Hermitian and positive semidefinite, and $\sum_s\pi_\rho(s)=1$ since the trace of $\rho$ is $1$. So every density $\rho$ on $\mathbb{C}^S$ does indeed define a probability distribution $\pi_\rho$ on $S$. We can also go in the other direction, and in more than one way.

	% \marginnote{\[
	% \begin{tikzcd}[ampersand replacement = \&]
	% \text{ a density on } \mathbb{C}^S \arrow[rr, bend left,gray] \&  \& \text{probabilities on } S \arrow[ll, bend left,thick]
	% \end{tikzcd}
	% \]}

\bigskip
\noindent\textit{\Large{Every probability distribution on $S$ defines densities on $\mathbb{C}^S$.}}
% in latex book class can replace \noindent\textbf with \subsubsection
\bigskip

\noindent Suppose a probability distribution $\pi$ on $S$ is already given. Let's turn our attention to density operators $\rho$ on $\mathbb{C}^S$ with the property that the Born distribution induced by $\rho$ coincides with the given probability distribution $\pi.$
\begin{equation}\label{eq:pi_rho}
\pi=\pi_\rho 
\end{equation}
As it turns out, there are \textit{multiple} ways to do this to build densities with this property. Below we share \textit{two} such ways to define a density operator on $\mathbb{C}^S$ from a probability distribution on $S$ so that Equation (\ref{eq:pi_rho}) is satisfied. These two ways can be thought of as two extremes in a range of options.
	\begin{enumerate}
		\item \textbf{Diagonal.} Given $\pi\colon S\to \mathbb{R}$ define $\rho_{\text{diag}}:=\sum_{s\in S}\pi(s)|s\>\<s|$. As a matrix, this is simply the operator with the probabilities $\pi(s)$ along its diagonal and zeros elsewhere.
		\[\rho_{\text{diag}}
		=\begin{bmatrix}
		\pi (s_1) & & & \\
		& \pi (s_2) & & \\
		& & \ddots & \\
		& & & \pi (s_n)
		\end{bmatrix}\]
			\marginnote[-2.5cm]{
			\[
			\begin{array}{ccc}
			\pi=\left(\frac{3}{5},\frac{1}{5},\frac{1}{5}\right)	&
			\leftrightsquigarrow						&
			\begin{bmatrix}
			\frac{3}{5} 	& 	0 	& 	0\\[5pt]
			0	&	\frac{1}{5} & 0 \\[5pt]
			0 & 0 & \frac{1}{5}
			\end{bmatrix} = \rho_\text{diag}
			\end{array}
			\]
			}
\noindent This is manifestly a density. It's symmetric with trace equal to $1$, and for any vector $|v\>=\sum_{s\in S}v(s)|s\>$ the inner product $\<v|\rho_{\text{diag}}|v\>$ is given by the sum  $\sum_{s\in S}|v(s)|^2\pi(s) \geq 0$, and so $\rho_{\textit{diag}}$ is positive semidefinite. It's also straightforward to see that $\rho_{\textit{diag}}$ satisfies Equation (\ref{eq:pi_rho}). The diagonal operator has the probabilities on its diagonal! In short, any finite list of nonnegative real numbers that add up to $1$ can be turned into a diagonal matrix, which is a density. Now let's give a second way to construct a density from $\pi$.

		\item \textbf{Orthogonal Projection.} Given $\pi\colon S\to\mathbb{R}$ define a unit vector by taking the sum of all the elements in $S$ and weighting each by the square root of its probability,
			\marginnote[.5cm]{Remember the vector in Equation (\ref{eq:psi}). It's the main character for the next three chapters!}
		\begin{equation}\label{eq:psi}
		|\psi\>=\sum_{s\in S}\sqrt{\pi(s)}|s\>
		\end{equation}
		and consider the orthogonal projection onto the line spanned by it,
		\begin{equation}\label{eq:rho}
		\rho_\pi:=|\psi\>\<\psi|.
		\end{equation}
		The operator $\rho_\pi$ is Hermitian since for any vectors $|v\>,|w\>\in \mathbb{C}^S$, the inner product $\<\rho_\pi v|w\>$ is $\<v|\psi\>\<\psi|w\>$, which is equal to $\<v|\rho_\pi w\>$. The operator is also positive semidefinite since the inner product $\<v|\psi\>\<\psi|v\>=|\<v|\psi\>|^2$ is nonnegative. Moreover the trace of $\rho_\pi$ is $\<\psi|\psi\>=1$. Further, Equation (\ref{eq:pi_rho}) is satisfied as claimed: 
		\[\pi_{\rho_\pi}(s)=\<s|\psi\>\<\psi|s\>=\left(\sqrt{\pi(s)}\right)^2=\pi(s).\]
		As a matrix, $\rho_\pi$ is computed as the outer product of $|\psi\>$ with itself.
			\marginnote{
			\begin{align*}
			\pi=\left(\frac{3}{5},\frac{1}{5},\frac{1}{5}\right)
			\quad &\leftrightsquigarrow \quad
			\begin{bmatrix}
			\sqrt{\frac{3}{5}} \\ \sqrt{\frac{1}{5}} \\ \sqrt{\frac{1}{5}}
			\end{bmatrix}
			\begin{bmatrix}
			\sqrt{\frac{3}{5}} & \sqrt{\frac{1}{5}} & \sqrt{\frac{1}{5}}
			\end{bmatrix}
			\\[10pt]
			&= 
			\begin{bmatrix}
			\frac{3}{5} 	& 	\frac{\sqrt 3}{5} 	& 	\frac{\sqrt 3}{5}\\[5pt]
			\frac{\sqrt 3}{5}	&	\frac{1}{5} & \frac{1}{5} \\[5pt]
			\frac{\sqrt 3}{5} & \frac{1}{5} & \frac{1}{5}
			\end{bmatrix}=\rho_\pi
			\end{align*}
			}
		\[
		\rho_\pi=
		\begin{bmatrix}
		\vdots \\ \sqrt{\pi(s)} \\ \vdots
		\end{bmatrix}
		\begin{bmatrix}
		\cdots & \sqrt{\pi(s)} & \cdots 
		\end{bmatrix}
		\]

	As a tensor network diagram, the covector $\<\psi|$ is denoted by reflecting the diagram for $|\psi\>$, and we'll denote the picture for $\rho_\pi=|\psi\>\<\psi|$ by stacking one on top of the other. 
	\begin{center}
	\begin{tikzpicture}
	\node at (0,0) {\includegraphics[scale=0.1]{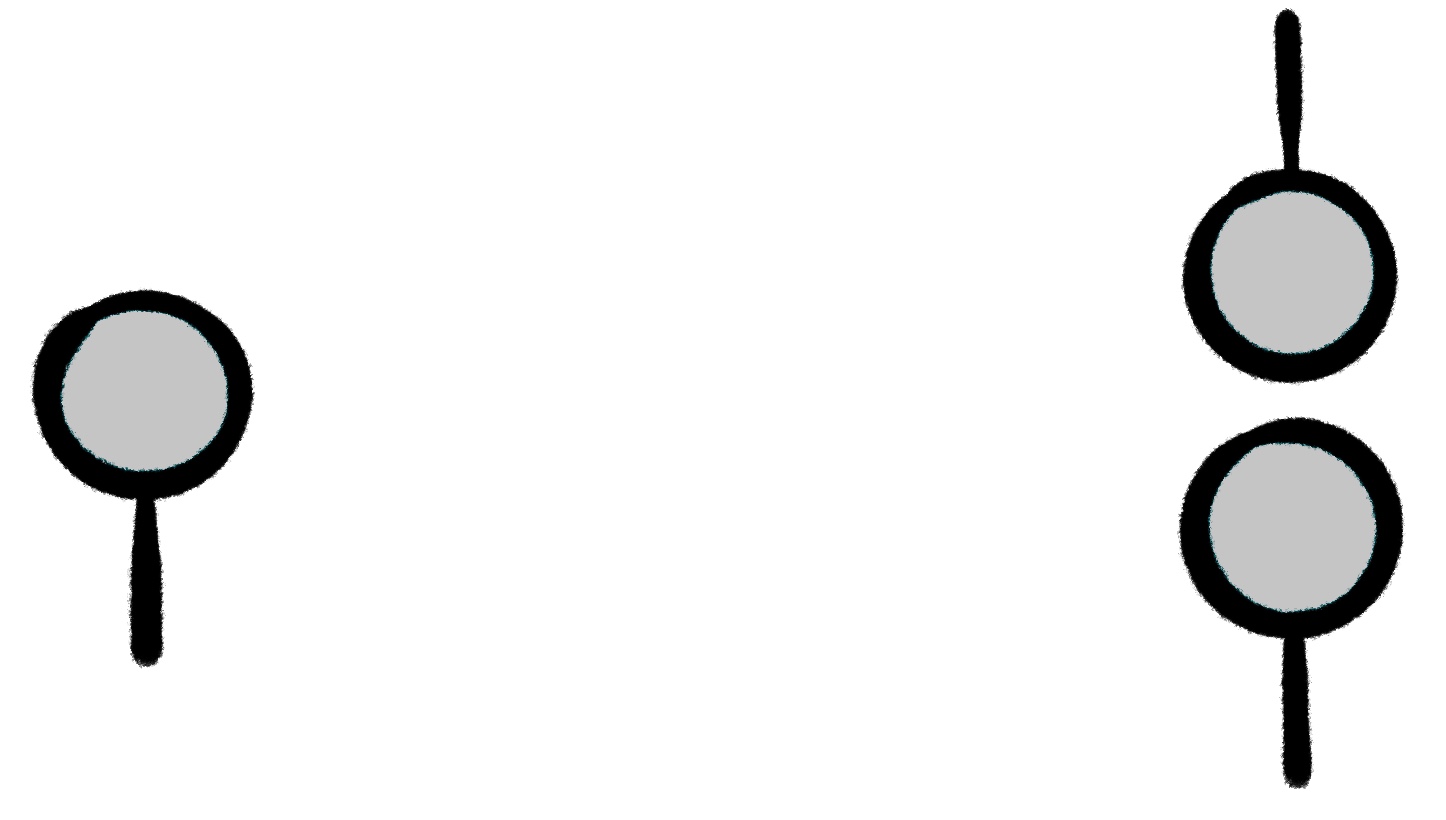}};
	\node at (-3,0) {$|\psi\>=$};
	\node at (-.8,0) {$\rightsquigarrow$};
	\node at (.8,0) {$|\psi\>\<\psi|=$};
	\end{tikzpicture}
	\end{center}
	\end{enumerate}

\newthought{For much of} this text, our main focus will be on the density operator $\rho_\pi=|\psi\>\<\psi|$, which is orthogonal projection onto the vector defined in Equation (\ref{eq:psi}). But think of diagonal operators and orthogonal projections as two extremes: $\rho_{\text{diag}}$ has maximal rank while $\rho_\pi$ has rank $1$. The language of quantum mechanics makes this more precise.

\begin{definition}\label{def:pure_mixed}
A density operator $\rho$ is called a \emph{pure quantum state}\index{quantum state!pure} if its rank is equal to $1$ and is called a \emph{mixed quantum state}\index{quantum state!mixed} otherwise. The degree of ``mixedness'' of a density is measured by its \emph{von Neumann entropy}\index{entropy!vonNeumann} $-\tr(\rho\ln\rho).$
\end{definition}

Let's spend some time on the terms just defined. First, as noted, rank $1$ densities are called \textit{pure quantum states}. Unit vectors are also referred to by this same name. The flexibility is understandable since a unit vector may be identified with orthogonal projection onto the line spanned by it. Knowing the projection operator is the same as knowing the unit vector. Here's another way to see why a ``state'' is an appropriate synonym for a vector. By definition, any unit vector $|\psi\>=\sum_s \psi(s)|s\>$ in $\mathbb{C}^S$ has the property that $\sum_s|\psi(s)|^2=1$. So its coordinates define a probability distribution on the set $S$: the probability of $s\in S$ is $|\psi(s)|^2=|\<s|\psi\>|^2$. For this reason, the set $S$ can be thought of as representing some system, the internal states of which are given by the probabilities $|\psi(s)|^2$. So we may occasionally refer to $|\psi\>$ or $|\psi\>\<\psi|$ as \emph{the state of the system $S$.}

But what about densities not arising as projections-onto-lines? If the rank of a density is larger than $1$, then it is called a \textit{mixed quantum state.} This, too, is descriptive. As with every Hermitian operator, a density $\rho$ on $\mathbb{C}^S$ has a \emph{spectral decomposition}.\index{spectral decomposition} That is, there exists an orthonormal set of eigenvectors $\{|e_1\>,\ldots,|e_r\>\}$ of $\rho$, where $r$ is the rank of $\rho$, so that
\begin{equation}\label{eq:spectral}
\rho=\sum_{i=1}^r\lambda_i|e_i\>\<e_i|.
\end{equation}
Each operator $|e_i\>\<e_i|$ is orthogonal projection onto the line spanned by the eigenvector $|e_i\>$, and so $\rho$ is a \textit{mixture} of pure states. As an aside, notice that the eigenvalues $\lambda_i$ of $\rho$ are real (since $\rho$ is Hermitian), nonnegative (since $\rho$ is positive semidefinite), and their sum is $1$ (since $\rho$ has unit trace).\sidenote{The trace is invariant under a change of basis. So choosing the eigenvectors $\{|e_i\>\}$ of $\rho$ as a basis for its image, the matrix representation of $\rho$ has its eigenvalues $\lambda_i$ along the diagonal and zeros elsewhere. The trace of that matrix is $\lambda_1+\cdots + \lambda_r=1$.} Therefore the $\lambda_i$ are probabilities. Concretely, they define a probability distribution on the set of eigenvectors $\{|e_i\>\}$. What's more, each eigenvector \textit{itself} defines a probability distribution on the original set $S$. The probability of an element $s\in S$ is given by the Born rule $\<s|e_i\>\<e_i|s\>=|\<s|e_i\>|^2$. Intuitively, just square the modulus of the $s$th entry of the column vector $|e_i\>$. Since $|e_i\>$ is a unit vector, the $|\<s|e_i\>|^2$ are probabilities---they add up to $1$.

	\begin{example}\label{ex:interpretation}
	Suppose $S=\{s_1,s_2,s_3\}$ is a three-element set and consider the following $3\times 3$ density operator on $\mathbb{C}^S\cong\mathbb{C}^3.$
	\[
	\rho = 
	\begin{bmatrix}
	\frac{1}{3} & \frac{1}{3} & 0\\[5pt]
	\frac{1}{3} & \frac{1}{3} & 0\\[5pt]
	0 & 0 & \frac{1}{3}
	\end{bmatrix}
	\]
	Its nonzero eigenvalues are $\lambda_1=\frac{2}{3}$ and $\lambda_2=\frac{1}{3}$ with corresponding eigenvectors
		\marginnote{We've analyzed the spectral information of this matrix before. Recall the discussion of colors and fruits and vegetables in Section \ref{ssec:memory}.}
	\[
	|e_1\> =
	\begin{bmatrix}
	\frac{1}{\sqrt{2}} \\[10pt] \frac{1}{\sqrt{2}} \\[10pt] 0
	\end{bmatrix}
	\qquad 
	|e_2\> = 
	\begin{bmatrix} 0 \\[5pt] 0 \\[5pt] 1
	\end{bmatrix}
	\]
	This has the following interpretation. With probability $\lambda_1=\frac{2}{3}$ the state of the system $S$ is $|e_1\>$, and with probability $\lambda_2=\frac{1}{3}$ the state is $|e_2\>$. Moreover,
		\begin{itemize}
			\item When $S$ is in the the state $|e_1\>$ then
				\begin{description}
					\item $s_1$ has probability $|\<s_1|e_1\>|^2=\left(\frac{1}{\sqrt 2}\right)=\frac{1}{2}$,

					\item $s_2$ has probability $|\<s_2|e_1\>|^2=\left(\frac{1}{\sqrt 2}\right)=\frac{1}{2}$,

					\item $s_3$ has probability $|\<s_3|e_1\>|^2=0$.\\
				\end{description}
			\item When $S$ is in the state $|e_2\>$ then
				\begin{description}
					\item $s_1$ has probability $|\<s_1|e_2\>|^2=0$,

					\item $s_2$ has probability $|\<s_2|e_2\>|^2=0$,

					\item $s_3$ has probability $|\<s_3|e_2\>|^2=1$.\\
				\end{description}
		\end{itemize}
		In short, the probability associated to any element $s\in S$ is completely determined by one of the unit vectors $|e_i\>.$ Think of the collection of these probabilities as a description of what's going on with the system $S$.
	\end{example} 

Before returning to the main text, let's say a few quick words about the third term in Definition \ref{def:pure_mixed}. The \textit{von Neumann entropy}\index{entropy!von Neumann} $-\tr(\rho \ln\rho)$ of a density $\rho$ coincides with the \textit{Shannon entropy}\index{entropy!Shannon}\sidenote[][-1cm]{To see this, consider the spectral decomposition $\rho=UDU^\dagger$, where $U$ is a unitary operator and $D$ is a diagonal operator. Equivalently, write $\rho=\sum_i\lambda_i|e_i\>\<e_i|$ where $|e_i\>$ are the columns of $U$ and $\lambda_i$ are the diagonals of $D$. Then since $\rho$ is diagonalizable, we have that $\ln\rho= U\ln(D)U^\dagger$, where $\ln(D)$ is a diagonal matrix with $\ln\lambda_i$ on the diagonal,  \[\ln(D)=\begin{bmatrix}\ln\lambda_1 &&&\\ &\ln\lambda_2&&&\\&&\ddots &\\&&&\ln\lambda_r\end{bmatrix}\] where $r$ is the rank of $\rho.$ Therefore $\rho\ln\rho=UD\ln(D)U^\dagger$ and so $-\tr\rho\ln\rho=-\sum_i\tr(|e_i\>\<e_i|)\lambda_i\ln\lambda_i$.} of the probability distribution defined by its eigenvalues $-\sum_{i=1}^r{\lambda_i\ln \lambda_i}$, where $r$ is the rank of $\rho$. Think of entropy as a measure of uncertainty. For instance, when $\rho$ is a mixed state, one thinks of the system (represented by the set) $S$ as being in any one of the $r$ states $|e_i\>$ with probability $\lambda_i$. This probabilistic uncertainty is captured by the entropy. In the extremal case when $r=1$ so that $\rho$ is a pure state, then it has only \textit{one} eigenvector $|e\>$ and thus $S$ can be in only one possible state. The probability that $S$ is in that state is the eigenvalue $\lambda = 1$. The von Neumann entropy is then $1\ln 1=0$, indicating the total lack of uncertainty.

\newthought{Let's summarize} the discussion so far. Given a probability distribution $\pi\colon S\to \mathbb{R}$ we've described two ways to define a density on $\mathbb{C}^S$ so that the probabilities $\pi(s)$ are on the diagonal of $\rho$ in the computational basis $\{|s\>\}$. One way to do this is by constructing the diagonal operator $\rho_{\text{diag}}=\sum_{s}\pi(s)|s\>\<s|$, which is a maximal rank mixed state with von Neumann entropy equal to $-\sum_{s}\pi(s)\ln \pi(s)$. At the other extreme, the orthogonal projection operator $\rho_\pi=|\psi\>\<\psi|$ is a pure state with zero von Neumann entropy, where $|\psi\>=\sum_s\sqrt{\pi(s)}|s\>$ is as defined in Equation (\ref{eq:psi}).\marginnote{There are other ways to write $\pi$ as a density, in addition to the two listed here. For instance, we could introduce a complex \emph{phase} $e^{i\theta}$ for some $\theta\in\mathbb{R}$ into $|\psi\>$.} In the context of physics, classical probability distributions are almost always modeled by \textit{diagonal} operators. In this work, however, we will always use the orthogonal projection operator $\rho_\pi$. The reason for this departure becomes clear when $\pi$ is a \textit{joint} probability distribution. As we'll soon see, the projection operator $\rho_\pi$ reduces to smaller operators that harness valuable statistical information about $\pi$. The reduction process is the quantum version of marginalization, called the partial trace.

\section{The Partial Trace}\label{sec:partial_trace}\index{partial trace}
The partial trace is an operation in linear algebra that returns information about subsystems. Imagine, for instance, a  system of many interacting ``particles'' or components. Given the state of this system, we may wish to know the state of a smaller subsystem---not in isolation, but rather---given that it interacts with its surrounding environment. The partial trace provides a way of zooming in
	\marginnote[-1cm]{
	\begin{tikzpicture}
	\node at (0,0) {\includegraphics[scale=0.06]{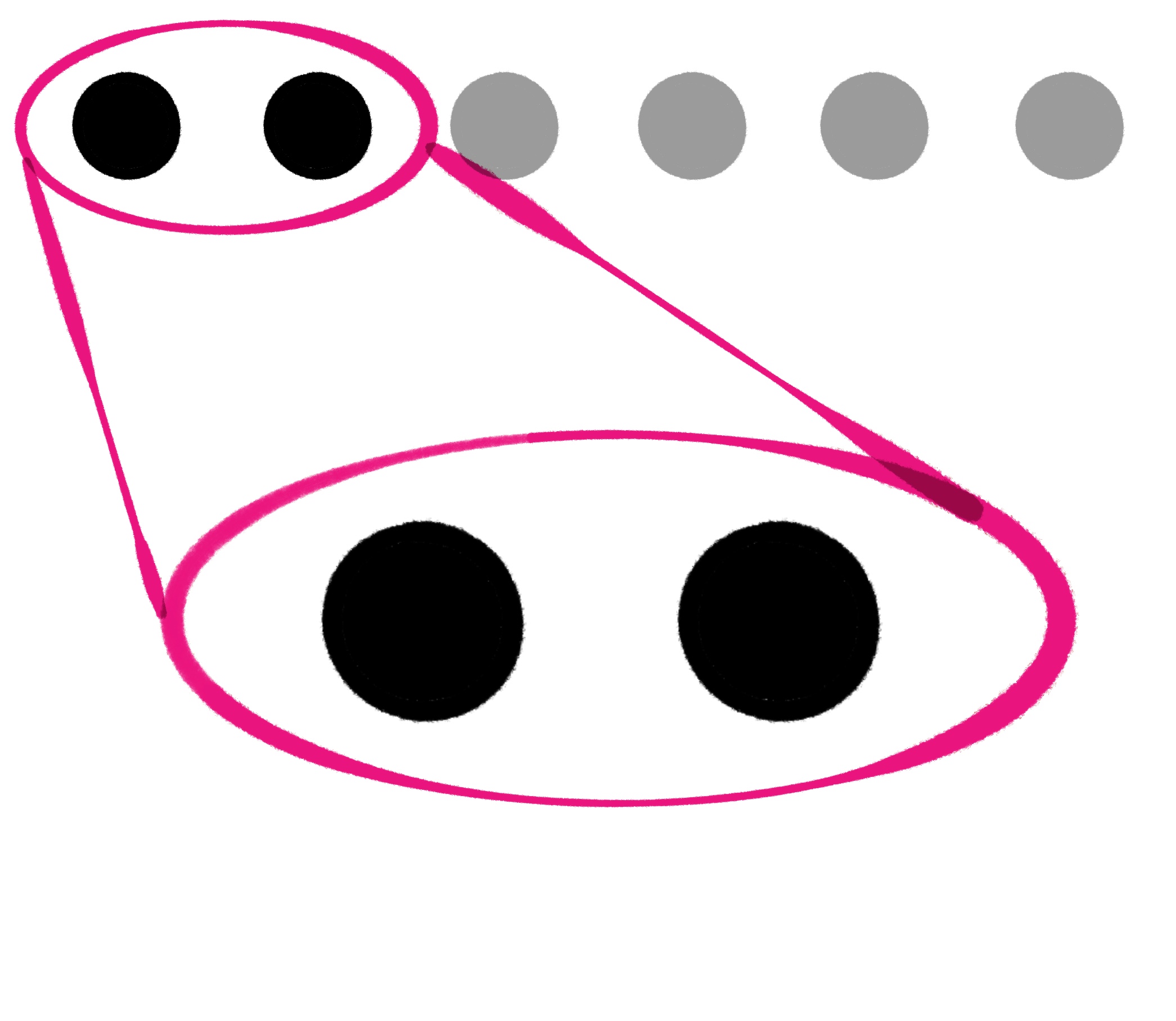}};
	\node at (0.2,-1.5) {\large\bfseries{\color{RubineRed}ZOOM}};
	\end{tikzpicture}}
and collecting such information. In quantum mechanics, \textit{a system of interacting components} is modeled with the tensor product, and the state of the system is an operator on the space. So let's consider the tensor product of two finite-dimensional vector spaces $V\otimes W$. ``Zooming in'' to one component, say $V$, can be thought of as constructing a mapping from $V\otimes W$ to $V$. But it turns out that in general there \textit{no} natural linear maps from a tensor product to each factor!
\[
	\begin{tikzcd}
	  & V\otimes W \arrow[ld, "\text{no!}"', dotted] \arrow[rd, "\text{no!}", dotted] &   \\
	V &                                                                               & W
	\end{tikzcd}
\]
We might have in mind the mapping $|v\>\otimes|w\>\mapsto|v\>$, but that assignment is not linear. More concretely, the projection\sidenote{The \emph{product} $V\times W$ is the vector space consisting of pairs $(|v\>,|w\>)$ where $|v\>\in V$ and $|w\>\in W$.} $p\colon V\times W\to V$ defined by $(|v\>,|w\>)\mapsto |v\>$ is not bilinear,
\[p(|v\>,|w\>+|z\>)=|v\>\neq 2|v\>=p(|v\>,|w\>) + p(|v\>,|z\>)\]
and so it does not extend to a linear map on the tensor product.
\[
	\begin{tikzcd}
	V\otimes W \arrow[r, "\text{no!}", dotted] & V \\
	V\times W \arrow[u, hook] \arrow[ru, "p"'] &  
	\end{tikzcd}
\]
That's not
		\marginnote[-2.8cm]{The diagram on the left refers to a defining property of the tensor product of vector spaces. The tensor product $V\otimes W$ is defined to be the vector space with the property that for any vector space $Z$ and any bilinear map $f\colon V\times W\to Z$, there is a unique linear map $\hat f\colon V\otimes W\to Z$ so that $\hat f(|v\>\otimes|w\>)=f(v,w)$; that is, so that the following diagram commutes: 
		\[
		\begin{tikzcd}[ampersand replacement = \&]
		V\otimes W \arrow[r, "\hat f", dashed] \& Z \\
		V\times W \arrow[u, hook] \arrow[ru, "f"'] \&  
		\end{tikzcd}
		\]
	The inclusion $V\times W\hookrightarrow V\otimes W$ is the assignment ${|(x,y)\>}\mapsto {|x\>\otimes|y\>}$ whenever $\{|x\>\}$ is a basis for $V$ and $\{|y\>\}$ is a basis for $W.$ This is called the \emph{universal property of the tensor product $V\otimes W.$}}
to say there are \textit{no} linear maps from a tensor product. There are plenty. If $\{|x\>\}$ is an orthonormal basis for $V$ and if $\{|y\>\}$ is an orthonormal basis for $W$, then for each $|y\>$ there is a linear map $V\otimes W\to V$ given by $\id_V\otimes\<y|$, which maps any pair $|x\>\otimes|y'\>$ to $|x\>$ if $y'=y$ and to $0$ otherwise. This assignment extends to a linear map, and yet it is not \textit{natural} because it depends on the choice of $|y\>.$ The point is there are generally no \textit{naturally-defined linear maps} from a tensor product to each factor. However, there \textit{are} natural linear maps after passing to endomorphisms of the spaces.
	\marginnote[1cm]{Recall that $\End(V)$ denotes the vector space of linear operators $V\to V$.}
	\[
	\begin{tikzcd}
	        & \End(V\otimes W) \arrow[ld] \arrow[rd] &         \\
	\End(V) &                                                           & \End(W)
	\end{tikzcd}
	\]
These natural maps are called \emph{partial traces}. The formal definition first requires a lemma.

\begin{lemma}Given finite-dimensional vector spaces $V$ and $W$ there is an isomorphism
\[\End(V\otimes W)\cong \End V\otimes \End W.\]
\end{lemma}
\begin{proof}
The proof quickly follows from the general fact that $\hom(A,B)\cong B\otimes A^*$ for finite-dimensional spaces $A$ and $B$. In particular, we have cannonical isomorphisms
\begin{align*}
\End(V\otimes W)&\cong (V\otimes W)\otimes (V\otimes W)^*\\
&\cong V\otimes W\otimes V^*\otimes W^*\\
&\cong V\otimes V^* \otimes W\otimes W^*\\
&\cong \End(V)\otimes \End(W).
\end{align*}
\end{proof}

\noindent With this isomorphism in mind, here is the definition.

\begin{definition}
Given finite dimensional vector spaces $V$ and $W$, there are two linear maps
	\marginnote{The trace associates to a linear operator a number. The \textit{partial} trace associates to a linear operator \textit{another} linear operator.}
	\[
	\begin{tikzcd}
	        & \End(V\otimes W) \arrow[ld, "\tr_W"'] \arrow[rd, "\tr_V"] &         \\
	\End(V) &                                                           & \End(W)
	\end{tikzcd}
	\]
defined for any $f\in \End(V)$ and $g\in \End(W)$ by
\[
	\tr_W(f\otimes g):=f\tr(g) \qquad \tr_V(f\otimes g):=g\tr(f).
\]
These maps are called \emph{partial traces}.\index{partial trace}
\end{definition}
	\marginnote[-.75cm]{To be explicit, the partial trace maps are obtained as the tensor product with the trace map and an identity.
	\begin{align*}
	\tr_W&:=\id_{\End(V)}\otimes \tr\\
	\tr_V&:=\tr\otimes \id_{\End(W)}
	\end{align*}
	So, for instance, we have
	\[
	\begin{tikzcd}[ampersand replacement = \&]
	\End(V)\otimes\End(W) \arrow[dd, "\id_{\End(V)}\otimes\tr"]\\ \\
	\End(V)\otimes\mathbb{C} \arrow[d, "\cong"]\\
	\End(V)
	\end{tikzcd}
	\]
	}
Notice why the partial trace repairs the lack of a natural linear map $\begin{tikzcd} V\otimes W\ar[r,dotted] &V\end{tikzcd}$. The assignment $\End(V)\otimes \End(W)\to \End(V)$ doesn't completely forget about $W$. Crucially, there is a trace of (that is, a bit of something left over from) $W$ in the expression $\tr_W(f~\otimes~g)=f\tr(g)$, namely the value $\tr(g)$! The presence of the trace of $g$ guarantees that the projection $p\colon \End V\times W\to \End V$ defined by $p(f,g)=f\tr g$ is bilinear and therefore lifts to give a linear map from the tensor product.
\[
	\begin{tikzcd}
	\End V\otimes \End W \arrow[r, "\tr_W", dashed]      & \End V \\
	\End V\times \End W \arrow[ru, "p"'] \arrow[u, hook] &       
	\end{tikzcd}
\]
We also have explicit, 
	\marginnote{Remember that $|x_iy_\alpha\>$ is shorthand for the tensor product $|x_i\>\otimes |y_\alpha\>$. Also notice that the sums in Equation (\ref{eq:partial_trace}) are analogous to the formula for the trace of a square matrix $M$,
	\[\tr M=\sum_i M_{ii}.\]
	In both cases, we sum over entries where the indices are the same. Since an operator $f$ on a tensor product is specified by \textit{four}, rather than two, indices $f_{i\alpha,j\beta}$, we can either sum over the common index $i=j$ or we can sum over the common index $\alpha=\beta$. Hence there are \textit{two} partial trace maps.}
basis-dependent expressions for the partial trace maps. Given orthonormal bases $\{|x_i\>\}$ for $V$ and $\{|y_\alpha\>\}$ for $W$, any operator $f\in\End(V\otimes W)$ is of the form
\[
	f=\sum_{\substack{i,\alpha \\ j,\beta}} f_{i\alpha,j\beta}|x_iy_\alpha\>\<x_jy_\beta|
\]
and so
\begin{equation}\label{eq:partial_trace}
	\tr_Wf = \sum_{\substack{i,j \\ \alpha}}f_{i\alpha,j\alpha}|x_i\>\<x_j|
	\qquad
	\tr_Vf = \sum_{\substack{\alpha,\beta \\ i}}f_{i\alpha,i\beta}|y_\alpha\>\<y_\beta|.
\end{equation}
\noindent Visually, the partial trace $\tr_Wf$ is obtained by computing the trace of submatrices of the matrix representing $f$. Up to a reordering of the sets $\{|x_i\>\}$ and $\{|y_\alpha\>\}$, the partial trace $\tr_Vf$ can be computed in a similar visual way.
\begin{equation*}
\begin{tikzpicture}
	\node at (0,0) 
		{$f = \left[
			\begin{array}{llllll}
			{\color{VioletRed}f_{1\alpha,1\alpha}} & f_{1\alpha,1\beta} & f_{1\alpha,1\gamma} & {\color{LimeGreen}f_{1\alpha,2\alpha}} & f_{1\alpha,2\beta} & f_{1\alpha,2\gamma} \\[15pt]
			f_{1\beta,1\alpha} & {\color{VioletRed}f_{1\beta,1\beta}} & f_{1\beta,1\gamma} & f_{1\beta,2\alpha} & {\color{LimeGreen}f_{1\beta,2\beta}} & f_{1\beta,2\gamma} \\[15pt]
			f_{1\gamma,1\alpha} & f_{1\gamma,1\beta} & {\color{VioletRed}f_{1\gamma,1\gamma}} & f_{1\gamma,2\alpha} & f_{1\gamma,2\beta} & {\color{LimeGreen}f_{1\gamma,2\gamma}} \\[15pt]
			{\color{cyan}f_{2\alpha,1\alpha}} & f_{2\alpha,1\beta} & f_{2\alpha,1\gamma} & {\color{Orange}f_{2\alpha,2\alpha}} & f_{2\alpha,2\beta} & f_{2\alpha,2\gamma} \\[15pt]
			f_{2\beta,1\alpha} & {\color{cyan}f_{2\beta,1\beta}} & f_{2\beta,1\gamma} & f_{2\beta,2\alpha} & {\color{Orange}f_{2\beta,2\beta}} & f_{2\beta,2\gamma} \\[15pt]
			f_{2\gamma,1\alpha} & f_{2\gamma,1\beta} & {\color{cyan}f_{2\gamma,1\gamma}} & f_{2\gamma,2\alpha} & f_{2\gamma,2\beta} & {\color{Orange}f_{2\gamma,2\gamma}}
			\end{array}
			\right]
		$};
	\node at (5.5,0) {$\longmapsto$};
	\node at (8.9,0)
		{$\tr_Wf = 
			\begin{bmatrix}
			\phantom{\cdots} & \phantom{\cdots} \\[18pt]
			\phantom{\cdots} &\phantom{\cdots}
			\end{bmatrix}
		$};
	\node at (7,2) 
		{$
		(\tr_Wf)_{{\color{VioletRed}11}} = \displaystyle\sum_\alpha {{\color{VioletRed}f_{1\alpha,1\alpha}}}
		$};
	\node at (7,-2) 
		{$
		(\tr_Wf)_{{\color{cyan}21}} = \displaystyle\sum_\alpha {{\color{cyan}f_{2\alpha,1\alpha}}}
		$};
	\node at (12,2) 
		{$
		(\tr_Wf)_{{\color{LimeGreen}12}} = \displaystyle\sum_\alpha {{\color{LimeGreen}f_{1\alpha,2\alpha}}}
		$};
	\node at (12,-2) 
		{$
		(\tr_Wf)_{{\color{Orange}22}} = \displaystyle\sum_\alpha {{\color{Orange}f_{2\alpha,2\alpha}}}
		$};
	\node[rectangle,fill={rgb, 255:red,222;green,62;blue,138},minimum size=5mm] at (9.2,.35){}; %pink square
	\node[rectangle,fill={rgb, 255:red,59;green,203;blue,62},minimum size=5mm] at (9.85,.35){}; %green square
	\node[rectangle,fill={rgb, 255:red,78;green,173;blue,234},minimum size=5mm] at (9.2,-.4){}; %blue square
	\node[rectangle,fill={rgb, 255:red,255;green,150;blue,56},minimum size=5mm] at (9.85,-.4){}; %orange square

	\draw[line join=round,
		decorate, decoration={
		zigzag,
		segment length=6,
		amplitude=1.5,post=lineto,
		post length=2pt
		}, ->, >=stealth, color={rgb, 255:red,78;green,173;blue,234}]
		(8,-1.5) -- (8.5,-.75); %blue arrow
	\draw[line join=round,
		decorate, decoration={
		zigzag,
		segment length=6,
		amplitude=1.5,post=lineto,
		post length=2pt
		}, ->, >=stealth, color={rgb, 255:red,222;green,62;blue,138}] (8,1.5) -- (8.5,.75); %pink arrow
	\draw[line join=round,
		decorate, decoration={
		zigzag,
		segment length=6,
		amplitude=1.5,post=lineto,
		post length=2pt
		}, ->, >=stealth, color={rgb, 255:red,59;green,203;blue,62}] (11,1.5) -- (10.5,.75); %green arrow
	\draw[line join=round,
		decorate, decoration={
		zigzag,
		segment length=6,
		amplitude=1.5,post=lineto,
		post length=2pt
		}, ->, >=stealth, color={rgb, 255:red,255;green,150;blue,56}] (11,-1.5) -- (10.5,-.75); %orange arrow

	\draw[color=gray] (.35,-2.8) -- (.35,2.8); %vertical line
	\draw[color=gray] (-3,0) -- (3.6,0); %horizontal line
\end{tikzpicture}
\end{equation*}
Tensor network diagrams provide yet another way to visualize the partial trace. It corresponds to joining wires along a common index.
\begin{center}
\begin{tikzpicture}[x=1.5cm]
\node at (0,0) {\includegraphics[scale=0.17]{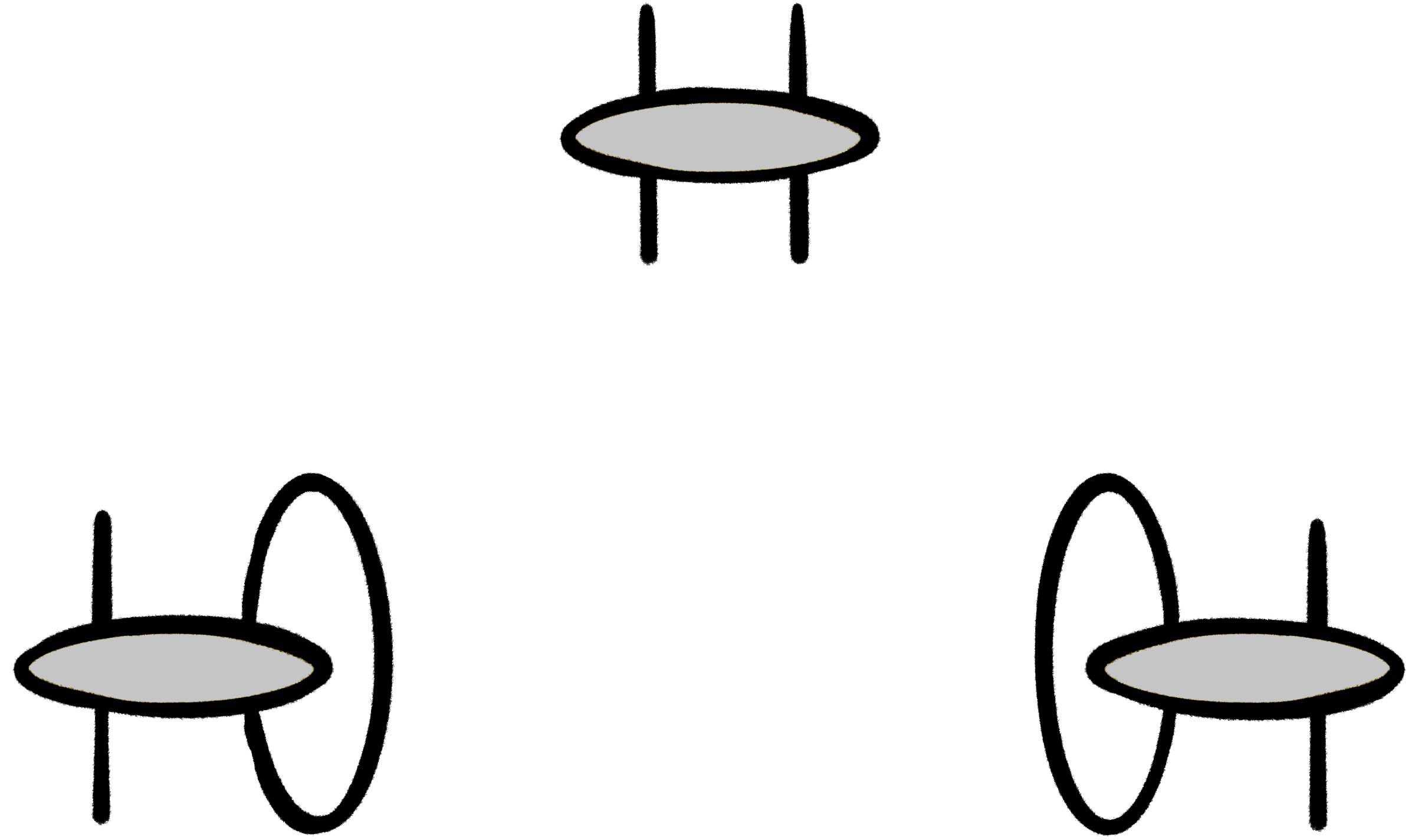}};
\node at (0,0) {$f$};
\node at (-1.8,-3) {$\tr_Wf$};
\node at (1.8,-3) {$\tr_Vf$};

\draw[thick, ->, >=stealth] (-1,1)--(-1.8,0);
\draw[thick, ->, >=stealth] (1,1)--(1.8,0);
\end{tikzpicture}
\end{center}

\noindent Importantly, the partial trace preserves several key properties of linear operators.
	\marginnote[-6cm]{Equations (\ref{eq:partial_trace}) follow directly from the definition, but the computation is a good exercise in working with bra-kets. For instance, the fact that $|x_iy_\alpha\>\<x_jy_\beta|=|x_i\>\<x_j|\otimes |y_\alpha\>\<y_\beta|$ as shown in Section \ref{ssec:notation} comes in handy. Here is the line-by-line derivation of $\tr_Wf$.
		\begin{align*}
			\tr_Wf &= \sum_{\substack{i,\alpha \\ j,\beta}} f_{i\alpha,j\beta} \tr_W|x_iy_\alpha\>\<x_jy_\beta|\\
			&= \sum_{\substack{i,\alpha \\ j,\beta}} f_{i\alpha,j\beta} \tr_W(|x_i\>\<x_j|\otimes |y_\alpha\>\<y_\beta|)\\
			&= \sum_{\substack{i,\alpha \\ j,\beta}} f_{i\alpha,j\beta} |x_i\>\<x_j|\otimes \tr|y_\alpha\>\<y_\beta|\\
			&= \sum_{\substack{i,\alpha \\ j,\beta}} f_{i\alpha,j\beta} |x_i\>\<x_j|\otimes \<y_\alpha|y_\beta\>\\
			&= \sum_{\substack{i,j \\ \alpha}} f_{i\alpha,j\alpha}|x_i\>\<x_j|.
		\end{align*}
	The second-to-last line uses orthonormality: $\<y_\alpha|y_\beta\>$ is $1$ if $\beta=\alpha$ and is $0$ otherwise. }
\begin{lemma} Let $f\in \End(V\otimes W)$. 
	\begin{itemize}
	\item If $f$ is Hermitian, then $\tr_Wf$ and $\tr_Vf$ are Hermitian.
	\item If $f$ is positive semidefinite, then $\tr_Wf$ and $\tr_Vf$ are positive semidefinite.
	\item The traces of $\tr_Wf$ and $\tr_Vf$ are equal to the trace of $f$.
	\end{itemize}
\end{lemma}
\begin{proof}
The proof is an exercise in the definition. Let $\{|x_i\>\}$ be an orthonormal basis for $V$ and let $\{|y_\alpha\>\}$ be an orthonormal basis for $W$. Let
\[f = \sum_{\substack{i,\alpha \\ j,\beta}}f_{i\alpha,j\beta}|x_iy_\alpha\>\<x_jy_\beta|\]
be any operator on $V\otimes W$, and let $f_V:=\tr_Wf$ be as in the left-hand side of Equation (\ref{eq:partial_trace}). If $f$ is Hermitian, then
\[
(f_V)_{ij} = \sum_{\alpha} f_{i\alpha,j\alpha} = \sum_{\alpha} \overline{f_{j\alpha,i\alpha}} = \overline{(f_V)_{ji}}
\]
and so $f_V$ is also Hermitian. If $f$ is positive semidefinite then
	\sidenote{This is an expansion of the inner product $\<\phi|f_V|\phi\>$, where $|\phi\>=\sum_{i,\alpha}\phi_{i\alpha}|x_iy_\alpha\>$ is any vector in $V\otimes W$.}
\[
\sum_{\substack{i,\alpha \\ j,\beta}} \overline{\phi_{i\alpha}}\phi_{j\beta}f_{i\alpha,j\beta} \geq 0 \qquad \text{ for any $\phi_{i\alpha},\phi_{j\beta}\in \mathbb{C}$.}
\]
So for any vector $|v\>=\sum_{k}v_k|x_k\>$ in $V$, the inner product $\<v|f_V|v\>$ is the sum $\sum_{i,j,\alpha}\overline{v_i}v_jf_{i\alpha,j\alpha}$ which is nonnegative by the above assumption. Finally, the trace of $f$ is $\sum_{i,\alpha}f_{i\alpha,i\alpha}$, which is also the trace of $f_V$.
\end{proof}

\noindent By the lemma, we conclude that the partial trace maps density operators to density operators. This leads to the following definition.

\begin{definition} If $\rho$ is a density operator on $V\otimes W$, then 
\[\rho_V:=\tr_W\rho \qquad\text{and}\qquad\rho_W:=\tr_V\rho\]
are called the \emph{reduced density operators}\index{density operator!reduced} associated to $\rho$.
\end{definition}
\marginnote{
  	\begin{center}
    \begin{tabular}{@{}l|l@{}}
      %\toprule
      \textbf{classical} & \textbf{quantum} \\
      \midrule
      \parbox[c]{1cm}{probability\\ distribution} & density operator \\[30pt]
      \parbox[c]{1cm}{marginal\\ probability\\distribution} & \parbox[c]{1cm}{reduced density\\ operator} \\[30pt]
       marginalizing & partial trace\\
     %\bottomrule
    \end{tabular}
  	\end{center}
	}
\noindent Just as density operators are the quantum version of probability distributions, think of reduced density operators as the quantum version of marginal probability distributions, and think of the partial trace as the quantum version of marginalizing. This is more than an analogy. Reduced densities contain classical marginal probability distributions along their diagonals. This follows immediately from the Born rule, a fact summarized in the next proposition.

\begin{proposition}\label{prop:marginals}
Let $X$ and $Y$ be finite sets, and let $\rho$ be any density operator on $\mathbb{C}^X\otimes \mathbb{C}^Y$. Then
\[\pi_{\rho_X}=(\pi_\rho)_X \qquad\text{and}\qquad \pi_{\rho_Y}=(\pi_\rho)_Y.\]
In other words, the probability distribution induced by the Born rule on $\rho_X$ is the marginalization of the joint probability distribution induced by the Born rule on $\rho$, and similarly for $\rho_Y.$
\end{proposition}
\begin{proof}
As noted in Section \ref{sec:density}, any density operator \[\rho=\sum_{\substack{xy\\x'y'}}\rho_{xy,x'y'}|xy\>\<x'y'|\] on $\mathbb{C}^X\otimes\mathbb{C}^Y$ defines a joint probability distribution $\pi_\rho\colon X\times Y\to \mathbb{R}$ by
\[\pi_\rho(x,y):=\<xy|\rho|xy\>=\rho_{xy,xy},\]
where $|xy\>$ is shorthand for $|x\>\otimes|y\>.$ The reduced density $\rho_X$ is obtained by a sum over $Y$. In particular, its $xx'$th entry is
\[\<x|\rho_X|x'\>=\sum_y\rho_{xy,x'y}.\]
Setting $x'=x$, the Born rule defines a probability distribution $\pi_{\rho_X}\colon X\to \mathbb{R}$ by
\[\pi_{\rho_X}(x):=\<x|\rho_X|x\>=\sum_y\rho_{xy,xy}=\sum_y\pi_\rho(x,y)=(\pi_\rho)_X(x)\]
where $(\pi_\rho)_X\colon X\to \mathbb{R}$ is the marginal probability distribution obtained from $\pi_\rho$.
\end{proof}

This fact will be emphasized again in the next section. There, we'll consider the special case when $\rho$ is a rank $1$ density operator, namely projection onto the unit vector defined from a classical probability distribution $\pi$ as in Equation (\ref{eq:psi}). In the case when $\pi$ is a \textit{joint} distribution, the reduced densities of $\rho$ will recover the marginal probabilities from $\pi$ along their diagonals, and moreover their eigenvectors and eigenvalues will contain conditional probabilistic information. Before arriving at these results, here is one more general result about reduced densities \cite[Problem 11.2]{kitaev2002classical}. It features what's sometimes called a \emph{Kraus representation} \cite{watrous2018theory} or \emph{operator-sum decomposition} or representation \cite{preskill229notes,kitaev2002classical}.  We won't need this Proposition until Section \ref{sec:entailment}, but now is a good time to record it.

\begin{proposition}\label{prop:povm}
Let $V$ and $W$ be finite-dimensional Hilbert spaces and let $\rho$ be any density operator on $V\otimes W.$ For each $\alpha\in\{1,\ldots,\dim(W)\}$ and $i\in\{1,\ldots,\dim(V)\}$ there are linear maps
\[
\begin{tikzcd}
             & V\otimes W \arrow[ld, "B_\alpha"'] \arrow[rd, "A_i"] &              \\
V &                                     & W
\end{tikzcd}
\]
such that the reduced densities can be written as
\[
\rho_V=\sum_\alpha B_\alpha \rho B_\alpha ^\dagger
\qquad\qquad
\rho_W=\sum_iA_i\rho A_i^\dagger
\]
and which satisfy $\sum_\alpha B_\alpha^\dagger B_\alpha= \sum_i A_i^\dagger A_i=\text{id}_{V\otimes W}$.
\end{proposition}
\begin{proof}
The proof is an exercise in the definition and in bra-ket notation. We'll prove the claim for $\rho_V$. Let $\{|x_i\>\}$ be an orthonormal basis for $V$ and let $\{|y_\alpha\>\}$ be an orthonormal basis for $W$. Define $B_\alpha =\id_V\otimes\<y_\alpha|$ to be projection onto a fixed vector. In other words,
\[B_\alpha(|x_j\>\otimes|y_\beta\>)=
\begin{cases}
|x_j\> &\text{if $\beta=\alpha$}\\
0 &\text{otherwise}
\end{cases} 
\]
The adjoint of $B_\alpha$ is tensoring with a fixed vector. 
\[B_\alpha^\dagger|x\>=|x\>\otimes|y_\alpha\>.
\] 
Let's suggestively denote this mapping by $B_\alpha^\dagger=\id_V\otimes|y_\alpha\>$. The goal is to show that the $ij$th entry of the operator $\sum_\gamma B_\gamma\rho B_\gamma^\dagger$ is equal to
\begin{equation}\label{eq:littlegoal}
\langle x_i|\sum_\gamma B_\gamma\rho B_\gamma^\dagger|x_j\> = \sum_\gamma \rho_{i\gamma,j\gamma}.
\end{equation}
Let us work on the inner product
\begin{equation}\label{eq:littleip}
\<x_i|B_\gamma\rho B_\gamma^\dagger|x_j\>
\end{equation}
from right to left. First, $B_\gamma^\dagger|x_j\>=|x_jy_\gamma\>$  and so  we're interested in the expression $\rho |x_jy_\gamma\>$ which, after expanding $\rho$ in coordinates, is equal to
\begin{equation*}
\rho|x_jy_\gamma\>=\sum_{\substack{i',\alpha \\ j',\beta}}\rho_{i'\alpha,j'\beta}|x_iy_\alpha\>\<x_{j'}y_\beta|x_jy_\gamma\>=\sum_{i',\alpha}\rho_{i'\alpha,j\gamma}|x_i'y_\alpha\>.
\end{equation*}
The inner product in (\ref{eq:littleip}) is therefore 
\begin{equation*}
\sum_{i',\alpha}\rho_{i'\alpha,j\gamma}\<x_iB_\gamma|x_{i'}y_\alpha\>
=\sum_{i'}\rho_{i'\gamma,j\gamma}\<x_i|x_{i'}\>=\rho_{i\gamma,i\gamma},
\end{equation*}
and summing over $\gamma$ recovers Equation (\ref{eq:littlegoal}). Finally, a quick check shows that for any $|x_i\>$ and $|y_\beta\>$
\[\sum_{\alpha }B_\alpha^\dagger B_\alpha(|x_j\>\otimes|y_\beta\>) = B_\beta^\dagger|x_j\> = |x_j\>\otimes|y_\beta\>\]
and so $\sum_\alpha B_\alpha^\dagger B_\alpha=\text{id}_{V\otimes W}$ as claimed.
\end{proof}

\begin{fullwidth}
\section{Reduced Densities of Pure Quantum States}\label{sec:red_density}
\end{fullwidth}
Let's continue the discussion of reduced density operators, now focusing on reduced densities of rank $1$ densities---pure quantum states. Reduced densities of these operators are especially interesting: they share the same spectrum, and they have particularly simple matrix representations. The goal of this section is to prove these two claims. The key is a simple fact about finite-dimensional vector spaces, namely that \textit{every vector in $V\otimes W$ corresponds to a linear map $V\to W$.} There are a few ways to see this, some of which we've already covered.
\medskip

\begin{description}
	\item \textbf{An explicit isomorphism.} We can describe an explicit pair of isomorphisms
	\begin{equation}
	V\otimes W \cong V^*\otimes W \cong \hom(V,W).
	\end{equation}
	Suppose $\{|x_1\>,\ldots,|x_n\>\}$ is a basis for $V$ and $\{|y_1\>,\ldots,|y_m\>\}$ is basis for $W$. The first isomorphism is the identification of a vector with its dual $|x\>\leftrightarrow \<x|$. For the second isomorphism, the assignment $\<x|\otimes |y\>\mapsto |y\>\<x|$ gives a map $V^*\otimes W\to\hom(V,W)$. Going in the other direction, any linear map $f\colon V\to W$ defines a vector $\sum_x\<x|\otimes f|x\>$ in $V^*\otimes W$, and a quick check shows these assignments are linear inverses.

	\item \textbf{Bend the wires around.} The isomorphism can be illustrated with tensor network diagrams by simply ``sliding the wire around.'' To elaborate, recall that a vector in a tensor product of spaces $V\otimes W$ is drawn as a node with two parallel edges, one edge representing $V$ and the other representing $W$. To distinguish between a vector space and its dual, we may further decorate the edges with arrows. Let's say \;\begin{tikzpicture}[baseline=-1.5pt] \draw[thick, witharrow](0,.25)--(0,-.25); \end{tikzpicture}\; represents the space $V$ and \;\begin{tikzpicture}[baseline=-1.5pt] \draw[thick, witharrow](0,-.25)--(0,.25); \end{tikzpicture}\; represents its dual $V^*$. Then to see that vectors in $V^*\otimes W$ correspond to vectors in $V\otimes W$ and hence linear maps $V\to W$, simply swap the first arrow and then slide the edge to the top.
	\begin{center}
	\begin{tikzpicture}[y=.5cm]
	\node at (0,0) {\includegraphics[scale=0.13]{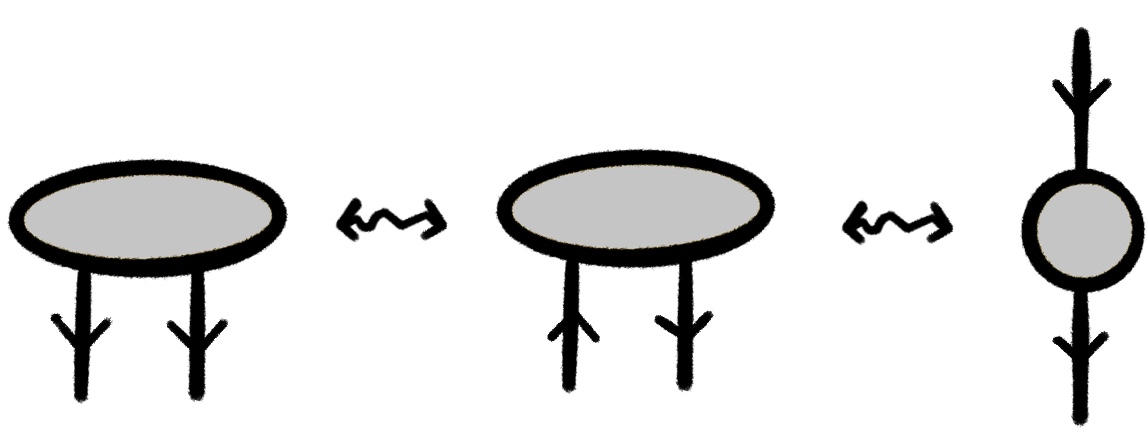}};
	\node at (-.8,.9) {swap};
	\node at (1.4,1) {slide};
	\end{tikzpicture}
	\end{center}

	\item \textbf{Assemble coefficients into a matrix.} Another way to see the correspondence between vectors in $V\otimes W$ and linear maps $V\to W$ is to rearrange the coefficients of any vector in $V\otimes W$ into a matrix. This is the option to keep in mind. It's as simple as reshaping rectangles. Let $|\psi\>\in V\otimes W$ be any vector so that
	\[
	|\psi\>=\sum_{x,y}\psi_{xy}|x\>\otimes |y\>,
	\]
	and assemble the coefficients $\psi_{xy}$ as the entries of an $m\times n$ matrix $M$. Informally, simply reshape the $nm\times 1$ column vector into a $m\times n$ rectangle.
	\[
	\begin{bmatrix} 1\\2\\3\\4\\5\\6
	\end{bmatrix}
	\quad \rightsquigarrow \quad
	\begin{bmatrix}
		1 & 4\\ 
		2 & 5\\
		3 & 6
	\end{bmatrix}
	\]
\end{description}

So suppose $|\psi\>\in V\otimes W$ is any vector and let $M$ be the $\dim(W)\times \dim(V)$ matrix whose entries are the coefficients of $|\psi\>$ as just described. As with any matrix, $M$ has a \emph{singular value decomposition} (SVD)\index{singular value decomposition} that is, a factorization of the form\sidenote[][-3cm]{\textit{Proof}. Given any $m\times n$ matrix $M,$ the matrix $M^\dagger M$ is Hermitian and therefore has a set of orthonormal eigenvectors $\{|u_1\>,\ldots,|u_n\>\}$ that form a basis for $\mathbb{C}^n$. The set of nonzero vectors $M|u_i\>$ is a basis for the image of $M$ but may not form an orthonormal set. So for each $i=1,\ldots,\text{rank}(M)$ define $|v_i\>:=M|u_i\>/\<u_i|M^\dagger M|u_i\>=M|u_i\>/\sqrt{\lambda_i}$ where $\lambda_i$ is the eigenvalue associated to $|u_i\>$. If $\text{rank}(M)<m$ then extend this set to an orthonormal basis $\{|v_1\>,\ldots,|v_m\>\}$ for $\mathbb{C}^m$. Lastly define $\sigma_i:=\sqrt{\lambda_i}$ and let $\Sigma$ be the $m\times n$ matrix with the $\sigma_i$ along its diagonal, let $U$ be the $n\times n$ matrix whose columns are the $|u_i\>$, and let $V$ be the $m\times m$ matrix whose columns are the $|v_i\>$.}
\begin{equation}\label{eq:svd}
M=V\Sigma U^\dagger.
\end{equation}
Here $U$ is an $n\times n$ unitary matrix with columns $|u\>$ called the $\emph{right singular vectors}$ of $M$; and $V$ is an $m\times m$ unitary matrix\sidenote[][1cm]{Let's not confuse the unitary operator $V$ with the vector space $V$ appearing on other pages. It's very common to use the letters $V$ and $\Sigma$ (or $D$) and $U$ for the SVD of a matrix. I'm sticking to this convention, and the context should clarify which ``$V$'' we mean.} with columns $|v\>$ called the \emph{left singular vectors} of $M$; and $\Sigma$ is an $n\times m$ matrix that, supposing $m\leq n$, is of the form $\Sigma=\left[\begin{array}{l|l} D & 0\end{array}\right]$ where $D$ is a diagonal $m\times m$ matrix whose nonzero entries $\sigma$ are the \emph{singular values} of $M$, and here $0$ denotes the $m\times(n-m)$ matrix of all zeros. Notice the last row of $U^\dagger$ makes no contribution to the matrix product $\Sigma U^\dagger$, since it is ``zeroed out'' by the submatrix of zeros in $\Sigma$. So the singular value decomposition of $M$ may be rewritten as
\begin{equation}\label{eq:svd_modified}
M=VD U_0^\dagger
\end{equation}
where $U_0^\dagger$ is the submatrix of $U$ consisting of only the first $m$ rows. 
\begin{center}
\begin{tikzpicture}
\node at (0,0) {\includegraphics[scale=0.17]{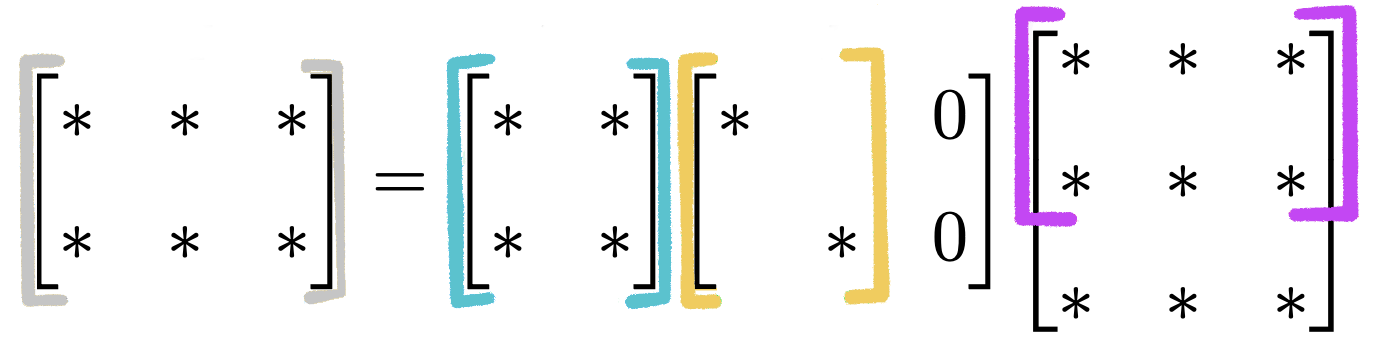}};
\node at (-3,-1.5) {$M$};
\node at (-1.75,-1.5) {=};
\node at (-.8,-1.5) {$V$};
\node at (.5,-1.5) {$D$};
\node at (3,-1.5) {$U_0^\dagger$};
\end{tikzpicture}
\end{center}
The matrix $U_0$ is no longer unitary, but it does satisfy $U_0U_0^\dagger=\id$, since its columns are orthonormal. So $U_0$ is a linear isometry, and we may\footnote{And we will. See Figure \ref{fig:bestproof}.} make use of this equality as well as $VV^\dagger=\id$.
	\begin{center}
	\includegraphics[scale=0.1]{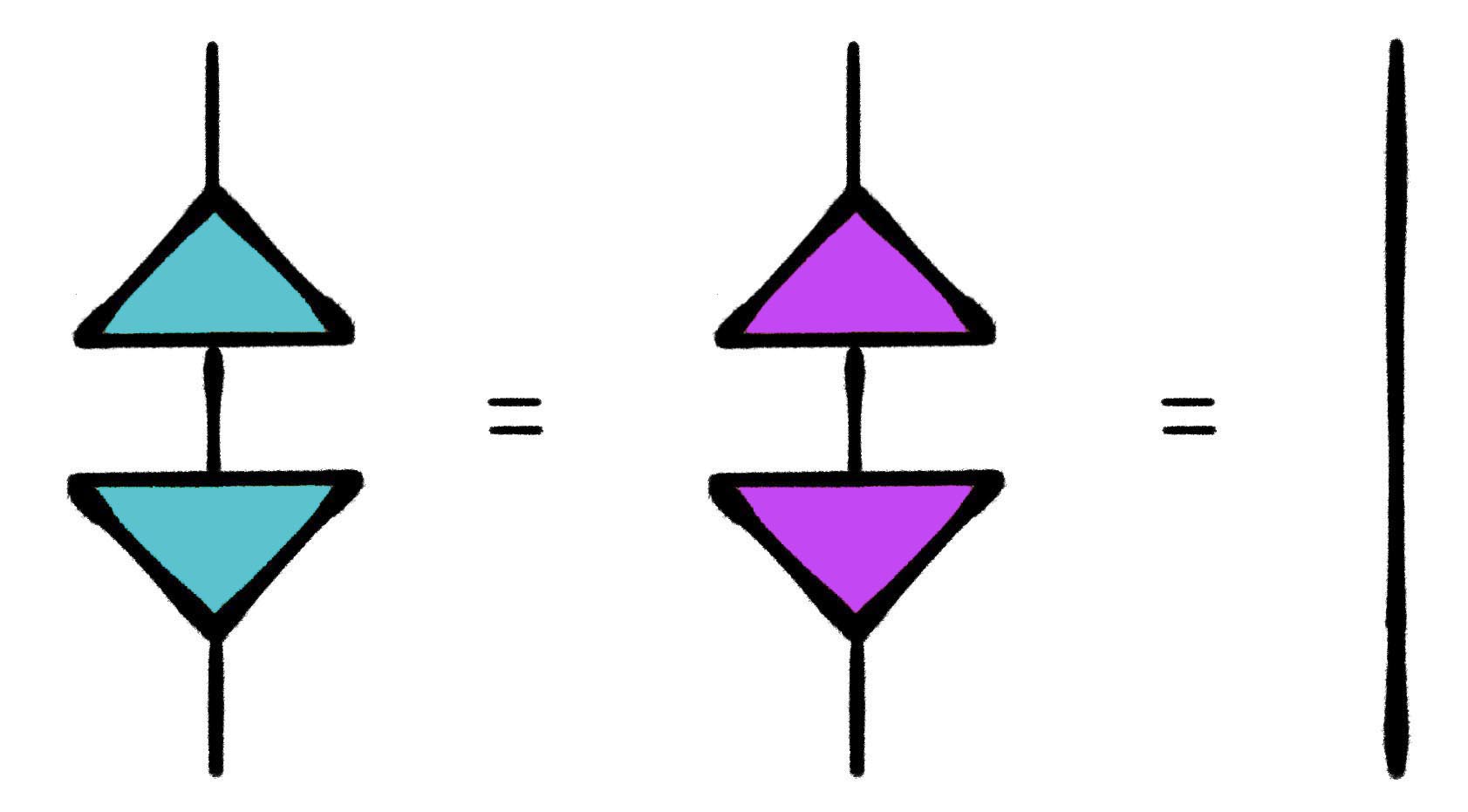}
	\end{center}
For simplicity let's always take the singular value decomposition of $M$ to refer to the cleaned-up version in Equation (\ref{eq:svd_modified}), writing $U$ in place of $U_0$. As tensor diagrams, we have the following picture.
	\marginnote{As explained in Section \ref{ssec:TNs}, triangles are reserved for linear isometries.}
\begin{center}
\begin{tikzpicture}[x=1.2cm]
\node at (0,0) {\includegraphics[scale=0.1]{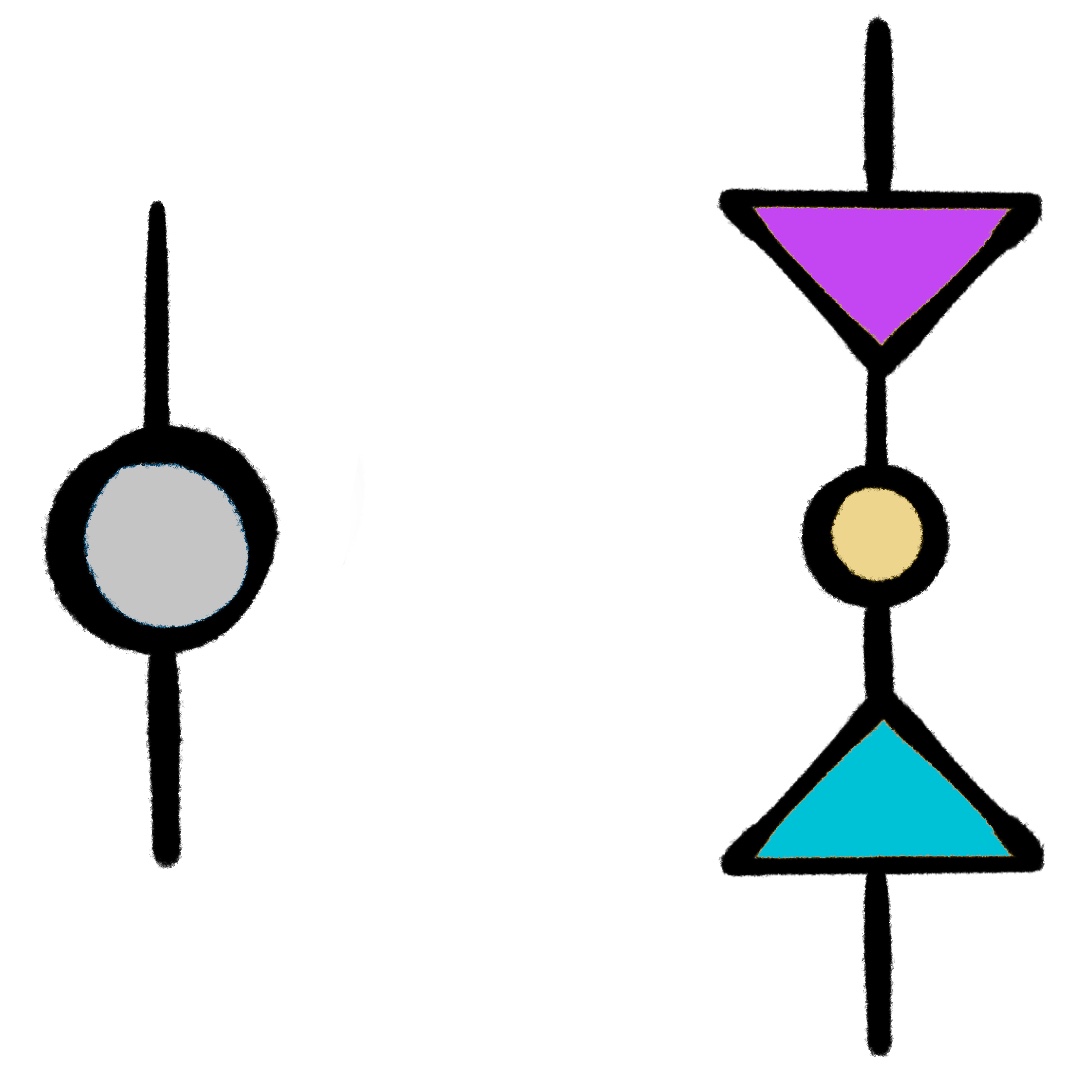}};
\node at (2,1) {$U^\dagger$};
\node at (2,0) {$D$};
\node at (2,-1) {$V$};
\node at (-2,0) {$M$};
\end{tikzpicture}
\end{center}

\noindent The next result follows immediately.

\begin{lemma} Let $V$ and $W$ be inner product spaces with $\dim V=n$ and $\dim W=m$ and suppose $m\leq n.$ For any vector $|\psi\>\in V\otimes W$ there exist orthonormal vectors $|e_1\>,\ldots |e_m\>$ in $V$ and $|f_1\>,\ldots,|f_m\>$ in $W$ so that
\begin{equation}\label{eq:schmidt}
	|\psi\> = \sum_{i=1}^m \sigma_i|f_i\>\otimes |e_i\>
\end{equation}
for nonnegative, real $\sigma_i$.
\end{lemma}
\begin{proof}
Let $\{|x_1\>,\ldots,|x_n\>\}$ be any orthonormal basis for $V$ and let $\{|y_1\>,\ldots,|y_m\>\}$ be any orthonormal basis for $W$ so that
\[|\psi\>=\sum_{i,\alpha}\psi_{i\alpha}|x_i\>\otimes |y_\alpha\>.\]
Let $M$ be the $m\times n$ matrix with $\alpha i$th entry $\psi_{i\alpha}$, and consider its singular value decomposition $M=VD U^\dagger$ where $V$ is an $m\times m$ matrix with orthonormal columns $|f\>$, and $U$ is an $n\times m$ matrix with orthonormal columns $|e\>$, and $D$ is an $m\times m$ diagonal matrix with entries $\sigma$. Then $M=\sum_{i=1}^m\sigma_i|f_i\>\<e_i|$, which defines the vector\footnote{Recall the discussion in Section \ref{ssec:notation} about tensor products versus outer products.} in Equation (\ref{eq:schmidt}).
\end{proof}

To summarize, any vector $|\psi\>$ in $V\otimes W$ corresponds to a linear map $M\colon V\to W$, whose SVD produces orthonormal sets $\{|e\>\}$ in $V$ and $\{|f\>\}$ in $W$ so that the coefficients of $|\psi\>$ with respect to the basis $\{|e\>\otimes|f\>\}$ are the singular values of $M$. This alternate representation of $|\psi\>$ is called its \emph{Schmidt decomposition}\index{Schmidt decomposition}. We use this in the next proposition, which has important consequences in the pages to come.

\begin{proposition}\label{prop:evects}
Let $V$ and $W$ be finite-dimensional inner product spaces and let $|\psi\>\in V\otimes W$ be a unit vector. The reduced densities $\rho_V$ and $\rho_W$ of the orthogonal projection operator $\rho=|\psi\>\<\psi|$ have the same eigenvalues, and there is a one-to-one correspondence between their eigenvectors.
\end{proposition}
\begin{proof}
Let $|\psi\>=\sum_{i=1}^m\sigma_i|e_i\>\otimes|f_i\>$ be the Schmidt decomposition of $|\psi\>$, where $m$ is the minimum of $\dim(W)$ and $\dim(V)$. Orthogonal projection onto $|\psi\>$ then has the following expression.
\begin{align*}
	\rho&=|\psi\>\<\psi|\\
	&=\left(\sum_{i=1}^m\sigma_i|e_i\>\otimes|f_i\>\right)\left(\sum_{j=1}^m\sigma_j\<f_j|\otimes\<e_j|\right)\\
	&=\sum_{i,j=1}^m\sigma_i\sigma_j|e_i\>\<e_j|\otimes |f_i\>\<f_j|.
\end{align*} 
Tracing out $W$ gives the reduced density operator $\rho_V=\tr_W\rho$, which is equal to
\begin{align}
	\tr_W\left(\sum_{i,j=1}^m\sigma_i\sigma_j|e_i\>\<e_j|\otimes |f_i\>\<f_j|\right)
	&= \sum_{i,j=1}^m \sigma_i\sigma_j|e_i\>\<e_j|\cdot \tr(|f_i\>\<f_j|) \nonumber \\
	&=\sum_{i,j=1}^m \sigma_i\sigma_j|e_i\>\<e_j|\<f_i|f_j\>\nonumber \\
	&=\sum_{i=1}^m\sigma_i^2|e_i\>\<e_i|. \label{eq:red_den1}
\end{align} 
Similarly, tracing out $V$ from $\rho$ gives the reduced density operator $\rho_W=\tr_V\rho$, which is equal to
\begin{align}
	\tr_V\left(\sum_{i,j=1}^m\sigma_i\sigma_j|e_i\>\<e_j|\otimes |f_i\>\<f_j|\right)
	&= \sum_{i,j=1}^m \sigma_i\sigma_j\tr(|e_i\>\<e_j|)\cdot |f_i\>\<f_j| \nonumber \\
	&=\sum_{i,j=1}^m \sigma_i\sigma_j\<e_i|e_j\>|f_i\>\<f_j|\nonumber \\
	&=\sum_{i=1}^m\sigma_i^2|f_i\>\<f_i|.\label{eq:red_den2}
\end{align} 
As a result, for all $i$ we have
\[
	\rho_V|e_i\>=\lambda_i|e_i\> 
	\qquad\text{and}\qquad
	\rho_W|f_i\>=\lambda_i|f_i\> 
\]
where $\lambda_i=\sigma_i^2$. Therefore the reduced densities have the same spectrum, and there is a bijection between their sets of eigenvectors
\[
	\begin{tikzcd}
	{|e_i\>} \ar[r,shift right] & {|f_i\>} \ar[l, shift right, "\lambda_i",swap]
	\end{tikzcd}
\]
which is provided by $M$ and its linear adjoint,
\begin{equation*}
M|e_i\>=\lambda_i|f_i\>
\qquad
M^\dagger|f_i\>=\lambda_i|e_i\>.
\end{equation*}
The two equalities above follow since $M=VDU^\dagger$ implies $MU=VD$ and $M^\dagger V=UD$. 
\end{proof}

\noindent So reduced densities of a pure quantum state $|\psi\>\<\psi|$ share the same set of eigenvalues, and there is a bijection between their eigenvectors. Equations (\ref{eq:red_den1}) and (\ref{eq:red_den2}) of the  proof give another important result, namely that the spectral decompositions of the reduced densities are
\begin{equation}\label{eq:UDU_VDV}
	\rho_V=UD^2 U^\dagger \qquad \rho_W=VD^2 V^\dagger.
\end{equation}
In other words, whenever $M$ is the $\dim(W)\times \dim(V)$ matrix associated to a unit vector $|\psi\>\in V\otimes W$, the eigenvectors of $\tr_W|\psi\>\<\psi|$ and $\tr_V|\psi\>\<\psi|$ are the singular vectors of $M,$ and their eigenvalues are the singular values of $M$. The following corollary is immediate.

\begin{corollary}\label{cor:red_densities_as_MM}
Let $V$ and $W$ be finite-dimensional inner product spaces, let $|\psi\>~\in~V\otimes W$ be a unit vector, and let $\rho=|\psi\>\<\psi|$ be orthogonal projection. There exists a matrix $M$ so that the reduced densities of $\rho$ are of the form
\begin{equation}
\rho_V = M^\dagger M \qquad \rho_W = MM^\dagger
\end{equation}
\end{corollary}
\begin{proof}
Let $M$ be the $\dim(W)\times\dim(V)$ matrix associated to $|\psi\>$. If $M$ has the singular value decomposition $M=VDU^\dagger$ then  Equations (\ref{eq:red_den1}) and (\ref{eq:red_den2}) imply that $\rho_V=UD^2 U^\dagger$ and $\rho_W=VD^2 V^\dagger$, and so
\begin{align*}
M^\dagger M &= (VDU^\dagger)^\dagger (VDU^\dagger)\\
&= U D^\dagger V^\dagger VD U^\dagger\\
&= U D^2 U^\dagger \\
&=\rho_V
\end{align*}
since $D^\dagger=D$ for a diagonal operator. Similarly,
\begin{align*}
M M^\dagger &=  (VDU^\dagger)(VDU^\dagger)^\dagger\\
&= V D^\dagger U^\dagger UD V^\dagger\\
&= V D^2 V^\dagger \\
&=\rho_W.
\end{align*}
\end{proof}
\noindent Figure \ref{fig:bestproof} captures this corollary in tensor diagram notation, where we have temporarily decorated the diagrams with arrows to show the flow of information.
\begin{figure}[h!]
\begin{tikzpicture}[y=3cm]
	\node at (0,0) {\includegraphics[scale=0.12]{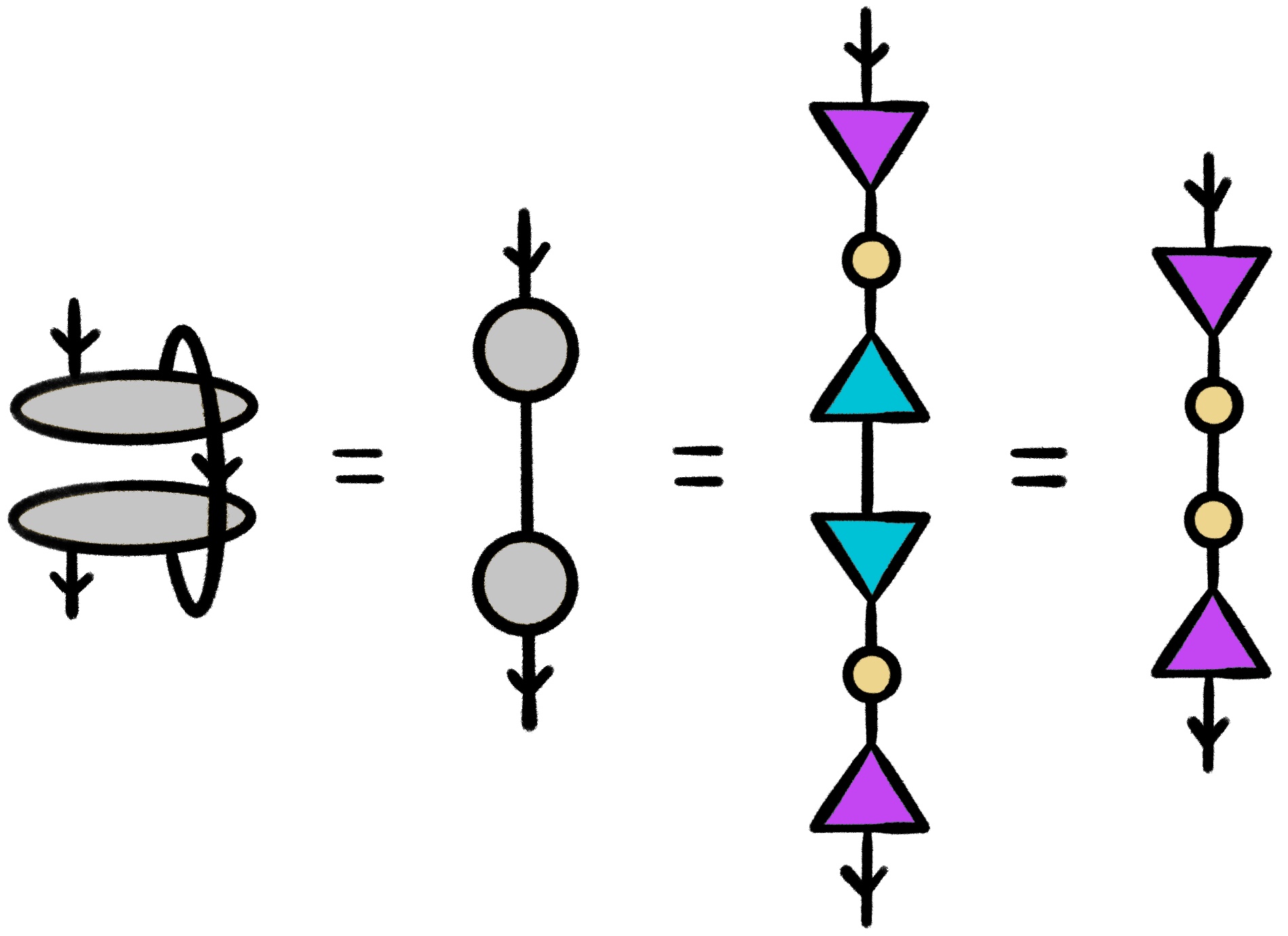}};
	\node at (-3.2,-1.1) {$\rho_V$};
	\node at (-.9,-1.1) {$M^\dagger M$};
	\node at (1.4,-1.1) {$U D^\dagger V^\dagger VD U^\dagger$};
	\node at (3.5,-1.1) {$U D^2 U^\dagger$};

	\node at (-3.2,-1.5) {$\rho_W$};
	\node at (-.9,-1.5) {$MM^\dagger$};
	\node at (1.4,-1.5) {$V D U^\dagger U D^\dagger V^\dagger$};
	\node at (3.5,-1.5) {$V D^2 V^\dagger$};
	\node at (0,-2.6) {\includegraphics[scale=0.12]{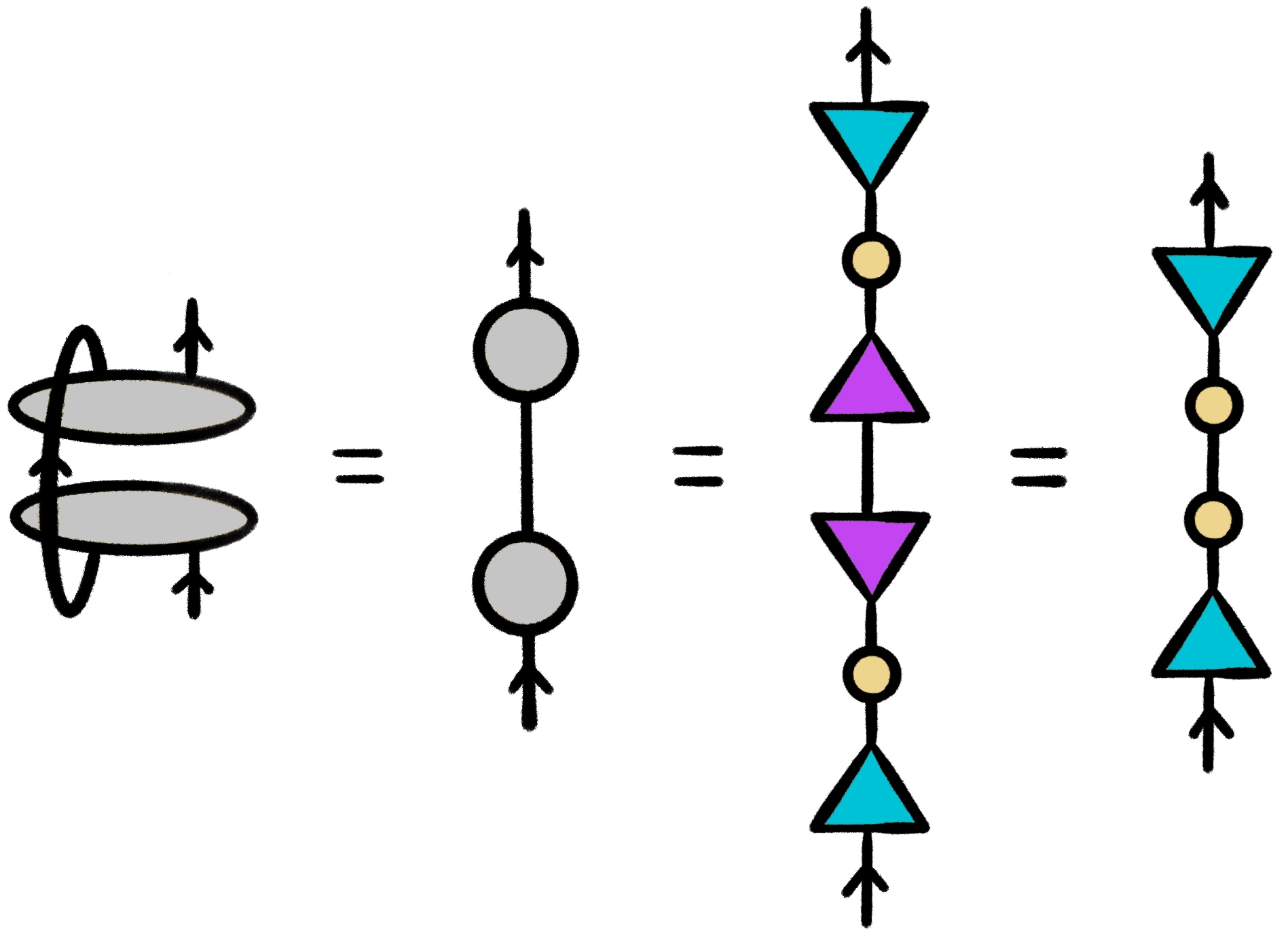}};
\end{tikzpicture}
\caption{Understanding the decompositions $\rho_V=M^\dagger M = UD^2 U^\dagger$ and $\rho_W=MM^\dagger = VD^2V^\dagger$ as tensor diagrams.}
\label{fig:bestproof}
\end{figure}
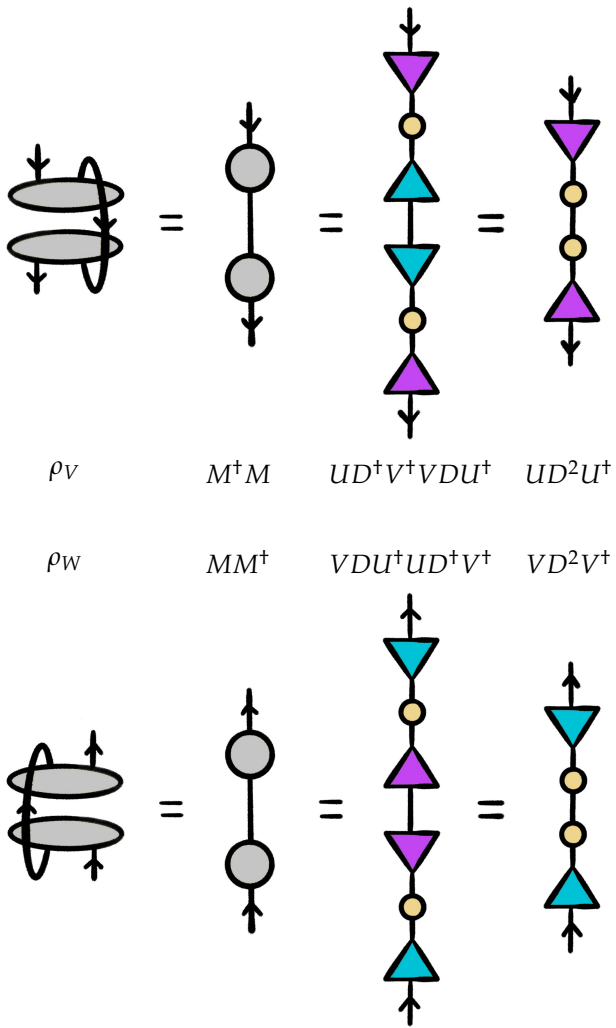
\newthought{So we get a lot} of mileage out of the singular value decomposition. Reduced densities of pure states have the same eigenvalues, they have the same number of eigenvectors, which are the singular vectors of a matrix $M$, and moreover their matrix representations are obtained as the composition of $M$ with its linear adjoint. Here's yet another result staring right at us: The original density operator $\rho=|\psi\>\<\psi|$ can be reconstructed from its reduced densities $\rho_V$ and $\rho_W$. In other words, if we have $\rho_V=UD^2U^\dagger$ and $\rho_W=VD^2V^\dagger$, then we can ``pick out'' $V$ and $D$ and $U$ and compose them to obtain $M=VDU^\dagger$ whose entries can be reassembled into the vector $|\psi\>$, as illustrated in Figure \ref{fig:reconstruct}. This may seem obvious, but imagine a scenario in which one does not have full access to the quantum state $\rho=|\psi\>\<\psi|$. Perhaps the system described by $\rho$ is too complicated to know precisely, or too large to store on a computer. It`s often easier to estimate the eigenvectors of $\rho_V$ and $\rho_W$---which are operators on a smaller, and hence more manageable, subsystem---perhaps discard those eigenvectors corresponding to the smallest eigenvalues, and then use this approximate spectral information to reconstruct the $\rho.$ This is precisely the approach we take in the machine learning problem of Chapter \ref{ch:application}. 
This section has collected a number of important tools. Here's the takeaway.
\begin{takeaway}\label{takeaway:densities}
Let $|\psi\>\in V\otimes W$ be a unit vector and consider the orthogonal projection operator $\rho=|\psi\>\<\psi|$. Then
	\begin{enumerate}
		\item the reduced density operators $\rho_V$ and $\rho_W$ have the same spectrum, and there is a one-to-one correspondence between their eigenvectors. (Proposition \ref{prop:evects})
	\end{enumerate}
Moreover, if $M=VDU^\dagger$ is the singular value decomposition of the $\dim(W)\times \dim(V)$ matrix corresponding to $|\psi\>$ then
	\begin{enumerate}\setcounter{enumi}{1}
	\item the columns of $U$ are the eigenvectors of $\rho_V$, the columns of $V$ are the eigenvectors of $\rho_W$, and the squares of the diagonal entries of $D$ are their common eigenvalues. (Equation \ref{eq:UDU_VDV})

	\item The reduced densities are obtained as the product of $M$ with its linear adjoint; that is, $\rho_V=M^\dagger M$ and $\rho_W= MM^\dagger.$ (Corollary \ref{cor:red_densities_as_MM})

	\item The original density $\rho$ may be reconstructed from the spectral decompositions of its reduced densities $\rho_V$ and $\rho_W$. (Figure \ref{fig:reconstruct})
	\end{enumerate}
\end{takeaway}

\begin{fullwidth}
\begin{figure}[h!]
\includegraphics[scale=0.2]{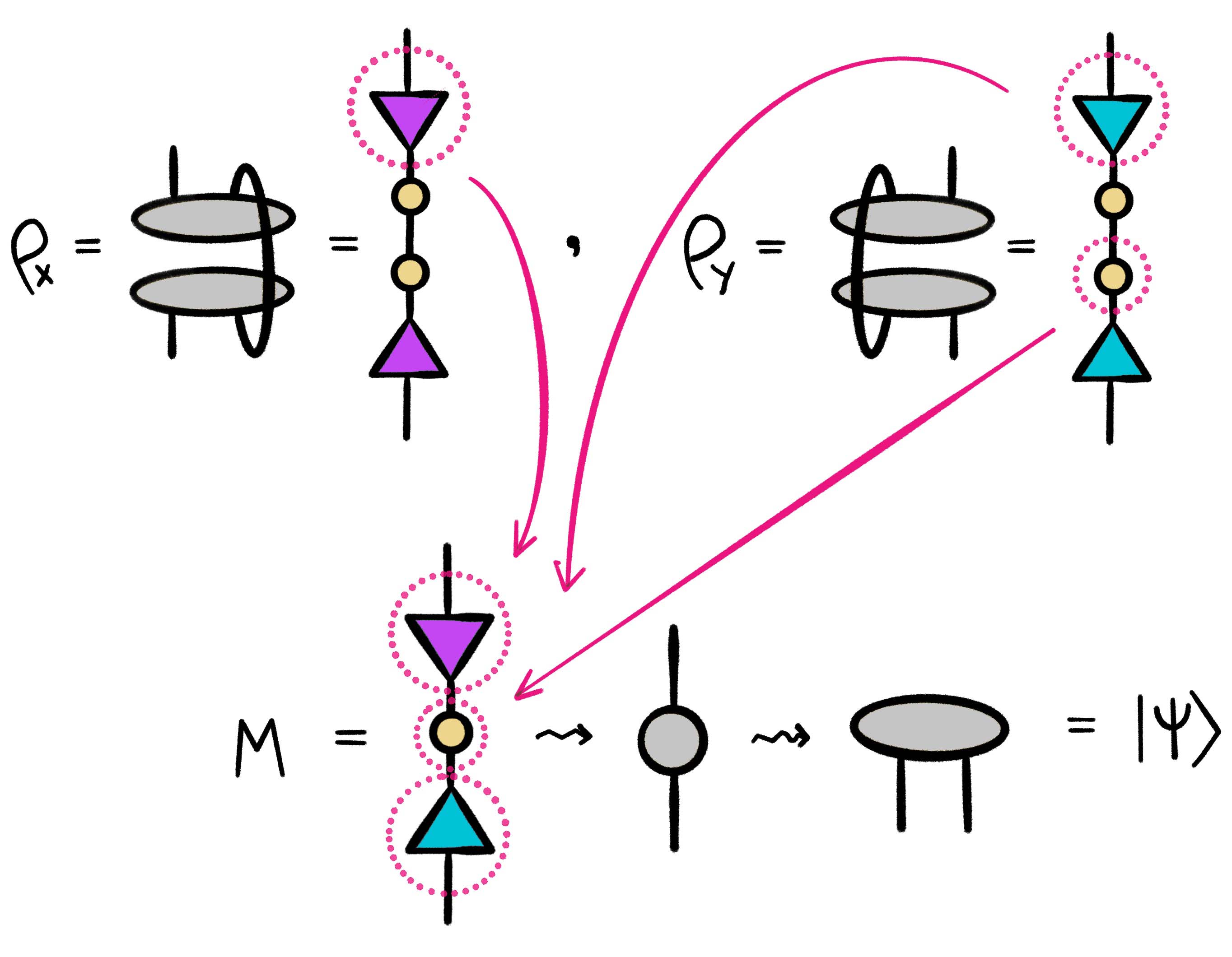}
\caption{Reconstructing $|\psi\>$ from the eigenvectors of $\rho_V$ and $\rho_W$ and their common eigenvalues.}
\label{fig:reconstruct}
\end{figure}
\end{fullwidth}
While on the topic of reduced densities, let's say a few brief words about the mathematics of entanglement.

\subsection{Entanglement}\label{ssec:entanglement}\index{entanglement}
In Section \ref{sec:density} it was noted that density operators come in different types---mixed and pure---the difference being measured by von Neumann entropy, all notions being mediated by rank. Like ``pure'' versus ``mixed,'' the notion of \textit{entanglement} is also related to rank, and the partial trace plays a key role.

\begin{definition}
A rank $1$ density operator is called an \emph{entangled quantum state}\index{quantum state!entangled} if the\marginnote{The use of ``the'' as in ``the rank of its reduced densities...'' is justified by Proposition \ref{prop:evects}. The reduced densities of a pure state always have the same rank!} rank of its reduced densities is greater than $1$, and it is not entangled otherwise.
\end{definition}

\noindent A related notion is that of entropy. The \emph{entanglement entropy}\index{entropy!entanglement} of a rank $1$ density is defined to be the Shannon entropy of the probability distribution defined by the (common) eigenvalues of its reduced densities, $-\sum_i\lambda_i\ln\lambda_i$. The entanglement entropy is zero precisely when the rank of the reduced densities is $1$; that is, when there's only one eigenvalue $\lambda=1.$ There's an equivalent  characterization of entanglement to know about. Let $|\psi\>\in V\otimes W$ be any unit vector and suppose $\rho=|\psi\>\<\psi|$ is the pure state given by orthogonal projection onto $|\psi\>$. Then $\rho$ is \textit{not} an entangled state---that is, its reduced densities have rank $1$---if and only if\sidenote[][-1cm]{\textit{Proof.} If $|\psi\>=|v\>\otimes|w\>$ for unit vectors $|v\>\in V$ and $|w\>\in W$ then the reduced densities of $|\psi\>\<\psi|$ are orthogonal projection onto $|v\>$ and $|w\>$, respectively, which both have rank $1$. Conversely, suppose $|\psi\>$ does not decompose as a simple tensor. Say, for example, $|\psi\>=|v\>\otimes|w\> + |v'\>\otimes|w'\>$ where $|v'\>$ is not a scalar multiple of $|v\>$ and $|w'\>$ is not a scalar multiple of $|w\>$. Then $\rho_V=|v\>\<v| + |v'\>\<v'|$, which has rank $2$, and similarly for $\rho_W.$}  $|\psi\>$ is of the form $|\psi\>=|v\>\otimes|w\>$ for some $|v\>\in V$ and $|w\>\in W$. Sometimes vectors of this form are called \emph{simple tensors}.\index{tensor!simple} The following proposition draws a simple connection between classical independence and the absence of entanglement.

\newpage
\begin{proposition}
Let $X$ and $Y$ be finite sets. If a joint probability distribution  $\pi\colon X\times Y\to \mathbb{R}$ satisfies
\[\pi(x,y)=\pi_X(x)\pi_Y(y)\]
for all $x\in X$ and $y\in Y$, then the orthogonal projection $\rho_\pi=|\psi\>\<\psi|$ is not entangled, where $|\psi\>=\sum_{x,y}\sqrt{\pi(x,y)}|xy\>$. Conversely, suppose $|a\>\in\mathbb{C}^X$ and $|b\>\in \mathbb{C}^Y$ are unit vectors and consider the orthogonal projection $\rho=|ab\>\<ab|$. Then the Born distribution induced by $\rho$ satisfies\sidenote{Equivalently, by Propostion \ref{prop:marginals}, we may express this in terms of the marginal distributions obtained from $\pi_\rho$; that is, \[\pi_\rho(x,y)=(\pi_\rho)_X(x)(\pi_{\rho})_Y(y).\]}
\[\pi_\rho(x,y)=\pi_{\rho_X}(x)\pi_{\rho_Y}(y)\]
where $\rho_X=\tr_Y\rho$ and $\rho_Y=\tr_X\rho$.
\end{proposition}
\begin{proof}
If a joint distribution satisfies $\pi(x,y)=\pi_X(x)\pi_Y(y)$, then the matrix representation for the reduced density $\rho_X=\tr_Y\rho_\pi$ is given by
\[\rho_X=\sum_y\sqrt{\pi(x,y)\pi(x',y)}=\sqrt{\pi_X(x)\pi_X(x')}\sum_y\pi_Y(y)=\sqrt{\pi_X(x)\pi_X(x')}\]
which shows that $\rho_X=\tr_Y\rho_\pi$ is orthogonal projection onto the vector $\sum_x\sqrt{\pi_X(x)}|x\>$. The reduced density $\rho_Y$ is obtained as a similar rank $1$ projection, and so $\rho_\pi$ not entangled. Conversely, consider the density $\rho=|ab\>\<ab|$ for unit vectors $|a\>\in\mathbb{C}^X$ and $|b\>\in\mathbb{C}^Y$. The Born distribution $\pi_\rho\colon X\times Y\to\mathbb{R}$ induced by $\rho$ is defined on any pair $(x,y)\in X\times Y$ by
\begin{align*}
\pi_\rho(x,y)&=\<xy|ab\>\<ab|xy\>\\[5pt]
&=\<x|a\>\<y|b\>\<a|x\>\<b|y\>\\[5pt]
&=|\<a|x\>|^2|\<b|y\>|^2\\[5pt]
&=\pi_{\rho_X}(x)\pi_{\rho_Y}(y).
\end{align*}
\end{proof}

\newthought{Either way} we think about it, let's be careful not to gloss over a couple of points. First, \textit{entanglement is only defined for density operators on a tensor product}. So if someone were to ask, ``Is my state $\rho$ on the inner product space $H$ entangled?'' the best response would be a question, ``Has a tensor decomposition been chosen for $H?$'' The point is that a state is said to be entangled \textit{with respect to} a chosen decomposition $H\cong V\otimes W.$ Second, \textit{we've only defined entanglement entropy for pure states}. If the rank of a density operator is greater than $1$, then its reduced densities may not have the same spectrum or the same rank, in which case the notion of entanglement entropy is not well-defined. 

For the work to come it will helpful to keep in mind this connection between \textit{entanglement} and the \textit{rank of reduced densities.} In fact, this will help to understand certain aspects of the application in Chapter \ref{ch:application}. Since the mathematics is fresh in our mind, let's take a quick behind-the-scenes look at the role that rank will play in the pages to come.

\newthought{Behind the Scenes.} Suppose $|\psi\>\in V\otimes W$ is a unit vector whose matrix has singular value decomposition  $VDU^\dagger$. In some cases (for instance, when searching for an optimal approximation to $|\psi\>$) it is advantageous to drop---that is, set equal to zero---those singular values that are smaller than some desired threshold. This truncation means we drop the corresponding eigenvectors of the reduced densities of $|\psi\>\<\psi|$, which therefore lowers the rank of the operators and hence decreases the entanglement associated to  $|\psi\>$. In tensor network diagram notation, this idea has the following illustration.
\begin{center}
\begin{tikzpicture}
\node at (0,0) {\includegraphics[scale=0.1]{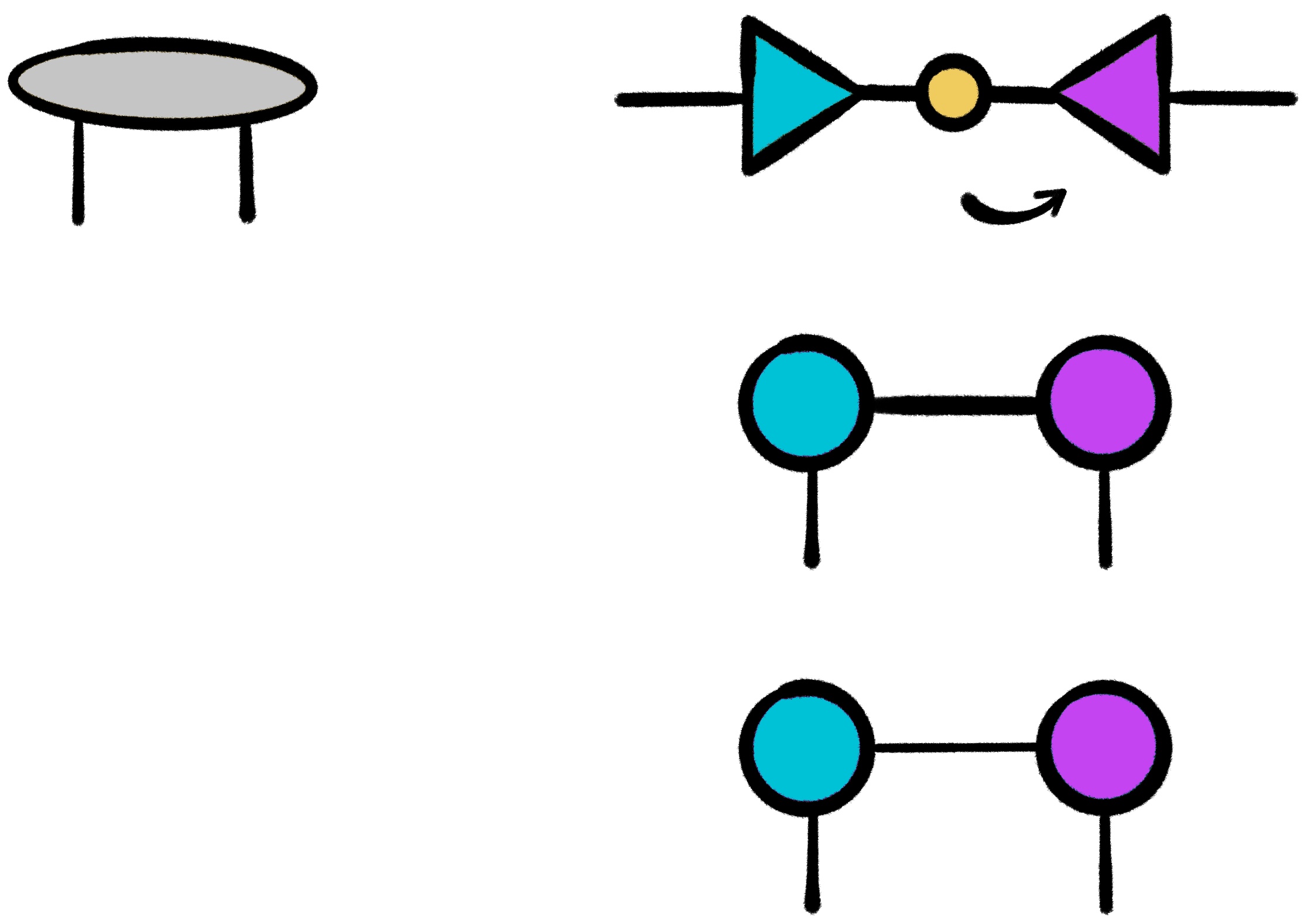}};
\node at (-4.6,2) {$|\psi\>=$};
\node at (1.8,2.8) {absorb};
\node at (1.8,.7) {\textbf{thick}};
\node at (1.8,-1.4) {thin};
\node at (-1,2) {$\overset{\text{SVD}}{\rightsquigarrow}$};
\node at (-1,.2) {$\rightsquigarrow$};
\node at (-1,-1.8) {$\overset{\text{truncate}}{\approx}$};
\end{tikzpicture}
\end{center}
In the second line, the matrix $D$ has been absorbed into $U^\dagger$ simply to keep the diagram clean. We've also drawn circles in lieu of triangles to suggest a connection with \textit{matrix product states.} (More on that below.) So the blue tensor is $V$ and the purple tensor is $DU^\dagger$. In the third line, we imagine dropping the smallest singular values. By doing so, the range of $U^\dagger$ becomes a smaller-dimensional space, which is represented by a thinner edge. We won't use a thick-edge/thin-edge convention going forward, but we introduce it here to share intuition. Think of this \textit{truncated SVD} as a technique for decreasing the flow of information. Intuitively, a thick edge (a large-dimensional space) allows more information to pass between the subsystems than does a thin edge (a small-dimensional space). Of course, \textit{no} information can pass between them if there's \textit{no} edge; that is, if $|\psi\>$ factors as a product (a ``simple tensor'') of two vectors. This brings to mind the descriptive quote at the beginning of this chapter, taken from the illuminating quantum max-flow/min-cut result of \cite{Cui2015}.
\begin{quote}
Networks transport classical things like power, water, oil, and cars. Tensor networks transport linear algebraic things like rank and entanglement and should be thought of as the quantum analogy.
\end{quote}
So we think of information about the two subsystems as being transported across the vector space represented by the edge joining the two nodes. How \textit{much} information can flow is determined by the dimension of that space, which is the rank of the operators.
\begin{center}
\begin{tikzpicture}
\node at (0,0) {\includegraphics[scale=0.1]{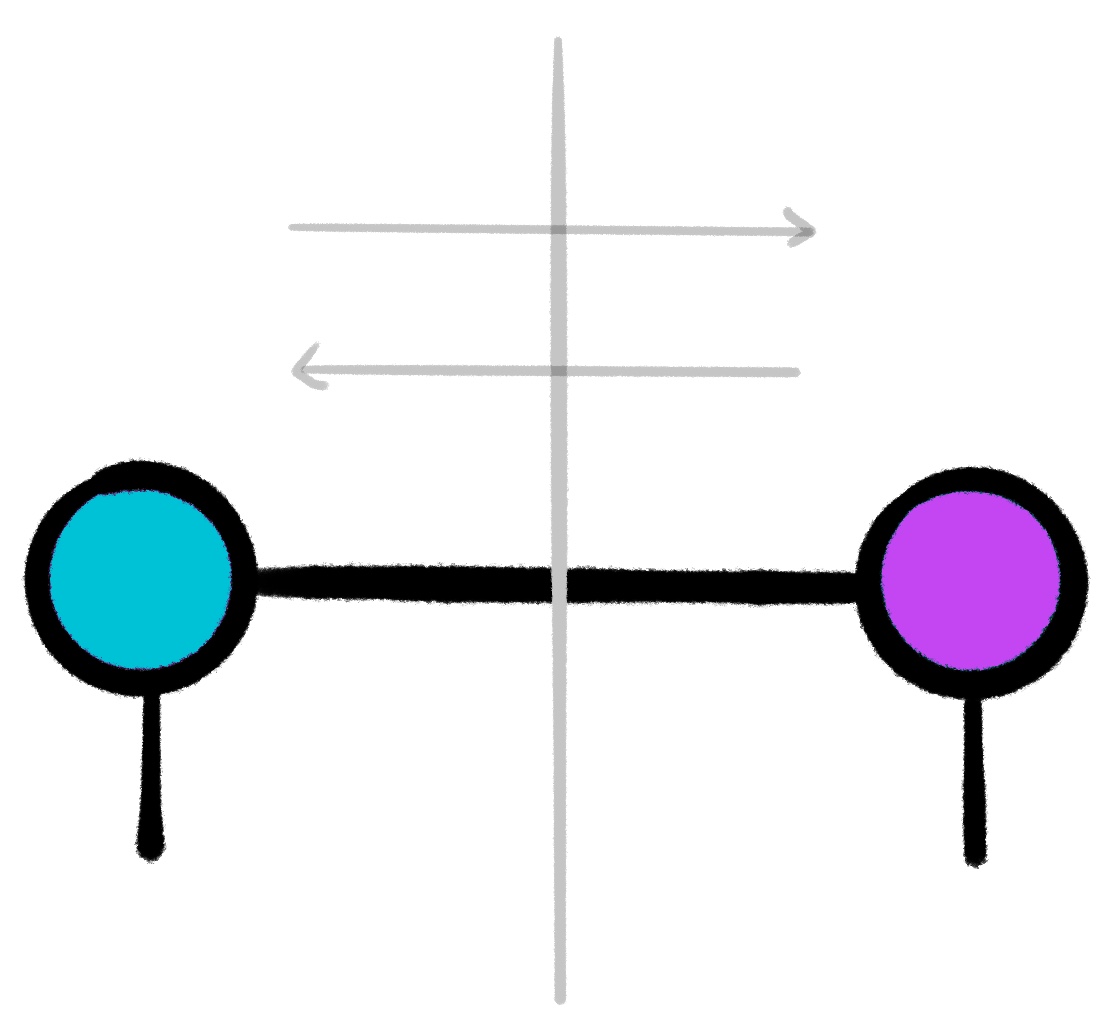}};
\node at (0,.75) {entanglement};
\node at (-1.6,-1.5) {left subsystem};
\node at (1.6,-1.5) {right subsystem};
\end{tikzpicture}
\end{center}
Think of the paradigm described here as a mini-version of how a matrix product state---introduced in Section \ref{ssec:TNs}---is formed. In general, one may start with a vector in a tensor product of \textit{many} vector spaces, $|\psi\>\in V^{\otimes N}.$ Any such vector corresponds to a linear map $V^{\otimes k}\to V^{\otimes {N-k}}$ for each $k<N$. A truncated SVD of that map gives us access to the most dominant eigenvectors of reduced density operators. One can imagine iteratively harnessing those eigenvectors for each $k$ into a tensor network with $N$ nodes. The SVD illustrated above is a special case of this idea.
\begin{center}
\begin{tikzpicture}
\node at (0,0) {\includegraphics[scale=0.1]{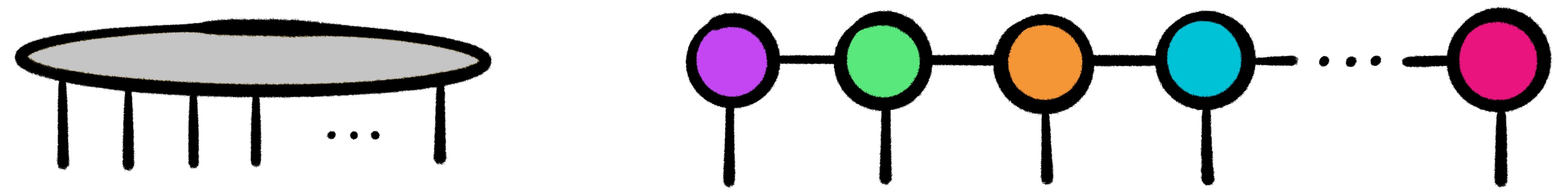}};
\node at (-1.3,0) {$\approx$};
\end{tikzpicture}
\end{center}
In fact, \textit{any} $|\psi\>$ can be factored \textit{exactly} by an MPS, just as any matrix can be factored \textit{exactly} by an SVD. (This is more than an analogy. It's a corollary.) But not all $|\psi\>$ can be approximated \textit{well} after dropping the lowest singular values at each step. That is, not all states $|\psi\>$ can be modeled well by a \textit{low rank, low entanglement} tensor. The application we describe in Chapter \ref{ch:application} applies to those $|\psi\>$ which do satisfy this hypothesis. But we'll cross that bridge when we get to it. Here's the takeaway for this brief discussion.

\newpage
\begin{takeaway}
A density operator $\rho$ on an inner product space $H$ is called
	\begin{itemize}
		\item a \emph{pure state} if the rank of $\rho$ is 1. If $H$ further decomposes as $H~\cong~V~\otimes~W$ then $\rho$ is said to be
			\begin{enumerate}
			\item \emph{entangled} if the rank of its reduced densities is greater than 1
			\item \emph{not entangled} if the rank of its reduced densities is equal to 1,
			\end{enumerate}
		\item a \emph{mixed state} if the rank of $\rho$ is greater than 1.
	\end{itemize}
\end{takeaway}

\newthought{In the next} chapter, we use the preliminaries accumulated in this chapter to revisit the idea that density operators on a tensor product are the quantum version of joint probability distributions, that their reduced density operators are the quantum version of marginal distributions, and that the partial trace is the quantum version of marginalizing. This analogy crystallizes when the density operator is orthogonal projection onto a vector $|\psi\>$ that arises from a classical probability distribution as in Equation (\ref{eq:psi}). We describe the details next.

\clearpage

%chapter 3
\chapter{A Passage from Classical to Quantum Probability}\label{ch:probability}
\epigraph{There does however seem to be some undefinable sense that a certain piece of mathematics is ``on to something,'' that it is a piece of a larger puzzle waiting to be explored further.}{Terence Tao \cite{tao2007good} }

\newthought{In this chapter,} we carefully outline the main procedure of this thesis by repackaging the theory in Chapter \ref{ch:preliminaries}. The starting ingredient is any classical joint probability distribution $\pi$, which defines a particular unit vector $|\psi\>$. The main advantage of $|\psi\>$ is that the spectral information of the reduced density operators of the rank $1$ density $|\psi\>\<\psi|$ organizes and encodes useful statistical information from the original distribution $\pi$. Think of the technique outlined here as a general recipe that can be applied to any finite set---so in particular, any \textit{dataset}. (We will revisit this applied perspective in Chapter \ref{ch:application}.) The procedure can be systematically described in a series of six steps.

\begin{enumerate}
	\item Start with any probability distribution $\pi$ on $X\times Y$.
	\item Represent the elements of $X$ and $Y$ as orthonormal basis vectors.
	\item Form the unit vector $|\psi\>=\sum_{x,y}\sqrt{\pi(x,y)}|x\>\otimes|y\>$ in $\mathbb{C}^X\otimes\mathbb{C}^Y$.
	\item Form the density operator $\rho=|\psi\>\<\psi|$.
	\item Compute the reduced density operators $\rho_X$ and $\rho_Y$ from $\rho.$
	\item Decode their spectral information.
\end{enumerate}
In Section \ref{sec:procedure} we will expand on each step in Section \ref{sec:procedure} and will illustrate with a simple example afterwards. The example will bring to light an especially simple scenario: when $\pi$ is an \textit{empirical} probability distribution, then the reduced densities of $|\psi\>\<\psi|$ can be computed from simple combinatorics of a bipartite graph associated to $\pi$. This will be explained in Section \ref{sec:graphs}. In Section \ref{sec:evectsFCA} we will revist the connection with formal concept analysis that was introduced early on in Section \ref{sec:ch1_FCA}. This chapter will close in Section \ref{sec:entailment} with some preliminary remarks on how this framework paves the way for modeling entailment and concept hierarchy in language using density operators.

\section{Reduced Densities from Classical Probability Distributions}\label{sec:procedure}
In this section, we will put the theory explaining the main procedure upfront and will illustrate with Example \ref{ex:main} afterwards. The example will rederive the elementary results in Section \ref{ssec:memory} using the language of quantum probabillity, thus explaining mathematics presented there. Here is the procedure for modeling any probability distribution as a pure quantum state and for decoding the information therein.

\bigskip
\noindent{\Large \itshape 1. Start with a joint probability distribution.}
\bigskip

\noindent Let $X$ and $Y$ be any finite sets and let $\pi\colon X\times Y\to\mathbb{R}$ be any probability distribution. Here we strongly emphasize the word \textit{any}. Although our language  is borrowed from quantum mechanics, it is not necessary to assume that $\pi$ has any quantum-like interpretation or properties.

\bigskip
\noindent{\Large \itshape 2. Represent the elements as orthonormal vectors.}
\bigskip

\noindent Choose an ordering the sets $X$ and $Y$, say $X=\{x_1,\ldots,x_n\}$ and $Y=\{y_1,\ldots,y_m\}$, and represent each element $x_i\in X$ as the $i$th standard basis vector
	\marginnote{The passage $x_i\mapsto|x_i\>$ might be called a ``one-hot encoding.''}
$|x_i\>\in\mathbb{C}^X\cong \mathbb{C}^n$, and identify each element $y_\alpha\in Y$ with the $\alpha$th standard basis vector $|y_\alpha\>\in\mathbb{C}^Y\cong\mathbb{C}^m$. Also consider the tensor product $\mathbb{C}^{X\times Y}\cong \mathbb{C}^X\otimes\mathbb{C}^Y$
	\marginnote{The isomorphism $\mathbb{C}^{X\times Y}\cong \mathbb{C}^X\otimes\mathbb{C}^Y$ is given by $|(x,y)\>\longleftrightarrow |x\>\otimes|y\>$.}
which is the $nm$-dimensional vector space whose bases are the tensor products $|x_i\>\otimes|y_\alpha\>$.

\medskip

\bigskip
\begin{fullwidth}
\noindent{\Large \itshape 3. Form the unit vector $|\psi\>$ defined in Equation (\ref{eq:psi}).}
\end{fullwidth}
\bigskip

\noindent To define the vector $|\psi\>$ in Equation (\ref{eq:psi}) we began with a set $S$. Here $S=X\times Y$ and so $|\psi\>\in\mathbb{C}^X\otimes\mathbb{C}^Y$ is the sum of all elements in $X\times Y$ weighted by the square roots of their probabilities.
\begin{equation}\label{eq:psi2}
|\psi\>=\sum_{i,\alpha}\sqrt{\pi(x_i,y_\alpha)}|x_i\>\otimes|y_\alpha\>
\end{equation}
As a tensor diagram, $|\psi\>$ is a node with two parallel edges.
\begin{center}
\begin{tikzpicture}
\node at (0,0) {\includegraphics[scale=0.12]{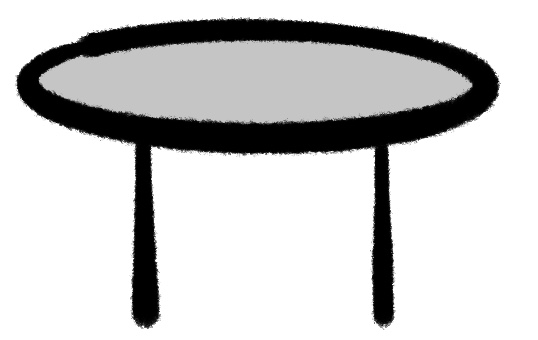}};
\node at (-1.5,0) {$|\psi\>=$};
\end{tikzpicture}
\end{center}

\bigskip
\begin{fullwidth}
\noindent{\Large \itshape 4. Form the orthogonal projection operator $|\psi\>\<\psi|$.}
\end{fullwidth}
\bigskip

\noindent Now we form the rank $1$ density operator associated to $|\psi\>$ as in Equation (\ref{eq:rho}). Computing $\rho=|\psi\>\<\psi|$ explicitly, one has
\begin{align*}
\rho &= |\psi\>\<\psi| \\
&= \left(\sum_{i,\alpha} \sqrt{\pi(x_i,y_\alpha)}|x_i\>\otimes|y_\beta\> \right)\left(\sum_{j,\beta} \sqrt{\pi(x_j,y_\beta)}\<x_j|\otimes\<y_\beta| \right)\\[10pt]
&=\sum_{\substack{ i,\alpha \\ j,\beta}} \sqrt{\pi(x_i,y_\alpha)\pi(x_j,y_\beta)}|x_i\>\<x_j|\otimes|y_\alpha\>\<y_\beta|. 
\end{align*}
Although the density $\rho$ carries no more information than $|\psi\>$ does, working with a linear operator allows us to call on the partial trace. This is the next step.
\begin{center}
\begin{tikzpicture}
\node at (0,0) {\includegraphics[scale=0.12]{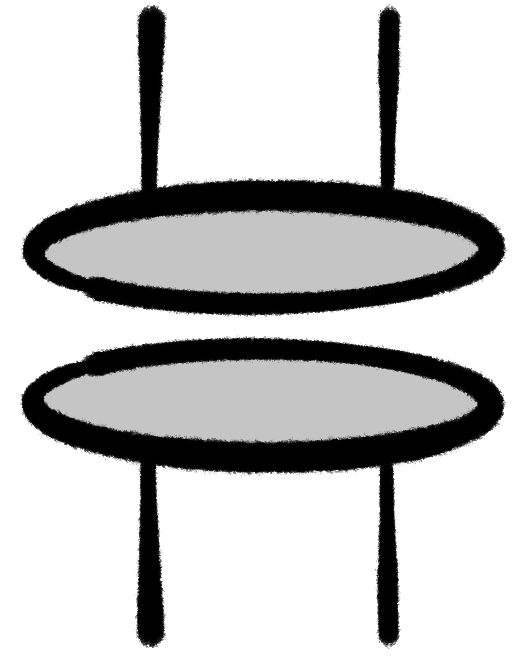}};
\node at (-1.5,0) {$\rho=$};
\end{tikzpicture}
\end{center}

\bigskip
\begin{fullwidth}
\noindent{\Large \itshape 5. Compute the reduced density operators associated to $|\psi\>\<\psi|.$}
\end{fullwidth}
\bigskip

\noindent Let's use the notation $\tr_Y:=\tr_{\mathbb{C}^Y}$ and $\tr_X:=\tr_{\mathbb{C}^X}$ and consider the reduced density operators associated to $\rho=|\psi\>\<\psi|$,
\begin{center}
\begin{tikzpicture}
\node at (0,0) {\includegraphics[scale=0.12]{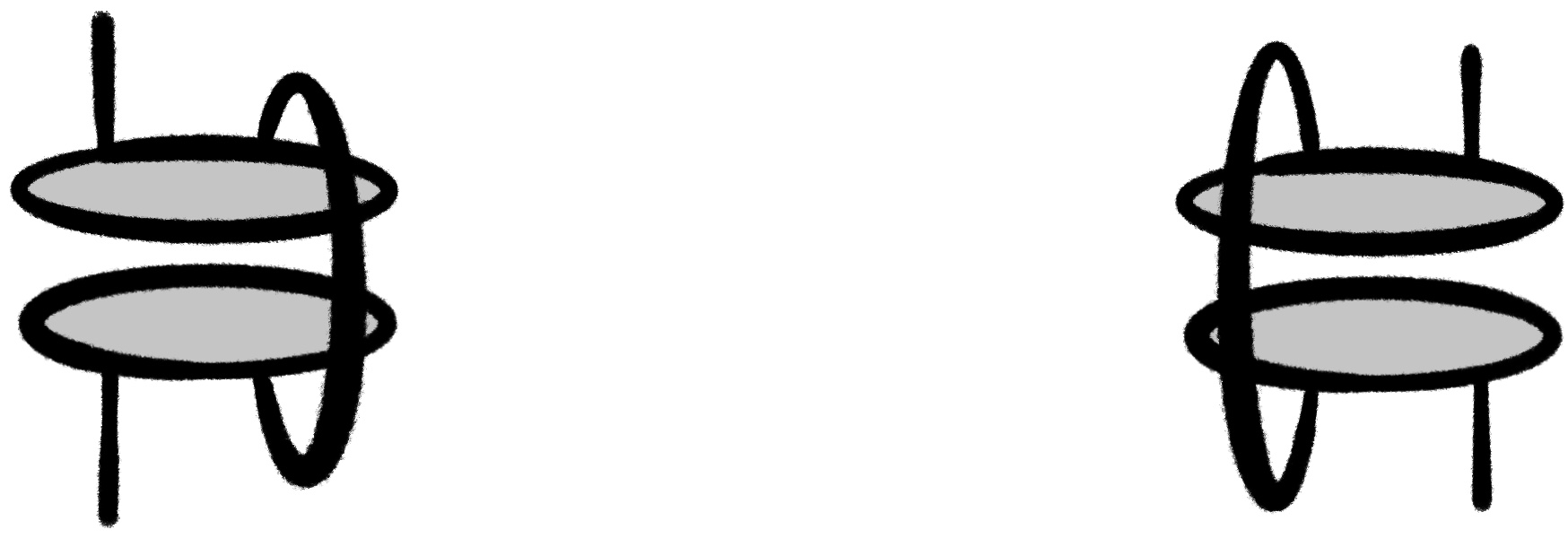}};
\node at (-5,0) {$\rho_X:=\tr_Y\rho=$};
\node at (.5,0) {$\rho_Y:=\tr_X\rho=$};
\end{tikzpicture}
\end{center}
We have are two ways to compute the reduced densities explicitly.\sidenote{In fact, we have at least \textit{four}. In addition to the two described here, there is also the operator sum decomposition in Proposition \ref{prop:povm}, as well as the combinatorial description given in Section \ref{sec:graphs} below.} The first is to recall from Corollary \ref{cor:red_densities_as_MM} that there exists a matrix $M$ so that
\[
\rho_X=M^\dagger M 
\qquad 
\rho_Y=MM^\dagger.
\]
The matrix $M$ is precisely the vector $|\psi\>\in \mathbb{C}^X\otimes \mathbb{C}^Y$ viewed as a linear map $\mathbb{C}^X\to\mathbb{C}^Y$. Explicitly, the entries of $M$ are the square roots of the probabilities,
\[
M_{\alpha i}:=\sqrt{\pi(x_i,y_\alpha)}.
\]
The second option is to apply the partial trace directly. Carefully understanding the latter is particularly helpful for understanding the theory, so let's proceed by this route. Applying the partial trace to $\rho$, the reduced density operator $\rho_X\colon \mathbb{C}^X\to\mathbb{C}^X$ has the following expression.
	\marginnote[2cm]{
	\[\rho_X = 
	\begin{bmatrix}
	\pi_X(x_1) & &  & * \\
	& \pi_X(x_2) & & \\
	& & \ddots &\\
	* & & & \pi_X(x_n)
	\end{bmatrix}
	\]
	}
\begin{align*}
\rho_X&=\tr_Y|\psi\>\<\psi| \\
&= \sum_{\substack{ i,\alpha \\ j,\beta }} \sqrt{\pi(x_i,y_\alpha)\pi(x_j,y_\beta)}\tr_Y(|x_i\>\<x_j|\otimes|y_\alpha\>\<y_\beta|) \\
&= \sum_{\substack{ i,\alpha \\ j,\beta }} \sqrt{\pi(x_i,y_\alpha)\pi(x_j,y_\beta)}|x_i\>\<x_j|\cdot  \tr|y_\alpha\>\<y_\beta| \\
&=\sum_{\substack{ i,j\\ \alpha }} \sqrt{\pi(x_i,y_\alpha)\pi(x_j,y_\alpha)}|x_i\>\<x_j|
\end{align*}
The last line follows since $\tr|y_\alpha\>\<y_\beta|=\<y_\alpha|y_\beta\>$ which is $1$ if $\alpha=\beta$ and is zero otherwise. The $ij$th entry of the operator $\rho_X$ is therefore
\begin{equation}\label{eq:rhox_offdiag}
(\rho_X)_{ij}=\sum_{\alpha}\sqrt{\pi(x_i,y_\alpha)\pi(x_j,y_\alpha)}.
\end{equation}
Notice that the $i$th diagonal entry is  $(\rho_X)_{ii}=\sum_\alpha\pi(x_i,y_\alpha)=\pi_X(x_i)$, and so $\rho_X$ recovers the classical marginal probability distribution $\pi_X\colon X\to\mathbb{R}$ along its diagonal. The off-diagonals of $\rho_X$ are also interesting. To explain, let's first adopt the following language.
\begin{definition} Given a pair $(x,y)\in X\times Y,$ call the element $x$ the \emph{prefix} of the pair and call the element $y$ the \emph{suffix} of the pair.
\end{definition}
\noindent Using this terminology, we are thinking of $\rho_X$ as an operator on the prefix subsystem. Its $ij$th off-diagonal entry then has a simple iterpretation---\textit{it measures the extent to which prefixes $x_i$ and $x_j$ share common suffixes in $Y$.} Indeed, the $ij$th off-diagonal of $\rho_X$ is zero precisely when either $\pi(x_i,y)=0$ or $\pi(x_j,y)=0$ for all $y$; that is, when $x_i$ and $x_j$ share no common suffixes in $Y$. On the other hand, the $ij$th entry receives a positive contribution for each suffix $y$ for which the probabilities of being paired with $x_i$ and $x_j$ are nonzero. In much the same way, the reduced density operator on the suffix subsystem $\rho_Y\colon\mathbb{C}^Y\to\mathbb{C}^Y$ is given by,
\begin{equation}
\rho_Y=\sum_{\substack{ \alpha,\beta \\ i }} \sqrt{\pi(x_i,y_\alpha)\pi(x_i,y_\beta)}|y_\alpha\>\<y_\beta|.
\end{equation}
The $\alpha\beta$th entry of the operator $\rho_Y$ is therefore
	\marginnote{
	\[\rho_Y = 
	\begin{bmatrix}
	\pi_Y(y_1) & &  & * \\
	& \pi_Y(y_2) & & \\
	& & \ddots &\\
	* & & & \pi_Y(y_m)
	\end{bmatrix}
	\]}
\begin{equation}\label{eq:rhoy_offdiag}
(\rho_Y)_{\alpha\beta}=\sum_{i}\sqrt{\pi(x_i,y_\alpha)\pi(x_i,y_\beta)}.
\end{equation}
In particular, the $\alpha$th diagonal entry is  $(\rho_Y)_{\alpha\alpha}=\sum_i\pi(x_i,y_\alpha)=\pi_Y(y_\alpha)$, and so we recover the classical marginal probability distribution $\pi_Y\colon Y\to\mathbb{R}$ along the diagonal of $\rho_Y$. Moreover, the $\alpha\beta$th off-diagonal entry measures the extent to which $y_\alpha$ and $y_\beta$ share common prefixes in $X$. 

\newthought{So the reduced} densities $\rho_X$ and $\rho_Y$ recover classical marginal probabilities along their diagonal. For this reason, we think of the reduced densities of $|\psi\>\<\psi|$ as the \textit{quantum version of marginal probability distributions}, and we think of the partial trace as the \textit{quantum version of marginalizing.} In general, the reduced densities will also have nonzero off-diagonal entries. These off diagonals encode information about subsystem interactions. From Equation \ref{eq:rhox_offdiag}, for instance, it is evident that the $ij$th off-diagonal of $\rho_X$ ``knows'' something about the complementary subsystem $Y$. This fact clearly distinguishes $\rho_X$ and $\rho_Y$ from their classical counterparts $\pi_X$ and $\pi_Y$, where all knowledge of the complementary subsystem is completely \textit{lost} when marginalizing. On the other hand, it is preserved when ``marginalizing'' with the partial trace. The preserved information is manifest in the nonzero off-diagonal entries of the reduced densities, the existence of which guarantees that $\rho_X$ and $\rho_Y$ will have interesting eigenvectors and eigenvalues. This spectral information therefore harnesses the subsystem interactions in a clean and organized fashion. Understanding this is the next and final step. But first, let us emphasize that the mathematics just described depends on our choice of $\rho=|\psi\>\<\psi|$ as orthogonal projection. If instead the original density was chosen to be the diagonal operator $\rho_{\text{diag}}=\sum_{i,\alpha}\pi(x_i,y_\alpha)|x_iy_\alpha\>\<x_iy_\alpha|$, then its reduced densities would \textit{themselves} be diagonal operators, with the marginal distributions along their diagonals and zeros elsewhere. The absence of off-diagonal entries would give us no information about subsystem interactions, and the eigenvectors of the reduced densities would simply recover the basis vectors $|x_i\>$ and $|y_\alpha\>$. This is the reason we prefer to model the joint distribution $\pi$ as a pure state $|\psi\>\<\psi|$ rather than a diagonal operator. Now, let's close with the last step.

\bigskip
\begin{fullwidth}
\noindent{\Large \itshape 6. Decode the spectral information.}
\end{fullwidth}
\bigskip

\noindent The reduced density operators $\rho_X$ and $\rho_Y$ of $\rho=|\psi\>\<\psi|$ recover the classical marginal distributions along their diagonals, and they contain extra information on their off-diagonals. We've interpreted this extra information as a measure of subsystem interactions, but there is a more principled way to both \textit{understand} this information and to \textit{harness it}. Indeed, the off diagonal entries directly contribute to the eigenvalues and eigenvectors of the reduced densities. To understand this, recall from Proposition \ref{prop:evects} that the reduced densities share the same eigenvalues $\lambda_i$ and there is a bijection between their sets of eigenvectors $\{|e_i\>\}\leftrightarrow \{|f_i\>\}$. So we write their spectral decompositions as
\[
\rho_X = \sum_{i=1}^m\lambda_i|e_i\>\<e_i| \qquad
\rho_Y = \sum_{i=1}^m\lambda_i|f_i\>\<f_i|.
\]
This leads to the following paradigm.
\begin{description}
		\marginnote[2cm]{Now is an opportune time to think back to the discussion on entanglement in Section \ref{ssec:entanglement}. As noted there, the notion of entanglement is defined by the (common) rank of the reduced densities $\rho_X$ and $\rho_Y$. Now it's hopefully clear why rank and entanglement go hand-in-hand. Knowing something about the prefix subsystem $X$---namely, that its state is $|e_i\>$---gives immediate knowledge about the suffix subsystem $Y$---namely, that its state is $|f_i\>$---and vice versa.}

	\item With probability $\lambda_i$ the prefix subsystem will be in the state defined by the $i$th eigenvector $|e_i\>$ of $\rho_X$, and the suffix subsystem will be in the state defined by the $i$th eigenvector $|f_i\>$ of $\rho_Y$. Here's what this means.

	\item Since the eigenvector $|e_i\>=\sum_{x\in X}e_i(x)|x\>$ may be written as a linear combination of the basis vectors $|x\>\in \mathbb{C}^X$, we obtain a probability distribution on the set of prefixes $X$. When the prefix subsystem is in the state $|e_i\>$, the probability of a prefix $x\in X$ is $\<x|e_i\>|^2=|e_i(x)|^2$.

	\item Likewise, the eigenvector $|f_i\>=\sum_{y\in Y}f_i(y)|y\>$ may be written as a linear combination of the basis vectors $|y\>\in \mathbb{C}^Y$. This defines a probability distribution on the set of suffixes $Y$. So when the suffix subsystem is in the state $|f_i\>$, the probability of a suffix $y\in Y$ is $|\<y|f_i\>|^2=|f_i(y)|^2$.
\end{description} 

\noindent This shows how to make sense of the individual eigenvectors and eigenvalues. But we may also wish to place this information in a broader context. Mathematically speaking, what really \textit{is} the information being encoded in the eigenvectors and eigenvalues of these reduced densities? The answer is simple. \textit{It must be akin to conditional probability.} The reason for this conclusion is found in a remark given towards the end of Section \ref{sec:red_density}. Recall from Figure \ref{fig:reconstruct} of that section that the original state $|\psi\>=\sum_{x,y}\sqrt{\pi(x,y)}|x\>\otimes|y\>$ may be fully reconstructed from the eigenvectors and eigenvalues of the reduced densities of $\rho=|\psi\>\<\psi|$.
	\marginnote[-3cm]{
	\includegraphics[width = \linewidth]{figures/ch2/reconstruct}
	}
Quite crucially, an analogous reconstruction \textit{does not exist} in classical probability theory. Consider the familiar equation relating joint probability with marginals and conditionals.
\begin{equation}
\pi(x,y)=\pi(x|y)\pi_Y(y)
\end{equation}
It is not possible to reconstruct the joint distribution given the marginal distributions $\pi_X$ and $\pi_Y$ alone. Conditional probability is needed as well. But this is not the case in the quantum version, where the state $\rho=|\psi\>\<\psi|$---and hence $\pi$ itself!---\textit{can} be fully reconstructed from the spectral decompositions of the reduced densities $\rho_X$ and $\rho_Y$. Therefore since $\rho_X$ and $\rho_Y$ contain the classical marginals along their diagonals \textit{as well as} extra information in their off diagonals, the extra information must account for conditional probability, which is the information lost in the classical process of marginalizing.
	\marginnote[-2cm]{The equation that is analogous to $\pi(x,y)=\pi(x|y)\pi_Y(y)$ is essentially the singular value decomposition
	\[M=VDU^\dagger,\] where $M$ is the matrix associated to $|\psi\>\in\mathbb{C}^X\otimes\mathbb{C}^Y$. In especially simple cases, the eigenvectors of $\rho_X$ and $\rho_Y$ (equivalently, the columns of $U$ and $V$) will define conditional probabilities on the nose. For instance, suppose $M\colon \mathbb{C}^X\to\mathbb{C}^Y$ has exactly one nonzero probability in each column, then look at its singular value decomposition.}
Let's illustrate the ideas so far with an extended example.

\subsection{The Motivating Example, Revisited}\label{ssec:revisit}
We'll now walk through each of the six steps of the main procedure, this time illustrating each step with an elementary example.
\begin{example}\label{ex:main}
We will revisit the opening example in Section \ref{sec:example}. There we began with a small toy corpus of text consisting of the following three phrases.
\begin{quote}
orange fruit \qquad green fruit \qquad purple vegetable
\end{quote}
Let's identify these phrases with a subset $T$ \marginnote{\smallcaps{Behind the Scenes.} Thinking ahead to Chapter \ref{ch:application}, the letter $T$ stands for ``training set.''} of the Cartesian product of the ordered sets  $X=\{\text{orange, green, purple}\}$ and $Y=\{\text{fruit, vegetable}\}$. The phrase $xy$ corresponds to the pair $(x,y)\in X\times Y$. Now we're in position to take the first step.

\bigskip
\begin{fullwidth}
\noindent{\Large \itshape 1. Start with a joint probability distribution.}
\end{fullwidth}
\bigskip

\noindent Consider the joint probability distribution $\pi\colon X\times Y\to\mathbb{R}$ defined by
	\marginnote{In Section \ref{sec:example} the joint distribution $\pi$ was depicted as a $2\times 3$ table.
	\[
	\begin{array}{r|ccc}
		& \text{orange} & \text{green} & \text{purple}\\[5pt]
	\hline\\[-8pt]
	\text{fruit} & \frac{1}{3} & \frac{1}{3} &0\\[5pt]
	\text{vegetable} & 0 & 0 & \frac{1}{3} \\
	\end{array}
	\]
	}
\begin{align*}
\pi(\text{orange, fruit}) &= 1/3\\
\pi(\text{green, fruit}) &= 1/3\\
\pi(\text{purple, vegetable}) &= 1/3
\end{align*}
with $\pi(x,y) = 0$ for all other pairs $(x,y)\in X\times Y$. As an aside, notice this probability distribution is an empirical one. The probability of a pair $(x,y)$ is simply the number of times it appears in the set $T$, divided by the size of $T$. Let's move on to the next step.

\newpage
\begin{fullwidth}
\noindent{\Large \itshape 2. Represent the elements as orthonormal vectors.}
\end{fullwidth}
\bigskip
The elements of $X=\{\text{orange, green, purple}\}$ and $Y=\{\text{fruit, vegetable}\}$ correspond to the following standard basis vectors in $\mathbb{C}^X\cong\mathbb{C}^3$ and in $\mathbb{C}^Y\cong\mathbb{C}^2$.
	\[
	\begin{array}{rrr}
	|\text{orange}\>= \begin{bmatrix}1\\0\\0 \end{bmatrix} &
	|\text{green}\> = \begin{bmatrix}0\\1\\0 \end{bmatrix} &
	|\text{purple}\>= \begin{bmatrix}0\\0\\1 \end{bmatrix}  \\[30pt]
	& |\text{fruit}\> = \begin{bmatrix}1\\0 \end{bmatrix} &
	|\text{vegetable}\> = \begin{bmatrix}0\\1 \end{bmatrix} 
	\end{array}
	\]
Moreover, for each $x\in X$ and $y\in Y$ the tensor product $|x\>\otimes|y\>$ is computed as an outer product. This gives a basis for the space $\mathbb{C}^X\otimes\mathbb{C}^Y$. For example, $|\text{orange}\>\otimes|\text{fruit}\>$ is given by the product of their vector representations.
\[
	|\text{orange}\>\otimes|\text{fruit}\>=
	\begin{bmatrix}1\\0\\0
	\end{bmatrix}
	\begin{bmatrix}1 & 0
	\end{bmatrix}=
	\begin{bmatrix} 1 &0\\ 0&0\\0&0
	\end{bmatrix}
\]
After stacking the two columns, this $3\times 2$ matrix reshapes as the column vector $\begin{bmatrix}1 &0&0&0&0&0 \end{bmatrix}^\top.$	Representations of the vectors $|\text{green}\>\otimes|\text{fruit}\>$ and $|\text{purple}\>\otimes|\text{vegetable}\>$ are obtained similarly, and so the three phrases in the subset $T\subseteq X\times Y$ correspond to the following three vectors in $\mathbb{C}^X\otimes\mathbb{C}^Y$.
\begin{align*}
|\text{orange}\>\otimes|\text{fruit}\> 
&= \begin{bmatrix}1 &0&0&0&0&0 \end{bmatrix}^\top\\
|\text{green}\>\otimes|\text{fruit}\>
&= \begin{bmatrix}0&1&0&0&0&0 \end{bmatrix}^\top \\
|\text{purple}\>\otimes|\text{vegetable}\>
&= \begin{bmatrix}0&0&0&0&0&1 \end{bmatrix}^\top
\end{align*}

\bigskip
\begin{fullwidth}
\noindent{\Large \itshape 3. Form the unit vector defined in Equation (\ref{eq:psi}).}
\end{fullwidth}
\bigskip
The unit vector defined by the joint probability distribution $\pi$ is a weighted sum of the three phrases in $T$.
\[
|\psi\>=\tfrac{1}{\sqrt 3}(|\text{orange}\>\otimes|\text{fruit}\>+ |\text{green}\>\otimes|\text{fruit}\> + |\text{purple}\>\otimes|\text{vegetable}\>)
\]
By identifying each term with its vector representation, we may represent $|\psi\>$ with the $6\times 1$ column vector,
\[
|\psi\>= \begin{bmatrix}\frac{1}{\sqrt 3} &\frac{1}{\sqrt 3} & 0 & 0 & 0 & \frac{1}{\sqrt 3} \end{bmatrix}^\top.
\]

\bigskip
\begin{fullwidth}
\noindent{\Large \itshape 4. Form the orthogonal projection operator $|\psi\>\<\psi|.$}
\end{fullwidth}
\bigskip
The explicit expression for orthogonal projection onto $|\psi\>\<\psi|$ consists of nine terms.
\begin{fullwidth}
\[
|\psi\>\<\psi|=\frac{1}{3}(|\text{orange}\>\<\text{orange}|\otimes |\text{fruit}\>\<\text{fruit}| + |\text{orange}\>\<\text{green}|\otimes |\text{fruit}\>\<\text{vegetable}| + \cdots)
\]
\end{fullwidth}
\smallskip
As written, it looks rather messy. It will be easier to read if we organize the terms into an array. Let's also adopt the following abbreviations.
\[
\begin{array}{lllll}
{\color{YellowOrange}o} =\text{orange} & {\color{Green}g}=\text{green} & {\color{Purple}p}=\text{purple} & F =\text{fruit} & V=\text{vegetable}
\end{array}
\]
Then the pure state $|\psi\>\<\psi|$ is as follows.
\begin{fullwidth}
\begin{equation}\label{eq:psipsi}
	|\psi\>\<\psi| = 
	\frac{1}{3}
	\left(
	\begin{array}{ccccc}
	|\o\>\<\o|\otimes |F\>\<F| &+& |\o\>\<\g|\otimes |F\>\<F| &+& |\o\>\<\p|\otimes |F\>\<V| \\[10pt]
	|\g\>\<\o|\otimes |F\>\<F| &+& |\g\>\<\g|\otimes |F\>\<F| &+& |\g\>\<\p|\otimes |F\>\<V| \\[10pt]
	|\p\>\<\o|\otimes |V\>\<F| &+& |\p\>\<\g|\otimes |V\>\<F| &+& |\p\>\<\p|\otimes |V\>\<V| \\
	\end{array}
	\right)
\end{equation}
\end{fullwidth}
As an array, $|\psi\>\<\psi|$ is computed as an outer product that results in a $6\times 6$ matrix.
	\[
	|\psi\>\<\psi| = 
	\begin{bmatrix}
	\frac{1}{3} & \frac{1}{3} & 0 & 0 & 0 & \frac{1}{3}\\[2pt]
	\frac{1}{3} & \frac{1}{3} & 0 & 0 & 0 & \frac{1}{3}\\[2pt]
	0&0&0&0&0&0\\[2pt]
	0&0&0&0&0&0\\[2pt]
	0&0&0&0&0&0\\[2pt]
	\frac{1}{3} & \frac{1}{3} & 0 & 0 & 0 & \frac{1}{3}
	\end{bmatrix}
	\]

\bigskip
\begin{fullwidth}
\noindent{\Large \itshape 5. Compute the reduced density operators associated to $|\psi\>\<\psi|$.}
\end{fullwidth}
\bigskip
Let's now compute the density operators $\rho_X$ and $\rho_Y$ by applying the partial trace to the projection operator in Equation (\ref{eq:psipsi}) above.
\[
\rho_X=\tr_Y|\psi\>\<\psi|\qquad
\rho_Y=\tr_X|\psi\>\<\psi|
\]
The idea is to compute $\tr_Y(|x_i\>\<x_j|\otimes|y_\alpha\>\<y_\beta|)=|x_i\>\<x_j|\<y_\alpha|y_\beta\>$ for each of the nine terms and add the result. Notice that only terms with $y_\alpha=y_\beta$ lend a contribution since $\<y_\alpha|y_\beta\>=0$ whenever $\alpha\neq \beta$. So the partial trace collects only those elements $x_i$ and $x_j$ that have common suffixes.\footnote{\smallcaps{Behind the Scenes.} This is the reason we choose to represent elements $x\in X$ and $y\in Y$ as orthonormal vectors, as opposed to some other representation. Otherwise this particular interpretation of the partial trace would not hold.} For example, among the three phrases
\begin{quote}
orange fruit \qquad green fruit \qquad purple vegetable
\end{quote}
the prefixes \textit{orange} and \textit{green} have exactly one shared suffix, namely \textit{fruit}, and so the partial trace returns $\tr_Y (|\o\>\<\g|\otimes |F\>\<F|)= |\o\>\<\g|\<F|F\>$ which is $|\o\>\<\g|$ since $\<F|F\>=1$. On the other hand, the words \textit{orange} and \textit{purple} have different suffixes, and so the partial trace returns $\tr_Y(|\o\>\<\p|\otimes |F\>\<V|) = |\o\>\<\p| \<F|V\>$ which is $0$ since $\<F|V\>=0$. Proceeding this way for each term comprising $|\psi\>\<\psi|$, the reduced density $\rho_X$ is as follows.
\begin{align*}
	\rho_X=\tr_Y|\psi\>\<\psi| &= 
	\frac{1}{3}
	\left(
	\begin{array}{ccccc}
	|\o\>\<\o|\otimes \<F|F\> &+& |\o\>\<\g|\otimes \<F|F\> &+& |\o\>\<\p|\otimes \cancel{\<F|V\>} \\[10pt]
	|\g\>\<\o|\otimes \<F|F\> &+& |\g\>\<\g|\otimes \<F|F\> &+& |\g\>\<\p|\otimes \cancel{\<F|V\>} \\[10pt]
	|\p\>\<\o|\otimes \cancel{\<V|F\>} &+& |\p\>\<\g|\otimes \cancel{\<V|F\>} &+& |\p\>\<\p|\otimes \<V|V\> \\
	\end{array}
	\right)\\[20pt]
	&=
	\frac{1}{3}(|\o\>\<\o|+ |\o\>\<\g| + |\g\>\<\o| + |\g\>\<\g| + |\p\>\<\p|)
\end{align*}
Each of the five operators in the previous line is computed as an outer product of standard basis vectors, so they have simple matrix representations,
\begin{equation*}
\begin{split}
	|\o\>\<\o| &= \begin{bsmallmatrix}1 & 0 &0\\0&0&0\\0&0&0\end{bsmallmatrix}\\[10pt]
	|\g\>\<\g| &= \begin{bsmallmatrix}0 & 0 &0\\0&1&0\\0&0&0\end{bsmallmatrix}
\end{split}\qquad\qquad
\begin{split}
	|\p\>\<\p| &= \begin{bsmallmatrix}0 & 0 &0\\0&0&0\\0&0&1\end{bsmallmatrix}\\[10pt]
	|\o\>\<\g| &= \begin{bsmallmatrix}0 & 1 &0\\0&0&0\\0&0&0\end{bsmallmatrix}
\end{split}
\end{equation*}
and so on. Therefore $\rho_X$ has the following matrix representation.
\begin{equation}\label{ex:rhox_example}
	\rho_X =
	\kbordermatrix{
    & \o & \g & \p\\
    \o & 1 & 1 & 0  \\[10pt]
    \g & 1 & 1 & 0  \\[10pt]
    \p & 0 & 0 & 1 \\
  	}\frac{1}{3}
\end{equation}
Notice the factor of $\frac{1}{3}$ on the right, which ensures that $\rho_X$ has unit trace. What's more, the diagonal of $\rho_X$ is the classical marginal probability distribution $\pi_X=\left(\frac{1}{3},\frac{1}{3},\frac{1}{3}\right)$ on $X$. Momentarily ignoring the normalization factor, we verify that the entries of $\rho_X$ have a combinatorial interpretation.

\begin{description}
	\item \textbf{Off diagonals count shared suffixes}. The $ij$th off-diagonal entry of $\rho_X$ counts the number of times the prefixes $x_i$ and $x_j$ share a common suffix in $Y$. For instance, \textit{orange} and \textit{green} have exactly one shared suffix---\textit{fruit}---and the corresponding off-diagonal entries are $1$. On the other hand, \textit{purple} shares no common suffixes with either \textit{orange} or \textit{green}, and the corresponding off-diagonal entries are zero. 

	\item \textbf{Diagonals count occurrences.} The $i$th diagonal entry counts the number of times $x_i$ appears in the set $T$. \textit{Orange, green,} and \textit{purple} each appear once, and so the diagonals are each $1$.
	\end{description}

\noindent Let's now compute the reduced density operator $\rho_Y=\tr_X|\psi\>\<\psi|\colon \mathbb{C}^Y\to\mathbb{C}^Y$, which can be understood similarly. When tracing out $\mathbb{C}^X$ from the nine terms of $|\psi\>\<\psi|$, only terms with $x_i=x_j$ will lend a contribution since $\tr_X(|x_i\>\<x_j|\otimes|y_\alpha\>\<y_\beta|)=\<x_i|x_j\>|y_\alpha\>\<y_\beta|$ and $\<x_i|x_j\>=0$ whenever $i\neq j$. Therefore the reduced density $\rho_Y$ is as follows.
\smallskip
\begin{align*}
	\rho_Y=\tr_X|\psi\>\<\psi| &= 
	\frac{1}{3}
	\left(
	\begin{array}{ccccc}
	\<\o|\o\>\otimes |F\>\<F| &+& \cancel{\<\o|\g\>}\otimes |F\>\<F| &+& \cancel{\<\o|\p\>}\otimes |F\>\<V| \\[10pt]
	\cancel{\<\g|\o\>}\otimes |F\>\<F| &+& \<\g|\g\>\otimes |F\>\<F| &+& \cancel{\<\g|\p\>}\otimes |F\>\<V| \\[10pt]
	\cancel{\<\p|\o\>}\otimes |V\>\<F| &+& \cancel{\<\p|\g\>}\otimes |V\>\<F| &+& \<\p|\p\>\otimes |V\>\<V| \\
	\end{array}
	\right)\\[20pt]
	&=
	\frac{1}{3}\left(|F\>\<F|+ |F\>\<F| + |V\>\<V|\right)
\end{align*}
Each operator in the previous line is computed as an outer product of standard basis vectors, so they have simple matrix representations.
\[
	|F\>\<F| = \begin{bmatrix}1 & 0 \\0&0\end{bmatrix}
	\qquad
	|V\>\<V| = \begin{bmatrix}0 & 0 \\0&1\end{bmatrix}
\]
Therefore $\rho_Y$ has the following matrix representation.
\begin{equation}\label{ex:rhoy_example}
	\rho_Y =
	\kbordermatrix{
	 & F & V\\
		F & 2 & 0  \\[10pt]
	 	V & 0 & 1  \\
  	}\frac{1}{3}
\end{equation}
Notice the factor of $\frac{1}{3}$ on the right, which ensures that $\rho_Y$ has unit trace. Moreover the diagonal of $\rho_Y$ is the classical marginal probability distribution $\pi_Y=\left(\frac{2}{3},\frac{1}{3}\right)$. Ignoring the normalization factor for the moment, we verify that the entries of $\rho_Y$ have a combinatorial interpretation.
	\begin{description}
	\item \textbf{Off diagonals count shared prefixes.} The $\alpha\beta$th off-diagonal entry counts the number of times the suffixes $y_\alpha$ and $y_\beta$ share a common prefix. In this example, \textit{fruit} and \textit{vegetable} are never described by the same color, and so both off diagonals are zero.

	\item \textbf{Diagonals count occurrences.} The $\alpha$th diagonal entry counts the number of times $y_\alpha$ appears in the corpus. \textit{Fruit} appears twice, and the corresponding entry is $2$. \textit{Vegetable} appears once, and the corresponding entry is $1$. 
	\end{description}
	\marginnote[-6cm]{As an aside, by Corollary \ref{cor:red_densities_as_MM}, the reduced densities may also be computed as a product of a matrix with its transpose. The vector $|\psi\>\in\mathbb{C}^X\otimes\mathbb{C}^Y$ computed in Step 3 corresponds to the following matrix, which represents a linear map $\mathbb{C}^X\to\mathbb{C}^Y$.
		\[
		M = 
		\begin{bmatrix}
		\frac{1}{\sqrt 3} & \frac{1}{\sqrt 3} & 0\\
		0 & 0 & \frac{1}{\sqrt 3}
		\end{bmatrix}
		\leftrightsquigarrow
		|\psi\>
		\]
	The product of $M$ with its adjoint gives us the reduced density operators of $\rho$.
		\begin{equation*}
			M^\dagger M = 
			\begin{bmatrix}
			\frac{1}{3} & \frac{1}{3} & 0\\[5pt]
			\frac{1}{3} & \frac{1}{3} & 0\\[5pt]
			0 & 0 & \frac{1}{3}\\[5pt]
			\end{bmatrix}
			=\rho_X
		\end{equation*}
		\begin{equation*}
			MM^\dagger=
			\begin{bmatrix}
			\frac{2}{3} & 0\\[5pt]
			0 & \frac{1}{3}
			\end{bmatrix} = \rho_Y
		\end{equation*}
	Either way, we've now recovered the matrices presented early on in the motivating example of Section \ref{ssec:memory}.}
So the reduced densities $\rho_X$ and $\rho_Y$ recover classical marginal probability along their diagonals, but they also carry additional information about interactions between the prefix and suffix subsystems. This information is harnessed neatly in their eigenvectors and eigenvalues. Decoding this information is the final step.

\bigskip
\begin{fullwidth}
\noindent{\Large \itshape 6. Decode the spectral information.}
\end{fullwidth}
\bigskip
Lastly, we decode the spectral information of the two reduced densities. Here are their eigenvalues and eigenvectors.
\begin{center}
   \begin{tabular}{@{}r|c|c@{}}
      %\toprule
      & \textbf{eigenvectors of $\rho_X$} & \textbf{eigenvectors of $\rho_Y$} \\
      \midrule
      $\lambda_1=\frac{2}{3}$ & $|e_1\>=\begin{bmatrix} \frac{1}{\sqrt{2}} \\[10pt] \frac{1}{\sqrt{2}}\\[10pt] 0  \end{bmatrix}\begin{matrix} \o \\[10pt] \g \\[10pt] \p  \end{matrix}$ & $|f_1\>=\begin{bmatrix} 1\\0\end{bmatrix} \begin{matrix} F\\V\end{matrix}$
      \\[30pt]
      \midrule
      $\lambda_2=\frac{1}{3}$ & $|e_2\>=\begin{bmatrix} 0 \\[10pt] 0 \\[10pt] 1  \end{bmatrix}\begin{matrix} \o \\[10pt] \g \\[10pt] \p  \end{matrix}$ & $|f_2\>=\begin{bmatrix} 0\\1\end{bmatrix}\begin{matrix} F\\V\end{matrix}$\\
     %\bottomrule
\end{tabular}
\end{center}
This has the following interpretation.
	\marginnote[-6cm]{
	\begin{center}
	\begin{tikzpicture}[y=.5cm,x=1.5cm]
	\node at (0,0) {\includegraphics[scale=0.15]{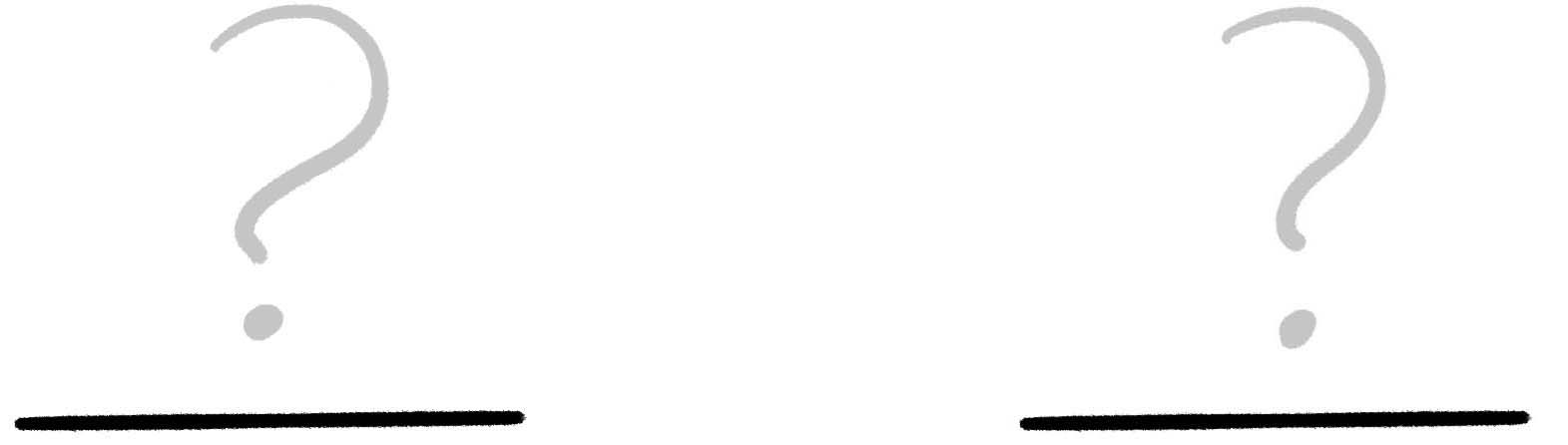}};
	\node at (-1,-2) {\textcolor{YellowOrange}{orange}};
	\node at (-1,-3) {\textcolor{Green}{green}};
	\node at (-1,-4) {\textcolor{Purple}{purple}};

	\node at (1,-2) {fruit};
	\node at (1,-3) {vegetable};
	\end{tikzpicture}
	\end{center}
	}
	\marginnote{
	\begin{center}
	\begin{tabular}{@{}c|c@{}}
	\textbf{66\%} & \textbf{33\%}\\
	\toprule
	\begin{tikzcd}
	\text{fruit} \arrow[d, "50\%"', bend right, shift right=3] \arrow[dd, "50\%", bend left=49, shift left=2] \\
	\text{\color{YellowOrange}orange}\\
	\text{\color{Green}green}                 
	\end{tikzcd}
	&
	\begin{tikzcd}
	\text{vegetable} \ar[dd,"100\%"]\\ \\
	\text{\color{Purple}purple}
	\end{tikzcd}
	\end{tabular}
	\end{center}
	}
\begin{itemize}
	\item With probability $\lambda_1=\frac{2}{3}$, the state of prefix subsystem is given by $|e_1\>$, and the state of the suffix subsystem is given by $|f_1\>$. The absolute squares of the entries of $|e_1\>$ define a probability distribution on $X.$ For instance, the probability of the prefixes \textit{orange} and \textit{green} are both $|\<e_1|\text{orange}\>|^2=|\<e_1|\text{green}\>|^2=\frac{1}{2}$, which are conditional probabilities given that the suffix of a phrase is \textit{fruit}.  Likewise, the absolute squares of the entries of $|f_1\>$ define a probability distribution on $Y$. For instance, the probability of the suffix \textit{fruit} is $|\<f_1|\text{fruit}\>|^2=1$, which is a conditional probability given that the prefix of a phrase is either \textit{orange} or \textit{green}.
		\marginnote{Compare this to Example \ref{ex:interpretation}, which was essentially this example without context.}
	\item With probability $\lambda_2=\frac{1}{3}$, the state of the prefix subsystem given by $|e_2\>$, and the state of the suffix subsystem is given by $|f_2\>$. The entries of $|e_2\>$ define a probability distribution on $X.$ The probability of the prefix \textit{purple} is $|\<e_2|\text{purple}\>|^2=1$, which is the probability of the prefix of a phrase being purple given that the suffix is \textit{vegetable}.  Likewise, the entries of $|f_2\>$ define a probability distribution on $Y$. The probability of the suffix \textit{vegetable} is $|\<f_2|\text{vegetable}\>|^2=1$, which is a conditional probability given that the prefix of a phrase is \textit{purple}.
  \end{itemize}
\end{example}

\newthought{This example concludes} a description of the main passage from classical to quantum probability $\pi\mapsto|\psi\>\<\psi|$. The punchline is that the eigenvalues and eigenvectors of the reduced densities of $|\psi\>\<\psi|$ recover classical marginal probabilities and also harness extra information that is akin to conditional probability. For the application to come, it will be good to know about an especially simple case that makes the theory even more interpretable. If the initial joint probability distribution $\pi\colon X\times Y\to\mathbb{R}$ is an empirical one, then it may be represented as a bipartite graph, and the entries of the reduced densities may be read off from the combinatorics of the graph. We will rely on this combinatorial understanding in the application of Chapter \ref{ch:application}, so let's collect the ideas here.

\section{Reduced Densities from Empirical Distributions}\label{sec:graphs}
We now take a quick look at the main procedure of Section \ref{sec:procedure} in the special case when the joint probability distribution is an empirical one. As before, let $X$ and $Y$ be finite sets and now choose a subset $T\subseteq X\times Y$. This defines an empirical distribution $\hat\pi\colon X\times Y\to\mathbb{R}$ by
\[
\hat\pi(x,y)= 
\begin{cases}
\frac{1}{|T|} & \text{if $(x,y)\in T$}\\[5pt]
0 &\text{otherwise}.
\end{cases}
\]
This probability distribution may be illustrated as a weighted bipartite graph. The sets $X$ and $Y$ provide the two sets of vertices, and an edge connecting $x\in X$ and $y\in Y$ is labeled by the probability $\hat\pi(x,y).$  To keep the graph clean, let's omit the edge between $x$ and $y$ whenever $\hat\pi(x,y)=0$ and let's always leave off the weights, knowing that the probability  associated to any edge is simply $1$ divided by the total number of edges in the graph.
\begin{center}
\includegraphics[scale=0.1]{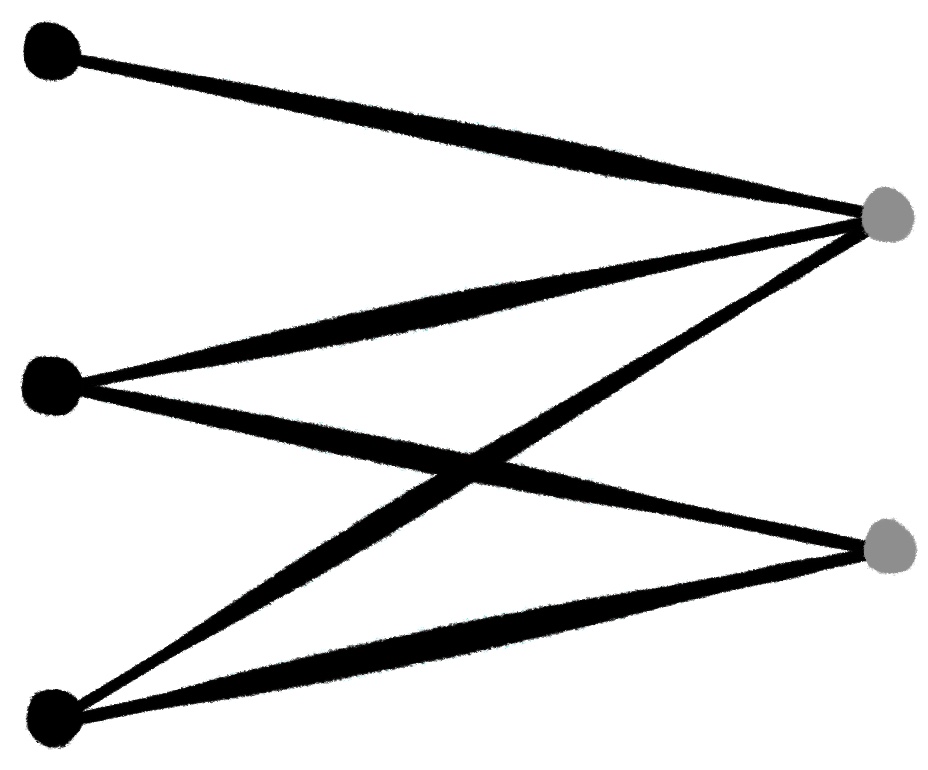}
\end{center}
The distribution $\hat\pi$ then defines the unit vector as in Equation (\ref{eq:psi2}).
\begin{equation}\label{eq:psi_emp}
|\psi\> = \frac{1}{\sqrt {|T|}}\sum_{(x_i,y_\alpha)\in T}|x_i\>\otimes|y_\alpha\>
\end{equation}
In this case, the entries of the reduced density matrices associated to $|\psi\>\<\psi|$ have especially clean, combinatorial interpretations. En route to understanding this, define the following nonnegative integers.
\begin{align*}
d(x_i,x_j)&:= |\{y\in Y\mid (x_i,y)\in T \text{ and } (x_j,y)\in T\}|\\[5pt]
d(y_\alpha,y_\beta)&:= |\{x\in X\mid (x,y_\alpha)\in T \text{ and } (x,y_\beta)\in T\}|
\end{align*}
In words, $d(x_i,x_j)$ is the number of suffixes that $x_i$ and $x_j$ have in common in $T$. In terms of the graph associated to $\hat\pi,$ it is the number of paths of length two between $x_i$ and $x_j$. If $i=j$, then $d(x_i,x_i)$ is simply the degree of vertex $x_i$.  Likewise, $d(y_\alpha,y_\beta)$ is the number prefixes that $y_\alpha$ and $y_\beta$ have in common in $T$. It is also the number of paths of length two between $y_\alpha$ and $y_\beta$ in the graph associated to $\hat\pi$. If $\alpha=\beta$ then $d(y_\alpha,y_\alpha)$ is the degree of vertex $y_\alpha$.  In this simple setup, the entries of the reduced density operators as defined in Equations (\ref{eq:rhox_offdiag}) and (\ref{eq:rhoy_offdiag}) specialize to
\begin{align}
(\rho_X)_{ij} &=\frac{d(x_i,x_j)}{|T|}\\[10pt]
(\rho_Y)_{\alpha\beta} &=\frac{d(y_\alpha,y_\beta)}{|T|}.
\end{align}

\noindent As a result, we have the following general combinatorial interpretation of the reduced densities, as seen in Example \ref{ex:main}.
	\begin{description}
	\item \textbf{Off diagonals count shared prefixes.} Up to the normalizing factor $1/|T|$, the $ij$th off-diagonal entry of $\rho_X$ counts the number of times $x_i$ and $x_j$ share a common suffix in $T$. The $\alpha\beta$th off diagonal entry of $\rho_Y$ counts the number of times $y_\alpha$ and $y_\beta$ share a common prefix in $T$. 

	\item \textbf{Diagonals count occurrences.} Up to the normalizing factor $1/|T|$, the $i$th diagonal of $\rho_X$ counts the number of times $x_i$ occurs as a prefix in $T$. The $\alpha$th diagonal entry of $\rho_Y$ counts the number of times $y_\alpha$ appears as a suffix in $T$.
	\end{description}

\noindent The upshot is that the entries of the reduced densities obtained from an empirical distribution may be read off rather quickly---just look at the graph. We will use this tactic when peering into the application of Chapter \ref{ch:application}, happily following Gromov's advice, ``If you don't understand, count!'' \cite{gromov2019}.

\begin{example}\label{ex:graph} Suppose $X$ is a set with three elements and $Y$ is a set with two elements and consider the following graph with 5 edges.
\begin{center}
\begin{tikzpicture}[y=1.2cm]
\node at (0,0) {\includegraphics[scale=0.1]{figures/ch2/graph3}};
\node at (-2,1) {$x_1$};
\node at (-2,0) {$x_2$};
\node at (-2,-1) {$x_3$};

\node at (2,.5) {$y_1$};
\node at (2,-.5) {$y_2$};
\end{tikzpicture}
\end{center}
Let $T\subseteq X\times Y$ be the set consisting of all pairs $(x,y)$ for which there is an edge joining $x$ and $y$. Then the reduced densities of the state $|\psi\>=\displaystyle \frac{1}{\sqrt 5}\sum_{(x_i,y_\alpha)\in T}|x_i\>\otimes|y_\alpha\>$ associated to this graph are
\begin{equation*}
	\rho_X =
	\kbordermatrix{
    & x_1 & x_2 & x_3\\
    x_1 & 1 & 1 & 1  \\[10pt]
    x_2 & 1 & 2 & 2  \\[10pt]
    x_3 & 1 & 2 & 2 \\
  	}\frac{1}{5}
  	\qquad\qquad
  	\rho_Y =
	\kbordermatrix{
	 & y_1 & y_2\\
		y_1 & 3 & 2  \\[10pt]
	 	y_2 & 2 & 2  \\
  	}\frac{1}{5}
\end{equation*}
To see this, take a look at $\rho_Y$, for the moment. Ignoring the factor of $\frac{1}{5}$, the diagonal entries $3$ and $2$ are the degrees of the vertices in $Y$.
	\begin{center}
	\begin{tikzpicture}[y=1.2cm]
	\node at (0,0) {\includegraphics[scale=0.1]{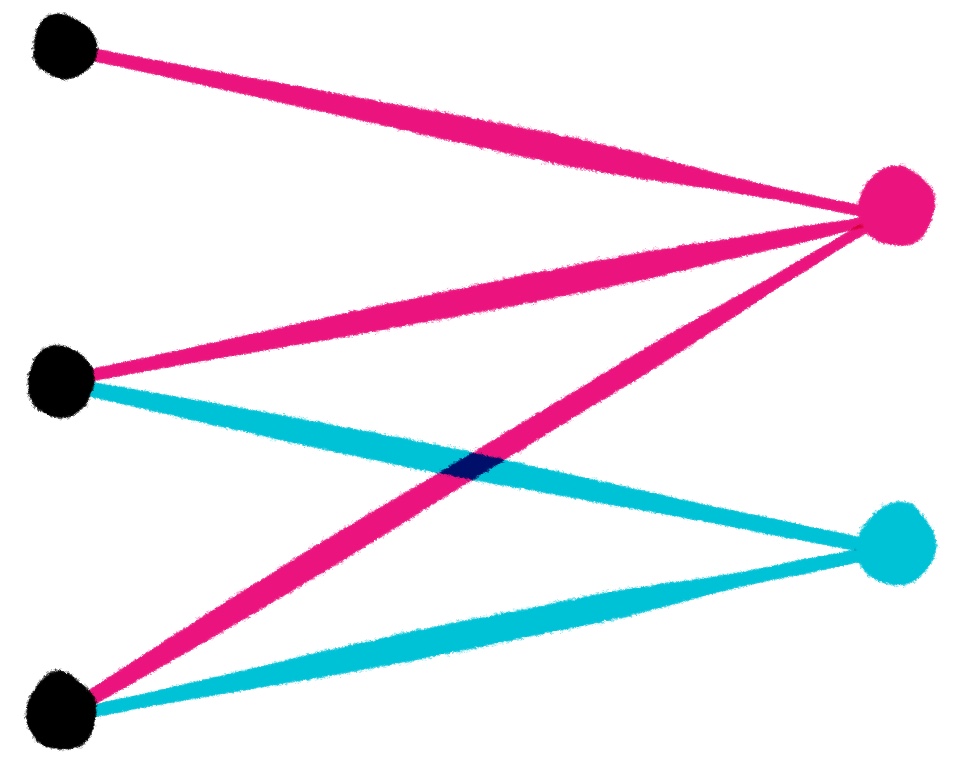}};
	\node at (2.5,.5) {degree $3$};
	\node at (2.5,-.5) {degree $2$};

	\node at (5,0) {$\begin{bmatrix}{\color{RubineRed}3}&2\\2 &{\color{cyan}2}\end{bmatrix}$};
	\end{tikzpicture}
	\end{center}
The off-diagonal entry of $\rho_Y$ is $2$, which is the number of paths of length two that connect the elements of $Y.$ 
	\begin{center}
	\begin{tikzpicture}[y=1.2cm]
	\node at (0,0) {\includegraphics[scale=0.1]{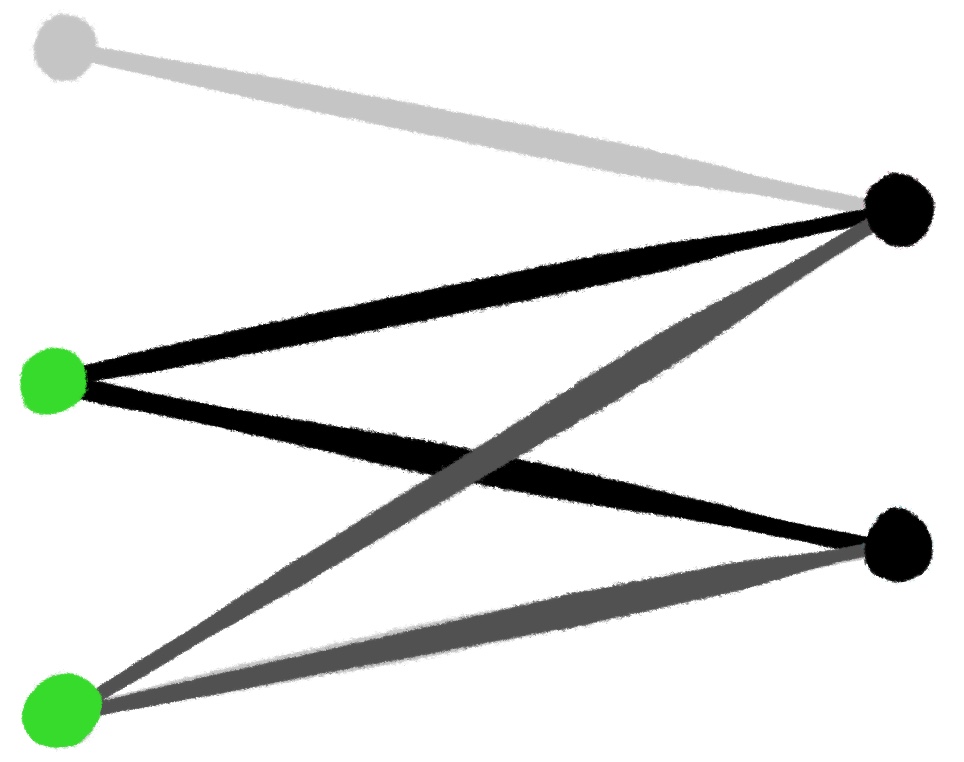}};
	\node at (-2.5,0) {degree $2$};
	\node at (-2.5,-1) {degree $2$};

	\node at (3.5,0) {$\begin{bmatrix}3&{\color{LimeGreen}2}\\{\color{LimeGreen}2} &2\end{bmatrix}$};
	\end{tikzpicture}
	\end{center}
The entries for $\rho_X$ are understood just as easily. Ignoring the factor of $\frac{1}{5}$, the diagonal entries $1,2,2$ are the degrees of the vertices in $X$.
	\begin{center}
	\begin{tikzpicture}[y=1.2cm]
	\node at (0,0) {\includegraphics[scale=0.1]{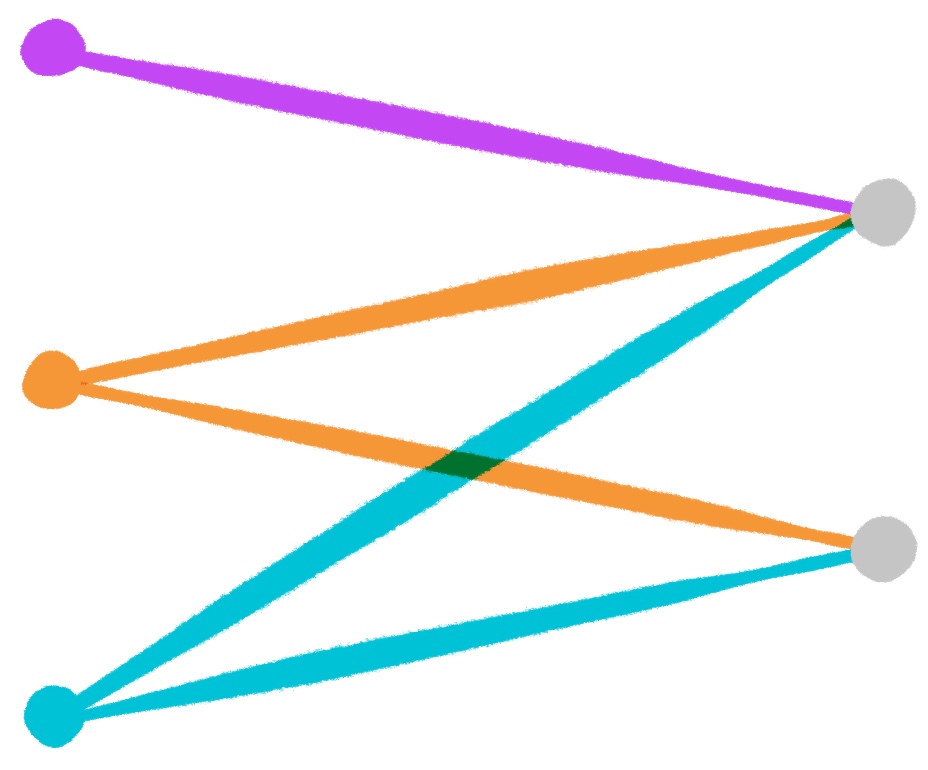}};
	\node at (-2.5,1) {degree $1$};
	\node at (-2.5,0) {degree $2$};
	\node at (-2.5,-1) {degree $2$};

	\node at (4,0) {$\begin{bmatrix}{\color{Purple}1}&1&1\\1&{\color{YellowOrange}2}&2\\1&2&{\color{cyan}2}\end{bmatrix}$};
	\end{tikzpicture}
	\end{center}
The off diagonals of $\rho_X$ are $1$ and $1$ and $2$. The first off diagonal counts the single path of length two that joins $x_1$ and $x_2$.
	\begin{center}
	\begin{tikzpicture}[y=1.2cm]
	\node at (0,0) {\includegraphics[scale=0.1]{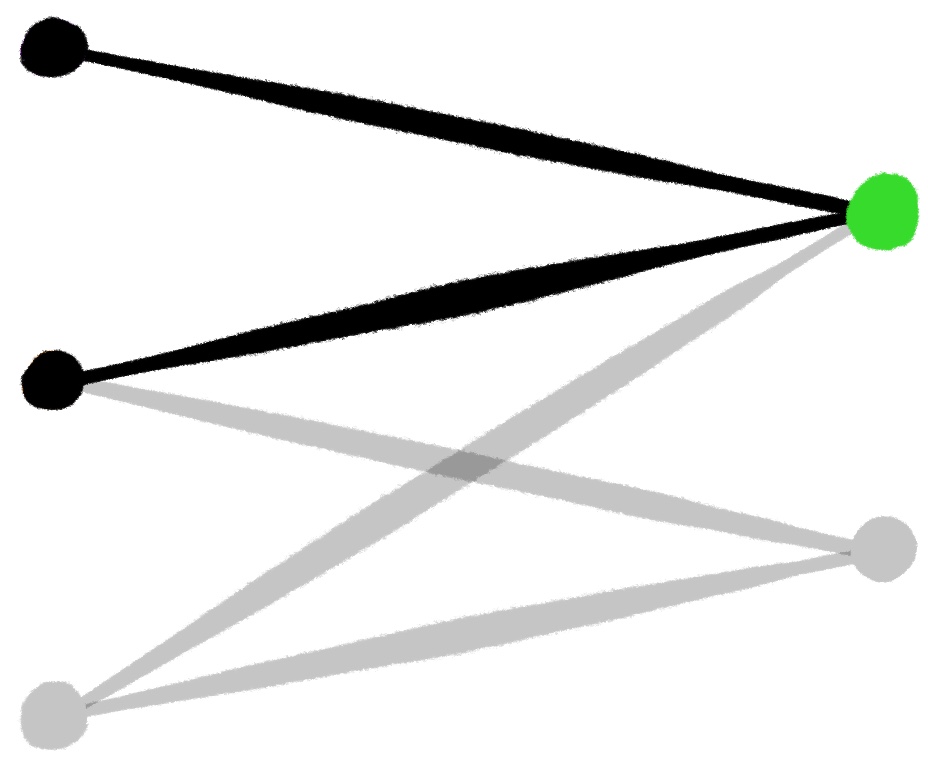}};
	\node at (2.5,.5) {degree $2$};

	\node at (5,0) {$\begin{bmatrix}1&{\color{LimeGreen}1}&1\\{\color{LimeGreen}1}&2&2\\1&2&2\end{bmatrix}$};
	\end{tikzpicture}
	\end{center}
The third off diagonal counts the two paths of length two that join $x_2$ and $x_3$.
	\begin{center}
	\begin{tikzpicture}[y=1.2cm]
	\node at (0,0) {\includegraphics[scale=0.1]{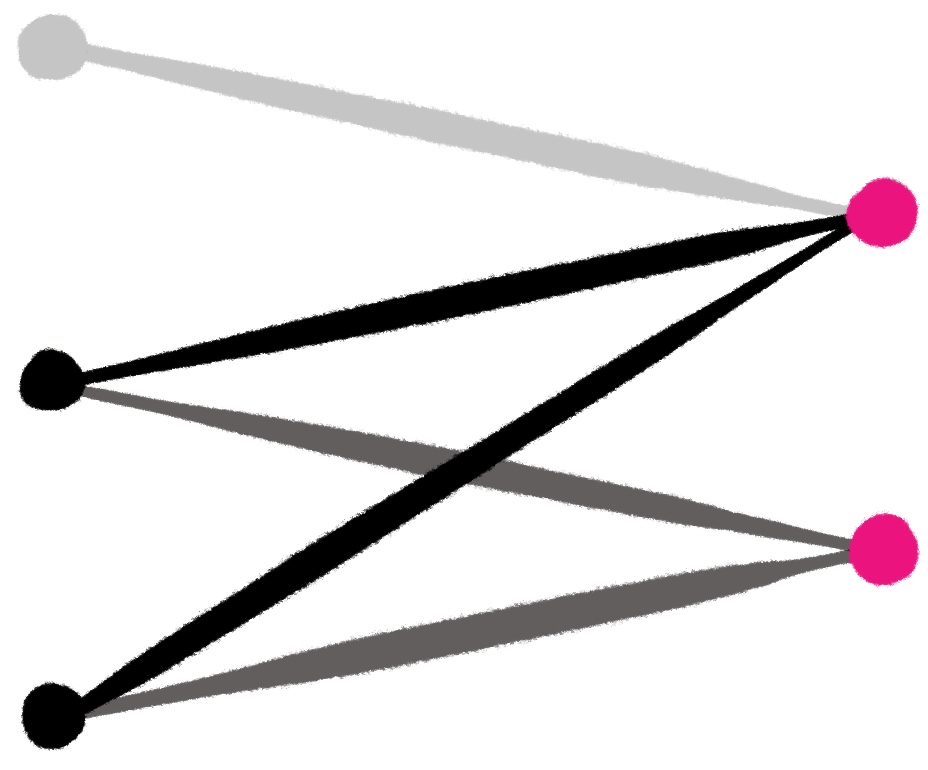}};
	\node at (2.5,.5) {degree $2$};
	\node at (2.5,-.5) {degree $2$};

	\node at (5,0) {$\begin{bmatrix}1&1&1\\1&2&{\color{RubineRed}2}\\1&{\color{RubineRed}2}&2\end{bmatrix}$};
	\end{tikzpicture}
	\end{center}
The second off diagonal counts the single path of length two that joins $x_1$ and $x_3$. You can imagine the picture. 
\end{example}
\noindent With these combinatorial interpretations in mind, notice that Example \ref{ex:main} can be greatly simplified.

\begin{example} The joint distribution on the toy corpus introduced in Example \ref{ex:main} is an empirical one. With $X=\{\text{orange, green, purple}\}$ and $Y=\{\text{fruit, vegetable}\}$ the example corpus corresponds to the following subset $T\subseteq X\times Y$.
\[T=\{(\text{orange, fruit}), (\text{green, fruit}), (\text{purple, vegetable})\}\]
As a graph, the empirical distribution $\hat\pi$ consists of three edges that are each weighted by $\frac{1}{3}$. 
\begin{center}
\begin{tikzpicture}[y=1.2cm, x=1.2cm]
\node at (0,0) {\includegraphics[scale=0.1]{figures/ch2/graph2}};
\node at (-2,1) {orange};
\node at (-2,0) {green};
\node at (-2,-1) {purple};

\node at (2,.5) {fruit};
\node at (2,-.5) {vegetable};
\end{tikzpicture}
\end{center}
Now recall the reduced density operators $\rho_X=\tr_Y|\psi\>\<\psi|$ and $\rho_Y=\tr_X|\psi\>\<\psi|$ derived in Example \ref{ex:main}.
\begin{equation*}
	\rho_X =
	\kbordermatrix{
    & \o & \g & \p\\
    \o & 1 & 1 & 0  \\[10pt]
    \g & 1 & 1 & 0  \\[10pt]
    \p & 0 & 0 & 1 \\
  	}\frac{1}{3}
  	\qquad\qquad
  	\rho_Y =
	\kbordermatrix{
	 & F & V\\
		F & 2 & 0  \\[10pt]
	 	V & 0 & 1  \\
  	}\frac{1}{3}
\end{equation*}
Observe that the diagonals of $\rho_X$ are the degrees of the three vertices $x\in X$, and the off-diagonals count paths of length two. For instance, the \textit{orange-green} and \textit{green-orange} off-diagonal entries are both $1$, as there is a single path of length two joining those vertices. On the other hand, there is \textit{no} path of length two joining \textit{purple} to any other vertex in $X$. This accounts for all entries in $\rho_X$ that are zero. The diagonals of $\rho_Y$ are likewise the degrees of the vertices $y\in Y$. Both off-diagonals are zero since \textit{fruit} and \textit{vegetable} have no prefixes in common. Equivalently, there is not path of length two between \textit{fruit} and \textit{vegetable}.
\end{example}

\noindent By now, we've covered a lot of information. The following takeaway summarizes the discussion.

\begin{takeaway}\label{takeaway:red_den} Given finite sets $X$ and $Y$, any joint probability distribution $\pi\colon X\times Y\to\mathbb{R}$ defines a unit vector in $\mathbb{C}^X\otimes\mathbb{C}^Y$ whose coefficients are the square roots of the probabilities
\[
|\psi\>= \sum_{(x,y)\in X\times Y}\sqrt{\pi(x,y)}|x\>\otimes |y\>.\]
Orthogonal projection onto this unit vector defines a density operator $\rho=|\psi\>\<\psi|$ whose reduced densities $\rho_X$ and $\rho_Y$ have the following properties.
	\begin{itemize}
	\item The classical marginal probability distribution $\pi_X\colon X\to\mathbb{R}$ is contained along the diagonal of $\rho_X$.
	\[(\rho_X)_{ii}=\sum_{\alpha}\pi(x_i,y_\alpha)=\pi_X(x_i)\]

	\item The classical marginal probability distribution $\pi_Y\colon Y\to\mathbb{R}$ is contained along the diagonal of $\rho_Y$.
	\[(\rho_Y)_{\alpha\alpha}=\sum_{i}\pi(x_i,y_\alpha)=\pi_Y(y_\alpha)\]

	\item The reduced densities  $\rho_X$ and $\rho_Y$ generally have nonzero off-diagonal entries that encode extra information about subsystem interactions. 
	\begin{align*}
	(\rho_X)_{ij}&=\sum_{\alpha}\sqrt{\pi(x_i,y_\alpha)\pi(x_j,y_\alpha)}
	\\[5pt]
	(\rho_Y)_{\alpha\beta}&=\sum_{i}\sqrt{\pi(x_i,y_\alpha)\pi(x_i,y_\beta)}
	\end{align*}
	This information contributes to the eigenvalues and eigenvectors of $\rho_X$ and $\rho_Y$, and it is akin to conditional probability.

	\item When $\pi$ is an empirical probability distribution, the matrix entries of $\rho_X$ and $\rho_Y$ have simple, combinatorial interpretations that can be read off from a bipartite graph associated to $\pi$. In particular, the off-diagonal entries count the number of shared suffixes and prefixes.
	\end{itemize}
\end{takeaway}

\newthought{At this point}, the theory branches off into a few directions.
\begin{description}
	\item \textbf{An application.} Now that we understand the information stored in reduced densities, we may return to the idea of reconstructing a pure quantum state from its reduced densities. In practice, how does one use information about smaller subsystems and their interactions to recover knowledge about a larger whole? As we'll see in Chapter \ref{ch:application}, the answer is exploited in a well-known procedure in physics---\textit{the density matrix renormalization group} (DMRG) procedure \cite{schollwock,white1992}. As it turns out, this ability to reconstruct a quantum state is immensely useful for tackling a familiar puzzle in an applied setting: Given some data generated by an unknown probability distribution, what is a good algorithm for building a \textit{model} of that distribution with the goal of generating new data from it? In Chapter \ref{ch:application}, we will share one such algorithm inspired by DMRG. It successively builds one of these so-called \textit{generative models} entirely from eigenvectors of reduced density operators. 

	\item \textbf{A link with formal concepts.} The graph-theoretic approach for understanding reduced densities described in Section \ref{sec:graphs} is directly connected with the graph-theoretical presentation of formal concepts introduced in Chapter \ref{ch:introduction}. Indeed, the observant reader will notice that the eigenvectors of the reduced densities that we computed in Example \ref{ex:main} actually \textit{recovered} the formal concepts in Example \ref{ex:FCA}! In fact, more is true---Example \ref{ex:main} recovered the classical formal concepts and further \textit{enriched} them with probabilities. We will spell this out more carefully below in Section \ref{sec:evectsFCA}.

	\item \textbf{A framework for entailment.} 	So far, we've only discussed probability distributions on a Cartesian product of two sets $X$ and $Y$. The results of this section still hold if $X$ and $Y$ each decompose as products of smaller sets. In other words, the theory holds for larger systems where the initial joint probability distribution is over a set of sequences of any fixed length. As we'll show in Section \ref{sec:entailment}, every sequence (and any subsequence) can be represented by a density operator on a multi-partite system obtained from the state $|\psi\>$ of Equation (\ref{eq:psi2}). The set of density operators on any given space form a partially-ordered set, and so there is a natural way to compare sequences based on the statistics encoded by the densities. When those sequences consist of \textit{words} in a natural language, this paves the way for modeling entailment and hierarchy.
\end{description}

Each of these points is worth expanding upon a little more. Below in Section \ref{sec:evectsFCA} we'll say more about the intriguing connection between eigenvectors and formal concepts. In Section \ref{sec:entailment} we will expand on the theory behind modeling entailment. In Chapter \ref{ch:application}, we will outline an algorithm that produces a generative model using the mathematics of this chapter. All three endeavors may be thought of as ``corollaries'' to the main ideas presented in Sections \ref{sec:procedure}--\ref{sec:graphs}.

\section{Eigenvectors versus Formal Concepts}\label{sec:evectsFCA}
Turning back a few pages to the formal concept construction in Example \ref{ex:FCA} is now encouraged. There we started with the same finite sets that we had in Example \ref{ex:main}, namely $X=\{\text{orange, green, purple}\}$ and $Y=\{\text{fruit, vegetable}\}$. The remaining ingredients in each example are completely analogous to each other. In Example \ref{ex:FCA} we started with a function $R\colon X\times Y\to \{0,1\}$, indicated by this table:
\begin{table}[h!]
	\centering
	\begin{tabular}{r|ccc}
		& \text{orange} & \text{green} & \text{purple}\\[5pt]
	\hline\\[-8pt]
	\text{fruit} & $1$ & $1$ & $0$\\[5pt]
	\text{vegetable} & $0$ & $0$& $1$
	\end{tabular}
\end{table}

\noindent In Example \ref{ex:main} we started with a probability distribution $\pi\colon X\times Y\to\mathbb{R}$, indicated by this table:
\begin{table}[h!]
	\centering
	\begin{tabular}{r|ccc}
		& \text{orange} & \text{green} & \text{purple}\\[5pt]
	\hline\\[-8pt]
	\text{fruit} & $\frac{1}{3}$ & $\frac{1}{3}$ & $0$\\[5pt]
	\text{vegetable} & $0$ & $0$& $\frac{1}{3}$
	\end{tabular}
\end{table}

\noindent Both $R$ and $\pi$ were then visualized as the same bipartite graph.
\begin{center}
\begin{tikzpicture}[y=1.2cm, x=1.2cm]
\node at (0,0) {\includegraphics[scale=0.1]{figures/ch2/graph2}};
\node at (-2,1) {orange};
\node at (-2,0) {green};
\node at (-2,-1) {purple};

\node at (2,.5) {fruit};
\node at (2,-.5) {vegetable};
\end{tikzpicture}
\end{center}
Using this graph in Example \ref{ex:FCA}, we first obtained a pair maps between power sets $\adj{f}{2^X}{2^Y}{g}$, and the fixed points of $fg$ and $gf$ gave two formal concepts.
	\marginnote[-.7cm]{\begin{center}
	\begin{tikzpicture}
	\node at (3,0) {\includegraphics[scale=0.13]{figures/ch1/concept2}};
	\node at (2,1.8) {$f$};
	\node at (4,1.8) {$g$};
	\node at (3,1.1) {fruit};

	\node at (1.2,-1.6) {$f$};
	\node at (4.8,-1.6) {$g$};
	\node at (3,-2) {vegetable};
	\end{tikzpicture}
	\end{center}}
\begin{align*}
	&(\{\text{orange, green}\},\{\text{fruit}\})\\[10pt]
	&(\{\text{purple}\},\{\text{vegetable}\})
\end{align*}
Using the same graph in Example \ref{ex:main}, we also obtained a pair of maps between vector spaces $\adj{M}{\mathbb{C}^X}{\mathbb{C}^Y}{M^\dagger}$, and the fixed points of $MM^\dagger$ and $M^\dagger M$ gave two pairs of eigenvectors, $\{|e_1\>,|f_1\>\}$ and $\{|e_2\>,|f_2\>\}$, where each pair is further ``tagged'' with an eigenvalue.
\[
\begin{array}{lll}
\{|e_1\>,|f_1\>\}= \left\{\tfrac{1}{\sqrt 2}\left(|\text{orange}\> + |\text{green}\>\right)\;,\; |\text{fruit}\>\right\} && \lambda_1 = \frac{2}{3} \\[10pt]
\{|e_2\>,|f_2\>\}= \left\{|\text{purple}\>\;,\; |\text{vegetable}\>\right\} && \lambda_2 = \frac{1}{3}
\end{array}
\]
These sets of eigenvectors are clearly a statistically-enriched version of the formal concepts. But the analogy goes deeper. Recall from Equation (\ref{eq:ab}) that $f$ and $g$ are defined as lifts of the following functions $a\colon X\to 2^Y$ and $b\colon Y\to 2^X$.
\begin{equation*}
	\begin{tikzcd}
	X\ar[r,"a"] & 2^Y \\[-20pt]
	x \ar[r,mapsto] & a(x):=\{y\in Y\mid R(x,y)=1\}\\[10pt] 
	Y\ar[r,"b"] & 2^X \\[-20pt]
	y \ar[r,mapsto] & b(y):=\{x\in X\mid R(x,y)=1\}
	\end{tikzcd}
\end{equation*}
The subset
	\marginnote[-3cm]{
	\[
	\begin{tikzcd}[ampersand replacement = \&]
	2^X \arrow[r, "f", dashed]   \& 2^Y \\
	X \arrow[u] \arrow[ru, "a"'] \&    
	\end{tikzcd}
	\]
	\bigskip
	\[
	\begin{tikzcd}[ampersand replacement = \&]
	2^Y \arrow[r, "g", dashed]   \& 2^X \\
	Y \arrow[u] \arrow[ru, "b"'] \&    
	\end{tikzcd}
	\]
	}
$a(x)$ can be identified with the $x$th column of the table (or $\{0,1\}$-valued matrix) associated to $R$, and the subset $b(y)$ can be identified with the $y$th row. Analogously, there are functions $\alpha\colon X\to\mathbb{C}^Y$ and $\beta\colon Y\to\mathbb{C}^X$ that lift to give the linear maps $M$ and $M^\dagger$. Recall from Proposition \ref{prop:povm} that for each $i=1,\ldots,|X|$ there is a map $A_i\in \hom(\mathbb{C}^X\otimes\mathbb{C}^Y,\mathbb{C}^Y)$ defined by $A_i:=\<x_i|\otimes \id_{\mathbb{C}^Y}$. So when $|\psi\>=\sum_{x,y}\sqrt{\pi(x,y)}|x\>\otimes|y\>$ is the state associated to $\pi$, ``evaluation at $|\psi\>$'' defines a function
\begin{equation*}
	\begin{tikzcd}
	X\ar[r,"\alpha"] & \mathbb{C}^Y \\[-20pt]
	x_i \ar[r,mapsto] & A_i|\psi\>
	\end{tikzcd}
\end{equation*}
Note that $A_i|\psi\>\in\mathbb{C}^Y$ is the $i$th column of $M$.
	\marginnote{
	\[
	\begin{tikzcd}[ampersand replacement = \&]
	\mathbb{C}^X \arrow[r, "M", dashed]   \& \mathbb{C}^Y \\
	X \arrow[u] \arrow[ru, "\alpha"'] \&    
	\end{tikzcd}
	\]
	\bigskip
	\[
	\begin{tikzcd}[ampersand replacement = \&]
	\mathbb{C}^Y \arrow[r, "M^\dagger", dashed]   \& \mathbb{C}^X \\
	Y \arrow[u] \arrow[ru, "\beta"'] \&    
	\end{tikzcd}
	\]
	}
By the universal property of vector spaces, $\alpha$ \textit{lifts} to a linear map $\mathbb{C}^X\to\mathbb{C}^Y$, which is precisely $M$.  In the same way, for each $\alpha=1,\ldots,|Y|$ there is a map $B_\alpha\in\hom(\mathbb{C}^X\otimes\mathbb{C}^Y,\mathbb{C}^X$) so that evaluation at $|\psi\>$ defines a function
\begin{equation*}
	\begin{tikzcd}
	Y\ar[r,"\beta"] & \mathbb{C}^X \\[-20pt]
	y_\alpha \ar[r,mapsto] & B_\alpha|\psi\>
	\end{tikzcd}
\end{equation*}
where $B_\alpha|\psi\>\in\mathbb{C}^X$ is the $\alpha$th row of $M$. Up to complex conjugation, $\beta$ lifts to the linear adjoint $M^\dagger\colon \mathbb{C}^Y\to\mathbb{C}^X$.% I have to say ``up to complex conjugation'' because B_\alpha|\psi\> is the \alpha'th row of $M$, which is DIFFERENT than the *complex conjugate* of the the \alpha'th row of $M$. The latter would truly lift to the adjoint M^\dagger. As it is, $\beta$ only lifts to M-transpose.

\newthought{So the parallels} between formal concepts and eigenvectors of reduced densities are striking. Even so, we may \textbf{not} conclude that they always coincide, or that the linear algebra always recovers the set theory. In fact, \textit{this is only true when the relation $R$ satisfies the property that its bipartite graph is a disjoint union of complete bipartite subgraphs,} or ``clusters.'' Then one expects that the eigenvectors of the reduced densities derived from the uniform probability distribution defined by the graph \textit{will} coincide with classical formal concepts, and each formal concept or pair of eigenvectors corresponds to a cluster. (The verification is a straightforward, though perhaps tedious, exercise in arithmetic.) So it's very easy to find an example where formal concepts and eigenvectors don't coincide. Consider the following graph, for instance, which contains four edges instead of three.
\begin{center}
\begin{tikzpicture}[y=1.2cm, x=1.2cm]
\node at (0,0) {\includegraphics[scale=0.1]{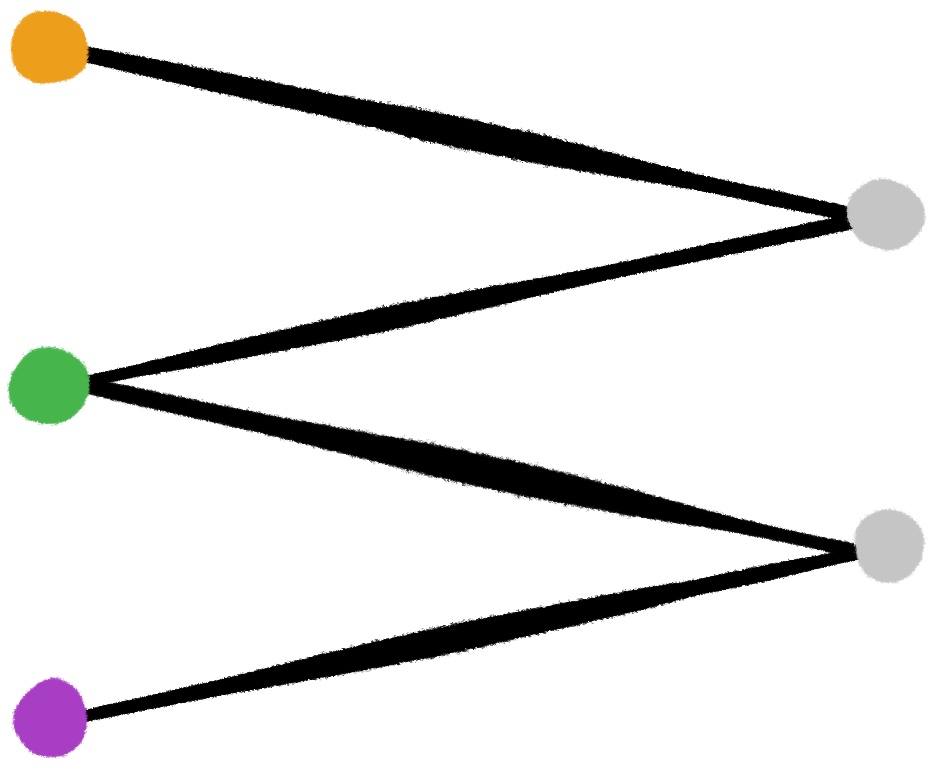}};
\node at (-2,1) {orange};
\node at (-2,0) {green};
\node at (-2,-1) {purple};

\node at (2,.5) {fruit};
\node at (2,-.5) {vegetable};
\end{tikzpicture}
\end{center}
It defines the following uniform probability distribution on $X\times Y$.
\begin{center}
	\begin{tabular}{r|ccc}
		& \text{orange} & \text{green} & \text{purple}\\[5pt]
	\hline\\[-8pt]
	\text{fruit} & $\frac{1}{4}$ & $\frac{1}{4}$ & $0$\\[5pt]
	\text{vegetable} & $0$ & $\frac{1}{4}$ & $\frac{1}{4}$
	\end{tabular}
\end{center}
By replacing each ``$\frac{1}{4}$'' with a ``$1$,'' this table also defines a relation $R\subset X\times Y$, which---as can be quickly verified---has the following three formal concepts.
\begin{align*}
	&(\{\text{orange, green}\},\{\text{fruit}\})\\[10pt]
	&(\{\text{green, purple}\},\{\text{vegetable}\})\\[10pt]
	&(\{\text{green}\},\{\text{fruit, vegetable}\})
\end{align*}
On the other hand, there are only \textit{two} pairs of eigenvectors from $MM^\dagger$ and $M^\dagger M$, where $M=\frac{1}{\sqrt 4}\begin{bsmallmatrix}1&1&0\\0&1&1\end{bsmallmatrix}$. Here they are, together with their associated eigenvalues.
\[
\begin{array}{lll}
\left\{\tfrac{1}{\sqrt 6}\left(|\text{orange}\> + |\text{green}\> + |\text{purple}\>\right)\;,\; \tfrac{1}{\sqrt 2}(|\text{fruit}\> + |\text{vegetable}\>)\right\} && \lambda_1 = \frac{3}{4} \\[10pt]
\left\{\tfrac{1}{\sqrt 2}(|\text{orange}\> - |\text{purple}\>)\;,\; \tfrac{1}{\sqrt 2}(|\text{fruit}\> - |\text{vegetable}\>)\right\} && \lambda_2 = \frac{1}{4}
\end{array}
\]
So in this simple example, the two notions---eigenvectors and formal concepts---diverge. (And notice the negatives!) 

\newthought{Let's close this section} with two remarks. First, although eigenvectors and formal concepts don't always coincide, we've taken time to discuss the connection between them because we think it is interesting, if not tantalizing. Second, there \textit{is} a sense in which the two notions are  clearly related.  But to see this we must lift our feet off the ground and view the situation from the perspective of category theory. Chapter \ref{ch:categorytheory} is fully devoted to this endeavor. In the mean time, let's now turn our attention to another ``corollary'' of the theory of Section \ref{sec:graphs}, namely modeling entailment with density operators.

\section{Modeling Entailment with Densities}\label{sec:entailment}
In some settings, the spectral information of reduced densities of pure states may serve as a proxy for ``meaning.'' Let's think in the context of language, for instance. To understand something about the meaning of the word \textit{orange}, it's helpful to know how often the word occurs in the language as well as the \textit{context} in which it appears. From a mathematical perspective, this idea is not unfamiliar. The Yoneda lemma in category theory informally says that a mathematical object is completely determined by the network of relationships it shares with other objects of its kind.\marginnote{The Yoneda lemma brings to mind the oft-cited quote by linguist John Firth: ``You shall know a word by the company it keeps.'' \cite{firth57synopsis} I like to think of this as the Yoneda Lemma for Language.} In the same way, think of the meaning of a word as being captured by the network of all expressions that the word fits into \textit{along with} with the statistics of those occurences.
\begin{center}
\includegraphics[width=\textwidth]{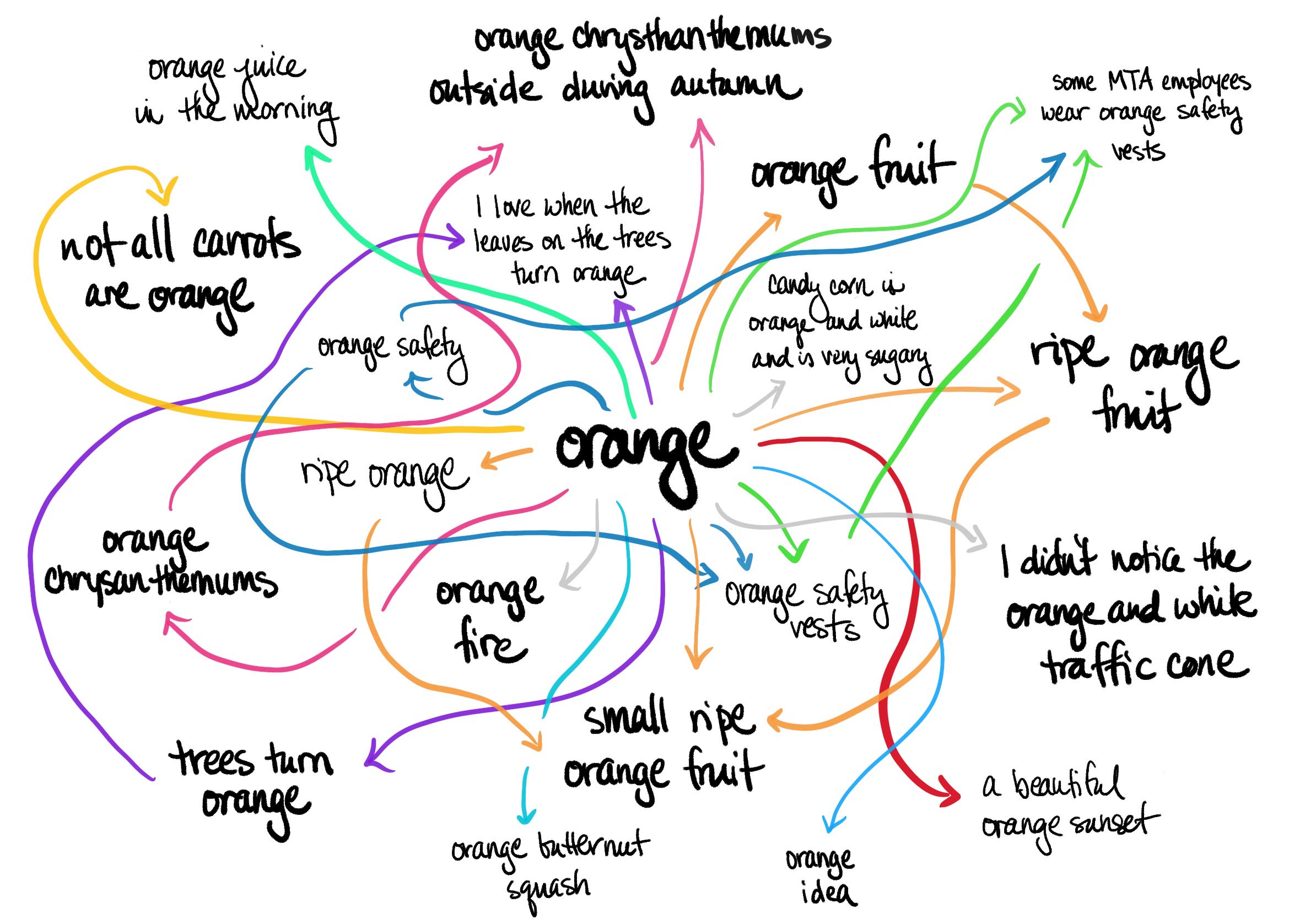}
\end{center}
As it turns out, density operators provide a natural framework for capturing this very idea. The goal of this section is to show that a word such as \textit{orange} naturally unravels as a weighted sum of all expressions containing that word, where the weights are conditional probabilities. More precisely, we'll see that any word or expression can be assigned a density, and the density for \textit{orange} decomposes as a weighted sum of densities, one for each expression containing \textit{orange}.  One might think of this as a \textit{Yoneda-enriched spectral decomposition}, which will give a clear candidate for modeling entailment and concept hierarchy using a simple poset structure on densities. 
\begin{center}
\includegraphics[width=\textwidth]{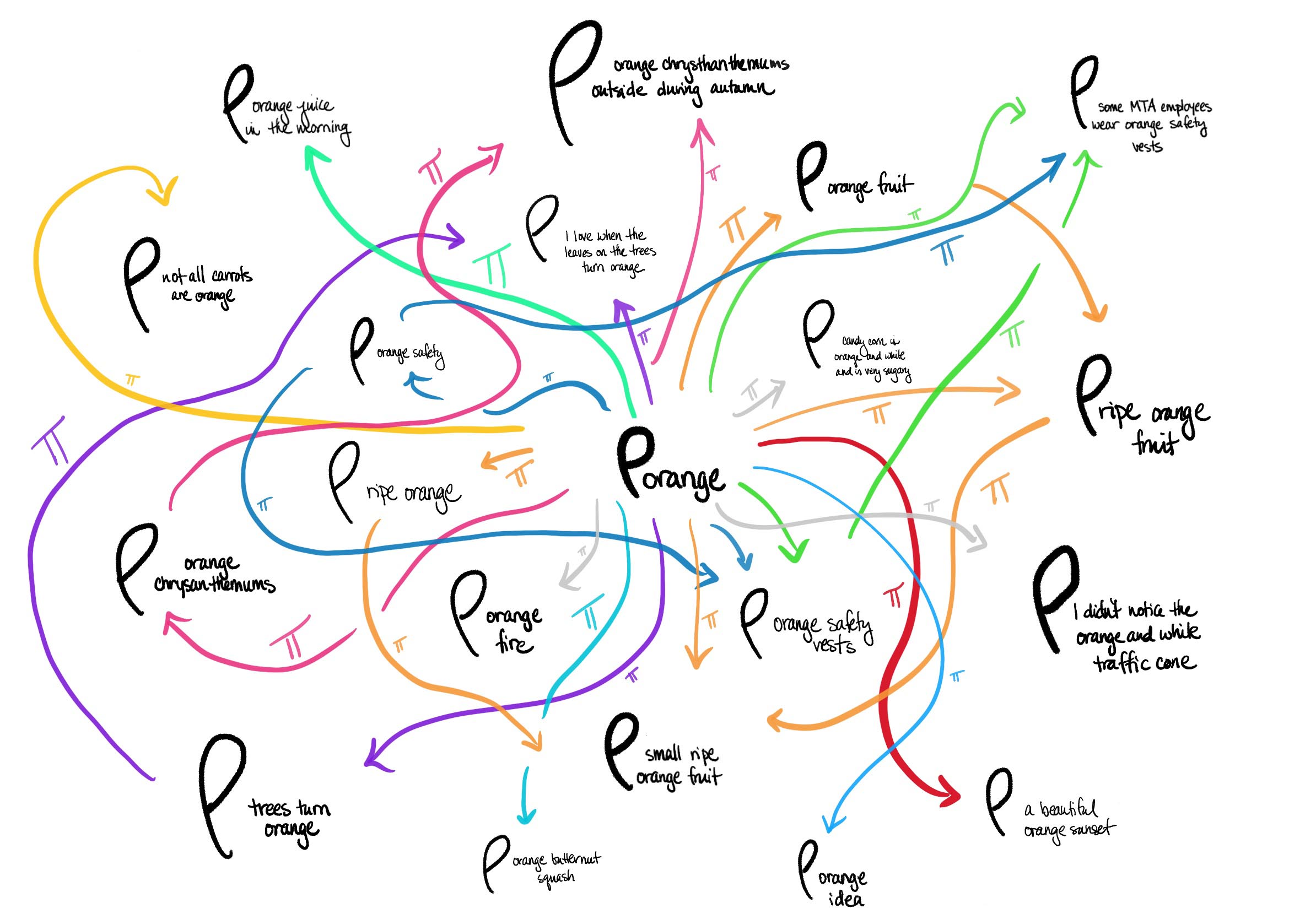}
\end{center}

The ideas in this section are the results of joint work with Yiannis Vlassopoulos, in which we further investigate entailment and concept hierarchy in language. The simple example in Section \ref{ssec:entailmentexample} is meant to spark the reader's interest for the extended theory appearing in a forthcoming paper \cite{taiyiannis}. Notably, Section \ref{ssec:entailmentexample}'s passage from words to densities, such as $\text{orange} \mapsto \rho_{\text{orange}}$, will be derived differently in \cite{taiyiannis}, and the latter greatly expands upon the general theory of Section \ref{ssec:generalentailment}. We hope the simple ideas here will whet your appetite for much more to come. Further connections between language modeling and tensor networks may be found in \cite{pestun2017,pestun2017b}. Others have also explored a density matrix model for entailment using the Lowener order. The authors in \cite{bankova2016} present one such framework based on the category theoretical approach of \cite{clark2010}. We may view the work below as an instantiation of the theory presented there. In particular, conditional probabilities present themselves as a natural candidate for entailment strength, and we are using statistics as a proxy for grammar.

\newthought{With this introduction} in mind, let us now describe how these ideas connect with the mathematics described in Section \ref{sec:procedure}. As noted then, the starting ingredient is always a probability distribution on a finite set $\pi\colon S\to\mathbb{R}$. When $S$ decomposes as an $N$-fold Cartesian product of a finite set $A$, say $S=A\times A \cdots\times A$, we may think of elements $s\in S$ as sequences of symbols from the ``alphabet'' $A$. The elements of $A$ may be thought of as individual characters or words or phrases in a language so that concatenations of them are longer expressions; $S$ contains all those of length $N$. The distribution $\pi$ then captures something of the statistics in the language built from $A$, and  the vector $|\psi\>=\sum_s\sqrt{\pi(s)}|s\>$ is the state of the language. As we will show below, knowledge of this state then gives an easy way to compare words and expressions from $A$. For example, the mathematics will give us a natural way to compare something that is ``orange'' with something that is ``small, ripe'' and ``orange,'' and to see that the latter is a specific case of the former. The key tool for making these kinds of comparisons is a passage from \textit{words} to \textit{reduced density operators}. 

Concretely, we aim to show that expressions from the alphabet $A$ can be represented by reduced density operators obtained from the orthogonal projection operator $|\psi\>\<\psi|$. These densities can then be compared to one another using a simple partial order. Given positive semidefinite operators $\sigma$ and $\tau$ on the same space, one has $\sigma\geq \tau$ if and only if $\sigma-\tau$ is positive semidefinite. This is called the \emph{Loewner order} on positive semidefinite operators, and it defines a partial order on the set of densities operating on a fixed space.

Here's the outline for this section. In Section \ref{ssec:entailmentexample}, we'll show that words and phrases, say  \textit{orange} and \textit{small ripe orange}, have density matrix representations, $\rho_{\text{orange}}$ and $\rho_{\text{small ripe orange}}$. Then we'll use Proposition \ref{prop:povm} to show that the density for a given word, say $\rho_{\text{orange}}$, decomposes as a weighted sum of densities corresponding to all expressions containing that word. Finally we'll show that this decomposition naturally gives a way to model entailment between various expressions using the Loewner order. In particular, we'll see that conditional probabilities from $\pi$ provide a natural measure of entailment strength:
\begin{equation}\label{eq:entailment}
\rho_{\text{orange}} \geq \pi(\text{small ripe orange $\mid$ orange}) \rho_{\text{small ripe orange}}
\end{equation}
Because the ideas here are relatively new, we will put an elementary example up front and derive the above inequality. Afterwards we'll describe the general theory in Section \ref{ssec:generalentailment}.

\subsection{An Extended Example}\label{ssec:entailmentexample}
In previous examples, we have worked with a toy corpus of text containing three phrases of length two: \textit{orange fruit, green fruit, purple vegetable}. In this example we will start with a different set of phrases that are a little longer. Suppose we have the following toy corpus of text containing five phrases of length four.
\begin{center}
\begin{tabular}{ll}
small ripe orange fruit & large rotten green vegetable\\
large ripe orange vegetable & small ripe orange vegetable\\
small rotten orange fruit &
\end{tabular}
\end{center}
Let's now follow the procedure outlined in Section \ref{sec:procedure} to obtain a probability distribution that can be modeled by a quantum state $|\psi\>$. As a first step, consider the following ordered sets.
\begin{equation*}
\begin{split}
A&:=\{\text{small, large}\}\\
B&:=\{\text{ripe, rotten}\}
\end{split}\qquad\qquad
\begin{split}
C&:=\{\text{orange, green}\}\\
Y&:=\{\text{fruit, vegetable}\}
\end{split}
\end{equation*}
Let $X=A\times B\times C$ so that elements in $X$ are \emph{prefixes} and elements in $Y$ are \emph{suffixes.} Identify each of the five phrases $xy$ in the corpus with its corresponding sequence $(x,y)$ and let $T\subseteq X\times Y$ be the set consisting of these pairs. This defines a probability distribution $\hat\pi\colon X\times Y\to\mathbb{R}$ by
\[
\hat\pi(x,y)= 
\begin{cases}
\frac{1}{5} &\text{if $(x,y)\in T$}\\
0 &\text{if $(x,y)\not\in T$.}
\end{cases}
\]
Now define $|\psi\>$ to be the unit vector obtained as the sum of all sequences in $T$ weighted by the square roots of their probabilities. Observe that $|\psi\>$ lies in the tensor product $\mathbb{C}^A\otimes\mathbb{C}^B\otimes\mathbb{C}^C\otimes\mathbb{C}^Y$ where each factor is isomorphic to $\mathbb{C}^2$.
\begin{align*}
|\psi\> &= \sum_{x,y}\sqrt{\hat\pi(x,y)}|x\>\otimes|y\>\\
&=\sum_{a,b,c,y}\sqrt{\hat\pi(a,b,c,y)}|a\>\otimes|b\>\otimes|c\>\otimes|y\>\\
&=\frac{1}{\sqrt 5}\left(
	\begin{array}{llllllll}
	|\text{small}\> & \otimes & |\text{ripe}\> & \otimes & |\text{orange}\> & \otimes &|\text{fruit}\> & +\\
	|\text{small}\> & \otimes & |\text{rotten}\> & \otimes & |\text{orange}\> & \otimes &|\text{fruit}\> & +\\
	|\text{large}\> & \otimes & |\text{ripe}\> & \otimes & |\text{orange}\> & \otimes &|\text{vegetable}\> & +\\
	|\text{large}\> & \otimes & |\text{rotten}\> & \otimes & |\text{green}\> & \otimes &|\text{vegetable}\> & +\\
	|\text{small}\> & \otimes & |\text{ripe}\> & \otimes & |\text{orange}\> & \otimes &|\text{vegetable}\> &
	\end{array}
\right)
\end{align*}
As a tensor diagram, $|\psi\>$ is a node with four parallel edges, the first three representing $\mathbb{C}^X$ and the fourth representing $\mathbb{C}^Y$.
\begin{center}
\begin{tikzpicture}
	\node at (-2,0) {$|\psi\> = $};
	\node at (0,0) {\includegraphics[scale=0.12]{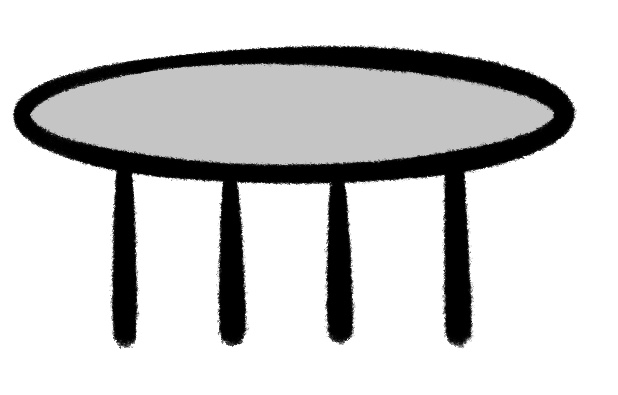}};
\end{tikzpicture}
\end{center}
With $|\psi\>$ in hand, our next goal is to compute the reduced density operator on the suffix subsystem $\rho_Y=\tr_X|\psi\>\<\psi|$. There are a few ways to do this, but let us refer to Proposition \ref{prop:povm}, which states that $\rho_Y$ decomposes as a sum of projection operators,
\begin{equation}\label{eq:rhoy_povm}
\rho_Y=\sum_{x\in X} A_x|\psi\>\<\psi|A_x^\dagger
\end{equation}
where each $A_x\colon\mathbb{C}^X\otimes\mathbb{C}^Y\to\mathbb{C}^Y$ is defined by $A_x:=\<x|\otimes\id_{\mathbb{C}^Y}$. That is, for any prefix-suffix pair $(x',y)\in X\times Y$  we have that $A_x(|x'\>\otimes|y\>)=\<x|x'\>|y\>$ which is equal to $|y\>$ if $x'=x$ and is equal to $0$ otherwise. So each term in Equation (\ref{eq:rhoy_povm}) involves a vector of the form
	\marginnote[-3cm]{The linear map $M$ associated to $|\psi\>$ has the following diagram,
	\begin{center}
	\begin{tikzpicture}
	\node at (-1.3,0) {$M=$};
	\node at (.3,0) {\includegraphics[scale=0.1]{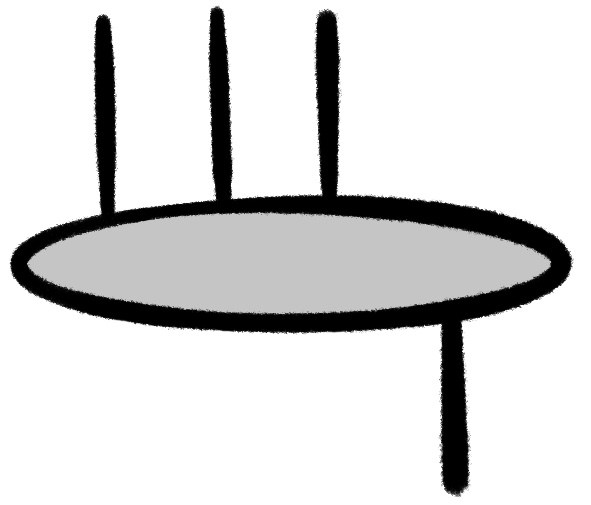}};
	\end{tikzpicture}
	\end{center}
	and the equation $\rho_Y=\tr_X|\psi\>\<\psi|=MM^\dagger$ has the following picture:
	\begin{center}
	\begin{tikzpicture}
		\node at (0,0) {\includegraphics[scale=0.1]{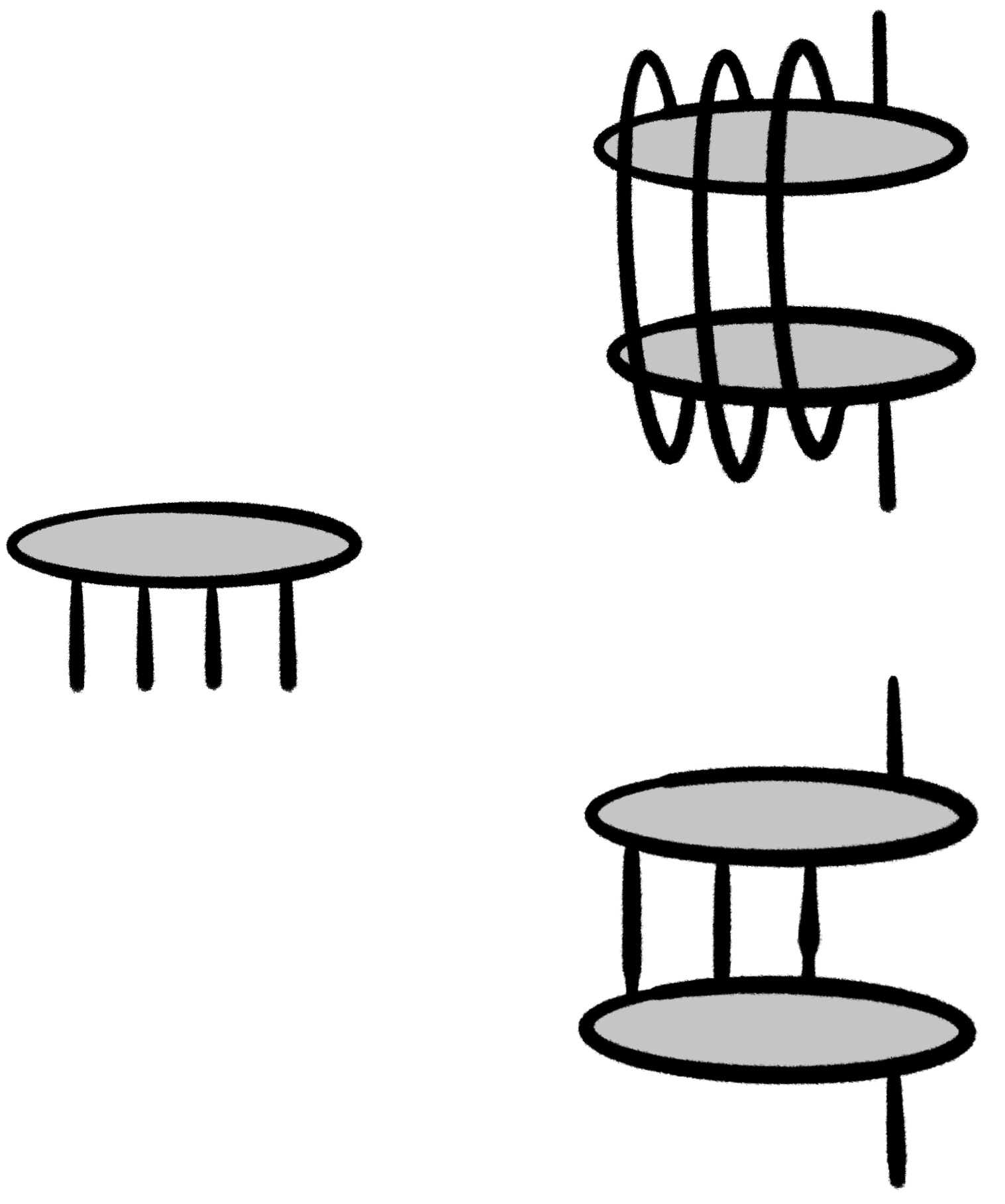}};
		\draw[thick, ->,>=stealth] (-1,1)--(0,2);
		\draw[thick, ->,>=stealth] (-1,-1)--(0,-2);
		\node at (-1.4,1.7) {$\tr_X|\psi\>\<\psi|$};
		\node at (-1.2,-1.5) {$MM^\dagger$};
	\end{tikzpicture}
	\end{center}
	The sum decomposition in Equation (\ref{eq:rhoy_povm}) gives yet \textit{another} expression for $\rho_Y$. It doesn't have a standard tensor diagram representation, but we will illustrate it as below.
	\begin{center}
	\begin{tikzpicture}
	\node at (0,0) {\includegraphics[width=\linewidth]{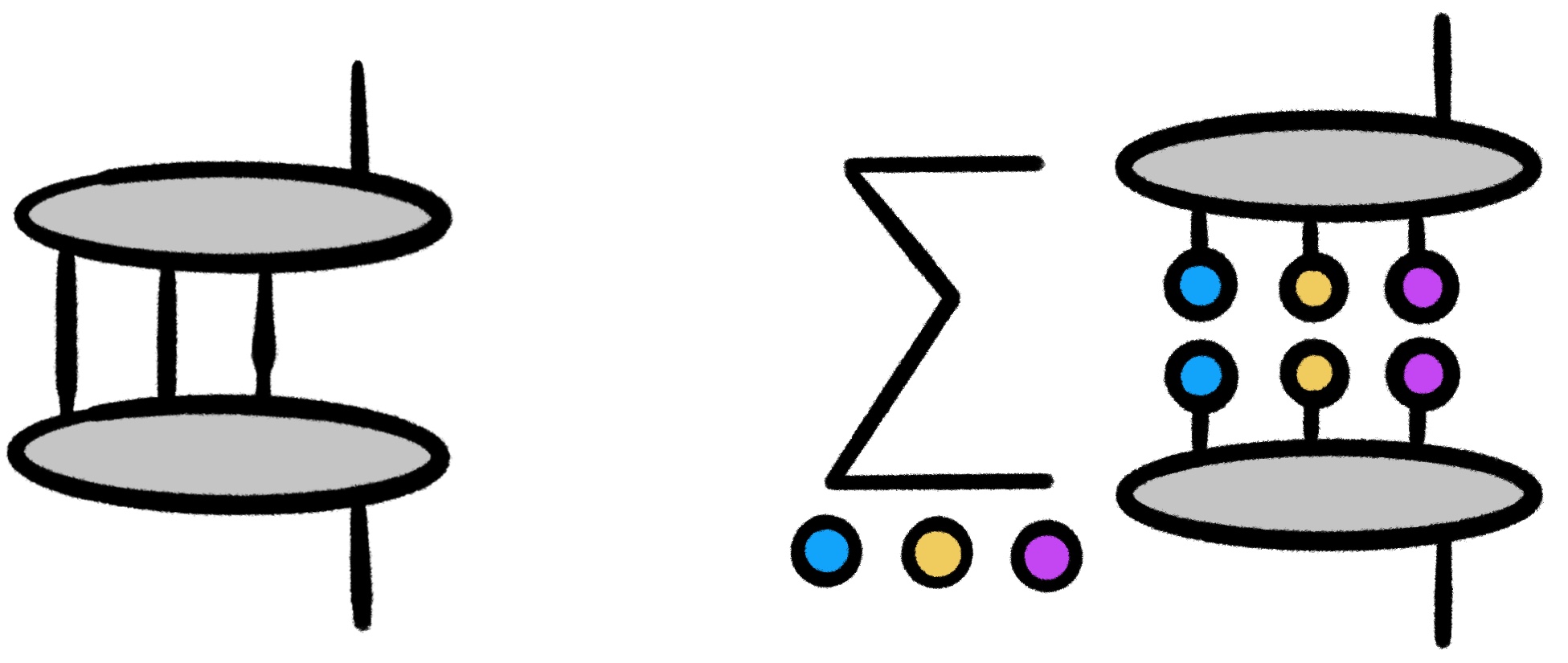}};
	\node at (-.5,0) {=};
	\end{tikzpicture}
	\end{center}
	In symbols,
	\[\rho_Y=\sum_{({\color{RoyalBlue}a},{\color{Goldenrod}b},{\color{Mulberry}c})\in A\times B\times C}M|{\color{RoyalBlue}a}{\color{Goldenrod}b}{\color{Mulberry}c}\>\<{\color{RoyalBlue}a}{\color{Goldenrod}b}{\color{Mulberry}c}|M^\dagger\]
	where we are using the fact that $A_i|\psi\>=M|x_i\>$ for each $i$, as remarked in the main text.
	}
\begin{align}\label{eq:Bi}
A_x|\psi\>&=\sum_y\sqrt{\hat\pi(x,y)}|y\>\\
&=\sqrt{\hat\pi(x,\text{fruit})}|\text{fruit}\> + \sqrt{\hat\pi(x,\text{vegetable})}|\text{vegetable}\>.\nonumber
\end{align}
Notice that the vector $A_x|\psi\>$ is the $x$th column of the matrix representation of the linear map $M\colon \mathbb{C}^X\to\mathbb{C}^Y$ associated to $|\psi\>\in \mathbb{C}^X\otimes\mathbb{C}^Y$, whose $yx$th entry is the probability $M_{yx}=\sqrt{\pi(x,y)}$. Equation (\ref{eq:rhoy_povm}) says that $\rho_Y$ decomposes as the \textit{sum of projections onto each of these columns.} This is a point worth emphasizing. Every density operator has a spectral decomposition and therefore decomposes as a sum of projection operators onto the spaces spanned by its eigenvectors. But Equation (\ref{eq:rhoy_povm}) is giving us \textit{a completely different} decomposition into projections, namely one onto each column of $M$. To get a better feel for this, let's write down the matrix $M$ explicitly. To do so we first need to choose an ordering on the set $X=\{x_1,x_2,\ldots,x_8\}$. Let's choose ordering below. 
\begin{equation*}
\begin{split}
\mathbf{x_1}&=(\textbf{small},\textbf{ripe},\textbf{{\color{YellowOrange}orange}})\\
\mathbf{x_2}&=(\textbf{large},\textbf{ripe},\textbf{{\color{YellowOrange}orange}})\\
\mathbf{x_3}&=(\textbf{small},\textbf{rotten},\textbf{{\color{YellowOrange}orange}})\\
x_4&=(\text{large},\text{rotten},\text{{\color{YellowOrange}orange}})\\
\end{split}
\qquad\qquad
\begin{split}
x_5&=(\text{small},\text{ripe},\text{{\color{Green}green}})\\
x_6&=(\text{large},\text{ripe},\text{{\color{Green}green}})\\
x_7&=(\text{small},\text{rotten},\text{{\color{Green}green}})\\
\mathbf{x_8}&=(\textbf{large},\textbf{rotten},\textbf{{\color{Green}green}})\\
\end{split}
\end{equation*}
Of these eight sequences, recall that only \textit{four} of them appear as prefixes in $T$, namely $x_1,x_2,x_3$ and $x_8$, indicated in boldface above. Notably, $x_1$ appears in $T$ twice. With this ordering, the matrix $M$ associated to $|\psi\>$ has the probability $\sqrt{\hat\pi(x_i,y_\alpha)}$ as the $\alpha i$th entry.
\[
M=\frac{1}{\sqrt 5}
\begin{bmatrix}
1 & 0 & 1 & 0 & 0&0&0&0\\
1&1&0&0&0&0&0&1
\end{bmatrix}
\]
Each of the eight columns correspond to one of the prefixes $x_i\in X$, and each of two rows correspond to one of the suffixes $y_\alpha\in Y$. There is a $1/\sqrt{5}$ in the $\alpha i$th entry if and only if $(x_i,y_\alpha)\in T$. Here's a decorated version of the matrix,
\[
M = \kbordermatrix{
    & x_1 & x_2 & x_3 & x_4 & x_5 & x_6 & x_7 & x_8 \\
    y_1 & 1&0&1&0&0&0&0&0  \\[10pt]
    y_2 & 1 & 1 & 0 &0&0&0&0&1
  	}\frac{1}{\sqrt 5}
  	\]
or more elaborately, 
\begin{center}
\begin{tikzpicture}
	\node at (0,0) {$M=\kbordermatrix{
		    & \text{small} & \text{ large}  & \text{ small} & \text{large} & \text{ small} & \text{large}  & \text{small} & \text{ large} \\
		    \text{fruit} & {\color{YellowOrange}1} & {\color{YellowOrange}0} & {\color{YellowOrange}1} & {\color{YellowOrange}0} &{\color{Green}0}&{\color{Green}0}&{\color{Green}0}&{\color{Green}0}  \\[10pt]
		    \text{vegetable} & {\color{YellowOrange}1} & {\color{YellowOrange}1} & {\color{YellowOrange}0} & {\color{YellowOrange}0} &{\color{Green}0}&{\color{Green}0}&{\color{Green}0}&{\color{Green}1} 
		  	}\frac{1}{\sqrt 5 }$};
	\draw (.9,-1.2) -- (.9,1.2); %mid vert
	\draw (-1.1,-1) -- (-1.1,1); %left vert
	\draw (2.9,-1) -- (2.9,1); %right vert
	\draw (-3,-.25) -- (4.8,-.25); %horizontal
	\node at (-1.1,-1.5) {orange};
	\node at (2.9,-1.5) {green};
	\node at (-2.2,1) {ripe};
	\node at (-.1,1) {rotten};
	\node at (2,1) {ripe};
	\node at (4,1) {rotten};
\end{tikzpicture}
\end{center}
Now we may quickly verify the claim that $A_i|\psi\>$ is the $i$th column of $M$. Recall that the natural map $Y\to \mathbb{C}^Y$ represents suffixes as standard basis vectors, $|\text{fruit}\>=\begin{bsmallmatrix}1\\0\end{bsmallmatrix}$ and $|\text{vegetable}\>=\begin{bsmallmatrix}0\\1\end{bsmallmatrix}$. Then Equation (\ref{eq:Bi}) implies that
\begin{align*}
A_1|\psi\>&=\tfrac{1}{\sqrt 5}(|\text{fruit}\>+|\text{vegetable}\>)=\frac{1}{\sqrt 5}\begin{bmatrix}1\\1\end{bmatrix}\\[10pt]
A_2|\psi\>&=\tfrac{1}{\sqrt 5}|\text{vegetable}\> = \frac{1}{\sqrt 5}\begin{bmatrix}0\\1\end{bmatrix}\\[10pt]
A_3|\psi\>&=\tfrac{1}{\sqrt 5}|\text{fruit}\>=\frac{1}{\sqrt 5}\begin{bmatrix}1\\0\end{bmatrix}\\[10pt]
A_8|\psi\>&=\tfrac{1}{\sqrt 5}|\text{vegetable}\>=\frac{1}{\sqrt 5}\begin{bmatrix}0\\1\end{bmatrix}
\end{align*}
and $A_i|\psi\>=0$ for any prefix $x_i$ that is not in $T$. So indeed $A_i|\psi\>$ is the $i$th column of $M$. Equivalently, we have that $A_i|\psi\>=M|x_i\>.$ So thinking of $M$ as a linear map between the prefix and suffix subsystems, $M$ maps every prefix in $T$ onto its suffix. For example, $M|x_2\>=M|\text{large ripe orange}\>=\tfrac{1}{\sqrt 5}|\text{vegetable}\>.$ Also observe that two prefixes have the same image under $M$ whenever they share the same suffix in $T$. For instance, the vectors $|x_2\>$ and $|x_8\>$ both map to $\tfrac{1}{\sqrt 5}|\text{vegetable}\>$ since the prefixes \textit{large ripe orange} and \textit{large rotten green} both have \textit{vegetable} as a suffix. The upshot is that Equation (\ref{eq:rhoy_povm}) gives a decomposition of $\rho_Y$ as a sum of projections onto these vectors. Explicitly,
\begin{align*}
\rho_Y &= 
A_1|\psi\>\<\psi|A_1^\dagger + 
A_2|\psi\>\<\psi|A_2^\dagger + 
A_3|\psi\>\<\psi|A_3^\dagger + 
A_8|\psi\>\<\psi|A_8^\dagger\\[10pt]
&=\frac{1}{5}\left(
\begin{bmatrix}1\\1\end{bmatrix}\begin{bmatrix}1 & 1\end{bmatrix} +
\begin{bmatrix}0\\1\end{bmatrix}\begin{bmatrix}0 & 1\end{bmatrix} +
\begin{bmatrix}1\\0\end{bmatrix}\begin{bmatrix}1 & 0\end{bmatrix} +
\begin{bmatrix}0\\1\end{bmatrix}\begin{bmatrix}0 & 1\end{bmatrix} \right) \\[20pt]
&=\frac{1}{5}\left(
\begin{bmatrix}1&1\\1&1\end{bmatrix} + 
\begin{bmatrix}0&0\\0&1\end{bmatrix} + 
\begin{bmatrix}1&0\\0&0\end{bmatrix} +
\begin{bmatrix}0&0\\0&1\end{bmatrix} \right)\\[20pt]
&=\frac{1}{5}
\begin{bmatrix} 
2 & 1 \\ 1 & 3
\end{bmatrix}.
\end{align*}
	\marginnote[-6cm]{\noindent 
	$x_1=$ (small, ripe, orange)\\\noindent
	$x_2=$ (large, ripe, orange)\\\noindent
	$x_3=$ (small, rotten, orange)\\\noindent
	$x_8=$ (large, rotten, green)
	}
\noindent Notice the diagonals of $\rho_Y$ count the number of occurrences of \textit{fruit} and \textit{vegetable} in the corpus, while the off-diagonal entry counts the number of shared prefixes. This is consistent with the combinatorial claims of Section \ref{sec:graphs}. In any case, the path for defining a density matrix representation for sequences is now clear. For each prefix $x_i\in T$ define the \emph{unnormalized reduced density operator} associated to $x_i$ to be
\[\hat\rho_{x_i}=A_i|\psi\>\<\psi|A_i^\dagger\]
	\marginnote{If $x_i=(a,b,c)$ so that $|abc\>=
	\begin{tikzpicture}[y=.3cm,x=.4cm,baseline=1pt]
	\draw[thick] (0,0) -- (0,-1);
	\draw[thick] (1,0) -- (1,-1);
	\draw[thick] (2,0) -- (2,-1);

	\node[tensor,minimum size=1mm, fill={rgb, 255:red,18;green,165;blue,252}] at (0,0) {}; 
	\node[tensor,minimum size=1mm,fill={rgb, 255:red,241;green,204;blue,95}] at (1,0) {}; 
	\node[tensor,minimum size=1mm,fill={rgb, 255:red,195;green,71;blue,243}] at (2,0) {}; 
	\end{tikzpicture}$
	then $\hat\rho_{x_i}$ is projection onto $A_i|\psi\>=M|abc\>$ which has the following picture.
	\begin{center}
	\begin{tikzpicture}
	\node at (0,0) {\includegraphics[scale=0.1]{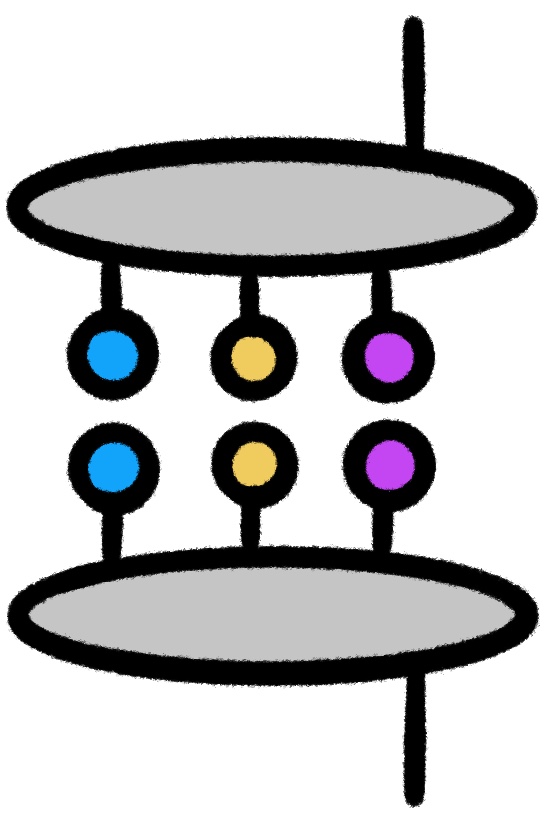}};
	\node at (-1.5,0) {$\rho_x=$};
	\end{tikzpicture}
	\end{center}
	}
As a projection operator, $\hat\rho_{x_i}$ is positive semidefinite and Hermitian, though it does not have unit trace. Should we want a properly normalized density operator, simply divide $\hat\rho_{x_i}$ by its trace, which is the marginal probability of $x_i$:
\[\rho_{x_i}:=\frac{1}{\hat\pi_X(x_i)}\hat\rho_{x_i}\]
To see this latter claim, recall that $A_i|\psi\>=\sum_y\sqrt{\hat\pi(x_i,y)}|y\>$ and so $A_i|\psi\>\<\psi|A_i^\dagger=\sum_{y,y'}\sqrt{\hat\pi(x_i,y)\hat\pi(x_i,y')}|y\>\<y'|$ which has trace equal to $\sum_y\hat\pi(x_i,y)=\hat\pi_X(x_i)$, where $\hat\pi_X\colon X\to \mathbb{R}$ is the marginal distribution obtained from $\hat\pi\colon X\times Y\to\mathbb{R}$. In the above example we obtain the following properly normalized reduced densities associated to each prefix,
	\marginnote[1cm]{
	\begin{align*}
	\hat\pi_X(x_1)&=2/5\\
	\hat\pi_X(x_2)&=1/5\\
	\hat\pi_X(x_3)&=1/5\\
	\hat\pi_X(x_8)&=1/5
	\end{align*}
	}
\begin{equation*}
\begin{split}
\rho_{x_1}&=\rho_{\text{small ripe orange}} = \frac{1}{2}\begin{bmatrix}1&1\\1&1\end{bmatrix}\\[10pt]
\rho_{x_2}&=\rho_{\text{large ripe orange}} = \begin{bmatrix}0&0\\0&1\end{bmatrix}
\end{split}
\qquad\qquad
\begin{split}
\rho_{x_3}&=\rho_{\text{small rotten orange}} =\begin{bmatrix}1&0\\0&0\end{bmatrix}\\[10pt]
\rho_{x_8}&=\rho_{\text{large rotten green}} = \begin{bmatrix}0&0\\0&1\end{bmatrix}
\end{split}
\end{equation*}
while the \textit{unnormalized} densities $\hat\rho_{x_i}$ are the four matrices whose sum is $\rho_Y$, as shown above. To summarize, every prefix $x_i$ in $T$ corresponds to a (normalized or unnormalized) density operator identified with one of the projections that sum to $\rho_Y$. 

Notice, however, that some sequences $x_i$ contain the same words. For instance, $x_1,x_2,$ and $x_3$ are all phrases that contain the word \textit{orange}. Intuitively, the sum $\hat\rho_{x_1}+\hat\rho_{x_2}+\hat\rho_{x_3}$ of all unnormalized densities containing \textit{orange}  should be defined as the unnormalized density $\hat\rho_{\text{orange}}$ associated to \textit{orange} itself, and similarly for $\hat\rho_{\text{green}}.$
\begin{align*}
\rho_Y = \overbrace{{\color{YellowOrange}A_1|\psi\>\<\psi|A_1^\dagger} + {\color{YellowOrange}A_2|\psi\>\<\psi|A_2^\dagger} + {\color{YellowOrange}A_3|\psi\>\<\psi|A_3^\dagger}}^{\hat\rho_{\text{orange}}} 
\quad +\quad 
\overbrace{{\color{Green}A_8|\psi\>\<\psi|A_8^\dagger}}^{\hat\rho_{\text{green}}}
\end{align*}
As tensor diagrams, this equation is
	\marginnote[-.5cm]{You might wonder why the tensor network diagram for $\hat\rho_{\text{green}}$ isn't the following:
	\begin{center}
	\includegraphics[scale=0.08]{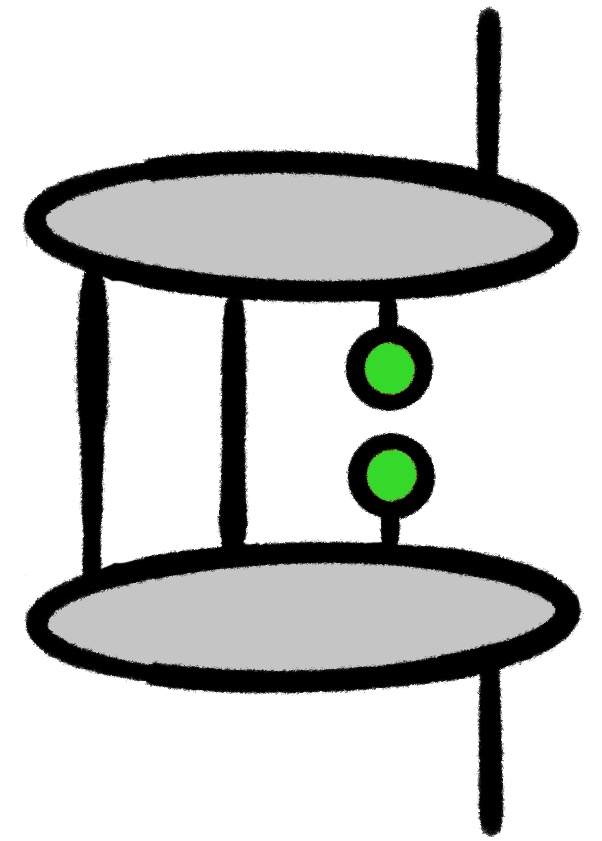}
	\end{center}
	In fact, it \textit{is}. The two edges connecting the gray nodes imply a sum over prefixes $x_i$ that contain the word ``green.'' But there is only \textit{one} such prefix, namely $x_8$, and so $\hat\rho_{\text{green}}$ is comprised of the single term $A_8|\psi\>\<\psi|A_8^\dagger$, which is projection onto the vector $A_8|\psi\>$. That projection operator is the diagram illustrated in the main text.
	}
\begin{center}
\begin{tikzpicture}
\node at (0,0) {\includegraphics[scale=0.1]{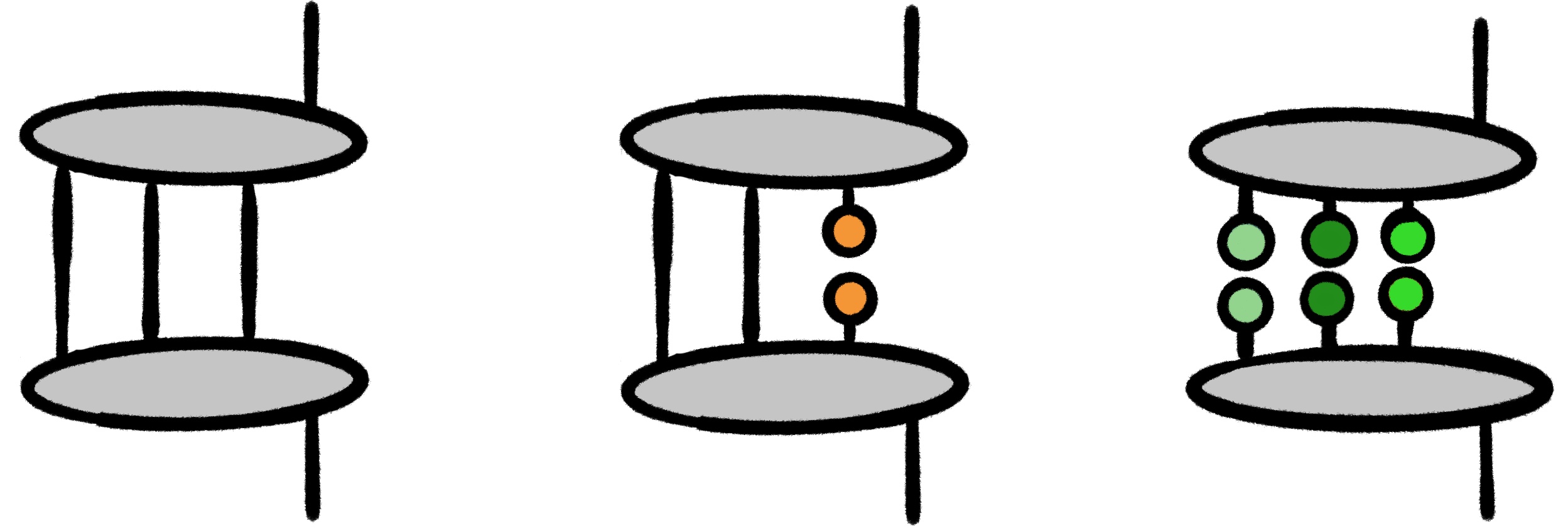}};
\node at (-2,0) {=};
\node at (1.8,0) {+};
\end{tikzpicture}
\end{center}
where
$|\text{orange}\>=
\begin{tikzpicture}[baseline=-5pt] 
	\draw[thick] (0,0)--(0,-.35) {}; 
	\node[tensor,minimum size=1mm,fill={rgb, 255:red,244;green,149;blue,54}] at (0,0) {};
\end{tikzpicture}$
and
$|\text{green}\>=
\begin{tikzpicture}[baseline=-5pt] \draw[thick] (0,0)--(0,-.35) {};
	\node[tensor,minimum size=1mm,fill={rgb, 255:red,55;green,219;blue,43}] at (0,0) {};
\end{tikzpicture}$.
Now for $\hat\rho_{\text{orange}}$ to have unit trace, we must divide by the marginal probability $\hat\pi_C(\text{orange})=4/5$. So let's make this a definition.\sidenote[][.5cm]{Here I am expanding the Cartesian product $X=A\times B\times C$ to recall that our original joint distribution is on the four-fold Cartesian product $\hat\pi\colon A\times B\times C\times Y\to\mathbb{C}$, which defines a marginal distribution $\pi_C\colon C\to\mathbb{R}$. A quick check shows that the trace of $\hat\rho_{\text{orange}}:=\hat\rho_{x_1}+\hat\rho_{x_2}+\hat\rho_{x_3}$ is computed as $\sum_{a,b,y}\hat\pi(a,b,\text{orange},y)$, which is the marginal probability $\pi_C(\text{orange}).$}
\begin{align*}
\rho_{\text{orange}}&:=\frac{1}{\hat\pi_C(\text{orange})}\hat\rho_{\text{orange}}\\[10pt]
&=\frac{1}{\hat\pi_C(\text{orange})}(\hat\rho_{x_1}+\hat\rho_{x_2}+\hat\rho_{x_3})\\[10pt]
&=\frac{1}{4}\begin{bmatrix}
2 & 1\\ 1 & 2
\end{bmatrix}
\end{align*}
In fact, this suggests a nicer way to write the result:
\begin{align*}
\rho_{\text{orange}}&=\frac{1}{\pi_C(\text{orange})}\sum_{i=1,2,3}\hat\rho_{x_i}\\[10pt]
&=\sum_{i=1,2,3}\frac{\pi_X(x_i)}{\pi_C(\text{orange})}\;\frac{\hat\rho_{x_i}}{\pi_X(x_i)}\\[10pt]
&= \sum_{i=1,2,3}\pi(x_i\mid\text{orange})\rho_{x_i}.
\end{align*}
Consequently, for any $i\in\{1,2,3\}$ we have the inequality
\[\rho_\text{orange}\geq \pi(x_i|\text{orange})\rho_{x_i}\]
since their difference is a sum of positive semidefinite operators weighted by nonnegative numbers. Setting $i=1$ we have now recovered Inequality (\ref{eq:entailment}), as was the goal:
\[\rho_\text{orange}\geq \pi(\text{small ripe orange}\mid\text{orange})\rho_{\text{small ripe orange}}.\]
In a similar vein, we can easily sum over those tuples $x_i=(a,b,c)\in A\times B\times C$ with $b=$\textit{ripe} as well as $c=$\textit{orange} to obtain the unnormalized density associated to the expression \textit{ripe orange},
\[
\hat\rho_{\text{ripe orange}}:= A_1|\psi\>\<\psi|A_1^\dagger + A_2|\psi\>\<\psi|A_2^\dagger = \frac{1}{5}\begin{bmatrix}1&1\\1&2\end{bmatrix}.
\] 
The difference $\hat\rho_{\text{orange}}-\hat\rho_{\text{ripe orange}}=\tfrac{1}{5}\begin{bsmallmatrix}0&0\\0&1\end{bsmallmatrix}\geq 0$ is a positive operator and therefore we have $\hat\rho_{\text{orange}}\geq \hat\rho_{\text{ripe orange}}\geq \hat\rho_{\text{small ripe orange}},$ which pairs well with the intuition that \textit{ripe orange} is a more specific instance of something that is \textit{orange} but is more general than something that is \textit{small ripe orange}. 
\begin{center}
\begin{tikzpicture}
\node at (0,0) {\includegraphics[scale=0.1]{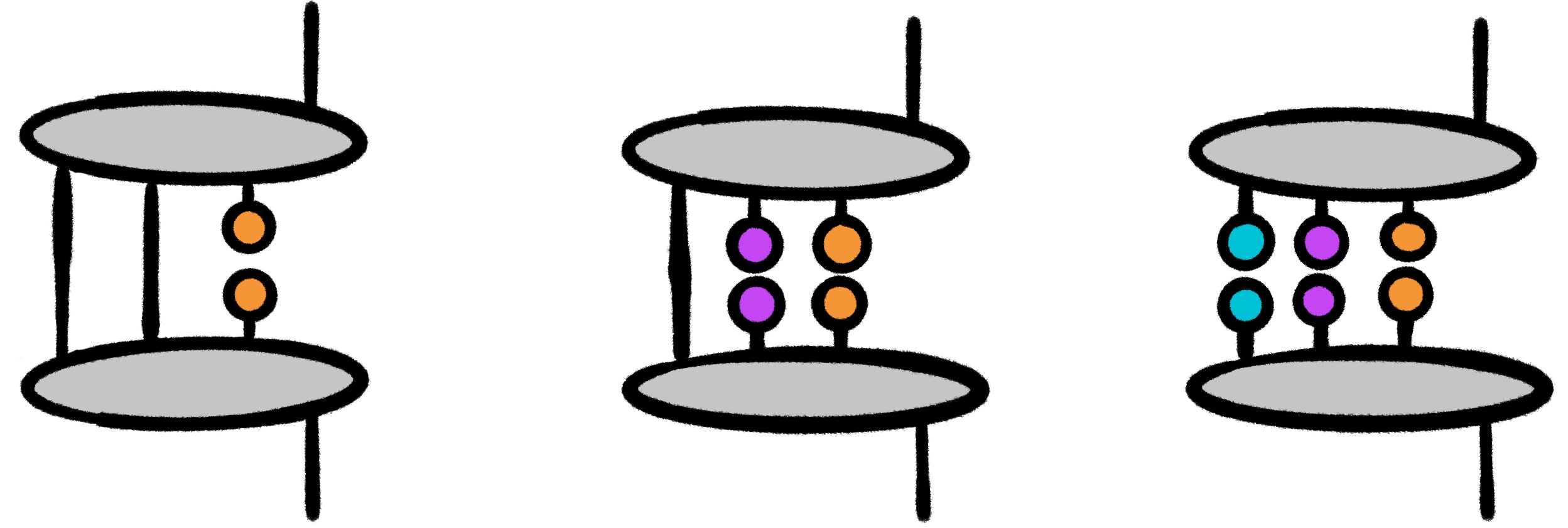}};
\node at (-1.8,0) {$\geq$};
\node at (1.7,0) {$\geq$};
\node at (-3.3,-2) {$\hat\rho_{\text{orange}}$};
\node at (0,-2) {$\hat\rho_{\text{ripe orange}}$};
\node at (3.2,-2) {$\hat\rho_{\text{small ripe orange}}$};
\end{tikzpicture}
\end{center}

\newthought{To summarize}, the reduced density associated to the word \textit{orange} decomposes as a weighted sum of densities, one associated to each expression that contains the word orange. Crucially, the weights are naturally found to be conditional probabilities and  may be thought of measuring the strength of entailment.

\subsection{The General Idea}\label{ssec:generalentailment}
%\marginnote{\tai{Probably don't need ``empirical'' in this section, or the previous. Remove all hats?}}
This extended example illustrates the general idea, which we summarize briefly here. For any finite set $A$, consider a subset of prefix-suffix pairs $T\subseteq X\times Y$ with $X=A^{N-1}=A\times \cdots\times A$ and $Y=A$. This defines an empirical probability distribution $\hat\pi\colon X\times Y\to\mathbb{R}$ and therefore a state $|\psi\>=\frac{1}{\sqrt{|T|}}\sum_{(x,y)\in T}|x\>\otimes|y\>\in \mathbb{C}^X\otimes\mathbb{C}^Y$. Choose an ordering on the set of prefixes, say $X=\{x_1,\ldots,x_n\}$. Then each prefix $x_i$ corresponds to an \textit{unnormalized reduced density operator}
\[\hat\rho_{x_i}:=M|x_i\>\<x_i|M^\dagger\]
where $M\colon \mathbb{C}^X\to\mathbb{C}^Y$ is the linear map associated to $|\psi\>$. (Equivalently, $\hat\rho_{x_i}:=A_i|\psi\>\<\psi|A_i^\dagger$ where $A_i$ is as in Proposition \ref{prop:povm}.) To obtain the properly normalized reduced density associated to $x_i$, simply divide by its marginal probability,
\[\rho_{x_i}:=\frac{1}{\hat\pi_X(x_i)}\hat\rho_{x_i}.\]
By Proposition \ref{prop:povm} the reduced density operator $\rho_Y$ on the suffix subsystem decomposes as a sum of these unnormalized projections,
\[\rho_Y=\sum_i\hat\rho_{x_i}.\]
By summing over a \textit{subset} of the indices $i$, we can obtain a reduced density associated to various subsequences in $X$. For instance, let $X(a)$ denote the set of all sequences $x_i$ that contain $a\in A$ as the \textit{last} entry. Then we define the \textit{unnormalized reduced density operator associated to $a$} to be
\[\hat\rho_a=\sum_{x\in X(a)}\hat\rho_x\]
This is not a density as its trace is the marginal probability $\hat\pi(a)<1$, so we normalize by dividing by the marginal probability of $a$ and obtain the \textit{reduced density operator associated to $a$},
	\marginnote{To emphasize a point mentioned earlier, any density such as $\rho_a$ can be written as a weighted sum of projections onto its eigenvectors, where the weights are eigenvalues. The decomposition for $\rho_a$ shown here is similar yet different. It is still a weighted sum of projections, but now the projections are onto vectors associated with each expression containing the word ``$a$,'' and the weights are conditional probabilities. For this reason, we might think of $\rho_a=\sum_{x\in X(a)}\pi(x|a)\rho_x$ as a probabilistically enriched spectral decomposition of $\rho_a$.}
\begin{align*}
\rho_a&:=\frac{1}{\hat\pi(a)}\hat\rho_a\\
&=\sum_{x\in X(a)}\frac{1}{\hat\pi(a)} \hat\rho_x\\
&=\sum_{x\in X(a)}\frac{\hat\pi(x)}{\hat\pi(a)}\frac{1}{\hat\pi(x)}\hat\rho_x\\
&=\sum_{x\in X(a)}\pi(x\mid a)\rho_x.
\end{align*}
As the last line shows, $\rho_a\geq \pi(x|a)\rho_x$ for any prefix $x$. We may think of this as witnessing the fact that $x$ is an expression containing $a$, thereby being a \textit{refinement} of $a$. There's nothing special about $a$ being the \textit{last} term in $x=(a_1,\ldots,a)$, nor is there anything special about $a$ being a singleton. Using the ideas here, one could associate a density $\rho_s$ to \textit{any} subsequence $s$ of any length in $X$, and one can obtain a similar ``unraveling'' of $\rho_s$ into a weighted sum of all expressions that contain $s$. Some care must be taken in \textit{where} $s$ occurs, but we won't provide the details here.
\begin{center}
\includegraphics[scale=0.17]{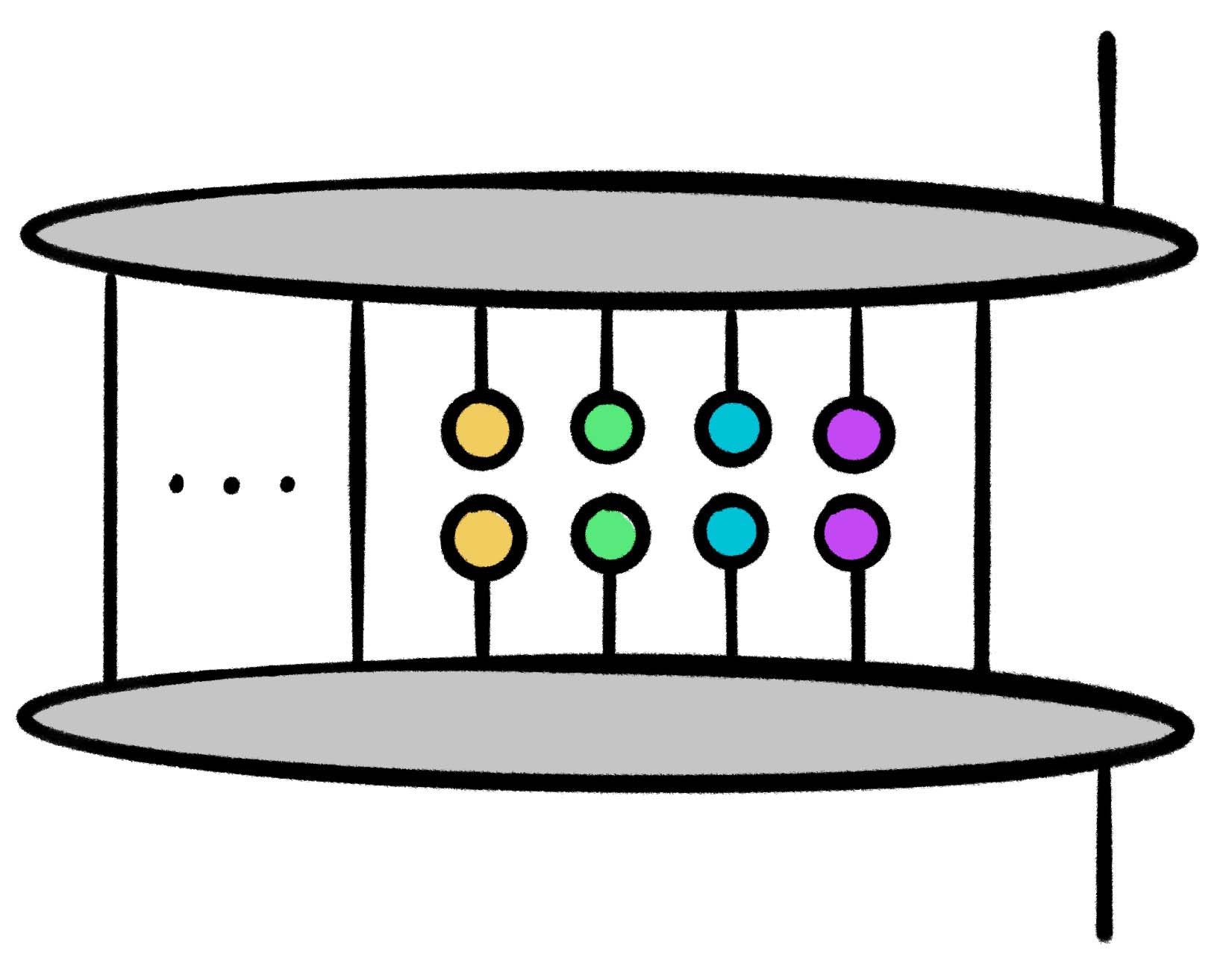}
\end{center}

\clearpage

%chapter 4
\chapter{An Application and an Experiment}\label{ch:application}

\epigraph{Language is a one-dimensional painting.}{\textit{Yiannis Vlassopoulos \cite{vlassopoulos2019}}}

\newthought{This chapter highlights} an application of the mathematics in Chapters \ref{ch:preliminaries} and \ref{ch:probability}. As we've seen, any probability distribution on a finite set can be represented as a rank $1$ density operator---a pure quantum state. We've described this idea in detail, and so we now anticipate the question, ``Why bother?'' That ism why go through the trouble of passing from basic probability to (possibly ultra large-dimensional) vector spaces, densities, and reduced densities? The answer is hidden in plain sight in Figure \ref{fig:reconstruct}.\marginnote{\includegraphics[width=\linewidth]{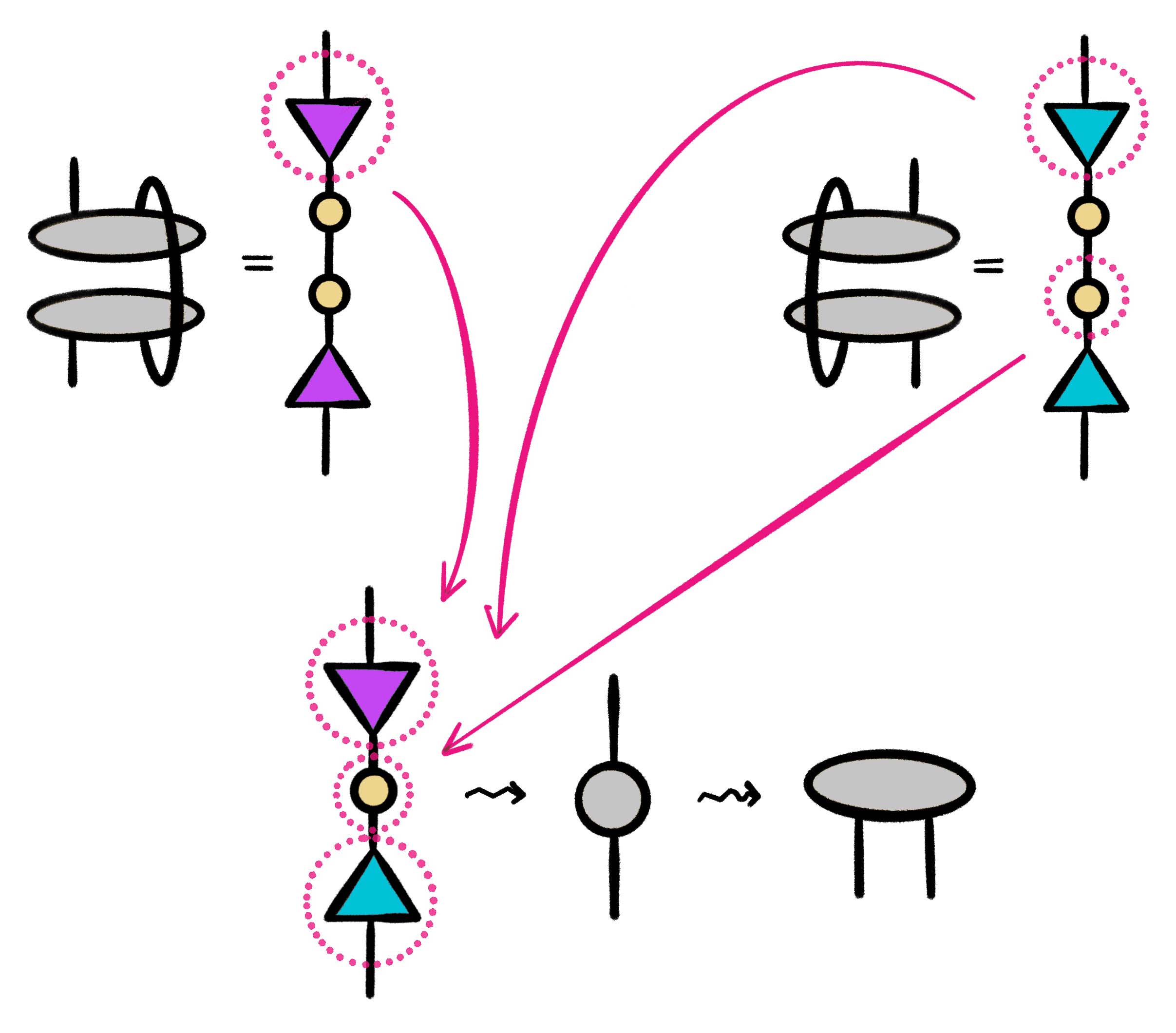}} That is, understanding how to piece together the eigenvectors of reduced densities to reconstruct a quantum state suggests a \textit{new} algorithm for reconstructing a classical joint probability distribution---namely by putting together smaller pieces of it in a highly principled way that knows something about how those pieces interact.

We share such a reconstruction algorithm below and illustrate it by performing an experiment on a well-known dataset. But first, to help put this application in a larger context, let's consider the following question. Suppose we have a set of data drawn from an unknown probability distribution $\pi$. How might we use these samples to define a \textit{new} probability distribution that estimates $\pi$, with the goal of generating new data from it? Here are a couple of examples to illustrate.
\begin{description}
	\item \textbf{Bitstrings}. Suppose our data points are bitstrings---sequences of $0$s and $1$s---all of length 16, for example. There are $2^{16}=65,536$ such strings, but suppose we only have a small fraction of them that were drawn from some distribution $\pi$. We may or may not know what $\pi$ is, but either way we'd like to use the statistics of the samples we \textit{do} have in order to estimate $\pi$. With the estimation in hand, we can then generate new bitstrings by sampling from the estimated distribution.

	\item \textbf{Natural language.} Suppose our data points are meaningful sentences; that is, sequences of words from the English alphabet, say. There are many of possibilities, of course, but suppose we only have a small fraction of samples. There is some probability distribution on natural language---the probability of the phrase ``orange fruit'' is higher than the probability of ``orange idea''---though we don't know it, exactly. But we might like to use the data we \textit{do} have to estimate the probability distribution on language in order to generate new, plausible text. 
\end{description}
In both of these examples, the data are sequential, meaning they are  sequences of symbols from some alphabet. The second example is a case of statistical language modeling. Both examples are instances of \textit{generative modeling}---the task of modeling probabilities with the goal of generating data. This falls under the wider genre of \textit{unsupervised machine learning.}

\newthought{In this chapter} we will illustrate how the extra information stored in the eigenvectors of reduced densities of a pure entangled state form the building blocks of a generative model. Here's the big-picture idea. The passage from classical to quantum probability outlined in Chapter \ref{ch:probability} gives a straightforward recipe for modeling any finite dataset $T$ as a quantum state $|\psi\>$, as long as the elements of $T$ are sequences; that is, as long as $T$ is a subset of a Cartesian product of finite sets.  Further supposing the elements $s\in T$ are samples generated from a (possibly unknown) probability distribution $\pi$, the dataset $T$ defines an empirical probability distribution $\hat\pi$. Think of $\hat\pi$ as a rough approximation to the \textit{true} distribution $\pi$.  As we know from Chapter \ref{ch:probability}, the empirical distribution $\hat\pi$ can be modeled by a unit vector $|\psi\>$ so that the probabilities are the coefficients, $\hat\pi(s)=|\<s|\psi\>|^2$. The machine learning problem is to find an algorithm that takes in $|\psi\>$ and outputs a \textit{new} unit vector, call it $|\psi_{\text{MPS}}\>$, so that the resulting probabilities $|\<s|\psi_{\text{MPS}}\>|^2$ are close to the true probabilities $\pi(s),$ for which the $\hat\pi(s)$ were merely an approximation. This new vector $|\psi_{\text{MPS}}\>$ is the desired generative model.
	\marginnote[-6cm]{
	\begin{center}
	\begin{tikzpicture}
	\node at (0,0) {\includegraphics[width=\linewidth]{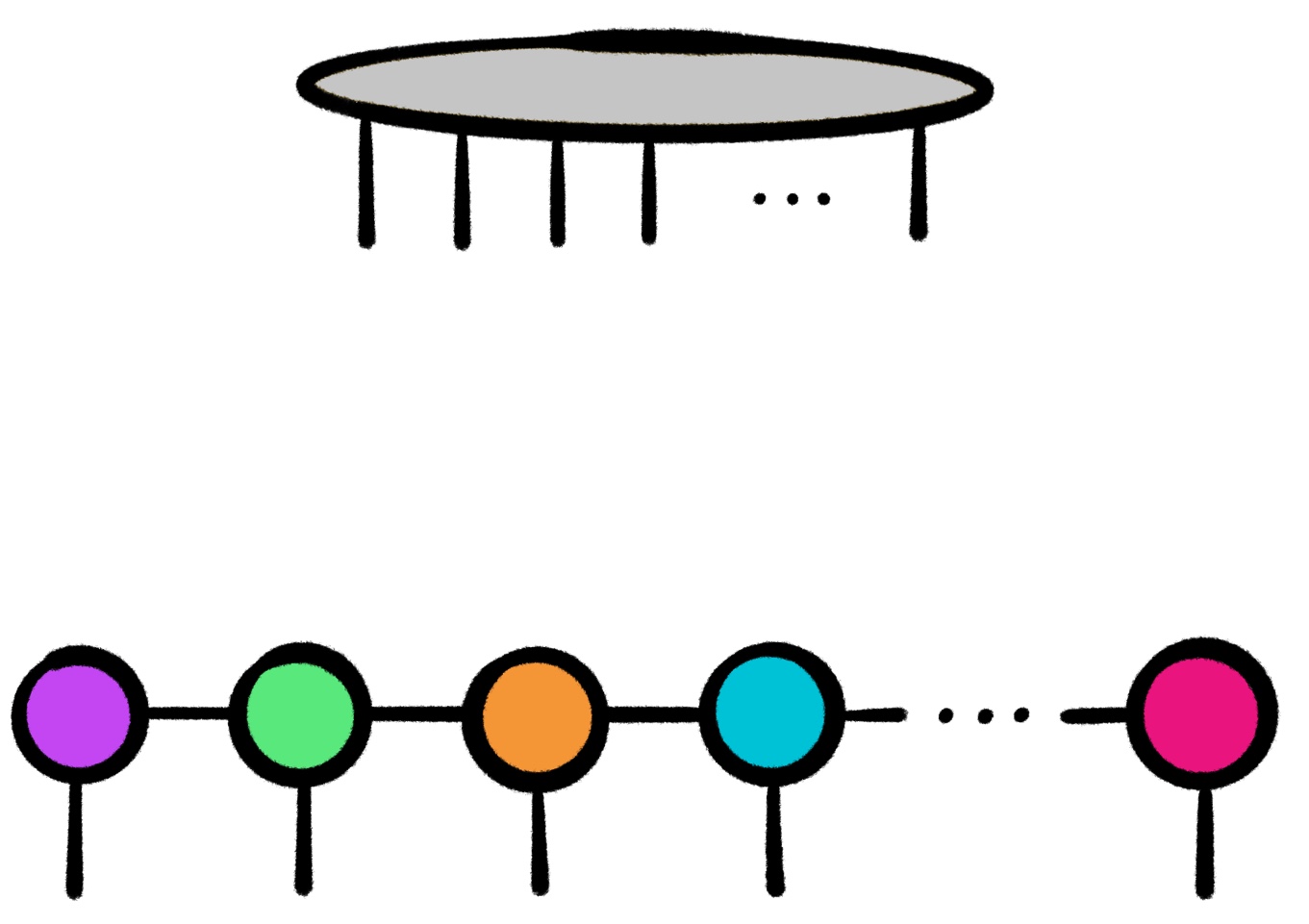}};
	\draw[line join=round,
		decorate, decoration={
		zigzag,
		segment length=10,
		amplitude=2,post=lineto,
		post length=2pt
		}, ->, >=stealth]
		(0,.5) -- (0,-.5); 
	\node at (0,2) {start with $|\psi\>$};
	\node at (0,-2) {finish with $|\psi_{\text{MPS}}\>$};
	\end{tikzpicture}
	\end{center}
	}

This is how the mathematics of Chapter \ref{ch:probability} arises in an applied setting. In this chapter, we'll give the bird's-eye view of the ideas. The first thing to know is that \textit{there exists} an algorithm that successfully accomplishes the task just described. It is motivated by and built from the mathematics of Chapters \ref{ch:preliminaries}--\ref{ch:probability}. More formally, it is inspired by the \textit{density matrix renormalization group} procedure---a procedure used in physics to construct the ground state of a Hamiltonian operator \cite{schollwock,white1992}. Here, we use it to reconstruct a joint probability distribution modeled as a quantum state. The algorithm we share is deterministic and requires only simple tools from linear algebra. As a result, it's possible to predict just how \textit{well} the algorithm will work given only the number of samples in the training set $T$. What's more, the vector $|\psi_{\text{MPS}}\>$ turns out to be an example of a tensor network known as an MPS (see Section \ref{ssec:TNs}).
		\marginnote[-3cm]{Although we are using the tools and language of quantum physics, it is \textit{not} necessary that the set or probability distribution being modeled have any interpretation of being quantum. But see Sidenote 2 in Section \ref{ssec:bird} for a remark about \textit{which} probability distributions the algorithm is intended.}
We wish to keep the focus on the theory-side-of-things, so rather than giving a detailed description of the algorithm, we will only share the main highlights. We'll then close with a few words about an experiment that tests it. As we'll soon see, the procedure to build $|\psi_{\text{MPS}}\>$ is like a machine that reaps what we've sown in $|\psi\>$: it sweeps across the training samples, builds reduced density operators from them, and glues their eigenvectors together to form a sequence of tensors that become $|\psi_{\text{MPS}}\>$. Readers are referred to our joint work with Stoudenmire and Terilla in \cite{bradley2019modeling} for any details omitted here.

%I gratefully attribute the algorithm to Stoudenmire and Terilla and refer readers to our joint work in \cite{bradley2019modeling} for any details omitted here.

\section{Building a Generative Model}\label{sec:experiment}
To build the model $|\psi_{\text{MPS}}\>$ we first need a training set $T$ and an initial probability distribution $\hat\pi$. We choose $T$ to be a set of even-parity bitstrings of length $N>2$, and will let $\hat\pi$ be the probability distribution uniformly concentrated on that set. Let's describe this more carefully.

\subsection{Initial Inputs}
		\marginnote{\smallcaps{Behind the Scenes.} Think of the set $A$ as an alphabet for some language, and think of $A^N$ as the set of all possible ``phrases'' (of length $N$) in that language. The word \textit{all} is important. Not every phrase will be a valid expression in the language. Indeed, not every concatenation of words in English produces a meaningful sentence. Here we are defining \textit{meaningfulness} to be based on probability. A phrase is meaningless if, essentially, the probability of it appearing in the language is zero, and it's meaningful otherwise. In this bitstring example, we are therefore declaring even-parity bitstrings to be the only meaningful expressions in this simple langauge of $0$s and $1$s.}
Begin with the finite set $A=\{0,1\}$ and consider the $N$-fold Cartesian product $A^N=A\times \cdots\times A$, which coincides with the set of all bitstrings of length $N$. Later in an experiment, we will ultimately take $N=16$, but let's be flexible with this number. In a few examples we'll suppose $N=5$ for simplicity. Now observe that bitstrings of length $N$ always come in two ``flavors'' based on their parity---even or odd. A bitstring $s\in A^N$ is called \emph{even} if it has an even number of $1$s, and it is called \emph{odd} if it has an odd number of $1$s. Let $E_N:=\{s\in A^N\mid s\text{ is even}\}$ be the set of even bitstrings; let $O_N:=\{s\in A^N\mid s\text{ is odd}\}$ be the set of odd bitstrings. So $00110\in E_5$ while $00111\in O_5$. Observe that $A^N=E_N\cup O_N$ and that $|E_N|=|O_N|=2^{N-1}$. Now suppose $T\subseteq E_N$ is a subset of $N_T$ even-parity bitstrings, for some number $N_T\leq 2^{N-1}$.  This defines an empirical probability distribution $\hat\pi\colon A^N\to\mathbb{R}$ by
\[
\hat\pi(s) =
\begin{cases}
\frac{1}{N_T} &\text{if $s\in T$},\\
0 &\text{otherwise}.
\end{cases}
\]
Think of $\hat\pi$ as a rough approximation of the ``true'' probability distribution $\pi\colon A^N\to \mathbb{R}$ uniformly concentrated on the set of all even bitstrings, where $\pi(s)=1/2^{N-1}$ if $s$ is even and $\pi(s)=0$ otherwise. In a real-world application, we would likely not know or have access to the target distribution $\pi$ (suppose the elements of $A$ were English words scraped from the Internet, for instance), and so the goal is to find a model for it. Of course we \textit{do} know $\pi$ exactly in this toy example, but let us suppose that we \textit{don't} for the sake of illustrating the idea.

Then the first step in finding a model for $\pi$ is the passage from sets to vector spaces outlined in Chapter \ref{ch:probability}. In particular, we start by representing the empirical distribution $\hat\pi$ as a pure quantum state. To do this, begin by letting $V=\mathbb{C}^A\cong\mathbb{C}^2$ be the two-dimensional vector space with basis vectors $|0\>=\begin{bsmallmatrix}1\\0\end{bsmallmatrix}$ and $|1\>=\begin{bsmallmatrix}0\\1\end{bsmallmatrix}$ so that every sequence $s\in T$ corresponds to a vector in $|s\>\in V^{\otimes N}\cong \mathbb{C}^2\otimes\cdots\otimes\mathbb{C}^2$. For example the sequence $00101\in E_5$ maps to the basis vector $|0\>\otimes|0\>\otimes|1\>\otimes|0\>\otimes|1\>\in V^{\otimes 5}$. The probability distribution $\hat\pi$ therefore defines a unit vector in $V^{\otimes N}$ as the sum of all samples in $T$ weighted by the square roots of their probabilities,
\begin{equation}\label{eq:psi_exp}
|\psi\>=\frac{1}{\sqrt N_T}\sum_{s\in T}|s\>.
\end{equation}
The tensor network diagram for $|\psi\>$ is the following.
\begin{center}
\begin{tikzpicture}
\node at (0,0) {\includegraphics[scale=0.1]{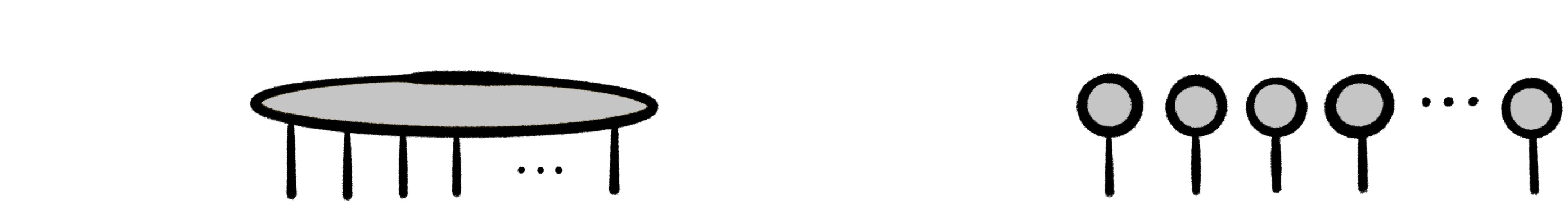}};
\node at (-4,0) {$|\psi\>\;=$};
\node at (-.1,0) {$=$};
\node at (1,0) {$ {\mathlarger{\frac{1}{\sqrt N_T}\sum }}$};
\end{tikzpicture}
\end{center}

\noindent This vector serves as input for the algorithm that outputs a new vector $|\psi_{\text{MPS}}\>\in V^{\otimes N}$. As explained in the introduction, $|\psi_{\text{MPS}}\>$ will have the property that its coefficients define a probability distribution on $A^N$ which is closer to the distribution for which $\hat\pi$ was merely an approximation; that is, $|\psi_{\text{MPS}}\>$ will satisfy the property that $|\<s|\psi_{MPS}\>|^2\approx \pi(s)=1/2^{N-1}$ for all even bitstrings $s$. We also mentioned that $|\psi_{\text{MPS}}\>$ will turn out to be an MPS, which is a particular tensor network factorization of $|\psi\>$. Think of this as the higher-dimensional analogue to approximating a matrix by a truncated SVD. We won't need the details, but interested readers may see the references in Section \ref{ssec:TNs} for more on MPS. In any case, the procedure to build $|\psi_{\text{MPS}}\>$ consists of $N$ steps where the task in step $k$ is to construct the $k$th tensor of $|\psi_{\text{MPS}}\>$. Conveniently, we already wrote the instruction manual in Chapter \ref{ch:probability}. In step $k$ we will define a reduced density operator whose eigenvectors assemble into the $k$th tensor of $|\psi_{\text{MPS}}\>$. Even better, the procedure is an inductive one. Steps $3,\ldots,N-1$ are all nearly identical to step $2$. 
\begin{center}
\begin{figure}[h!]
\begin{tikzpicture}
\node at (0,0) {\includegraphics[scale=0.1]{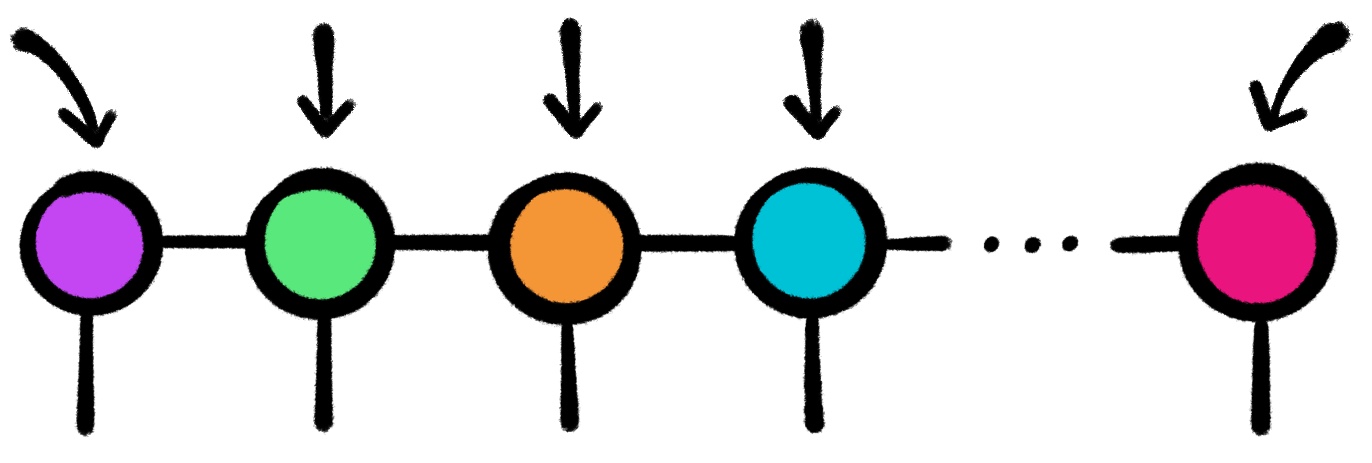}};
\node at(-2.5,1) {\footnotesize step $1$};
\node at(-1.4,1) {\footnotesize step $2$};
\node at(-.3,1) {\footnotesize step $3$};
\node at(.7,1) {\footnotesize step $4$};
\node at(2.8,1) {\footnotesize step $N$};
\node at (-3.2,0) {$|\psi_{\text{MPS}}\>=$};
\end{tikzpicture}
\caption{Building $|\psi_{\text{MPS}}\>$ by constructing each tensor individually}
\label{fig:mps}
\end{figure}
\end{center}

\noindent To streamline the discussion, we will only explain steps $1$ and $2$, being sure to share examples and intuition along the way. This abridged approach will equip interested readers to read through the original work in \cite{bradley2019modeling}, which contains all of the details that we omit. So let's now proceed to describe the construction of the first two tensors of $|\psi_{\text{MPS}}\>$ in Figure \ref{fig:mps}.

\subsection{A Bird's-Eye View}\label{ssec:bird}
The first step is mostly a formality. Define the first tensor comprising $|\psi_{\text{MPS}}\>$ to be the identity operator $\id_V$ on $V=\mathbb{C}^2$. As an array, it is simply the $2\times 2$ identity matrix, \begin{tikzpicture}[baseline=-.2pt]  \draw[thick] (-.5,0)--(.5,0) {}; \node[tensor,fill={rgb, 255:red,195;green,70;blue,242}] at (0,0) {};\end{tikzpicture} $=\begin{bsmallmatrix}1&0\\0&1\end{bsmallmatrix}$. The second step is much more interesting.

\bigskip
\noindent{\Large \itshape Step 2: Apply the Recipe in Chapter \ref{ch:probability}}.
\bigskip

\noindent  Let's now build the second tensor of $|\psi_{\text{MPS}}\>$ in Figure \ref{fig:mps}. Doing so requires a few steps. First, we start by ``modifying'' $|\psi\>$ by composing the first edge with the identity operator. Call this vector $|\psi_2\>$. 
\begin{center}
\begin{tikzpicture}
\node at (0,0) {\includegraphics[scale=.15]{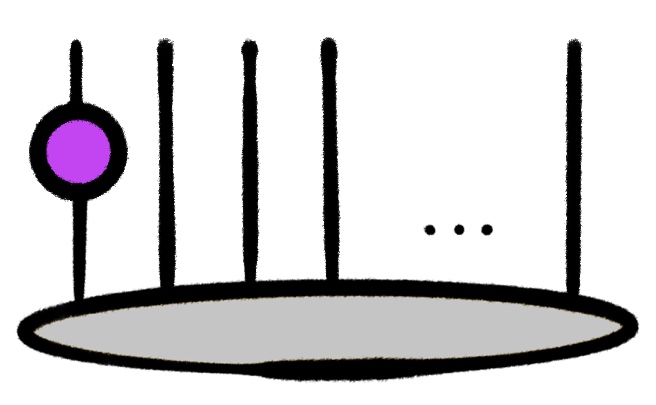}};
\node at (-2.5,-.5) {$|\psi_2\>=$};
\end{tikzpicture}
\end{center}
Here's how to understand $|\psi_2\>$ more carefully. Start with $|\psi\>\in V^{\otimes N}$ and reshape it into a linear map $M\colon V\to V^{\otimes N-1}$. Then precompose this map with the identity operator $M\id_V\colon V\to  V^{\otimes N-1}$ and reshape it back into a vector in $V^{\otimes N}.$ This is the vector we're calling $|\psi_2\>$. See Figure \ref{fig:step1}.
	\begin{center}
	\begin{figure}[h!]
	\begin{tikzpicture}
	\node at (0,0) {\includegraphics[scale=0.1]{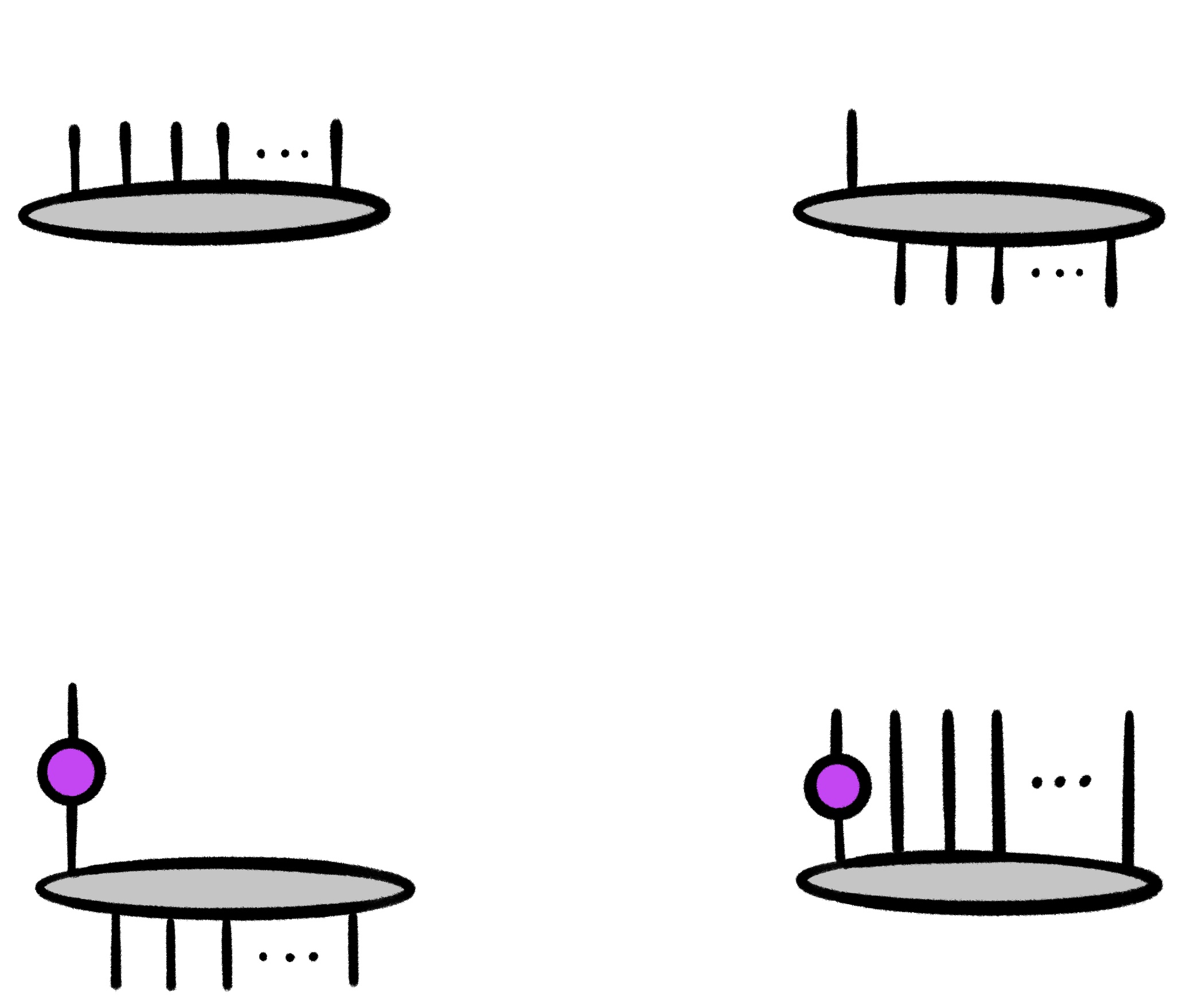}};
	\node at (-4.5,2) {$|\psi\>=$};
	\node at (4.5,2) {$=M$};
	\node at (-4.5,-2) {$M\id_V=$};
	\node at (4.5,-2) {$=|\psi_2\>$};
	\draw[thick, line join=round,
			decorate, decoration={
			zigzag,
			segment length=10,
			amplitude=2,post=lineto,
			post length=2pt
			}, ->, >=stealth]
			(-.8,2) -- (.8,2); 
	\draw[thick, line join=round,
			decorate, decoration={
			zigzag,
			segment length=10,
			amplitude=2,post=lineto,
			post length=2pt
			}, ->, >=stealth]
			(-.8,-2.2) -- (.8,-2.2); 
	\draw[thick, line join=round,
			decorate, decoration={
			zigzag,
			segment length=10,
			amplitude=2,post=lineto,
			post length=2pt
			}, ->, >=stealth]
			(1,.5) -- (-1,-1); 
	\end{tikzpicture}
	\caption{Obtaining the vector $|\psi_2\>$ from $|\psi\>$}
	\label{fig:step1}
	\end{figure}
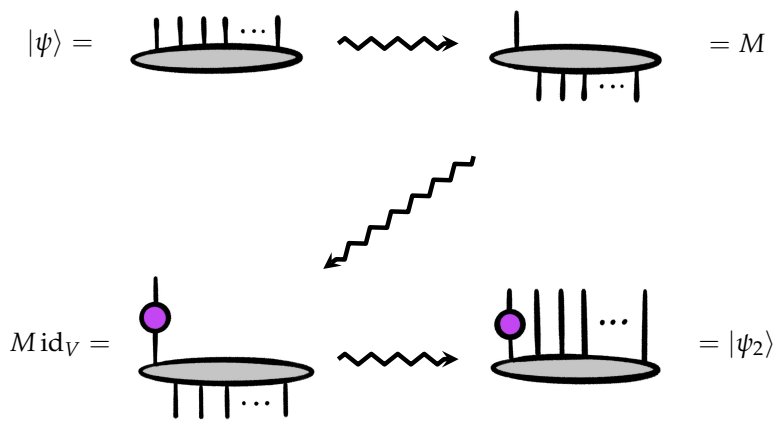
	\end{center}
Of course $|\psi_2\>=|\psi\>$ since composition with the identity changes nothing, but we have belabored this point because another version of it will appear soon. In any case, we are now in familiar territory. Consider the orthogonal projection operator $|\psi_2\>\<\psi_2|\colon V^{\otimes N}\to V^{\otimes N}$ and apply the partial trace to the last $N-2$ factors.\sidenote[][-8cm]{\smallcaps{Behind the Scenes.} The reason we choose to trace out the last $N-2$ factors, as opposed to some other number, is due to an important practical point: We're dealing with high-dimensional linear algebra! The state $|\psi\>$ lives in a space that grows exponentially in $N$ and orthogonal projection onto it is an operator on this space. Fully realizing it and trying to find a direct approximation won't be feasible. This motivates the use of \textit{reduced} densities. Indeed, it's more manageable to work with operators on smaller subsystems, which can then give us knowledge of the whole. To do this, we are exploiting the fact that our data are \textit{sequences}, which naturally gives us a way to subdivide the system into smaller ones. }
%Indeed, every bitstring $s$ of length $N$ can be ``sliced'' to obtain a prefix-suffix pair $s=(x,y)$, where $x$ is a bitstring of length $k$ and $y$ is a bitstring of length $N-k$. The set of all prefixes forms a \textit{prefix} subsystem $A=V^{\otimes k}$, and the set of all suffixes forms a \textit{suffix} subsystem $B=V^{\otimes N-k}$. This provides a tensor decomposition $V^{\otimes N}=A\otimes B$, and the algorithm described here is concerned with reduced densities on the prefix system $\rho_A.$ To ensure that all is manageable, we`ll want to keep the dimension of the image of $\rho_A$ small. This means we won't use $\rho_A$ \textit{exactly}. Instead we'll use a low-rank approximation of it. See Sidenote 2 for more on this.} 
This gives a reduced density operator $\rho_2:=\tr_{V^{\otimes N-2}}|\psi_2\>\<\psi_2|\colon V\otimes V\to V\otimes V$.
\begin{center}
\begin{tikzpicture}
\node at (0,0) {\includegraphics[scale=.15]{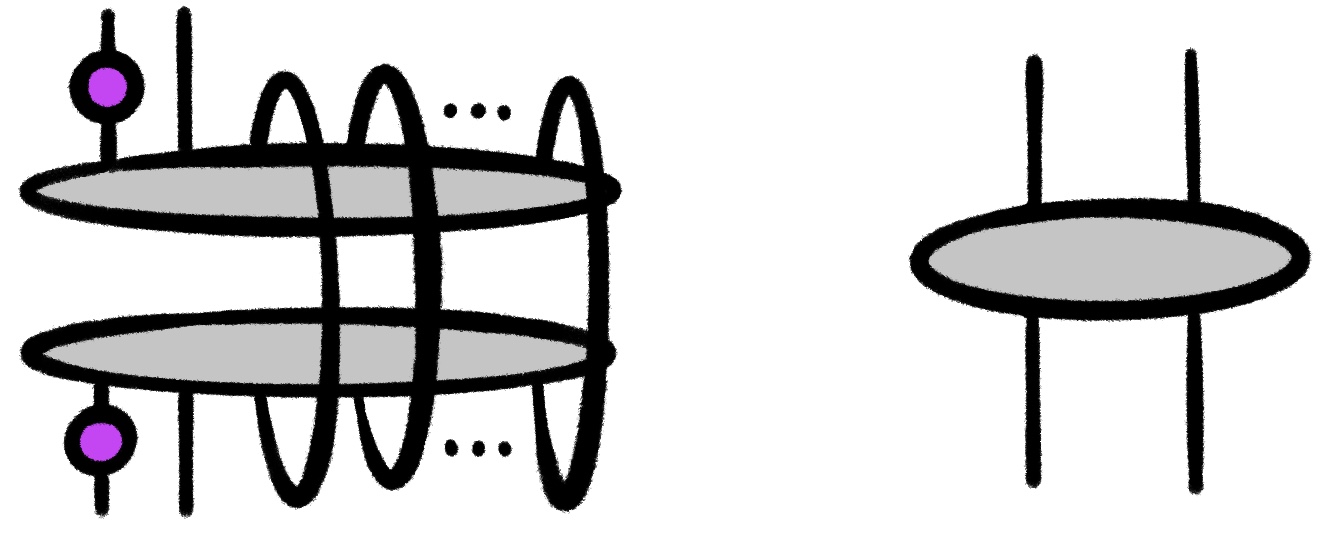}};
\node at (-4,0) {$\rho_2\;=$};
\node at (.6,0) {$=$};
\end{tikzpicture}
\end{center}
Since $\rho_2$ is a density, it has a spectral decomposition $\rho_2=UDU^\dagger$. The matrix $D$ is $4\times 4$ diagonal containing the eigenvalues of $\rho_2$, and $U$ is a $4\times 4$ unitary operator with the four eigenvectors of $\rho_2$ as its columns. In general $\rho_2$ will be full rank, but for reasons explained in Sidenotes 1 and 2 let us \textit{drop} the eigenvectors associated to the two smallest eigenvalues.\sidenote[][-8cm]{Resuming the conversation in Sidenote 1, we are forcing our reduced density to operate on a subsystem that is two-dimensional. (Example \ref{ex:perfection} will shed light on why the number \textit{two} was chosen.) To borrow language from Section \ref{ssec:entanglement}, the algorithm assumes $|\psi\>$ has low entanglement properties. To borrow language from machine learning, this is our \textit{model hypothesis.}

The reason we keep the rank \textit{low} is due to a hypothesis that is more philosophical than mathematical. It is the idea that data systems that arise in nature are not random but rather have underlying structure. For instance, natural language is not comprised of random sequences of characters---words have meaning! So the algorithm being described gives a way to summarize information in a structured dataset. In particular, the eigenvalues of $\rho_2$ stratify the importance of information in the subsystem on which it operators. By choosing to keep only the largest eigenvalues, we have control over how fine-tuned our summaries will be. For instance, dropping the smaller eigenvalues is akin to getting rid of sampling errors and thereby extracting the inherent meaning and structure beneath. As we'll see in Section \ref{sec:theexperiment}, this allows one to learn interesting probability distributions extremely well from a remarkably small amount of samples.}
% The algorithm described here builds a model for probability distributions $\pi$ whose corresponding quantum state $|\psi\>\<\psi|$ has \textit{low rank} reduce densities $\rho_A$ and $\rho_B$. To borrow language from Section \ref{ssec:entanglement}, the algorithm assumes $|\psi\>$ has low entanglement properties. To borrow language from machine learning, this is our \textit{model hypothesis.}
Use the same letter $U$ to denote the resulting  $4\times 2$ matrix. If $W$ is the span of these top two eigenvectors, then we have a linear isometry $U\colon W\to V\otimes V$ whose adjoint is $U^\dagger \colon V\otimes V\to W$. Note that the $2\times 2$ diagonal containing the two largest eigenvalues of $\rho_2$---let's also call it $D$---is an isomorphism of $W$.
\begin{center}
\begin{tikzpicture}
\node at (0,0) {\includegraphics[scale=0.15]{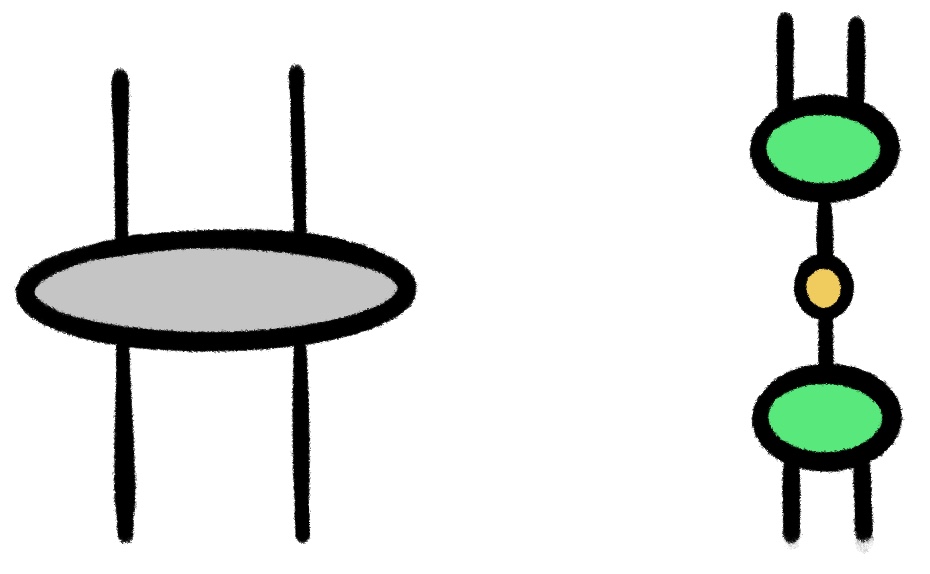}};
\node at (-2.8,0) {$\rho_2=$};
\node at (.5,0) {=};
\node at (3,1) {$U^\dagger$};
\node at (3,0) {$D$};
\node at (3,-1) {$U$};
\end{tikzpicture}
\end{center}
Now here's the point of step $2$: we take $U$ to be the \textit{second} tensor comprising $|\psi_{\text{MPS}}\>$.
	\marginnote{\normalsize\[
	\begin{tikzcd}[ampersand replacement = \&]
	V\otimes V \arrow[d, "U^\dagger"'] \arrow[r, "\rho_2"] \& V\otimes V              \\
	W \arrow[r, "D"]                                       \& W \arrow[u, "U"', hook]
	\end{tikzcd}\]}
\begin{center}
\begin{tikzpicture}
\node at (0,0) {\includegraphics[scale=0.1]{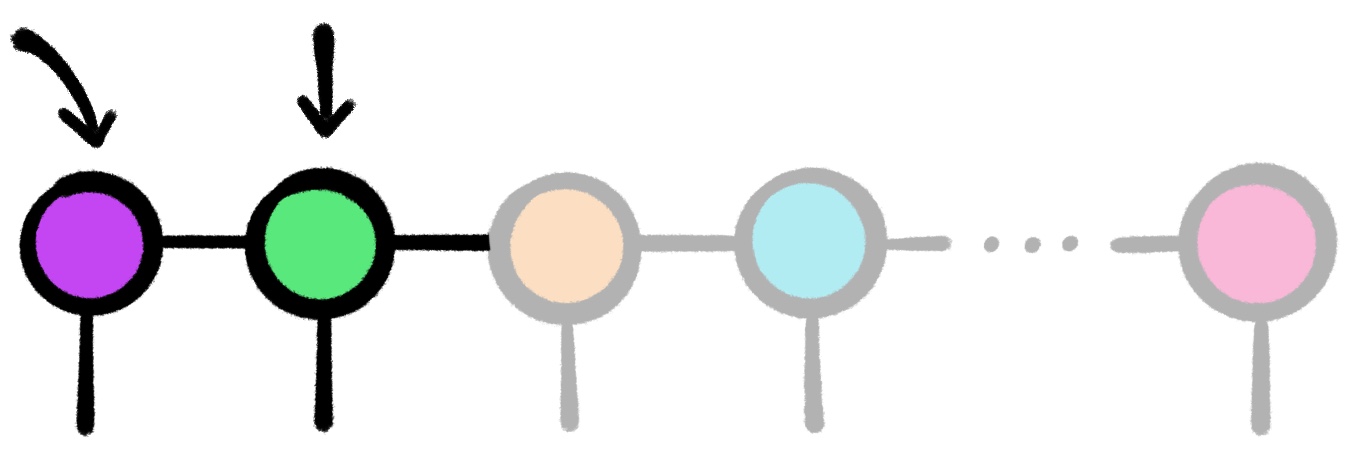}};
\node at(-2.5,1) {$\id_V$};
\node at(-1.3,1) {$U$};
\end{tikzpicture}
\end{center}
To understand what this means, think back to the initial vector $|\psi_2\>\in V^{\otimes N}$. We may define a \textit{new} vector $|\psi_3\>$ by composing $U$ with the first two edges of $|\psi_2\>$. 
\begin{center}
\begin{tikzpicture}
\node at (0,0) {\includegraphics[scale=0.15]{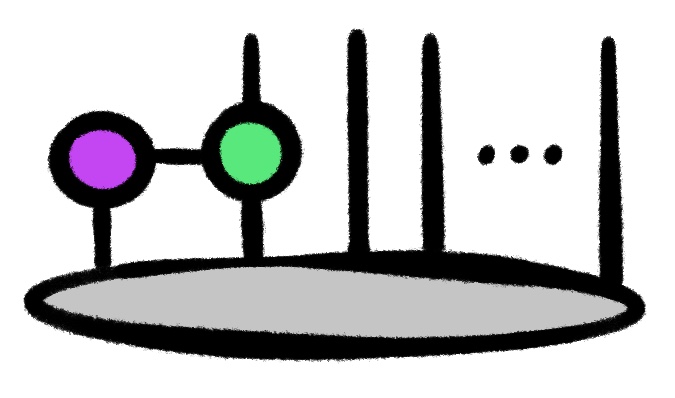}};
\node at (-2.5,-.5) {$|\psi_3\>=$};
\end{tikzpicture}
\end{center}
Said more carefully, start with $|\psi_2\>$ and view it as a linear mapping $M_2\colon V^{\otimes 2}\to V^{\otimes N-2}$, then precompose this mapping with $U$, then reshape the resulting composition $M_2U\colon W\to V^{\otimes N-2}$ into the vector $|\psi_3\>\in W\otimes V^{\otimes N-2}$. See Figure \ref{fig:step2}.
\begin{center}
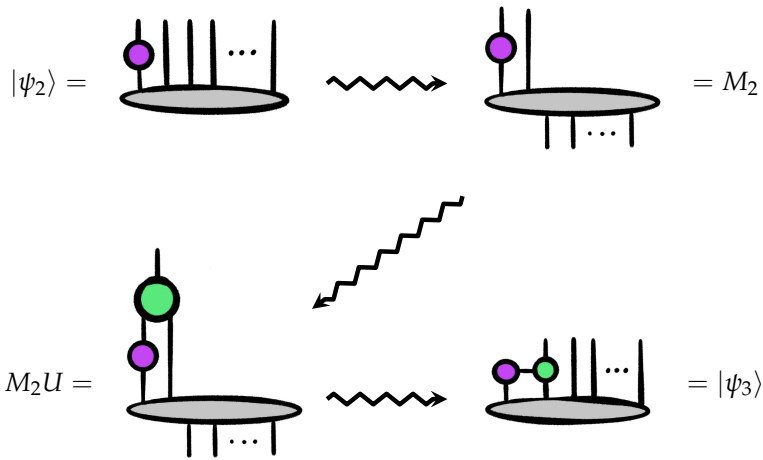
\begin{figure}[h!]
\begin{tikzpicture}
\node at (0,0) {\includegraphics[scale=0.1]{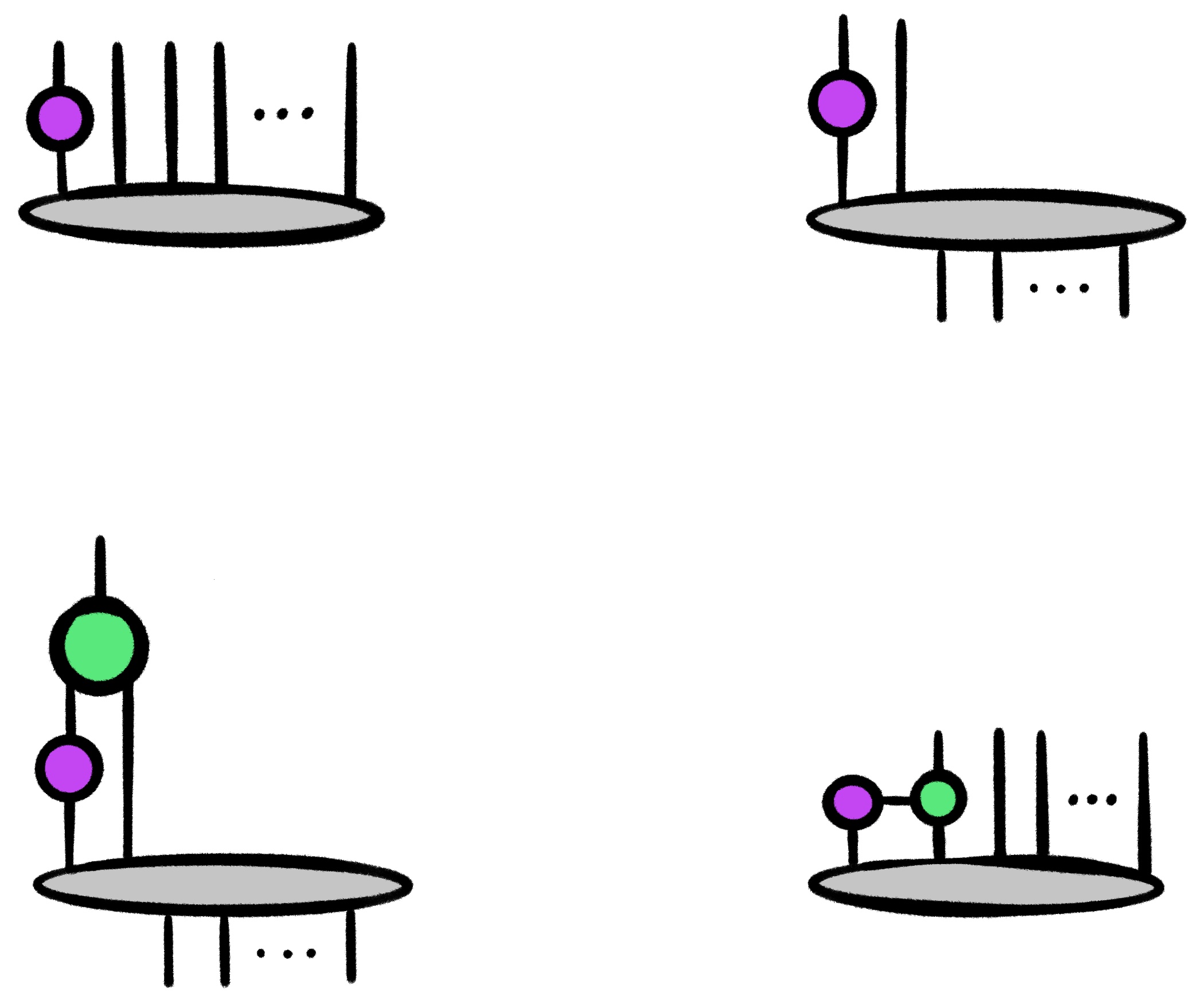}};
\node at (-4.5,2) {$|\psi_2\>=$};
\node at (4.5,2) {$=M_2$};
\node at (-4.5,-2) {$M_2U=$};
\node at (4.5,-2) {$=|\psi_3\>$};
\draw[thick, line join=round,
		decorate, decoration={
		zigzag,
		segment length=10,
		amplitude=2,post=lineto,
		post length=2pt
		}, ->, >=stealth]
		(-.8,2) -- (.8,2); 
\draw[thick, line join=round,
		decorate, decoration={
		zigzag,
		segment length=10,
		amplitude=2,post=lineto,
		post length=2pt
		}, ->, >=stealth]
		(-.8,-2.2) -- (.8,-2.2); 
\draw[thick, line join=round,
		decorate, decoration={
		zigzag,
		segment length=10,
		amplitude=2,post=lineto,
		post length=2pt
		}, ->, >=stealth]
		(1,.5) -- (-1,-1); 
\end{tikzpicture}
\caption{Obtaining the vector $|\psi_3\>$ from $|\psi_2\>$}
\label{fig:step2}
\end{figure}
\end{center}
Now you can see the MPS starting to form within the diagram associated to $|\psi_3\>$. This concludes step $2$ of the algorithm.

\newthought{Perhaps now you} have an idea of what step $3$ will be. The similarities between Figure \ref{fig:step1} and Figure \ref{fig:step2} suggest the pattern. Let's quickly make the connection. Step $3$ starts with the vector $|\psi_3\>$. The first two ``sites'' of $|\psi_3\>\<\psi_3|$ remain untouched while the partial trace is applied to the remaining sites. This gives a reduced density operator $\rho_3$, whose spectral decomposition is $\rho_3=U_3D_3U_3^\dagger$
\begin{center}
\begin{tikzpicture}
\node at (0,0) {\includegraphics[scale=.15]{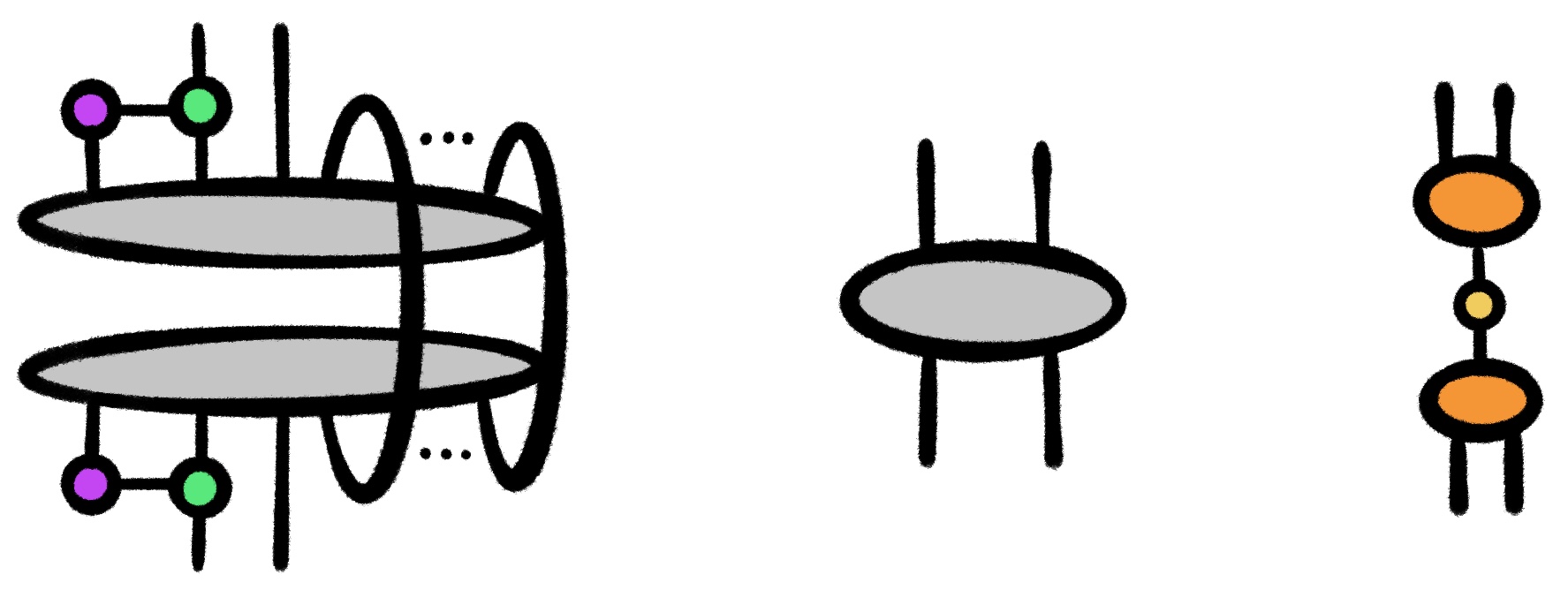}};
\node at (-5,0) {$\rho_3\;=$};
\node at (-.4,0) {$=$};
\node at (3,0) {$=$};
\end{tikzpicture}
\end{center}
We take $U_3$ to be the $4\times 2$ array whose two columns are the two eigenvectors of $\rho_3$ associated to its two largest eigenvalues. This forms the third tensor of $|\psi_{\text{MPS}}\>$. It also tells us how to define the vector $|\psi_4\>$ that initiates step $4$, which leads to $U_4$, and so on. See Figure \ref{fig:step3}.
\begin{center}
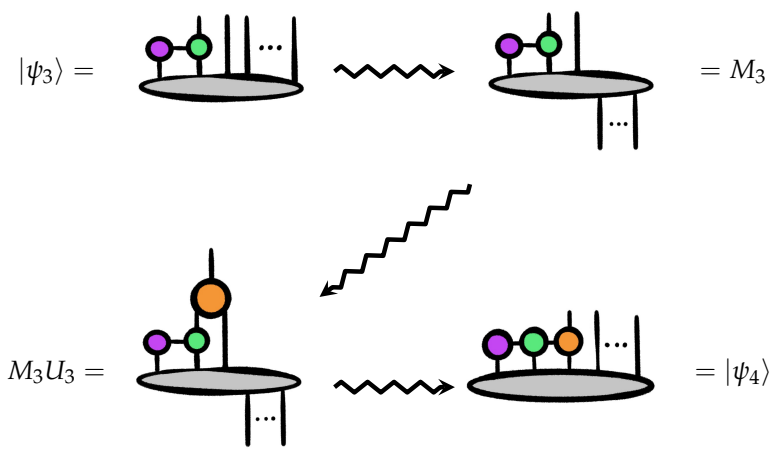
\begin{figure}[h!]
\begin{tikzpicture}
\node at (0,0) {\includegraphics[scale=0.1]{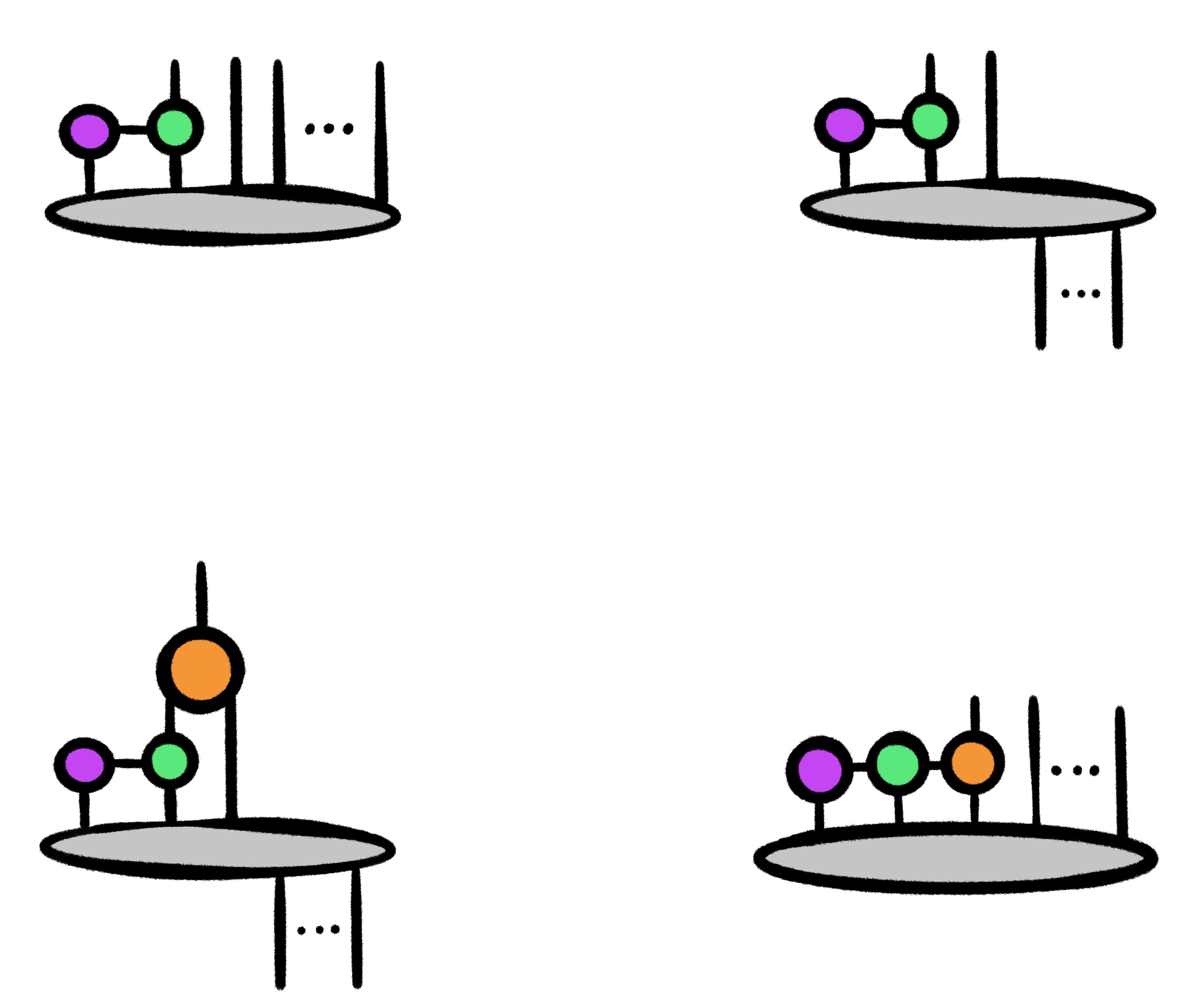}};
\node at (-4.5,2) {$|\psi_3\>=$};
\node at (4.5,2) {$=M_3$};
\node at (-4.5,-2) {$M_3U_3=$};
\node at (4.5,-2) {$=|\psi_4\>$};
\draw[thick, line join=round,
		decorate, decoration={
		zigzag,
		segment length=10,
		amplitude=2,post=lineto,
		post length=2pt
		}, ->, >=stealth]
		(-.8,2) -- (.8,2); 
\draw[thick, line join=round,
		decorate, decoration={
		zigzag,
		segment length=10,
		amplitude=2,post=lineto,
		post length=2pt
		}, ->, >=stealth]
		(-.8,-2.2) -- (.8,-2.2); 
\draw[thick, line join=round,
		decorate, decoration={
		zigzag,
		segment length=10,
		amplitude=2,post=lineto,
		post length=2pt
		}, ->, >=stealth]
		(1,.5) -- (-1,-1); 
\end{tikzpicture}
\caption{Obtaining the vector $|\psi_4\>$ from $|\psi_3\>$}
\label{fig:step3}
\end{figure}
\end{center}
Continuing the pattern at each step, the upshot is that the arrays of eigenvectors $U=U_2,U_3,\ldots, U_N$ assemble into the desired vector $|\psi_{\text{MPS}}\>$.
\begin{center}
\begin{tikzpicture}
\node at (0,0) {\includegraphics[scale=0.1]{figures/ch4/mpsch4}};
\node at(-2.5,1) {$\id_V$};
\node at(-1.3,1) {$U_2$};
\node at(-.4,1) {$U_3$};
\node at(.5,1) {$U_4$};
\node at(2.5,1) {$U_N$};
\end{tikzpicture}
\end{center}
The claim is that $|\psi_{\text{MPS}}\>$ has the property that the distribution induced by the Born rule is close to the true uniform distribution on even-parity bitstrings, $|\<s|\psi{MPS}\>|^2\approx1/2^{N-1}$, for which the empirical probabilities $\hat\pi(s)=1/N_T$ were merely an approximation. Let's now discuss the intuition for \textit{why} this is so, and then we'll report experimental results that illustrate the claim.

\section{Intuition: Why Does This Work?}

Intuition is gained through an example. Let's take a closer inspection at the eigenvectors of the reduced density operator $\rho_2$ obtained in step 2 by working with a simple example.

\begin{example}\label{ex:perfection}
Fix $N=5$ and suppose the training set contains the \textit{entire} set of even-parity bitstrings, $T=E_N$. The goal for this example is to compute the eigenvectors of $\rho_2$ explicitly. So first we'll need to find the matrix representation of $\rho_2$. To do so, recall that the procedure described in Section \ref{ssec:bird} distinguished between the first two bits of a sample, and the remaining $N-2$ bits. In light of this, let's think of every even-parity bitstring $s\in T$ as a  prefix-suffix pair $s=(x,y)$, where $x\in A^2$ is a bitstring of length $2$ and $y\in A^{3}$ is a bitstring of length $3$.  Since $s$ must have even parity, $(x,y)\in T$ if and only if $x$ and $y$ both have the same parity. There is a nice way to represent $T$ visually---it corresponds to a bipartite graph. The sets of prefixes and suffixes define the two sets of vertices, and there is an edge joining $x$ and $y$ if and only if their concatenation $xy$ is even. For example, the training set $T=E_5$ is illustrated by the following graph.
	\begin{center}
	\begin{tikzpicture}
	\node at (0,0) {\includegraphics[scale=0.15]{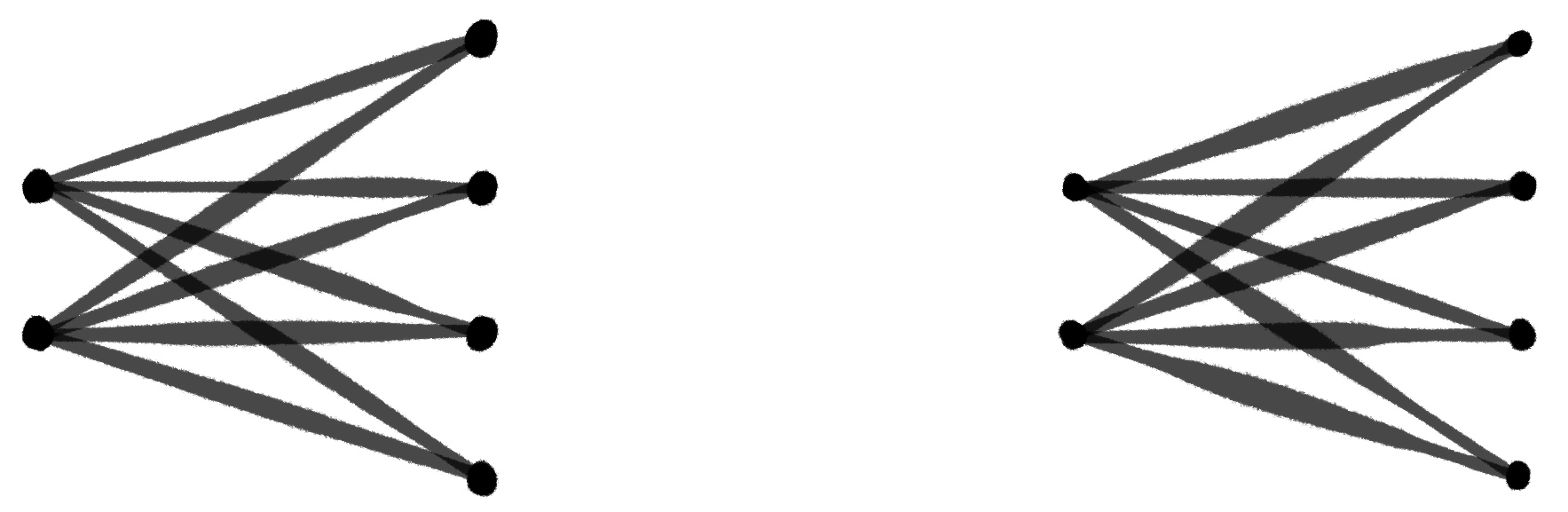}};

	\node at (-1.2,1.3) {$000$};
	\node at (-1.2,.4) {$110$};
	\node at (-1.2,-.5) {$011$};
	\node at (-1.2,-1.3) {$101$};

	\node at (-5,.4) {$00$};
	\node at (-5,-.5) {$11$};

	\node at (1.2,.4) {$01$};
	\node at (1.2,-.5) {$10$};

	\node at (5,1.3) {$100$};
	\node at (5,.4) {$010$};
	\node at (5,-.5) {$001$};
	\node at (5,-1.3) {$111$};
	\end{tikzpicture}
	\end{center}
The graph has two disconnected components because there is no sample $(x,y)$ in $T$ where $x$ and $y$ have different parity. For example, there is no edge joining $00$ and $111$. So the left component of the graph contains all \textit{even-even} prefix-suffix pairs, while the right component contains all \textit{odd-odd} prefix-suffix pairs. Now recall from Section \ref{sec:graphs} that the entries of reduced density operators defined by such a graph may be read off directly \textit{from} the graph. To be clear, $\rho_2$ is an operator on the prefix subsystem $V\otimes V$, a four-dimensional space whose bases are the four bitstrings of length two, $|00\>,|11\>,|01\>$ and $|01\>$. As a $4\times 4$ matrix, the rows and columns of $\rho_2$ correspond to these four prefixes. Choosing the ordering $A^2=\{00,11,01,10\}$ we may write $\rho_2$ as
\begin{equation*}
\rho_2\quad=\quad
\kbordermatrix{
	& {\color{gray}00} & {\color{gray}11} & {\color{gray}01} & {\color{gray}10}\\
	{\color{gray}00} & \text{deg $00$} & s_e & 0 & 0  \\[5pt]
	{\color{gray}11} & s_e & \text{deg $11$} & 0 & 0  \\[5pt]
	{\color{gray}01} & 0 & 0 & \text{deg $01$} & s_o  \\[5pt]
	{\color{gray}10} & 0 & 0 & s_o & \text{deg $10$}
}\frac{1}{2^4}
\quad=\quad
\frac{1}{2^4}
\begin{bmatrix}
4 & 4&0&0\\
4&4&0&0\\
0&0&4&4\\
0&0&4&4
\end{bmatrix}
\end{equation*}
The diagonal entries are the degrees of the four prefixes. The off diagonal $s_e$ is the number of suffixes that the even prefixes $00$ and $11$ have in common. It is also the number of paths of length two that join $00$ and $11$, and is equivalently the number of suffixes with degree 2 in the graph's left component.
	\marginnote[-1cm]{\smallcaps{Behind the Scenes.} By choosing the ordering $A^2=\{00,11,01,10\}$ I am associating the $i$th element to the $i$th standard basis vector in $\mathbb{C}^2\otimes\mathbb{C}^2\cong\mathbb{C}^4$. So $|00\>:=\begin{bmatrix}1&0&0&0\end{bmatrix}^\top$ and $|11\>:=\begin{bmatrix}0&1&0&0\end{bmatrix}^\top$ and so on. This does not match with the convention that $|ab\>=|a\>\otimes|b\>$ since, for instance, $|11\>$ would be $|1\>\otimes|1\>=\begin{bsmallmatrix}0\\1\end{bsmallmatrix}\otimes \begin{bsmallmatrix}0\\1\end{bsmallmatrix}$ which gives $\begin{bmatrix}0&0&0&1\end{bmatrix}^\top$. We've temporarily suspended this convention and are sticking with the choice above for the purpose of sharing intuition. With this choice, the two blocks in $\rho_2$ correspond to the two components of the graph, which is nice!}
The off-diagonal $s_o$ is likewise the number of suffixes that the odd prefixes $01$ and $10$ have in common, which is the number of paths of length two that join $01$ and $10$, which also is the number of degree two suffixes in the graph's right component. The matrix is normalized by the total number of edges in the graph, $2^4$, which guarantees that $\rho_2$ has trace $1.$ So the nonzero entries of $\rho_2$ are all equal to $4/2^4$. It is therefore a rank $2$ operator, and each nonzero $2\times 2$ block contributes one eigenvalue and one eigenvector. Both eigenvalues are $\frac{1}{2}$. The eigenvector from the first block is the normalized sum of even-parity bitstrings of length $2$, let's call it $|E_2\>$.
\[
|E_2\>=\frac{|00\>+|11\>}{\sqrt 2}
=\begin{bmatrix}\frac{1}{\sqrt 2}&\frac{1}{\sqrt 2}&0&0\end{bmatrix}^\top
\]
 The second eigenvector is the normalized sum of odd-parity bitstrings of length $2$, let's call it $|O_2\>$.
 \[
 |O_2\>=\frac{|01\>+|01\>}{\sqrt 2}
 = \begin{bmatrix}0&0&\frac{1}{\sqrt 2}&\frac{1}{\sqrt 2}\end{bmatrix}^\top
 \]
These two eigenvectors capture the two ``concepts'' that exist in our simple bitstring language---evenness and oddness! Moreover, they are precisely the  two\sidenote{We refer to ``the two'' columns since $D$ is really a $4\times 4$ diagonal, but the third and fourth diagonal entries are both $0$, and so the third and fourth columns of $U$ are ``zeroed out.''} columns in the matrix $U$ obtained in the spectral decomposition of $\rho_2=UDU^\dagger$.
\[
\rho_2=
\begin{bmatrix}
\frac{1}{\sqrt 2} & 0 \\[10pt]
\frac{1}{\sqrt 2} & 0 \\[10pt]
0 & \frac{1}{\sqrt 2} \\[10pt]
0 & \frac{1}{\sqrt 2} 
\end{bmatrix}
\begin{bmatrix}
\frac{1}{2} & 0 \\
0&\frac{1}{2}
\end{bmatrix}
\begin{bmatrix}
\frac{1}{\sqrt 2} & \frac{1}{\sqrt 2} & 0 & 0\\[5pt]
0 & 0 & \frac{1}{\sqrt 2} & \frac{1}{\sqrt 2}
\end{bmatrix}
\]
To gain further intuition, let's turn our attention to the adjoint of $U$. Let $W$ denote the two-dimensional space spanned by $|E_2\>$ and $|O_2\>$ and observe that the linear mapping $U^\dagger \colon V\otimes V\to W$ acts as a ``summarizer.'' It maps each prefix onto its parity.
\begin{equation*}
\begin{split}
U^\dagger|00\>=\tfrac{1}{\sqrt 2} |0\> \quad \text{(even)}\\[5pt]
U^\dagger|11\>=\tfrac{1}{\sqrt 2} |0\>\quad \text{(even)}
\end{split}
\qquad
\begin{split}
U^\dagger|01\>=\tfrac{1}{\sqrt 2} |1\>\quad \text{(odd)}\\[5pt]
U^\dagger|10\>=\tfrac{1}{\sqrt 2} |1\>\quad \text{(odd)}
\end{split}
\end{equation*}
The essential observation is two-fold:
\begin{description}
	\item First, $|00\>$ and $|11\>$ have the same image under $U^\dagger$. \textit{This captures the idea that $00$ and $11$ always share the same suffixes} in the dataset $T$, namely those with even parity.

	\item Likewise, $|01\>$ and $|10\>$ also have the same image under $U^\dagger$. \textit{This captures the fact that $01$ and $10$ always share the same suffixes} in $T$, namely those with odd parity.
\end{description}
The point is that $U^\dagger$ summarizes the prefixes while knowing information about which suffixes they can be paired with. This feature is precisely that which was emphasized in Chapter \ref{ch:probability}: \textit{Eigenvectors of reduced densities of pure entangled states capture information about complementary subsystems}.\sidenote{There's a second half to this statement, namely that \textit{the information stored in the eigenvectors can be pieced together to reconstruct the full state.} We will circle back to this point at the end of the chapter.} It's precisely this information that's lost when marginalizing to subsystems in the classical way---recall Takeaway \ref{takeaway:red_den} and the motivating example in Section \ref{ssec:memory}. 

Let us take this a step further by making the connection with tensor network diagrams. Note that the $2\times 4$ matrix representing $U ^\dagger\colon V\otimes V\to W$ can be reshaped into a $2\times 2 \times 2$ tensor of order three, which is essentially $\tfrac{1}{\sqrt 2}\begin{bsmallmatrix}1&1\\0&0\end{bsmallmatrix}$ stacked next to $\tfrac{1}{\sqrt 2}\begin{bsmallmatrix}0&0\\1&1\end{bsmallmatrix}$. Schematically, we have this picture.
\begin{center}
\includegraphics[scale=0.18]{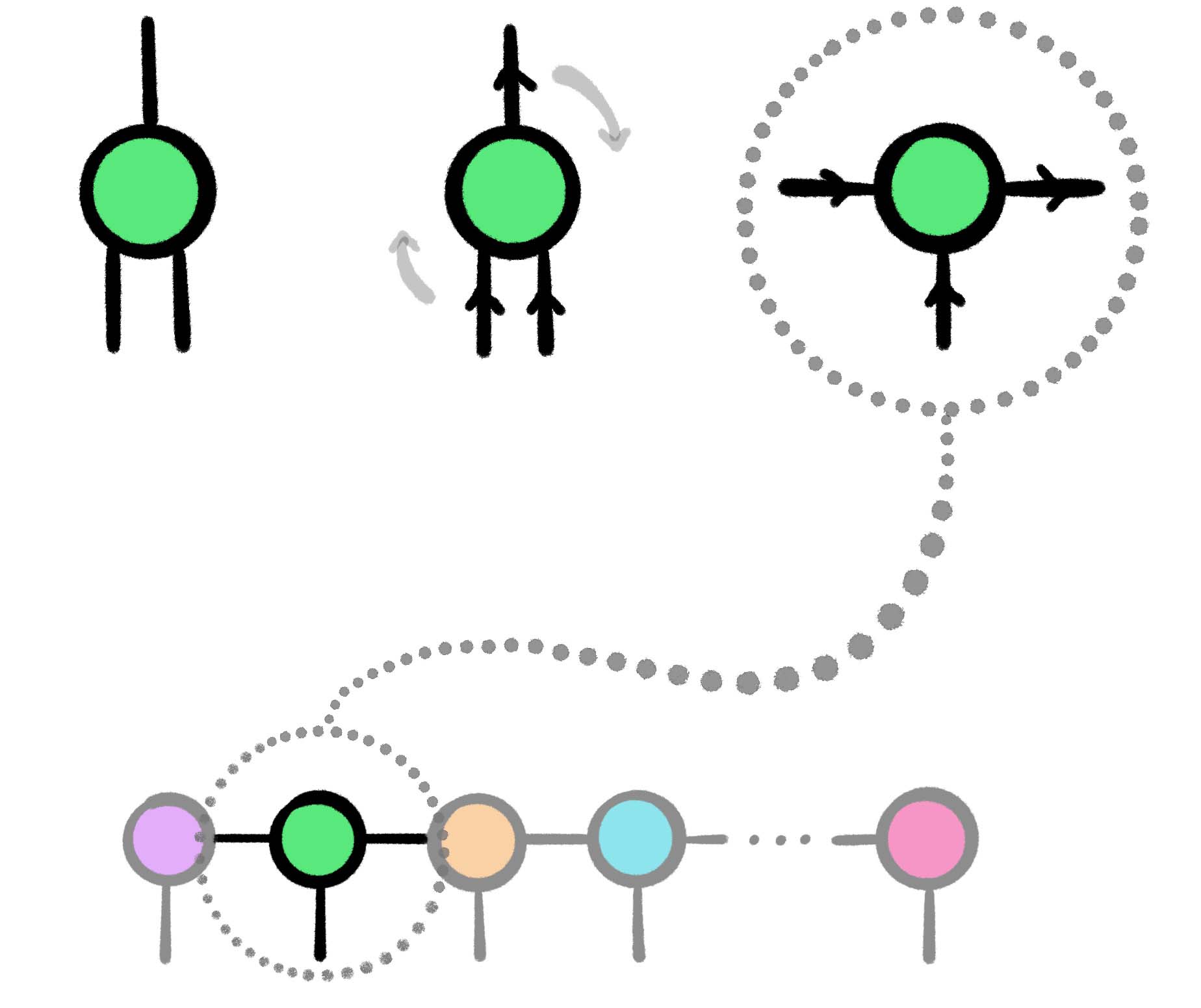}
\end{center}
In this way, we can envision the green tensor as a machine. It receives a pair of incoming bits and outputs the parity of their concatenation.\sidenote[][-1cm]{
	For example $|01\>$ maps to $|1\>$ since $01$ has odd parity and the output $1$ is odd!
	\begin{center}
	\begin{tikzpicture}
	\node at (0,0) {\includegraphics[scale=0.1]{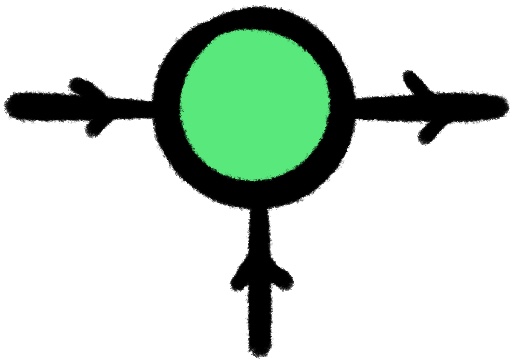}};
	\node at (-1.5,.2) {$|0\>$};
	\node at (0,-1) {$|1\>$};
	\node at (1.5,.2) {$|1\>$};
	\end{tikzpicture}
	\end{center}
	} In fact, the remaining tensors comprising $|\psi_{\text{MPS}}\>$ all behave in the same way.
\end{example}

\newthought{This example illustrates} the general idea. A typical dataset $T$ will not include all of the bitstrings in $E_N$, and so $\rho_2$ will likely be full rank. Fortunately, we can still understand its matrix entries explicitly in terms of a graph, and we still collect its top two eigenvectors $|E_2'\>$ and $|O_2'\>$ associated to its two largest eigenvalues.
	\marginnote[-1cm]{As a simple example with $N=5$, the graph below is associated to a seven-sample set.
	\begin{center}
	\begin{tikzpicture}[y=.5cm,x=.5cm]
	\node at (-3,0) {\includegraphics[scale=0.08]{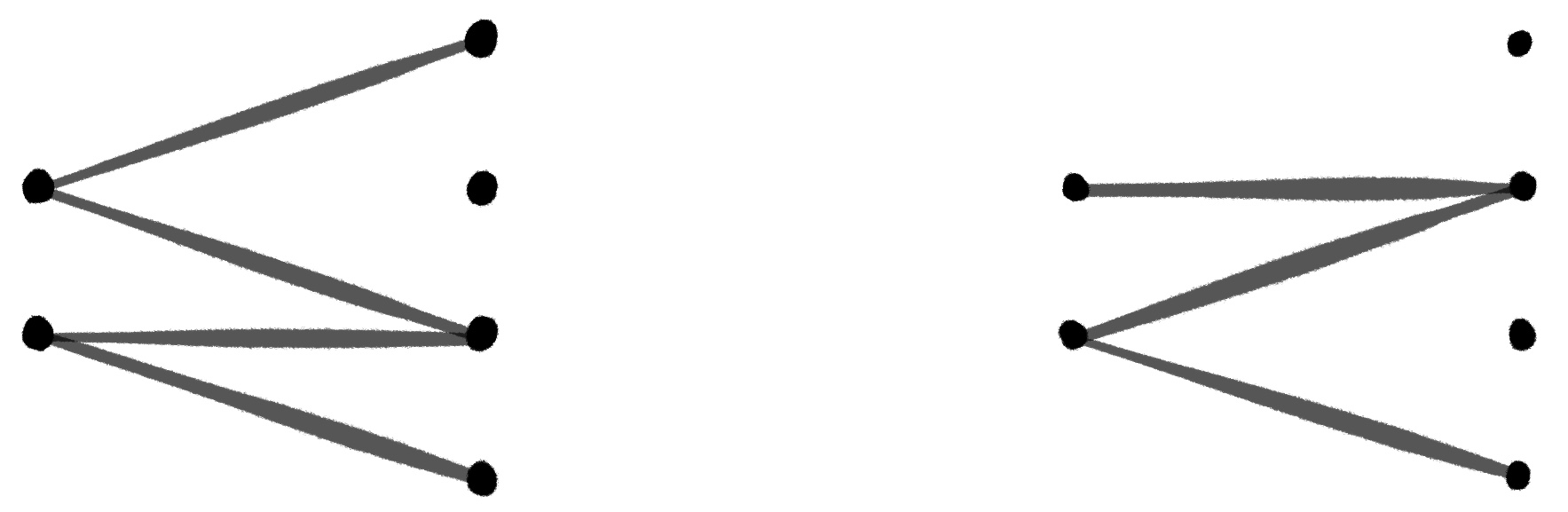}};

	\node at (-4,1.3) {$000$};
	\node at (-4,.4) {$110$};
	\node at (-4,-.5) {$011$};
	\node at (-4,-1.3) {$101$};

	\node at (-8.5,.4) {$00$};
	\node at (-8.5,-.5) {$11$};

	\node at (-1.8,.4) {$01$};
	\node at (-1.8,-.5) {$10$};

	\node at (2.5,1.3) {$000$};
	\node at (2.5,.4) {$110$};
	\node at (2.5,-.5) {$011$};
	\node at (2.5,-1.3) {$101$};
	\end{tikzpicture}
	\end{center}
	When $|\psi\>\in V^{\otimes 5}$ is the weighted sum of these seven samples, the resulting reduced density on the prefix subsystem is
	\[
	\rho_2=\tr_{V^{\otimes 3}}|\psi\>\<\psi|=\frac{1}{7}
	\begin{bmatrix}
	2 & 1 & 0 & 0\\
	1 & 2 & 0 & 0\\
	0 & 0 & 1 & 1\\
	0 & 0 & 1 & 2
	\end{bmatrix}
	\]
	}
Concretely, we may write the reduced density $\rho_2$ as
\[
\rho_2=\frac{1}{N_T}
\begin{bmatrix}
d_1 & s_e & 0 & 0\\
s_3 & d_2 & 0 & 0\\
0 & 0 & d_3 & s_o\\
0 & 0 & s_o & d_4
\end{bmatrix}
\]
The number of training samples $N_T=|T|$ is the total number of edges in the graph, and the matrix entries are nonnegative integers that have combinatorial interpretations. The diagonal entries count degrees of prefixes, and the off-diagonals count the number of paths of length two in each component of the graph. 
	\begin{center}
	\begin{tabular}{ll|ll}
		$d_1$ is the degree of $00$ &&& $s_e$ is the number of suffixes that\\
		$d_2$ is the degree of $11$ &&& \qquad  $00$ and $11$ have in common \\
		$d_3$ is the degree of $01$ &&& $s_o$ is the number of suffixes that \\
		$d_4$ is the degree of $10$ &&& \qquad  $01$ and $10$ have in common
	\end{tabular}
	\end{center}
The two eigenvectors $|E_2'\>$ and $|O_2'\>$ associated to the two largest eigenvalues of $\rho_2$ therefore have explicit expressions in terms of the graph's information. Concretely, we may write $|E_2'\>=\begin{bmatrix}\cos\theta & \sin\theta &0&0\end{bmatrix}^\top$ and $|O_2'\>=\begin{bmatrix}0&0&\cos\phi&\sin\phi\end{bmatrix}^\top$ where $\theta$ and $\phi$ are angles defined explicitly in terms of the combinatorics of the graph.\sidenote[][-4cm]{Knowing $\theta$ and $\phi$ exactly is not necessary for this discussion, but here they are: \begin{align*}
\theta&=\arctan\left(\frac{2s_e}{\sqrt{G_e^2+4s_e^2} + G_e}\right)\\
\phi&=\arctan\left( \frac{2s_o}{\sqrt{G_o^2+4s_o^2} + G_o}\right)
\end{align*}
where $G_e=d_1-d_2$ and $G_o=d_3-d_4$ are the gaps between the diagonals of each block of $\rho_2$. See \cite[Section 5]{bradley2019modeling} for more.} 
In the truncated spectral decomposition $\rho_2=UDU^\dagger$, the $2\times 4$  matrix $U^\dagger$ has $|E_2'\>$ and $|O_2'\>$ as its two rows. It therefore acts as an \textit{imperfect} summarizer. The error---or deviation from perfection---is measured by the values $\theta$ and $\phi$. For instance, if $\theta=\phi=\pi/4$, then we recover the ideal eigenvectors $|E_2'\>=|E_2\>$ and $|O_2'\>=|O_2\>$, and in that case $U^\dagger$ is a perfect summarizer as in Example \ref{ex:perfection}. 

\newthought{As it turns out}, \textit{every} tensor comprising $|\psi_{\text{MPS}}\>$ (with the minor exception of the first and last\sidenote{The
	first tensor is the identity map, as we previously described. To understand the last tensor, recall the procedure described in Section \ref{ssec:bird}, and observe that the final vector $|\psi_N\>$ can be reshaped into a map between a $2$-dimensional space (the span of the top two eigenvectors of $\rho_{N-1}$) and another $2$-dimensional space, $V\cong\mathbb{C}^A$. We take this mapping to be the final MPS tensor.
	\begin{center}
	\begin{tikzpicture}
	\node at (0,0) {\includegraphics[scale=0.147]{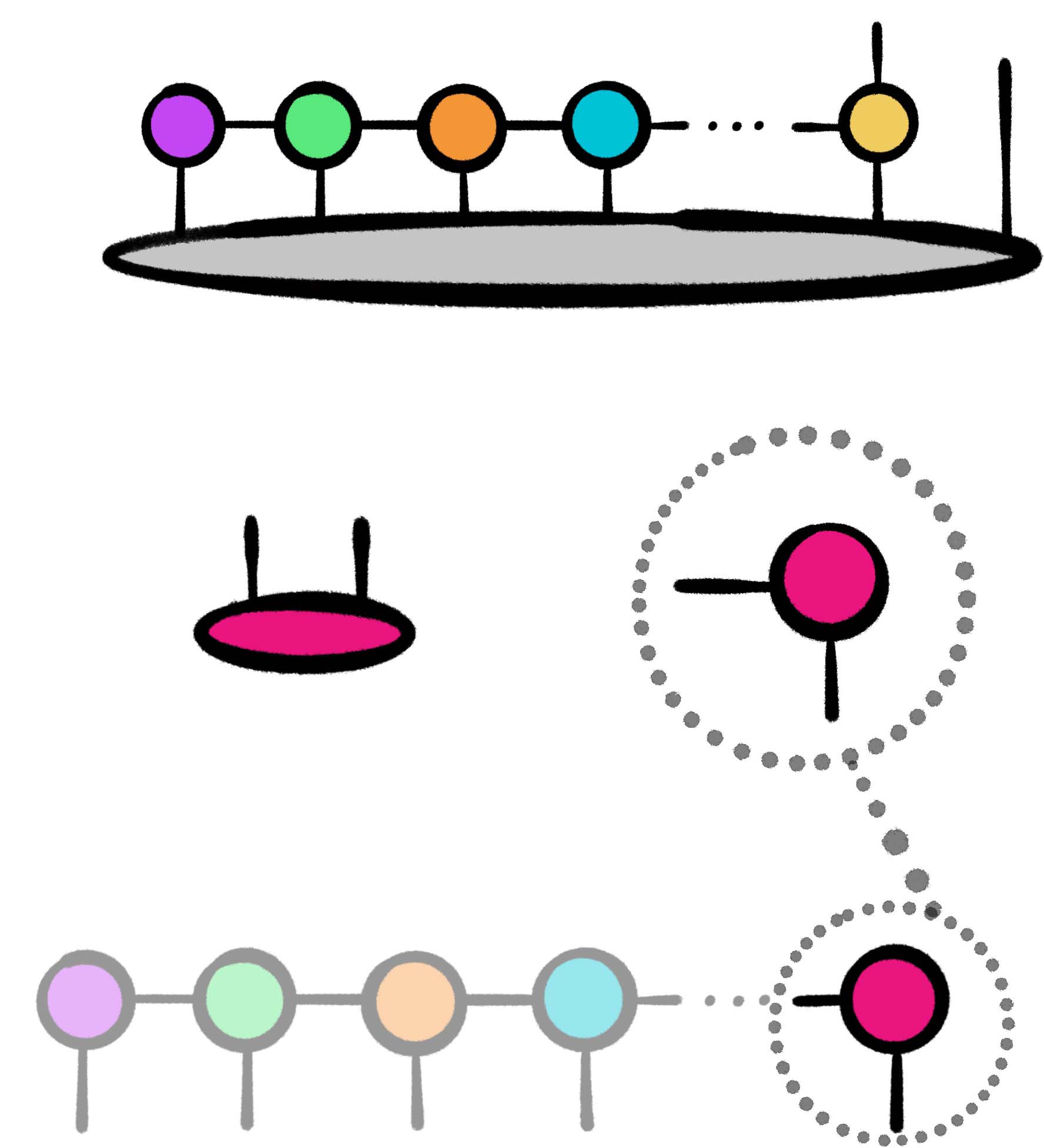}};
	\node at (-2.6,1.5) {$|\psi_N\>=$};
	\node at (-2.2,-.25) {=};
	\node at (0,-.25) {$\rightsquigarrow$};
	\end{tikzpicture}
	\end{center}})
behaves like the summarizer described here. That is, at step $k$ we collect a tensor $U_k$ built from eigenvectors $|E_k'\>$ and $|O_k'\>$ whose entries are determined by angles $\theta_k$ and $\phi_k$ that are read off explicitly from the combinatorics of a bipartite graph associated to the training set $T$. With this, we have a theoretical understanding of $|\psi_{\text{MPS}}\>$ both at a low level---the tensors can be understood entry-wise---and at a high level---the tensors are eigenvectors of reduced densities, which harness prefix-suffix interactions. 

As a closing remark, know that Example \ref{ex:perfection} showcased the \textit{idealized} scenario where each tensor is a perfect summarizer. When this happens (that is, when every $U_k^\dagger$ maps a small bitstring onto its parity) one finds that $|\psi_{\text{MPS}}\>$ is \textit{equal} to the normalized sum of all even bitstrings of length $N$,
\[
|E_N\>=\frac{1}{\sqrt{2^{N-1}}}\sum_{s \text{ even}}|s\>.
\]
This is the this best-case scenario, for the probability of any even bitstring $s$ is given by $|\<s|\psi_{\text{MPS}}\>|^2=|\<s|E_N\>|^2=1/2^{N-1}$, which is the target probability $\pi(s)$ identified at the opening of this section. Situations that are less than ideal will not have equality $|\psi_{\text{MPS}}\>\neq |E_N\>$, but the inner product $\<E_N|\psi_{\text{MPS}}\>$ (or some expression involving it) serves a measure of how close the model $|\psi_{\text{MPS}}\>$ is to the goal $|E_N\>$. As a teaser, here's a summary of an experiment that illustrates this.

\section{The Experiment}\label{sec:theexperiment}
As a concrete experiment, we worked with bitstrings of length $N=16$. We built a theoretical vector $|\psi_{\text{MPS}}\>$ as described in Section \ref{ssec:bird} and compared it to an \textit{experimental} version built with the algorithm using the ITensor library \cite{ITensor}; the code is available on Github.\sidenote{\url{https://github.com/emstoudenmire/parity}} For a fixed fraction $0<f\leq 0.2$, the experimental model was built by running the algorithm on ten different datasets $T$, each containing $N_T=f2^{N-1}$ bitstrings of length $16$. We then compared the resulting vector $|\psi_{\text{MPS}}\>$ with the target vector $|E_N\>$ by computing the Bhattacharyya distance\sidenote{The Bhattacharyya distance is a measure between two probability distributions, computed as the negative log of their ``inner product.'' That is, suppose $p,q\colon S\to \mathbb{R}$ are probability distributions on a finite set $S$. The \emph{Bhattacharyya distance}\index{Bhattacharyya distance} between them is $d_B(p,q):=-\ln\left(\sum_s\sqrt{p(s)q(s)}\right)$. When $p(s)=|\<s|\psi_{\text{MPS}}\>|^2$ and $q(s)=|\<s|E_N\>|^2$ then $d_B(p,q)=-\ln\<\psi_{\text{MPS}}|E_N\>.$} between them. The $x$-axis in Figure \ref{fig:experiment} is the range of fractions $0<f\leq 0.2$, and the Bhattacharyya distances are plotted vertically. Figure \ref{fig:experiment2} gives a zoomed-in look for $.15\leq f\leq 0.2$. In both figures, each black dot represents one of these values, and the pink curve is their experimental average. The green curve is our theoretical prediction using the mathematics described above.\sidenote[][1cm]{Obtaining the green curve is not an elementary calculation, and the story behind it is quite interesting. It involves understanding the entries of a number of reduced densities, $\rho_3,\rho_4,\ldots,\rho_{N}$, similar to the one shown in Example \ref{ex:perfection}. As it turns out, this task is rather complicated, so instead of calculating the entries \textit{exactly}, we instead find estimates for their expected values. The angles $\theta_k$ and $\phi_k$ that were mentioned in the main text play a role in finding these values. These expected entries then contribute to the expected eigenvectors of the operators, and these eigenvectors are used to find the green curve.
}
\begin{center}
\begin{figure}[h!]
\includegraphics[width=\linewidth]{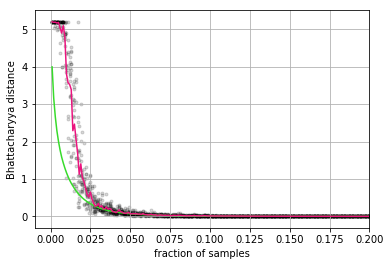}
\caption{The experimental average (pink) and theoretical prediction (green).}
\label{fig:experiment}
\end{figure}
\end{center}
\begin{center}
\begin{figure}[h!]
\includegraphics[width=\linewidth]{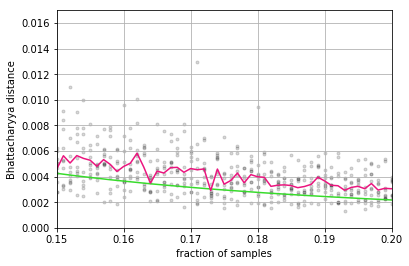}
\caption{Taking a closer look for $.15\leq f \leq 0.2$.}
\label{fig:experiment2}
\end{figure}
\end{center}
Notice that experiment (pink) and theory (green) go hand-in-hand, as the two curves grow closer together at the end. Hence we can predict how good the model $|\psi_{\text{MPS}}\>$ will be just given the fraction $f$ of samples used to train it. The plots tend \textit{downward} as one would expect: the more samples used, the better the model will be. Remarkably, the model does quite well training with just 2.5\%! In other words, with just a small fraction of training samples, the model state $|\psi_{\text{MPS}}\>$ is quite close to the true state $|E_N\>$. 

\newthought{This experiment is thus} an instantiation of the theory illustrated in Figure \ref{fig:reconstruct} and summarized in Takeaway \ref{takeaway:densities} of Chapter \ref{ch:preliminaries}. Succinctly put, the eigenvectors of reduced densities harness interactions between small subsystems (here, prefix-suffix subsystems), and this information can be pieced together to reconstruct the state (or a close approximation of it) of the full system. Here's the takeaway.

\begin{takeaway}
Reduced densities of pure entangled states harness information about interactions between subsystems. When those subsystems are \textit{sequences} (that is, a one-dimensional string of ``particles''), knowledge about the state of a subsystem and its complement can be pieced together to reconstruct the state of the whole.
\end{takeaway}

\clearpage

%chapter 5
\chapter{Fixed Points, Categorically}\label{ch:categorytheory}

\epigraph{When you see the same beautiful idea pop up everywhere, you begin to think that it is pointing to some deeper truth you haven't yet grasped.}{\textit{Francis Su \cite{su2020}}}

\newthought{Chapter \ref{ch:introduction} motivated}  the investigation of eigenvectors of reduced densities by drawing an analogy with formal concepts---an analogy that was revisited in Section \ref{sec:evectsFCA}. This chapter aims to put that analogy on firmer categorical ground. To do so, we will introduce a \textit{third} construction---a categorical one---that is parallel to both formal concepts and eigenvectors. The category theory will specialize to give formal concepts in a certain case. It will not specialize to the linear algebra, but there is a clear, common categorical thread connecting them, which will be discussed at the end of the chapter.
  \marginnote{This chapter assumes familiarity with the basics of category theory (categories, functors, natural transformations, (co)limits, and adjunctions), though I'll occasionally provide extra commentary for interested readers. For thorough introductions, see \cite{riehl2017category, leinster2014basic, awodey2010category,mac2013categories,spivak2014category, fong2019invitation}. For a lighter treatment of the basics of category theory, see the articles on Math3ma, under the ``Category Theory'' tag \cite{math3ma}. }
 %\tai{``What if we replace probabilities with \textit{possibilities?}''}

Here is an overview of the ideas to come. We will consider three familiar mathematical objects: a vector space with basis set $X$, the power set of a set $X$, and the category of presheaves on a small category $\C$.
\[\mathbb{C}^X \qquad 2^X \qquad \mathsf{Set}^{\C^{\op}}\]
These three objects are related in that all are ``free'' in a sense that can be made precise. Moreover, all three fit into a construction that is strikingly similar to that featured in Chapter \ref{ch:introduction}. To recall, any matrix $M$ corresponds to a function $M\colon X\times Y\to \mathbb{C}$ where $X$ and $Y$ are finite sets. The matrix represents a linear map  $M\colon \mathbb{C}^X\to \mathbb{C}^Y$, which has an adjoint going in the other direction $M^\dagger \colon \mathbb{C}^Y\to\mathbb{C}^X$. The eigenvectors of $M^\dagger M$ and $MM^\dagger$ capture interesting  information about the values $M(i,j)$, especially when they determine a probability distribution on $X\times Y.$ Interestingly, a nearly identical story holds when the complex numbers $\mathbb{C}$ are exchanged for the poset $2$ of \textit{truth values}  (which we define later) or for the category $\mathsf{Set}$. In this chapter we will carefully unwind this claim. 
  \marginnote{Recall that an \emph{adjunction} consists of a pair of functors $\adj{F}{\C}{\D}{G}$ such that there is a natural isomorphism of hom sets \[\C(Fc,d)\cong \D(c,Gd)\] for all objects $c$ in $\C$ and $d$ in $\D$. Equivalently, an adjunction consists of functors $F$ and $G$ together with natural transformations $\eta\colon \id_\C\Rightarrow GF$, called the \emph{unit}, and $\epsilon\colon FG\Rightarrow \id_\D$, called the \emph{counit}, such that the following diagrams commute
  \[
  \begin{tikzcd}[ampersand replacement =\&]
  F \arrow[r, "F\circ\eta", Rightarrow] \arrow[rd, "\id_F"', Rightarrow] \& FGF \arrow[d, "\epsilon \circ F", Rightarrow] \&  \& G \arrow[r, "\eta\circ G", Rightarrow] \arrow[rd, "\id_G"', Rightarrow] \& GFG \arrow[d, "G\circ \epsilon", Rightarrow] \\  \& F                               \&  \&  \& G           
  \end{tikzcd}
  \]
  Here $F\circ \eta$ denotes the natural transformation whose components are of the form $F(\eta_c)\colon F(c)\to FGF(c),$ for each object $c$ in $\C$, and similarly for the others.}

We start in Section \ref{sec:5_LA} by recalling the universal property of a free vector space on a set and will use it to recast the elementary linear algebra of Section \ref{sec:ch1_linalg} in a categorical light. Section \ref{sec:5_CT} discusses another free object with a very similar universal property: the free (co)completion of a category. Unwinding the theory, one finds a number of strong analogies with linear algebra and a clear dictionary between the two worlds. In particular, we will see that categorical version of a matrix $M\colon X\times Y\to\mathbb{C}$ is a functor $R\colon \C^{\op}\times\D\to\mathsf{Set}$; the categorical version of adjoint linear maps $\mathbb{C}^X\rightleftarrows \mathbb{C}^Y$ is an adjunction between (co)presheaf categories $\mathsf{Set}^{\C^{\op}}\rightleftarrows (\mathsf{Set}^\D)^{\op}$; and the categorical version of eigenvectors is a generalized \textit{Isbell completion}. Section \ref{sec:5_FCA} specializes this to the setting of \textit{enriched category theory} by replacing $\mathsf{Set}$ with the poset (category) of truth values $2$. A functor is then a relation $R\colon X\times Y\to 2$; the free (co)completion of a set, viewed as a discrete category, is the free join/meet semilattice on that set; the adjunction between (co)presheaf categories becomes an adjunction of join/meet-preserving maps $2^{X^{\op}}\rightleftarrows (2^Y)^{\op}$; and the generalized Isbell completion gives formal concepts.
\begin{fullwidth}
\begin{center}
\begin{table}[h!]
\begin{tabular}{@{}c|c|c@{}}
\toprule
\textbf{Section \ref{sec:5_LA}} & \textbf{Section \ref{sec:5_CT}} & \textbf{Section \ref{sec:5_FCA}} \\
\midrule
\textsc{linear algebra} & \textsc{category theory} & \textsc{enriched category theory} \\
eigenvectors of reduced densities & generalized Isbell completion & formal concepts\\
\bottomrule
\end{tabular}
\end{table}
\end{center}
\end{fullwidth}

Sections \ref{sec:5_LA}, \ref{sec:5_CT}, and \ref{sec:5_FCA} are meant to feel repetitious. All three are essentially the same story told in different contexts, the motif being that \textit{fixed points of the composition of a morphism with its ``adjoint'' are interesting.} Section \ref{sec:5behind} closes with a remark on the category theory that summarizes this. Here's a quick note on notation. As seen already, sans serif is used for categories $\C$. Also $\C(c,c')$ is used to denote the hom set of morphisms from an object $c$ to an object $c'$ in $\C$. Given categories $\C$ and $\D$, a pair of adjoint functors will be denoted by $\adj{F}{\C}{\D}{G}$ or by $F\dashv G.$

\section{Revisiting Eigenvectors}\label{sec:5_LA}
The motivating construction was already presented in Section \ref{sec:ch1_linalg}, but think of that as the appetizer. In this section, we revisit the elementary linear algebra, this time being mindful of the category theory behind the scenes. \marginnote{I gratefully acknowledge Simon Willerton for his clear exposition on this topic. Indeed, this section---and much of this chapter---was inspired by the excellent writings in \cite{willerton2013} as well as the companion articles \cite{willerton2013b,willerton2014, willerton2015legendrefenchel}.} To start, let $\mathsf{Vect}$ be the category of complex vector spaces and linear transformations, and let $\mathsf{Set}$ be the category of sets and functions. There is a free-forgetful adjunction
\[\adj{U}{\mathsf{Vect}}{\mathsf{Set}}{F}.\]
The functor $U\colon \mathsf{Vect}\to\mathsf{Set}$ associates to any vector space $V$ its underlying set $UV$. The functor $F\colon \mathsf{Set}\to\mathsf{Vect}$ associates to a set $X$ the free vector space $FX$ on $X$. It is ``free'' in that the unit natural transformation $\eta\colon \id_{\mathsf{Set}}\Longrightarrow UF$ of this adjunction fits into the following familiar universal property.

\begin{UP_vect} The \textit{free vector space on a set $X$} is a vector space $FX$ with the property that for any vector space $W$ and for any function $f\colon X\to UW$ there is a unique linear map $\hat f\colon FX\to W$ so that the following diagram commutes
\[
  \begin{tikzcd}
  UFX \arrow[r, "U\hat f", dashed]                  & UW \\
  X \arrow[ru, "f"'] \arrow[u, "\eta_X"] &   
  \end{tikzcd}
\]
\end{UP_vect}
\noindent Concretely, $FX$ is the vector space whose basis is $X$. A vector in $FX$ is then a complex-valued function $v$ on $X$ with the property that $v(x)\neq 0$ for finitely many $x\in X$. The formal sum
\begin{equation}\label{eq:LA_yoneda}
|v\>:=\sum_x v(x)x
\end{equation}
is then guaranteed to be finite and may be thought of as the vector $|v\>$. If $X$ is finite, then vectors in $FX$ may be identified with \textit{all} functions from $X$. In either case, the function $\eta_X$ associates to an element $x\in X$ the independent basis vector $|x\rangle$, and the space $FX$ is sometimes denoted $\mathbb{C}^X$. On any vector $|v\>=\sum_x v(x) |x\>$ in $FX$ the unique lift of any function $f\colon X\to UW$ is defined by
\begin{equation}\label{eq:LA_lift}
\hat f|v\>=\sum_xv(x)f(x).
\end{equation}
Oftentimes, the $U$s are omitted from the diagram and one sees the following.
\[
  \begin{tikzcd}
  \mathbb{C}^X \arrow[r, "\hat f", dashed]                  & W \\
  X \arrow[ru, "f"'] \arrow[u, "\eta_X"] &   
  \end{tikzcd}
\]
But this is not the best practice. The above diagram is not in the correct---or in any!---category. For example, the domain of $f\colon X\to W$ is a set while its codomain is a vector space. The diagrams in Takeaway  \ref{takeaway:LA} also suffered from this ambiguity. Going forward, let's avoid this imprecision. Instead we'll always have in mind the adjunction $\adj{U}{\mathsf{Vect}}{\mathsf{Set}}{F}$ and insert a ``$U$'' where appropriate. Let's also avoid writing ``$FX=\mathbb{C}^X$.'' As noted earlier, vectors in $FX$ may be identified with finitely supported functions $X\to \mathbb{C}$, which explains the convenient exponential notation $\mathbb{C}^X$. But we must be careful. The phrase ``a function $X\to \mathbb{C}$'' has the same problem mentioned above.  On one hand, $X$ is a set; on the other hand, $\mathbb{C}$ is a vector space. An arrow $X\to \mathbb{C}$ is not a morphism in any category, and the exponential $\mathbb{C}^X$ of a vector space by a set is not an object in any category. Rather, here's the better observation: the \textit{underlying set} of the free vector space on $X$ coincides with functions $X\to U\mathbb{C}$. This is summarized in the following takeaway.
\begin{takeaway}\label{takeaway:5LA} Given the free-forgetful adjunction $\adj{U}{\mathsf{Vect}}{\mathsf{Set}}{F},$ the free vector space $FX$ on a set $X$ has the property that
  \begin{equation}\label{eq:EforLA}
  UFX=U\mathbb{C}^X,
  \end{equation}
where $U\mathbb{C}^X$ denotes the set of ``finite'' functions $X\to U\mathbb{C}$, which is all function if $X$ is finite or finitely-supported functions otherwise. In short, the adjunction $F\dashv U$ has the property that there is a special object in $\mathsf{Vect}$, namely $\mathbb{C}$, such that \textit{all} free vector spaces are, in some sense, built up from it.
\end{takeaway}

Now we are close to rediscovering Takeaway \ref{takeaway:LA} in Chapter \ref{ch:introduction}. Suppose $X$ and $Y$ are finite sets and consider any function $M\colon X\times Y\to U\mathbb{C}$. The category of sets is Cartesian closed, and so we have a product-hom adjunction and hence the following bijections.
  \marginnote[2cm]{A category $\C$ is called \emph{Cartesian closed} if for every pair of objects $x$ and $z$ there is an object $z^x$ in $\C$ so that for all objects $y$ there is a natural bijection \[\C(x\times y,z)\cong \C(y,z^x)\] The object $z^x$ is called the \emph{internal hom} and should be thought of as a $\C$'s worth of morphisms from $x$ to $z$.}
\[
  \begin{tikzcd}
  {\mathsf{Set}(X,U\mathbb{C}^Y)} & \cong & {\mathsf{Set}(X\times Y,U\mathbb{C})}                 & \cong & {\mathsf{Set}(Y,U\mathbb{C}^X)} \\[-20pt]
  \alpha             &       & M \arrow[rr, maps to] \arrow[ll, maps to] &       & \beta            
  \end{tikzcd}
\]
Under this adjunction, the function $\alpha\colon X\to U\mathbb{C}^Y$ is defined on each $x\in X$ by $\alpha x(y)= M(x,y)$. Likewise define the function $\beta\colon Y\to U\mathbb{C}^X$ on each $y\in Y$ by $\beta y(x)=\overline{M(x,y)}$. By the universal property for vector spaces, $\alpha$ lifts uniquely to a linear map $M\colon FX\to FY$, and $\beta$ lifts uniquely to a linear map $M^\dagger \colon FY\to FX$ so that the following diagrams commute.
\begin{equation}\label{diagrams:correctlifts_LA}
  \begin{tikzcd}
  U\mathbb{C}^X \arrow[r, "UM", dashed]         & U\mathbb{C}^Y && U\mathbb{C}^Y \arrow[r, "UM^\dagger ", dashed]         & U\mathbb{C}^X \\
  X \arrow[ru, "\alpha"'] \arrow[u, hook] &     && Y \arrow[u, hook] \arrow[ru, "\beta"'] &    
  \end{tikzcd}
\end{equation}
What's more, $M$ and $M^\dagger$ are linear adjoints so that
\begin{equation}\label{eq5:adjoint}
\langle Mv,w\rangle = \langle v,M^\dagger w\rangle \qquad \text{for all $v\in FX$ and $w\in FY$.}
\end{equation}
If we drop the ``$U$s'' from the diagrams, then we recover the contents of Takeaway \ref{takeaway:LA}. When the function $M\colon X\times Y\to U\mathbb{C}$ further satisfies the condition that $\sum_{x,y}|M(x,y)|^2=1$, then it represents a pure quantum state in $FX\otimes FY$, and we've shown in detail in Chapter \ref{ch:probability} that the eigenvectors of the reduced density operators $MM^\dagger $ and $M^\dagger M$ harness valuable information from that state.  

These ideas, though elementary, set the stage for the remainder of this chapter. Here are the key points to have in mind. We started with the free-forgetful adjunction $\adj{U}{\mathsf{Vect}}{\mathsf{Set}}{F}$ and observed that $\mathsf{Vect}$ contains a special object, namely $\mathbb{C}$, with the property that the free space $FX$ on a set $X$ always ``looks like'' functions from $X$ to $U\mathbb{C}$. As a result, every  function $X\times Y\to U\mathbb{C}$ gives rise to a pair of adjoint maps $\adj{M}{FX}{FY}{M^\dagger}$, whose composites $M^\dagger M$ and $MM^\dagger$ have interesting fixed points.\sidenote{That is, they have interesting \textit{eigenvectors}. So we may restate our claim is, ``The one-dimensional invariant subspaces of $M^\dagger M$ and $MM^\dagger$ are interesting.''} This very same template will reappear in Section \ref{sec:5_CT} and again in Section \ref{sec:5_FCA}. More precisely, the idea that ``fixed points of a map with its adjoint are interesting'' is strikingly similar to a known scenario in category theory---there is nearly a word-for-word analogy between the two.  Even better, the categorical version recovers the formal concept construction in Section \ref{sec:ch1_FCA} as a special case. Let's discuss this category theory now.

\section{Free (Co)Completions}\label{sec:5_CT}
%Briefly introduce limits/colimits first, a la Math3ma abridged?

To construct the free vector space on a set $X$, we added all formal sums of elements in $X$. Category theory generalizes sums through \emph{colimits}, and it generalizes multiplication through \emph{limits}. We refer to any basic text for the formal definitions \cite{riehl2017category}. But here are three \textit{informal} things to know.
  \begin{enumerate}
    \item \textbf{Intuition.} Colimits and limits subsume many constructions in mathematics. Unions of sets and joins in a poset are two examples of colimits. Cartesian products of sets and meets in a poset are two examples of limits. Informally, colimits have a ``glue-y'' feel to them, whereas limits have an ``equation-y'' feel to them.

    \item \textbf{Vocabulary.} One always takes the (co)limit \textit{of} something. That ``something'' is a \emph{diagram} in a category $\C$, which is synonymous with a functor whose codomain is $\C$. The diagram is called \emph{small} if the domain of the functor is a small category. (See the next bullet.) A category is called \emph{cocomplete} if it contains  colimits of all small diagrams; it is called \emph{complete} if it contains limits of all small diagrams.

    \item \textbf{More Vocabulary.} A category is called \emph{small} if the morphisms of the category form a set. It is called \emph{locally small} if the morphisms between any two objects form a set. As an example, the category with two objects and a single nonidentity arrow $\bullet\to\bullet$ is small and hence locally small, whereas $\mathsf{Set}$ is locally small but not small.
  \end{enumerate} 

\newthought{As in the linear} algebra of Section \ref{sec:5_LA}, the starting ingredient for this section is an adjunction. Let $\mathsf{CAT}$ denote the category whose objects are locally small categories and whose morphisms are functors. Let $\mathsf{CocompCAT}$ denote the category of locally small cocomplete categories and cocontinuous functors between them.\sidenote{A functor is said to be \emph{cocontinuous} if it takes colimits to colimits.} Both $\mathsf{CAT}$ and $\mathsf{CocompCAT}$ are actually 2-categories; that is, there are morphisms (natural transformations) between the morphisms (functors). There is a free-forgetful ``adjunction'' between these 2-categories.
\[\adj{U}{\mathsf{CocompCAT}}{\mathsf{CAT}}{F}.\]
The functor $U\colon \mathsf{CocompCAT}\to \mathsf{CAT}$ associates to any locally small cocomplete category $\D$ its underlying category $U\D$. The functor $F\colon \mathsf{CAT}\to\mathsf{CocompCAT}$ associates to a locally small category $\C$ the \emph{free cocomplete category} $F\C$ on $\C$. The setup here isn't quite an adjunction in the usual sense, as 2-categorical considerations prevent it from being so. Instead, it is called a \emph{biadjunction}, which is the appropriate notion of adjunctions in this higher-categorical setting. But here's the main idea for us: just as every set gives rise to a free vector space, so every category $\C$ gives rise to a free cocomplete category $F\C$ \cite{Day_2007,kelly200,nlab:free_cocompletion}. Crucially, $F\C$ is free in that unit $\eta\colon \id_{\mathsf{CAT}}\Longrightarrow UF$ of the biadjunction has the following universal property.

\begin{UP_cocomplete}The \textit{free cocompletion of a category $\C$} is a cocomplete category $F\C$ with the property that for any other cocomplete category $\D$ and for any functor $f\colon \C\to U\D$, there is a unique cocontinuous functor $\hat f\colon F\C \to \mathsf{D}$ so that the following diagram commutes up to natural isomorphism.
\[
  \begin{tikzcd}
  UF\C \arrow[r, "U\hat f", dashed] & U\mathsf{D} \\
  \C \arrow[ru, "f"'] \arrow[u,"\eta_\C"]                 &   
  \end{tikzcd}
\]
\end{UP_cocomplete}
Let's unwind this statement. First, the category $F\C$ has a familiar construction. If $\C$ is small, then its free cocompletion is the category of all contravariant functors from $\C$ to $\mathsf{Set}$, also called \emph{presheaves} \cite{Day_2007}. It is typically written as\sidenote{Every category $\C$ has an \emph{opposite category} $\C^{\op}$ whose objects are the same as the objects in $\C$. A morphism $c\to c'$ in $\C^{\op}$ is defined to be a morphism $c'\to c$ in $\C$. As a result, for any category $\D$, a contravariant functor $\C\to \D$ is a functor $\C^{\op}\to \D$.}
\[F\C=\mathsf{Set}^{\C^{\op}}.\]
If the category $\C$ is \textit{not} small, then the construction of $F\C$ needs a minor adjustment, which will be explained shortly. But first, know that the unit $\eta_\C\colon \C\to UF\C$ of the biadjunction is the \emph{Yoneda embedding} $c\mapsto \C(-,c)$. Oftentimes the ``$U$''s are omitted from the diagram and one typically sees
\[
  \begin{tikzcd}
  \mathsf{Set}^{\C^{\op}}\arrow[r, "\hat f", dashed] & \mathsf{D} \\
  \C \arrow[ru, "f"'] \arrow[u,"\eta_\C"]                 &   
  \end{tikzcd}
\]
But let's be careful to insert a ``$U$'' where appropriate. In particular---and for reasons explained below---it's better not to write ``$F\C=\mathsf{Set}^{\C^{\op}}$,'' even though the notation may be familiar. We will suggest a better observation in Takeaway \ref{takeaway:5CT} later on. 

Now, how are we to understand the free cocompletion $F\C$ when $\C$ is a locally small (possibly large) category? Size issues prevent the presheaf category from being its free cocompletion, so a minor adjustment must be made. To motivate the idea, it will help to think back to the familiar linear algebra of Section \ref{sec:5_LA}. Recall that the free vector space on a set $X$ does not coincide with \textit{all} functions $X\to U\mathbb{C}$ but rather only those functions with finite support. In a completely analogous way, if $\C$ is not small, then its free cocompletion does not coincide with \textit{all} functors $\mathsf{C}^{\op}\to\mathsf{Set}$ but rather only those functors that are ``finite,'' or more precisely, \emph{small}. (Intuitively, a presheaf is said to be small if it can be written as a small colimit of representable functors.) For a precise definition and a thorough treatment of these ideas see \cite{Day_2007}.

Continuing to unwind the universal property, how should we understand the unique cocontinuous extension $\hat f$ along the Yoneda embedding? To aid intuition, we draw another comparison with linear algebra. Recall from the universal property for vector spaces that any function on a set $X$ lifts to a linear map on the free vector space on $X$ that agrees with the function on basis elements in $X$. Simply put, the lift is a \textit{linear extension} of the function. In a completely analogous way, the universal property above says that any functor $f\colon \C\to U\D$ has a ``cocontinuous extension'' $\hat f\colon F\C\to \D$ that agrees with $f$ on \emph{representable functors}, which are very much like basis vectors of the presheaf category $\mathsf{Set}^{\C^{\op}}$.\marginnote{``The representables are the prime numbers of presheaves.'' \cite{leinster2014basic}} More formally, \textit{any presheaf $V\colon \C^{\op}\to \mathsf{Set}$ is a colimit of representable functors}, which are functors of the form $\mathsf{C}(-,c)\colon \C^{\op}\to\mathsf{Set}$ for some object $c$ in $\C$ \cite[Theorem 6.5.8]{riehl2017category}, \cite[Theorem 6.2.17]{leinster2014basic}. There are a few ways to express this statement notationally. Here's one way that uses notation that we have not defined, but which is very suggestive \cite[Proposition 2.2.1]{loregian2015coend}.\sidenote{See also \cite[Proposition 2.2]{nlab:co-yoneda_lemma} and \cite[Section 2]{nlab:free_cocompletion}.} For any presheaf $V$, one has
\begin{equation}\label{eq:CT_yoneda}
V\cong \int^{c\in\C} \C(-,c)\cdot Vc,
\end{equation}
and the lift $\hat f$ in the universal property is then defined by
\begin{equation}\label{eq:CT_lift}
\hat{f}V:= \int^{c\in\C} fc \cdot Vc.
\end{equation}
The integral $\int^{c\in \C}$ is a type of a colimit called a \textit{coend}. The dot $\cdot$ is called a \textit{copower}, and the expression $\C(-,c)\cdot Xc$ denotes the $Xc$-indexed coproduct of the hom-functor $\C(-,c)$ with itself, and similarly for $fc\cdot Vc$. The particulars are not needed here, but we mention this simply to point out the similarities between Equation (\ref{eq:CT_yoneda}) and Equation (\ref{eq:LA_yoneda}),
\[
\begin{array}{c@{\hspace{3ex}}|@{\hspace{3ex}}c}
|v\>=\displaystyle\sum_{x\in X} v(x)|x\>
&
V\cong \displaystyle\int^{c\in\C} \C(-,c)\cdot Vc
\end{array} 
\]
and between Equation (\ref{eq:LA_lift}) and Equation (\ref{eq:CT_lift}).
\[
\begin{array}{c@{\hspace{3ex}}|@{\hspace{3ex}}c}
\hat f|v\>=\displaystyle\sum_{x\in X}v(x)f(x)
&
\hat{f}V= \displaystyle\int^{c\in\C} fc \cdot Vc
\end{array} 
\]
The analogy is clear: presheaves are \textit{like} vectors, representable functors are \textit{like} basis vectors, colimits are \textit{like} sums, cocontinuous functors are \textit{like} linear maps, and so on.\sidenote[][-3cm]{Here's yet another analogy. If $M$ is an $n\times m$ matrix and $N$ is an $m\times p$ matrix, then $MN$ is an $n\times p$ matrix whose $ij$th entry is obtained by ``tracing out'' a common index $\sum_k M_{ik}N_{kj}$. The categorical version of a matrix $X\times Y\to U\mathbb{C}$ is a \textit{profunctor}, which is the name given to a functor of the form $\C^{\op}\times D\to \mathsf{Set}$.  The composition of two profuctors $F\colon \C^{\op}\times \D\to\mathsf{Set}$ and $G\colon \D^{\op}\times \E\to\mathsf{Set}$ is given by a coend, which is a colimit over the common index $FG=\displaystyle \int^{d\in \D}F(-,d)\times G(d,-).$ More generally, one could imagine composing profunctors several variables, which brings to mind tensor contraction. I first learned of this analogy from \cite{Kissinger2011}.}

\newthought{In summary}, the free cocompletion $F\C$ of a category $\C$ is very much like the free vector space $FX$ on a set $X$. As suggested above, it's better not to use the notation ``$F\C=\mathsf{Set}^{\C^{\op}}.''$ The category $\mathsf{Set}$ has many nice properties, one of them being that it is cocomplete. Since an arbitrary category $\C$ may \textit{not} be cocomplete, an arrow $\C^{\op}\to \mathsf{Set}$ is not well-typed. Instead, it's better to say that a presheaf is a functor from $\C^{\op}$ to the \textit{underlying category} of the cocomplete category $\mathsf{Set}$, and that the \textit{underlying category} of the free cocompletion $F\C$ coincides with these presheaves. This is summarized in the takeaway below. Compare it with the linear algebra in Takeaway \ref{takeaway:5LA}.

\begin{takeaway}\label{takeaway:5CT}
Given the free-forgetful biadjunction $\adj{U}{\mathsf{CocompCAT}}{\mathsf{CAT}}{F}$, the free cocompletion $F\C$ of a category $\C$ has the property that
\begin{equation}
UF\C=U\mathsf{Set}^{\C^{\op}},
\end{equation}
where $U\mathsf{Set}^{\C^{\op}}$ denotes the category of ``small'' presheaves $\C^{\op}\to U\mathsf{Set}$, namely all presheaves if $\C$ is small or small presheaves otherwise. In short, the adjunction $F\dashv U$ has the property that there is a special object in $\mathsf{CocompCAT}$, namely $\mathsf{Set}$, such that all free cocomplete categories are, in some sense, built up from it.
\end{takeaway}

As expected, there is a dual discussion obtained by ``reversing all the arrows.'' It's worth mentioning briefly. Let $\mathsf{CompCAT}$ denote the 2-category of locally small complete categories, continuous functors,\sidenote{A functor is said to be \emph{continuous} if it takes limits to limits.} and natural transformations between them. There is a free-forgetful biadjunction 
\[\adj{U}{\mathsf{CompCAT}}{\mathsf{CAT}}{\bar F}.\]
The functor $U\colon \mathsf{CompCAT}\to \mathsf{CAT}$ associates to any complete category $\C$ its underlying category $U\C$. The functor $\bar F\colon \mathsf{CAT}\to\mathsf{CompCAT}$ associates to a category $\C$ the \emph{free complete category} $\bar F\C$ on $\C$. It is free in that the unit $\eta\colon \id_{\mathsf{CAT}}\Longrightarrow U\bar F$ of this biadjunction satisfies the following universal property.

\begin{UP_complete}The \textit{free completion of a category $\C$} is a complete category $\bar F\C$ with the property that for any complete category $\D$ and for any functor $g\colon \C\to U\D$, there is a unique continuous functor $\hat g\colon \bar F\to \D$ so that the following diagram commutes up to natural isomorphism.
\[
  \begin{tikzcd}
  U \bar F\C \arrow[r, "U\hat g", dashed] & U\D \\
  \C \arrow[ru, "g"'] \arrow[u,"\eta_{\C}"]                 &   
  \end{tikzcd}
\]
\end{UP_complete}
\noindent The free completion $\bar F\C$ is constructed as\sidenote{One may find difficulty in locating this explicit statement in the literature since it is a formal consequence of the free cocompletion construction. One place to start is the nLab's article on free completions \cite{nlab:free_completion}.} the opposite category of \textit{co}presheaves on $\C$, typically denoted $(\mathsf{Set}^{\C})^{\op}$. So $\bar F \C$ is the category whose objects are (small) functors $f\colon \C\to U\mathsf{Set}$, and where a morphism between copresheaves $f\to f'$ is a natural transformation going in the other direction $f'\Longrightarrow f.$  In particular, when $\C$ is a small category $\bar F\C$ coincides with \textit{all} copresheaves $\C\to U\mathsf{Set}$. In either case, the free completion of $\C$ has the property that 
\[ U\bar F \C = (U\mathsf{Set}^{\C})^{\op},\]
where the right-hand side is the category of small (and thus all when $\C$ is small) functors $\C\to U\mathsf{Set}$. The functor $\eta_{\C}\colon \C\to U\bar F\C$ in the universal property is the \textit{dual} Yoneda embedding that sends an object $c\in \C$ to its representable functor $\C(c,-)$, and the unique continuous extension $\hat g$ is defined dually to the unique cocontinuous extension in Equation (\ref{eq:CT_lift}).

%Use the fact that any copresheaf $W$ can be written as a limit of representable functors $W\cong\int_{c\in C} \C(c,-)^{Wc}$, where the notation $\C(c,-)^{Wc}$ is called a \textit{power} and denotes the $Wc$-indexed product of copies of $\C(c,-)$  \tai{(Fosco has them flipped; also weird that he takes copower as product! Maybe omit details? or put in margin since NO ONE ever says this.)}. Perhaps not surprisingly, then, the unique lift $\hat g$ of any functor $g\colon \C\to U\D$ is given by \[\hat gW=\int_{c\in \C} {gc}^{Wc}.\]

\newthought{There is quite} a bit of heavy machinery here, but we won't need to work with it explicitly. We mentioned it simply to call attention to the fact that \textit{there exists} well-developed theory that connects with familiar ideas in linear algebra and---as we'll see in the next section---formal concept analysis. En route to making this connection clearer, let us close with one final remark. Suppose $\C$ is a small category. We have the following two observations.
\begin{description}
\item \textit{The free cocomplete category $F\C$ is not only cocomplete, it is also complete.}

\item \textit{The free complete category $\bar F \C$ is not only complete, it is also cocomplete.}
\end{description}
The reason behind both statements is simple. The category of (contravariant) functors $\C^\to U\mathsf{Set}$ inherits (co)completion from $\mathsf{Set}$, which is both complete and cocomplete. The (co)limit of any diagram of functors $\C^{\op}\to\mathsf{Set}$ is obtained by evaluating them pointwise and then computing the (co)limit in $\mathsf{Set}$. This is analogous to how the set of functions from a set $X\to\mathbb{C}$ inherits a vector space structure from $\mathbb{C}$ by pointwise evaluation. These simple observations lead to our final analogy with linear algebra.\sidenote{The following discussion should be compared with the analogous discussion in Section \ref{sec:5_LA}.}

Let $\D$ be a small category. Any functor $M\colon C^{\op}\times \D\to U\mathsf{Set}$
gives rise to functors $A$ and $B$ on each factor, 
\[
  \begin{tikzcd}
  {\mathsf{Cat}(\C^{\op},U\mathsf{Set}^\D)} & & {\mathsf{Cat}(\C^{\op}\times \D,U\mathsf{Set})}\arrow[ll]\arrow[rr]                 & & {\mathsf{Cat}(\D,U\mathsf{Set}^{\C^{\op}})} \\[-20pt]
  A             &       & M \arrow[rr, maps to] \arrow[ll, maps to] &       & B            
  \end{tikzcd}
\]
Notice that functors $\C^{\op}\to U\mathsf{Set}^\D$ on the left-hand side are in bijection with functors $\C\to (U\mathsf{Set}^\D)^{\op}$. Under this correspondence, the functor $A\colon \C\to (U\mathsf{Set}^\D)^{\op}$ obtained from $M$ is defined on an object $c$ in $\C$ by $Ac(d)=M(c,d)$ for each object $d$ in $\D$. Likewise the functor $B\colon \D\to U\mathsf{Set}^{\C^{\op}}$ is defined on each $d$ by $Bd(c)=M(c,d)$. By the universal property for free cocompletions, $A$ lifts uniquely to a cocontinuous functor $M^*\colon F\C\to \bar F\D$, and by the universal property for free completions, $B$ lifts uniquely to a continuous functor $M_*\colon \bar F\D\to F\C$ so that the following diagrams commute,
\begin{equation}\label{diagrams:correctlifts_CT}
  \begin{tikzcd}
  U\mathsf{Set}^{\C^{\op}} \arrow[r, "UM^*", dashed]         & (U\mathsf{Set}^\D)^{\op} && (U\mathsf{Set}^\D)^{\op} \arrow[r, "UM_*", dashed]         & U\mathsf{Set}^{\C^{\op}} \\
  \C \arrow[ru, "A"'] \arrow[u, hook] &     && \D \arrow[u, hook] \arrow[ru, "B"'] &    
  \end{tikzcd}
\end{equation}
See the excellent exposition in \cite{willerton2013} for explicit expressions for $M^*$ and $M_*$ in the enriched setting, and compare this with the analogous discussion in Section \ref{sec:5_LA} leading up to the two diagrams in (\ref{diagrams:correctlifts_LA}). The climax here is that the functors $M^*$ and $M_*$ are adjoints. That is, there is a natural isomorphism $(\mathsf{Set}^\D)^{\op}(M^*V,W)\cong \mathsf{Set}^{\C^{\op}}(V,M_*W)$ for all presheaves $V$ and copresheaves $W$. Here we have lifted our self-imposed ban on writing $F\C=\mathsf{Set}^{\C^{\op}}$ to call attention to the ``op'' on the left-hand side of the isomorphism. A morphism $M^*V\to W$ in $(\mathsf{Set}^\D)^{\op}$ is a morphism $W\to M^*V$ in $\mathsf{Set}^\D$, and so we may rewrite the correspondence as
\[\mathsf{Set}^\D(W,M^*V)\cong \mathsf{Set}^{\C^{\op}}(V,M_*W).\]
The functors $M_*$ and $M^*$ are thus said to be \textbf{mutually right adjoints} that come with ``unit'' natural transformations  $\eta\colon \id \Longrightarrow M_*M^*$ and $\epsilon \colon \id \Longrightarrow M_*M^*$ \cite[Definition 4.3.1]{riehl2017category}, and their invariant subcategories, or ``fixed points,'' now arise as the topic of interest.\marginnote{Let's say an object $a$ in a category $\C$ is \textit{invariant} under a functor $F\colon\C\to\D$ if there is an isomorphism $a\overset{\cong}{\longrightarrow} Fa.$} Let $\text{Fix}(M_*M^*)$ be the category whose objects are objects $c$ in $\C$ such that the unit $\eta_c\colon c \overset{\cong}{\longrightarrow} M_*M^*c$ is an isomorphism. Likewise, let $\text{Fix}(M^*M_*)$ denote the category hose objects are objects $d$ in $\D$ so that $\epsilon_d\colon d \overset{\cong}{\longrightarrow} M^*M_*d$ is an isomorphism. There is an equivalence of categories\sidenote{This is a standard exercise. Every adjunction restricts to an equivalence of categories in this way.}
\[
\text{Fix}(M_*M^*)\cong
\text{Fix}(M^*M_*)
\] 
and their objects are called the \textbf{nucleus} of $M$. When $\D=\C$ and $M\colon \C^{\op}\times \C\to\mathsf{Set}$ is the hom-functor $\C(-,-)$, the adjunction $M^*\dashv M_*$ is called the \textbf{Isbell adjunction} or Isbell conjugation, and its nucleus is called the \emph{Isbell completion} of $\C$. If $\mathsf{Set}$ is further replaced by the category $2$ of truth values (defined in the next section) and the category $\C$ is a poset $P$, then this recovers the \emph{Dedekind MacNeille completion of $P$}. If $P=\mathbb{Q}$ with the usual ordering, then the Isbell completion gives the completion of $\mathbb{Q}$ by Dedekind cuts. If instead of replacing $\mathsf{Set}$ by $2,$ we were to replace it with $[0,\infty]$ with an appropriate categorical structure, this recovers the \emph{tight span} of a metric space. For more on these constructions see \cite{willerton2015legendrefenchel,pavlovic2012quantitative,isbell1966structure,Lawvere86takingcategories,Elliott2017OnTF,willerton2013tight}. By now we've covered a lot of information in a small amount of space. Here's the takeaway.

%and the category $\text{Fix}(M_*M^*)\cong \text{Fix}(M^*M_*)$ is called the \textbf{Isbell evelope} or \textbf{Isbell completion} of $\C$ \cite{Lawvere86takingcategories}. 

\begin{takeaway}\label{takeaway:5CTa}
Given small categories $\C$ and $\D$, any functor $M\colon \C^{\op}\times \D\to U\mathsf{Set}$ induces two functors $A\colon \C\to U\mathsf{Set}^{\C^{\op}}$ and $B\colon \D\to (U\mathsf{Set}^\D)^{\op}$ that lift to a cocontinuous functor $M^*$ and a continuous functor $M_*$ so that the following diagrams commute.
\begin{equation}\label{diagrams:presheaf_lifts}
  \begin{tikzcd}
  U\mathsf{Set}^{\C^{\op}} \arrow[r, "M^*", dashed]         & (U\mathsf{Set}^{\D})^{\op} && (U\mathsf{Set}^{\D})^{\op} \arrow[r, "M_*", dashed]         & U\mathsf{Set}^{\C^{\op}} \\
  \C \arrow[ru, "A"'] \arrow[u, hook] &     && \D \arrow[u, hook] \arrow[ru, "B"'] &    
  \end{tikzcd}
\end{equation}
Moreover $M^*$ and $M_*$ satisfy
\[\mathsf{Set}^\D(W,M^*V)\cong \mathsf{Set}^{\C^{\op}}(V,M_*W)\]
for all presheaves $V$ in $\mathsf{Set}^{\C^{\op}}$ and copresheaves $W$ in $\mathsf{Set}^{D}$, and the invariant subcategories of the compositions $M_*M^*$ and $M^*M_*$ are the nucleus of $M,$ and there is an equivalence between them.
\end{takeaway}
Take note of the unmistakable similarity with the linear algebra in Takeaway \ref{takeaway:LA}, shown in the margin for convenience.
  \marginnote[-3.5cm]{\textbf{Takeaway 2. }{\itshape  Given finite sets $X$ and $Y$, any function $M\colon X\times Y\to \mathbb{C}$ induces two functions $\alpha\colon X\to \mathbb{C}^Y$ and $\beta\colon Y\to\mathbb{C}^X$ that lift to linear maps $M$ and $M^\dagger$ so that the following diagrams commute.
  \[
    \begin{tikzcd}[ampersand replacement =\&]
    \mathbb{C}^X \arrow[r, "M", dashed]         \& \mathbb{C}^Y \&\& \mathbb{C}^Y \arrow[r, "M^\dagger", dashed]         \& \mathbb{C}^X \\
    X \arrow[ru, "\alpha"'] \arrow[u, hook] \&     \&\& Y \arrow[u, hook] \arrow[ru, "\beta"'] \&    
    \end{tikzcd}
  \]
  Moreover, $M$ and $M^\dagger$ satisfy 
  \[
  \langle Mv,w\rangle = \langle v,M^\dagger w\rangle
  \]
  for all $v\in\mathbb{C}^X$ and $w\in\mathbb{C}^Y$, and the one-dimensional invariant subspaces of the compositions $M^\dagger M$ and $MM^\dagger$ are their eigenvectors. Moreover, there is a one-to-one correspondence between them.}
  }
In one case, a function $M\colon X\times Y\to U\mathbb{C}$ defines a pair of linear adjoints $\adj{M}{\mathbb{C}^X}{\mathbb{C}^Y}{M^\dagger}$, and the eigenvectors of $M^\dagger M$ and $MM^\dagger$ we've discussed at length in Chapters \ref{ch:probability} and \ref{ch:application}. In another case, a functor $M\colon \C^{\op}\times \D\to U\mathsf{Set}$ defines a pair of categorical adjoints $\adj{M^*}{\mathsf{Set}^{\C^{\op}}}{(\mathsf{Set}^\D)^{\op}}{M_*}$, and the one-dimensional invariant subspaces of $M_*M^*$ and $M^*M_*$ are a known construction. Naturally, one is left wondering whether the dictionary between the category theory and linear algebra is more than a coincidence. Is it? Towards an answer, here's something to know. The category theory in this section may be promoted to an \textit{enriched} version. In other words, (co)presheaves and (co)completions have an analogous version when $\mathsf{Set}$ is replaced by any other symmetric monoidal closed category $\mathcal{V}$ that is both complete and cocomplete. This is a part of \emph{enriched category theory}, which stems from the idea that a set of morphisms in a category may actually be another object in that category.\sidenote{For example, the \textit{set} of linear operators is itself a \textit{vector space}, the \textit{set} of continuous functions is itself a \textit{topological space}, and so on.}

In view of this, it might seem desirable to ask that $\mathbb{C}$ be an enriching category, so that the linear algebra in Section \ref{sec:5_LA} is subsumed by the category theory in this section. But despite the many parallels, it is not known---to my knowledge---whether $\mathbb{C}$ may be viewed as a symmetric monoidal category which is also closed, complete, and cocomplete, so that one recovers the basic constructions in linear algebra as a special case. But perhaps one shouldn't expect this anyway. Colimits, for instance, are ``a remarkable enhancement of the concept of addition'' \cite{baez_2020}, and freely adding linear algebraic \textit{structure} is different than freely adding limits and colimits, which are closer to \textit{property-like structures} \cite[Section 4]{nlab:completion}. But all is not lost. Although linear algebra does not appear to be a special case of enriched category theory, it turns out that the formal concept construction discussed in Sections \ref{sec:ch1_FCA} and \ref{sec:evectsFCA} \textit{is}. This is explained in the next section.

Let us give a quick preview of the ideas. If $\mathsf{Set}$ is replaced by the poset (category) of \textit{truth values}, then free (co)completions reduce to familiar constructions in order theory. Briefly, any set $X$ can be viewed as a discrete category---the objects are the elements in $X$ and there are no nonidentity morphisms. The free cocompletion of $X$ coincides with the \emph{free join semilattice} on $X$, and the free completion of $X$ coincides with the \emph{free meet semilattice} on $X$. Moreover, there is an adjunction between these two semilattices, and the invariant subsets of the adjoint maps are formal concepts.  Because the enriched category theory reduces constructions that are both elementary and familiar, we will not introduce, use, or assume familiarity with the language of enriched category theory. Instead, we will only give a somewhat abridged, unenriched presentation with the simple goal of sharing the big-picture ideas. Readers are referred to \cite{lawvere1973,willerton2014,willerton2015legendrefenchel,Elliott2017OnTF,kelly1982basic} for more details. So while there is no claim that eigenvectors of reduced density operators are the fixed points of functors between (co)presheaf categories, it \textit{is} true of formal concepts. Understanding this will then shed light on the category theory that \textit{does} unite eigenvectors, generalized Isbell completions, and formal concepts.

\section{Revisiting Formal Concepts}\label{sec:5_FCA}
The present goal is specialize the category theory of Section \ref{sec:5_CT} with the aim of rederiving the definition of a formal concept given in Section \ref{sec:ch1_FCA}. Perhaps not obvious at first, this section is essentially a repeat of Section \ref{sec:5_CT} but with $\mathsf{Set}$ replaced by ``$2$.'' This will be explained soon. For now the starting point is, as before, an adjunction.

Let $\mathsf{Join}$ denote the category of join semilattices and join-preserving functions. A \emph{join semilattice} is a set $P$ with a partial order $\leq$ such that the join (the least upper bound) of any nonempty finite subset of $P$ exists. A function between join semilattices $f\colon P\to Q$ is said to \emph{preserve joins} if $f(p\vee p')=fp\vee fp'$ for all $p,p'\in P.$ Any join-preserving function automatically respects the partial order.\sidenote{Note that $p=p\vee p'$ if and only if $p'\leq p$, so if $f$ is join-preserving and if  $p'\leq p$, then $fp=f(p\vee p')=fp\vee fp'$, which implies $fp'\leq fp.$} This information can also be expressed in the language of category theory. Every poset $P$ is a category. Objects are the elements $p,p'$ of $P$, and a hom set $P(p,p')$ contains exactly one morphism $p\to p'$ if and only if $p\leq p'$. The colimit of any diagram in $P$ is the join of the elements, if it exists. So ``$P$ is a cocomplete category'' is synonymous with ``$P$ is a joint semilattice.'' An order-preserving function between posets is a functor, and a join-preserving function between join semilattices is a cocontinuous functor.\sidenote{There's a dual version to this. The \textit{limit} of any diagram in $P$ is the \textit{meet} (the greatest lower bound) of the elements, if it exists. So $P$ is a \emph{meet semilattice} precisely when it is a complete category. A meet-preserving function between meet semilattices is a continuous functor.} We won't need this perspective now, but being aware of it may be helpful in interpreting some of the results to come. We will, however, need the following fact. There is a free-forgetful adjunction 
\[\adj{U}{\mathsf{Join}}{\mathsf{Set}}{F}.\]
The functor $U\colon \mathsf{Join}\to\mathsf{Set}$ forgets partial orders and hence joins, and it associates to any join semilattice $P$ its underlying set $UP$. The functor $F\colon \mathsf{Set}\to\mathsf{Join}$ associates to a set $X$ the free join semilattice $FX$ on $X$. It is free in that the unit natural transformation $\eta\colon \id_{\mathsf{Set}}\Longrightarrow UF$ of this adjunction satisfies the following universal property.
\begin{UP_join}
The \textit{free join semilattice on a set $X$} is a join semilattice $FX$ with the property that for any join semilattice $P$ and for any function $f\colon X\to UP$ there is a unique join-preserving function $\hat f\colon FX\to P$ so that the following diagram commutes.
\[
  \begin{tikzcd}
  UFX \arrow[r, "U \hat f",dashed]                    & UP \\
  X \arrow[ru, "f"'] \arrow[u, "\eta_X"] &   
  \end{tikzcd}
\]
For any $A\in FX$ the unique lift in Equation (\ref{eq:CT_lift}) specializes in this context to give
\begin{equation}\label{eq:join_lift}
\hat f A:=\bigvee_{x\in A}fx.
\end{equation}
\end{UP_join}
Concretely, $FX$ is the poset consisting of all finite subsets of $X$ ordered by inclusion.\marginnote{As a technical aside, we are including the emptyset in $FX$, and so it may be referred to as the free \textit{bounded} join semilattice on $X$ since $\varnothing \leq A$ for every subset $A$ of $X$.}
%\tai{[\href{https://en.wikipedia.org/wiki/Semilattice\#Free_semilattices}{Wikipedia's} definition says have to take all NONEMPTY finit subsets. I want to include the empty set! E.g. two objects may have no attributes in common. So technically, I think I want $FX$ to be the free \href{https://orion.math.iastate.edu/jdhsmith/class/M567Defn.htm}{``bounded''} join semilattice on $X$, which means the partial order satisfies the additional property that there exists an object $\varnothing \in FX$ such that $\varnothing \leq A$ for all $A$. But I don't want to introduce more finicky language!]} 
The join of finite subsets $A,B\subseteq X$ is their union $A\vee B:=A\cup B.$ The function $\eta_X$ associates to an element $x\in X$ the singleton set $\{x\}$, and verifying that $\hat f$ preserves joins is a simple check. Elements in $FX$ may be identified with finitely-supported functions from $X$ to the two-element set $\{0,1\}$. When $X$ is finite this coincides with \textit{all} functions from $X$. Either way, it may be tempting to write ``$FX=2^X$,'' but let's resist the urge  for the same reason we resist writing $\mathbb{C}^X$ and $\mathsf{Set}^{\C^{\op}}$. That is, just as the complex numbers $\mathbb{C}$ and the category $\mathsf{Set}$ are special objects in $\mathsf{Vect}$ and $\mathsf{CAT}$, respectively, so ``$2$'' is a special object in $\mathsf{Join}$. We take a brief interlude to discuss this.

\subsection{Interlude: $2$, Categorically}\label{digression:2}
Start with the definition
\[2:=\{0\leq 1\}.\]
So $2$ is a poset. In fact, it is a join semilattice. For all elements $a$ and $b$ in $2$, their join $a\vee b$ exists:
\[0\vee 0 = 0 \qquad 0\vee 1 = 1\vee 0 =  1\vee 1 = 1.\]
So $2$ is an object in $\mathsf{Join}$, and its underlying set is $U2=\{0,1\}$. In fact, there's more: $2$ is also a \emph{meet semilattice}. For all elements $a$ and $b$ in $2$, their meet $a\wedge b$ exists:
\[0\wedge 0 = 0\wedge 1 = 1\wedge 0 = 0 \qquad 1\wedge 1 = 1.\]
So $2$ is also an object in $\mathsf{Meet}$, which is the category of meet semilattices and meet-preserving maps between them.\marginnote{That $2$ is both a join and a meet semilattice is completely analogous to the fact that $\mathsf{Set}$ is both cocomplete and complete.} Think of the elements of $2$ as representing \emph{truth values} with $0$ as false and $1$ as true, and think of meets as logical ``and'' and of joins as ``or.'' This pairs well with the intuition that a relation $R\subseteq X\times Y$ defines any function $R\colon X\times Y\to U2$, which provides a ``yes or no'' answer to the question, ``Are $x$ and $y$ related?'' As with any poset, $2$ is also a category. It has two objects with a single nonidentity morphism between them.
\[
  \begin{tikzcd}
  0 \arrow[r] \arrow[loop, distance=2em, in=125, out=55] & 1 \arrow[loop, distance=2em, in=55, out=125]
  \end{tikzcd}
\]
Think of the arrow $0\to 1$ as logical entailment. In this way, $2$ is quite rich in the eyes of category theory: it is symmetric monoidal (the monoidal product is ``and'') and it is closed\sidenote{Informally, a category is \emph{closed} if the morphisms between any two objects form another \textit{object} in that category. That object is called an \textit{internal hom.}} (the internal hom is implication). It is also both complete and cocomplete, for in any poset, the limit of any diagram is---if it exists---the meet of the elements in the diagram, and the colimit of any diagram is---if it exists---the join of the elements. So every meet (join) semilattice is a complete (cocomplete) category. In short, $2$ it is an excellent example of a category that one can ``do enriched category theory in.'' Though interesting, we omit the details. The enriched constructions ultimately recover familiar constructions in order theory, and the latter is sufficient for our goals. But, as a quick aside, here's the main idea connecting the two:
\begin{quote}
A subset $A\subseteq X$ is completely determined by its characteristic function $\bar A\colon X\to U2$ defined by $\bar Ax=1$ if and only if $x\in A$.
\end{quote}
In other words, knowing a subset $A$ is the same thing as knowing the truth value of the question ``Is $x\in A$?'' for every $x\in X$. Enriched category theory capitalizes on this. It's worth digressing a little more to understand just how. There is a precise sense in which every poset $P$ is enriched (a word that we still haven't, and won't, define) over the category $2$. Here's the basic idea. Rather than considering a hom-\textit{set} $P(p,p')$, which contains a single element $p\to p'$ if and only if $p\leq p'$, one instead considers an \textit{object} $2(p,p')$ in $2$, which is thought of as a hom-\textit{truth value}. The truth value answers the question ``Is $p\leq p'$?'' with $1$ when the answer is yes and with $0$ when the answer is no. As an example, if $P$ is the free join semilattice on a set $X$, then for any finite subsets $A,B\subseteq X$, the truth value between them indicates whether or not $B$ contains $A$.
\begin{equation}
  2(A,B)=
  \begin{cases}
  1 &\text{if $A\subseteq B$}\\
  0 &\text{otherwise}.
  \end{cases}
\end{equation}
Since $A$ and $B$ ultimately correspond to functions $\bar A,\bar B\colon X\to U2$, the hom-object $2(A,B)$ arises from pointwise evaluation.
\begin{equation}
  2(\bar A,\bar B)=
  \begin{cases}
  1 &\text{if $\bar Ax \leq \bar Bx$ for all $x\in X$}\\
  0 &\text{otherwise}
  \end{cases}
\end{equation}
So just as the dominance of $B$ in $A\subseteq B$ is indicated by a truth value $2(A,B)$, so the dominance of $Bx$ in $Ax\leq Bx$ is also indicated by a truth value, which we can denote by $[Ax,Bx]$. More precisely, $2$ can be equipped with a pairing $[-,-]\colon 2\times 2\to 2$ called the \emph{internal hom of 2}. Continuing to borrow the language of category theory, functions $\bar A,\bar B\colon X\to U2$ are examples of \emph{enriched presheaves},\footnote{More generally, a presheaf on any poset $P$ is simply an order-reversing function $P^{\op}\to 2$, where $P^{\op}$ is the poset obtained by reversing the order on $P$. A finite set $X$ is an example of a poset with $x\leq x'$ if and only if $x=x'$. In this way, $X^{\op}=X$ and thus a presheaf $X^{\op}\to 2$ is nothing more than a function $X\to U2$.} and the existence of a hom truth value $2(\bar A,\bar B)$ between them is the first step in discovering that $FX$ is a \emph{category enriched over $2$}.

After working through the details, it turns out that \textit{the free join semilattice $FX$ is quite literally the free cocompletion of $X$}, in a precise, enriched sense; that the join in Equation (\ref{eq:join_lift}) corresponds to a presheaf (i.e. a truth-valued function) analogous to Equation (\ref{eq:CT_lift}); and that the inclusion $\eta_X\colon X\to UFX$ sending an element to its singleton set $\{x\}$ is ultimately an enrichment of the Yoneda embedding. For more on enrichment over truth values, see Lawvere's classic \cite{lawvere1973} as well as the informal article \cite{willerton2014} and the PhD thesis \cite[Appendix A]{Elliott2017OnTF}. For more on enriched category theory in general, see \cite{riehl2014,kelly1982basic}. To summarize, the purpose of this digression is simply to call attention to analogous constructions
\[\mathbb{C}^X \qquad \mathsf{Set}^{\C^{\op}} \qquad \V^{\C^{\op}} \qquad 2^X\]
for a finite set $X$ and small category $\C$. Here $\V$ is a sufficiently nice category that can replace $\mathsf{Set}$ in the constructions in Section \ref{sec:5_CT}. Crucially, $2$ is an example of such a $\V$, and this chapter is essentially a repeat of Section \ref{sec:5_CT} in this truth-enriched setting.

\bigskip
\noindent\textbf{End of Interlude \ref{digression:2}.}
\bigskip

\noindent This is a good time to summarize the takeaway of the discussion so far. Compare the following with Takeaways \ref{takeaway:5LA} and \ref{takeaway:5CT}.

\begin{takeaway}\label{takeaway:5FCA} Given the free-forgetful adjunction $\adj{U}{\mathsf{Join}}{\mathsf{Set}}{F}$, the free join semilattice $FX$ on a set $X$ has the property that
  \begin{equation}\label{eq:EforFCA}
  UFX=U2^X,
  \end{equation}
where $U2^X$ denotes the set of finitely-supported functions (and hence all functions if $X$ is finite) $X\to U2$. To emphasize, the adjunction $F\dashv U$ has the property that there is a special object in $\mathsf{Join}$, namely 2, such that all join semilattices are, in some sense, built up from it.
\end{takeaway}
 
As expected, there is a dual discussion obtained by ``reversing all the arrows.'' In Section \ref{sec:5_CT} this led to the free completion of a category. It specializes in this context to the \textit{free meet semilattice of a set}.  Explicitly, there is a free-forgetful adjunction
\[\adj{U}{\mathsf{Meet}}{\mathsf{Set}}{\bar F}\]
where $\mathsf{Meet}$ is the category of meet semilattices and meet-preserving functions. It is free in that the unit $\eta\colon \id_{\mathsf{Set}}\Longrightarrow U\bar F$ of this adjunction satisfies the following universal property.
\begin{UP_meet}
The \textit{free meet semilattice on a set $X$} is a meet semilattice $\bar FX$ such that for any meet semilattice $P$ and for any function $g\colon X\to UP$ there is a unique meet-preserving function $\hat g\colon \bar FX\to P$ so that the following diagram commutes.
\[
  \begin{tikzcd}
  U\bar FX \arrow[r, "U\hat g"]                    & UP \\
  X \arrow[ru, "g"'] \arrow[u, "\eta_X"] &   
  \end{tikzcd}
\]
Explicitly, for any $B\in \bar FX$ the unique lift is given by
\begin{equation}\label{eq:meet_lift}
\hat gB:=\bigwedge_{x\in B}gx.
\end{equation}
\end{UP_meet}
\noindent The free meet semilattice on $X$ is constructed as the \textit{opposite} of its free join semilattice.
\begin{equation}\label{eq:join_op}
\bar F X=(FX)^{\op}
\end{equation}
In other words, $\bar FX$ consists of all 
%\tai{again, should be nonempty, but I don't want it to be.}
finite subsets of $X$ with the superset order, rather than inclusion. The meet of subsets $A,B\subseteq X$ in $\bar F$ is therefore their union $A\wedge B:=A\cup B.$ The elements of $\bar FX$ may be identified with the set $U2^Y$ of all finitely-supported functions $Y\to U2$; that is,
\[U\bar FY=U2^Y= UFY.\] 
The unit of the adjunction $\eta_X\colon X\to U\bar FX$ is again the inclusion of an element $x\in X$ to the singleton set $\{x\}$. A quick calculation verifies that $(FX)^{\op}$ does indeed satisfy the universal property and that the map given in Equation (\ref{eq:meet_lift}) preserves meets,\sidenote{Note that any meet-preserving map $(FY)^{\op}\to P$ is automatically order-\textit{reversing}. This follows from the fact that $A\cup B=B$ if and only if $A\subseteq B$ and that $p\wedge p'=p'$ if and only $p'\leq p$, for all $p,p'\in P$ and $A,B\in FY$.}  meaning that $\hat g(A\cup B)=\hat gA\cap \hat gB$.

\newthought{Let's take inventory} of the results so far. 
\begin{itemize}
  \item Subsets $X\to U2$ of a finite set $X$ are like functors $\C^{\op}\to \mathsf{Set}$ on a small category $\C$. In fact, one finds they are \textit{truth-enriched} functors.
  \item Free join (meet) semilattices are like free cocompletions (completions). In fact, they are \textit{truth-enriched} (co)completions.
  \item The inclusion $x\mapsto \{x\}$ is like the Yoneda embedding $c\mapsto \C(-,c)$. In fact, it arises from the \textit{truth-enriched} Yoneda embedding.
\end{itemize}
In light of this dictionary, the following may come as no surprise. Recall from Section \ref{sec:5_CT} that the free cocompletion of $\C$ was also \textit{complete}, and that its free completion was also \textit{cocomplete}. This led to the nice diagrams in Takeaway \ref{takeaway:5CTa}. The truth-enriched version has the same feature. 
\begin{description}
  \item \textit{The free join semilattice $FX$ on $X$ is also a meet semilattice.}
  \item \textit{The free meet semilattice $(FX)^{\op}$ on $X$ is also a join semilattice.}
\end{description}
The meet of subsets in $FX$ is their intersection, and the join of subsets in $(FX)^{\op}$ is also their intersection. These two facts go hand-in-hand in the following consideration.\marginnote{Compare the discussion here with the analogous discussions in Sections \ref{sec:5_LA} and \ref{sec:5_CT}.} Let $X$ and $Y$ be finite sets and consider any function $R\colon X\times Y\to U2$. The category of sets is Cartesian closed, and so we have a product-hom adjunction and hence the following bijections.
\[
  \begin{tikzcd}
  {\mathsf{Set}(X,U2^Y)} & \cong & {\mathsf{Set}(X\times Y,U2)}                 & \cong & {\mathsf{Set}(Y,U2^X)} \\[-20pt]
  a             &       & R \arrow[rr, maps to] \arrow[ll, maps to] &       & b            
  \end{tikzcd}
\]
Under this adjunction, the function $a\colon X\to U2^Y$ is defined on each $x\in X$ by
\begin{equation}\label{eq:ab2}
ax(y)=
\begin{cases}
1 &\text{if $R(x,y)=1$}\\
0 &\text{if $R(x,y)=0$},
\end{cases}
\end{equation}
and the function $b\colon Y\to U2^X$ is defined similarly.\marginnote{Notice the preimages $ax^{-1}(1)$ and $by^{-1}(1)$ recover the sets $\{y\in Y\mid R(x,y)=1\}$ and $\{x\in X\mid R(x,y)=1\}$ that appear in Equation (\ref{eq:ab}) in Chapter \ref{ch:introduction}.} By the universal properties for free join and meet semilattices, $a$ lifts uniquely to a join-preserving function $f\colon FX\to (FY)^{\op}$, while $b$ lifts uniquely to a meet-preserving function $g\colon (FY)^{\op}\to FX$ so that the following diagrams commute:
\begin{equation}\label{diagrams:FCA}
  \begin{tikzcd}
  U2^X \arrow[r, "Uf", dashed]         & U2^Y && U2^Y \arrow[r, "Ug", dashed]         & U2^X \\
  X \arrow[ru, "a"'] \arrow[u, hook] &     && Y \arrow[u, hook] \arrow[ru, "b"'] &    
  \end{tikzcd}
\end{equation}
Explicitly, for any $A\subseteq X$ and $B\subseteq Y$,
\[
  fA= \bigcap_{x\in A}ax 
  \qquad\text{and}\qquad
  gB=\bigcap_{y\in B} by
\]
which recovers Equations (\ref{eq:f}) and (\ref{eq:g}) in Chapter \ref{ch:introduction}. Compare this with the analogous discussion in Section \ref{sec:5_CT} leading up to the two diagrams in (\ref{diagrams:correctlifts_CT}). As was the case then, the climax is that $f$ and $g$ form a categorical adjunction---order-reversing functions are contravariant functors!---and so there is a bijection $(FY)^{\op}(fA,B)\cong FX(A,gB)$ for all $A\in FX$ and $B\in FY$. This is the same as having a bijection
\[FY(B,fA)\cong FX(A,gB)\]
which is equivalent to the original claim in (\ref{galois}) that
\[B \subseteq fA  \quad \text{if and only if} \quad A\subseteq gB.\]
A \textbf{formal concept} is defined to be a pair of sets $(A,B)$ for which equality holds. As noted in Chapter \ref{ch:introduction}, such pairs coincide with fixed points of the compositions $fg$ and $gf$ and, to borrow the language of Section \ref{sec:5_CT}, is the \textit{nucleus} of $R$.
%\tai{[It seems that $f$ and $g$ shouldn't be composable since they live in different categories: $f\in \mathsf{Join}$ and $g\in \mathsf{Meet}$. And yet there's no reason why we can't compose them. How to reconcile this?]}

\newthought{The past three}
  \marginnote{One notable difference between the linear algerba and the other constructions is that eigenvectors come with \textit{eigenvalues}. This is absent from the categorical version of the story, and so it could be interesting to generalize the Isbell completion yet further to allow for ``eigenvalues of profunctors.'' 
}
sections have featured a single construction that resurfaces in different contexts. We've briefly mentioned the \textit{Isbell completion} of a small category $\C$, which is the nucleus of the hom functor $\C(-,-)\colon \C^{\op}\times \C\to U\mathsf{Set}$. When $\mathsf{Set}$ is replaced by $2$ and when $\C$ is a poset $P$, then the nucleus of $P(-,-)\colon P^{\op}\times P\to U2$ is the \textit{Dedekind MacNeille completion of $P$}. The Isbell completion has a generalization when considering \textit{any} functor $M\colon \C^{\op}\times \D\to U\mathsf{Set}$ for any small category $\D$. When $\mathsf{Set}$ is replaced by $2$ and when $\C=X$ and $\D=Y$ are sets (viewed as discrete posets) then the nucleus of the relation $R:=M\colon X\times Y\to U2$ contains the formal concepts of $R$. By way of analogy, this thesis has replaced $2$ by the complex numbers $\mathbb{C}$ to study the ``nucleus'' of the matrix $M\colon X\times Y\to U\mathbb{C}$, in the special case when the values $|M(x,y)|^2$ define a probability distribution. In the next section, we close with a brief comment on the shared blueprint behind these constructions.
% \begin{fullwidth}
% \begin{center}
% % \begin{tikzpicture}[y=2cm,x=2cm]
% % \node at (0,0) {\includegraphics[width=\linewidth]{figures/ch5/chart}};

% % \node at (-.1,1.2) {\parbox[c]{3cm}{\centering\bfseries generalized\\ Isbell completion}};
% % 	\node at (0,.6) {\parbox[c]{4cm}{\centering nucleus of\\ $M\colon \C^{\op}\times \D\to U\mathsf{Set}$}};

% % \node at (-3.3,1.25) {\textbf{formal concepts}};
% % 	\node at (-3.2,.6) {\parbox[c]{4cm}{\centering nucleus of\\ $R\colon X\times Y\to U2$}};

% % \node at (3.2,1.25) {\textbf{Isbell completion}};
% % 	\node at (3.2,.6) {\parbox[c]{5cm}{\centering nucleus of\\ $\C(-,-)\colon \C^{\op}\times \C\to U\mathsf{Set}$}};

% % \node at (-2.35,-1) {\textbf{this thesis}};
% % 	\node at (-2.35,-1.8) {\parbox[c]{4cm}{\centering ``nucleus'' of\\ $M\colon X\times Y\to U\mathbb{C}$}};

% % \node at (3.1,-1.2) {\parbox[c]{4cm}{\bfseries \centering Dedekind-MacNeille\\ completion}};
% % 	\node at (3.2,-1.8) {\parbox[c]{4cm}{\centering nucleus of\\ $P(-,-)\colon P^{\op}\times P\to U2$}};
% % \end{tikzpicture}
% \end{center}
% \end{fullwidth}

\section{A Special Adjunction}\label{sec:5behind} 
Let's now zoom back a bit and summarize the overarching theme of Sections \ref{sec:5_LA}, \ref{sec:5_CT}, and \ref{sec:5_FCA}. Each begins with an adjunction $\adj{U}{\C}{\D}{F}$ where the category $\C$ contains a special object---call it $E$---with the property that $UFX$ ``looks like'' morphisms $X\to UE$ for every object $X$ in $\D$. Then for any objects $X$ and $Y$ in $\D$, a morphism $X\times Y\to UE$ induces two morphisms $X\to UFX$ and $Y\to UFY$ which have universal lifts $FX\to FY$ and $FY\to FX$ in $\C$, which are adjoints in some sense. In each case, ``fixed points'' of the composition of these adjoint maps are interesting.  More concretely,
%\tai{[Problem: there's only \textit{one} free-forgetful adjunction with vector spaces, but there are \textit{two} with joins/meets. Not sure how to get around this? \textbf{Maybe I could just say: modulo limits vs colimits? Or just outline this as a program for further investigation.}]}
\begin{itemize}
  \item When the adjunction is $\adj{U}{\mathsf{Vect}}{\mathsf{Set}}{F}$, the special object is $\mathbb{C}$. The universal lifts $FX\rightleftarrows FY$ are linear maps, and their one-dimensional invariant subspaces are \emph{eigenvectors}.

  \item When the biadjunction is $\adj{U}{\mathsf{(Co)compCAT}}{\mathsf{CAT}}{F}$, the special object is $\mathsf{Set}$. The universal lifts $F\C\rightleftarrows F\D$ are (co)continuous functors, and their invariant subcategories are \emph{nuclei}.

  \item When the adjunction is $\adj{U}{\mathsf{Join/Meet}}{\mathsf{Set}}{F}$, the special object is $2$. The universal lifts $FX\rightleftarrows (FY)^{\op}$ are join/meet-preserving  functions, and their invariant subsets are \emph{formal concepts}.
\end{itemize}
Modulo colimits versus limits, let us point out an interesting blueprint.

\subsection{A Blueprint}

Let $\D$ be a category with binary products, let $\C$ be any category, let $U\colon \C\to\D$ be a \textbf{faithful functor}
%\marginnote{\tai{define faithful + say forgetful functors are usually faithful!}} 
and consider an adjunction $\adj{U}{\C}{\D}{F}$ with the following property: There exists an object $E$ in $\C$ such that for each object $X$ in $\D$ there exists an object $UE^X$ in $\D$ together with an \textbf{evaluation morphism} \[\eval\colon UE^X\times X\to UE\] that is \textbf{exponential}; that is, for any object $Y\in\D$ and any morphism $M~\colon~Y~\times~X~\to~UE$ there is a unique morphism $m\colon Y\to UE^X$ so that the following diagram commutes
\[
  \begin{tikzcd}
  Y\times X \arrow[r, "M"] \arrow[d, "m\times \id_X"', dashed] & UE \\
  UE^X\times X \arrow[ru, "\eval"']                                     &   
  \end{tikzcd}
\]
In other words, there is a bijection
\begin{equation}\label{eq:exp}
\D(X\times Y,UE)\cong \D(Y,UE^X).
\end{equation}
Moreover assume that
\begin{equation}\label{eq:UE}
UFX\cong UE^X \qquad\text{for every $X\in \D$.}
\end{equation}
To understand this better, let's further suppose that $\D$ has a terminal object $I$.\marginnote{An object $I$ in a category $\D$ is called \emph{terminal} if for every object $d$ in $\D$ there is exactly one morphism $d\to I$.} Then under the bijection in Equation (\ref{eq:exp}),
\[\D(I,UE^X)\cong \D(X\times I,UE)\cong \D(X,UE)\]
Therefore $UE^X$ may be identified with morphisms $X\to UE$ in $\D$. The unit natural transformation $i\colon \id_\D \Longrightarrow UF$ of the adjunction $\adj{U}{\C}{\D}{F}$ satisfies a universal property.

\bigskip
\noindent\textbf{Universal Property of the Unit}. For any object $C$ in $\C$ and any morphism $f\colon X\to UC$ in $\D$, there is a unique morphism $\hat f\colon FX\to C$ in $\C$ so that the diagram commutes.
\[
  \begin{tikzcd}
  UFX \arrow[r, dashed, "U\hat f"]       & UC \\
  X \arrow[ru, "f"'] \arrow[u, "i"] &   
  \end{tikzcd}
\]

In the special case when $C=FY$ for some object $Y$ in $\D$, morphisms $X\to UFY$ in $\D$ correspond to morphisms $FX\to FY$ in $\C$,
\begin{equation}\label{eq:bij1}
\C(FX,FY)\cong \D(X,UE^Y).
\end{equation}
Here's the diagram:
  \marginnote{Compare this diagram with Diagrams \ref{diagrams:correctlifts_LA}, \ref{diagrams:presheaf_lifts}, and \ref{diagrams:FCA} where the special object $E$ is $\mathbb{C}, \mathsf{Set},$ and $2$, respectively.}
\[
  \begin{tikzcd}
  UE^X \arrow[r, dashed, "U\hat f"]       & UE^Y \\
  X \arrow[ru, "f"'] \arrow[u, "i"] &   
  \end{tikzcd}
\]
% \begin{example}
% Here are two examples.
%   \begin{enumerate}[i.]
%   \item (\textbf{Linear Algebra}) In the free-forgetful adjunction $\adj{U}{\mathsf{Vect}}{\mathsf{Set}}{F}$, functions $X\to UFY$ correspond to linear maps $FX\to FY$. In the diagram, one usually leaves off the $U$s and writes $\mathbb{C}^X$ for $FX$.
%   \[
%   \begin{tikzcd}
%   \mathbb{C}^X\arrow[r, dashed, "\tilde f"]       & \mathbb{C}^Y \\
%   X \arrow[ru, "f"'] \arrow[u, "i"] &   
%   \end{tikzcd}
% \]\\

%   \item (\textbf{Posets}) In the free-forgetful adjunction $\adj{U}{\mathsf{Meet}}{\mathsf{FinSet}}{2^{(-)}}$, functions $X\to U2^Y$ correspond to meet semilattice homomorphisms (i.e. meet-preserving, order-preserving functions) $2^X\to 2^Y$. Again, we usually leave off the $U$s in the diagram:
%     \[
%     \begin{tikzcd}
%     2^X\arrow[r, dashed, "\tilde f"]       & 2^Y \\
%     X \arrow[ru, "f"'] \arrow[u, "i"] &   
%     \end{tikzcd}
%   \]
% \end{enumerate}
% \end{example}
% To recap, the adjunction $F\dashv U$ gives the bijection
% \begin{equation}\label{eq:bij1}
% \hom_\C(FX,FY)\cong \hom_\D(X,UE^Y).
% \end{equation}
By exchanging $X$ and $Y$ there is also a bijection
\begin{equation}\label{eq:bij2}
\C(FY,FX)\cong \D(Y,UE^X)
\end{equation}
Now a priori, there's no relationship between $\C(FX,FY)$ and $\C(FY,FX)$.  But in fact they are isomorphic because $UE^X$ is exponential, as in Isomorphism (\ref{eq:exp}). Explicitly,
\begin{align*}
\C(FX,FY) &\overset{\text{by }(\ref{eq:bij1})}{\cong}  \D(X,UE^Y)\\[5pt]
 &\overset{\text{by }(\ref{eq:exp})}{\cong} \D(X\times Y,UE) \\[5pt] 
 &\overset{\text{by }(\ref{eq:exp})}{\cong} \D(Y,UE^X) \\[5pt]
 &\overset{\text{by } (\ref{eq:bij2})}{\cong} \C(FY,FX).
\end{align*}
As a result, morphisms $X\times Y\to UE$ in $\D$ correspond to morphisms $FX\rightleftarrows FY$~in~$\C$, and it is apparent that the ``fixed points'' of the composite morphisms carry interesting information encoded in  $X~\times~Y~\to~UE$.

\clearpage

%-------------------------------BIBLIOGRAPHY---------------------------------- %
%https://tex.stackexchange.com/questions/224871/tufte-book-with-standard-natbib-style-author-year-citations

\backmatter

\chapter*{Conclusion}
\addcontentsline{toc}{chapter}{Conclusion}

\newthought{In this thesis}, we set out to understand mathematical structure that has both algebraic and statistical properties, and basic tools from quantum probability theory have played a key role. In approaching the interface of algebra and statistics, we've seen that the eigenvalues and eigenvectors of reduced densities of a particular pure quantum state harness algebraic and statistical information in a highly principled way. In a further attempt to understand the marriage of these ideas, we think back to the density matrix model for entailment sketched in Section \ref{sec:entailment}. As it turns out, the way in which algebra and statistics meet in that setting is not incompatible with the categorical perspectives in Chapter \ref{ch:categorytheory}. This will be expanded on in the forthcoming paper \cite{taiyiannis}, which is part of an ongoing, promising investigation into this mathematical structure.
\clearpage
\thispagestyle{empty}

\bibliographystyle{alpha} 
\begin{fullwidth}
\addcontentsline{toc}{chapter}{Bibliography}
\bibliography{thesis}
\end{fullwidth}

%\printindex

\end{document}